\documentclass[11pt]{article}
\usepackage[utf8]{inputenc}
\usepackage{amsfonts,latexsym,amsthm,amssymb,amsmath,amscd,euscript,mathrsfs,graphicx}
\usepackage{framed}
\usepackage{fullpage}
\usepackage{color}
\usepackage{enumerate}
\usepackage{mathtools}
\usepackage[hidelinks]{hyperref}
\usepackage{thm-restate}

\usepackage{algorithm}
\usepackage[noend]{algpseudocode}
\usepackage{todonotes}

\usepackage{newpxmath}
\usepackage{newpxtext}
\usepackage{xspace}
\usepackage{cleveref}
\usepackage{paralist}

\newtheorem{theorem}{Theorem}[section]
\newtheorem{proposition}[theorem]{Proposition}
\newtheorem{lemma}[theorem]{Lemma}
\newtheorem{corollary}[theorem]{Corollary}

\newtheorem{claim}[theorem]{Claim}
\newtheorem{fact}[theorem]{Fact}
\theoremstyle{definition}
\newtheorem{definition}[theorem]{Definition}

\newtheorem{assumption}[theorem]{Assumption}

\theoremstyle{remark}
\newtheorem*{remark}{Remark}

\DeclareMathOperator{\E}{\mbox{\bf E}}

\newcommand{\pr}[2][]{\mathrm{Pr}\ifthenelse{\not\equal{}{#1}}{_{#1}}{}\!\left[#2\right]}
\newcommand{\R}{\mathbb{R}}
\newcommand{\N}{\mathbb{N}}

\newcommand{\dtv}{d_{\mathrm {TV}}}

\DeclareMathOperator{\vol}{vol}

\providecommand{\poly}{\operatorname*{poly}}

\newcommand{\mper}{\, .}
\newcommand{\mcom}{\, ,}

\newcommand{\proves}{\vdash}
\newcommand{\entails}{\vDash}

\newcommand{\paren}[1]{(#1)}
\newcommand{\Paren}[1]{\left(#1\right)}

\newcommand{\brac}[1]{[#1]}
\newcommand{\Brac}[1]{\left[#1\right]}

\newcommand{\abs}[1]{\lvert#1\rvert}
\newcommand{\Abs}[1]{\left\lvert#1\right\rvert}

\newcommand{\set}[1]{\{#1\}}
\newcommand{\Set}[1]{\left\{#1\right\}}

\newcommand{\norm}[1]{\lVert#1\rVert}
\newcommand{\Norm}[1]{\left\lVert#1\right\rVert}

\newcommand{\normt}[1]{\norm{#1}_2}

\newcommand{\normi}[1]{\norm{#1}_\infty}
\newcommand{\Normi}[1]{\Norm{#1}_\infty}

\newcommand{\iprod}[1]{\langle#1\rangle}

\newcommand{\normsch}[1]{\vert\kern-0.25ex\vert\kern-0.25ex\vert #1 
\vert\kern-0.25ex\vert\kern-0.25ex\vert}
\newcommand{\Normsch}[1]{\left\vert\kern-0.25ex\left\vert\kern-0.25ex\left\vert #1 
\right\vert\kern-0.25ex\right\vert\kern-0.25ex\right\vert}
\newcommand{\Psymb}{\mathbb{P}}

\DeclareMathOperator*{\ProbOp}{\Psymb}
\renewcommand{\Pr}{\ProbOp}

\newcommand{\suchthat}{\;\middle\vert\;}
\newcommand\bdot\bullet
\DeclareMathOperator{\Tr}{Tr}
\newcommand{\iid}{i.i.d.\xspace}
\newcommand{\cA}{\mathcal A}
\newcommand{\cB}{\mathcal B}

\newcommand{\cD}{\mathcal D}
\newcommand{\cE}{\mathcal E}
\newcommand{\cF}{\mathcal F}

\newcommand{\cL}{\mathcal L}

\newcommand{\cN}{\mathcal N}
\newcommand{\cO}{\mathcal O}
\newcommand{\cP}{\mathcal P}
\newcommand{\cQ}{\mathcal Q}
\newcommand{\cR}{\mathcal R}
\newcommand{\cS}{\mathcal S}
\newcommand{\cT}{\mathcal T}

\newcommand{\cV}{\mathcal V}
\newcommand{\cW}{\mathcal W}
\newcommand{\cX}{\mathcal{X}}
\newcommand{\cY}{\mathcal Y}

\newcommand*{\transpose}[1]{{#1}{}^{\mkern-1.5mu\mathsf{T}}}
\DeclareMathOperator{\Size}{Size}

\newcommand{\BI}{\mathbb I}

\newcommand{\BP}{\mathbb P}

\newcommand{\BR}{\mathbb R}
\newcommand{\BZ}{\mathbb Z}

\newcommand{\blue}{\textcolor{blue}}

\newcommand{\eps}{\varepsilon}
\newcommand{\e}{\eps}
\renewcommand{\epsilon}{\varepsilon}
\newcommand{\pE}{\tilde{\mathbf{E}}}

\newcommand{\muhat}{\hat{\mu}}
\newcommand{\tmu}{\widetilde{\mu}}
\newcommand{\tcO}{\widetilde{O}}

\newcommand{\tSigma}{\widetilde{\Sigma}}

\title{Robustness Implies Privacy in Statistical Estimation\footnote{Authors are listed in alphabetical order.}}
\author{Samuel B. Hopkins\thanks{\texttt{samhop@mit.edu}.} \\ MIT EECS \and Gautam Kamath\thanks{\texttt{g@csail.mit.edu}. Supported by an NSERC Discovery
Grant, an unrestricted gift from Google, and a University of Waterloo startup grant.} \\ University of Waterloo \and Mahbod Majid\thanks{\texttt{m2majid@uwaterloo.ca}. Supported by an NSERC Discovery Grant.} \\ University of Waterloo \and Shyam Narayanan\thanks{\texttt{shyamsn@mit.edu}. Supported by an NSF Graduate Fellowship, the NSF TRIPODS Program (award
DMS-2022448), and a Google Fellowship.} \\ MIT EECS}

\date{\today}

\begin{document}

\maketitle

\begin{abstract}
  We study the relationship between adversarial robustness and differential privacy in high-dimensional algorithmic statistics.
  We give the first \emph{black-box reduction from privacy to robustness} which can produce private estimators with optimal tradeoffs among sample complexity, accuracy, and privacy for a wide range of fundamental high-dimensional parameter estimation problems, including mean and covariance estimation.
  We show that this reduction can be implemented in polynomial time in some important special cases.
  In particular, using nearly-optimal polynomial-time robust estimators for the mean and covariance of high-dimensional Gaussians which are based on the Sum-of-Squares method, we design the first polynomial-time private estimators for these problems with nearly-optimal samples-accuracy-privacy tradeoffs.
  Our algorithms are also robust to a nearly optimal fraction of adversarially-corrupted samples.
\end{abstract}

\thispagestyle{empty}

\newpage
\thispagestyle{empty}
\tableofcontents
\thispagestyle{empty}

\newpage

\setcounter{page}{1}

\section{Introduction}

\emph{Parameter estimation} is a fundamental statistical task: given samples $X_1,\ldots,X_n$ from a distribution $p_{\theta}(X)$ belonging to a known family of distributions $\cP$ and indexed by a parameter vector $\theta \in \Theta \subseteq \R^D$, and for a given a norm $\| \cdot \|$, the goal is find $\hat{\theta}$ such that $\|\theta - \hat{\theta}\|$ is as small as possible.
Two important desiderata for parameter estimation algorithms are:

\medskip

\noindent \emph{Robustness:} If an $\eta$-fraction of $X_1,\ldots,X_n$ are adversarially corrupted, we would nonetheless like to estimate $\theta$.
  This \emph{strong contamination model} for robust parameter estimation dates from the 1960's, but has recently been under intense study from an algorithmic perspective, especially in the high-dimensional setting where $X_1,\ldots,X_n \in \R^d$ for large $d$.
  Thanks to these efforts, we now know efficient algorithms for a wide range of high-dimensional parameter estimation problems which enjoy optimal or nearly-optimal accuracy/sample complexity guarantees.

\medskip

\noindent \emph{Privacy:} A \emph{differentially private (DP)}~\cite{DworkMNS06} algorithm protects the privacy of individuals represented in a dataset $X_1,\ldots,X_n$ by guaranteeing that the distribution of outputs of the algorithm given $X_1,\ldots,X_n$ is statistically close to the distribution it would generate given $X_1',\ldots,X_n'$, where $X_1',\ldots,X_n'$ differs from $X_1,\ldots,X_n$ on any one sample $X_i$.

\medskip

Privacy and robustness are intuitively related: both place requirements on the behavior of an algorithm when one or several inputs are adversarially perturbed.
Already by 2009, Dwork and Lei recognized that ``robust statistical estimators present an excellent starting point for differentially private estimators''~\cite{DworkL09}.
More recent works continue to leverage ideas from robust estimation to design private estimation procedures \cite{BunKSW19, KamathSU20, BrownGSUZ21,RamsayC21, KothariMV22,LiuKO22,HopkinsKM22,GeorgievH22,RamsayJC22} -- these works address both sample complexity and computationally efficient algorithms.

Despite robustness being useful as a tool in privacy, the relationship between robustness and privacy remains murky.
Consequently, for many high-dimensional estimation tasks, we know polynomial-time algorithms which obtain (nearly) optimal tradeoffs among accuracy, sample complexity, and robustness, but known private algorithms either require exponential time or give suboptimal
%\todo{Shyam: far-from-optimal seems like a stretch (suboptimal certainly) but its only suboptimal in $\log 1/\delta$ parameters if you just care about privacy}
tradeoffs among accuracy, sample complexity, and privacy.
Indeed, this is the case even for \emph{learning the mean of a high-dimensional (sub-)Gaussian distribution,
%\todo{Shyam: I put sub in parantheses since we don't even know optimal for Gaussians}
and for learning a high-dimensional Gaussian in total variation distance.}

We contribute a new technique to design private estimators using robust ones, leading to:

\medskip

\noindent \emph{The first black-box reduction from private to robust estimation:}
    Prior works using robust estimators to design private ones are \emph{white box,} relying on properties of those estimators beyond robustness.
    Black-box privacy techniques such as the Gaussian and Laplace mechanisms are widely used, but so far do not yield private algorithms for high-dimensional estimation tasks with optimal accuracy-samples-privacy tradeoffs, even when applied to optimal robust estimators.
    \emph{For tasks including mean and covariance estimation and regression, using any robust estimator with an optimal accuracy-samples-robustness tradeoff, our reduction gives a private estimator with optimal accuracy-samples-privacy tradeoff.}
    
    Our basic black-box reduction yields estimators satisfying \emph{pure} DP, which work assuming $\Theta$ is bounded, and which don't necessarily admit efficient algorithms.
Two additional properties of an underlying robust estimator can lead to potential improvements in the resulting private estimator:

\begin{enumerate}
    \item If $\Theta$ is convex and the robust estimator is based on the \emph{Sum of Squares} (SoS) method, the resulting private estimator can often be implemented in polynomial time.
    \item If the robust estimator satisfies a stronger \emph{worst-case} robustness property, satisfied by many high-dimensional robust estimators, we can remove the assumption that $\Theta$ is bounded, at the additional (necessary) expense of weakening from pure to \emph{approximate} DP guarantees.
\end{enumerate}

\medskip

\noindent \emph{The first polynomial-time algorithms to learn high-dimensional Gaussian distributions with nearly-optimal sample complexity subject to differential privacy:}
Using SoS-based robust algorithms and our privacy-to-robustness reduction, we obtain polynomial-time estimators with nearly-optimal accuracy-samples-privacy tradeoffs, for both pure and approximate DP, for learning the mean and/or covariance of a high-dimensional Gaussian, and for learning a high-dimensional Gaussian in total variation.
In addition, our private algorithms enjoy near-optimal levels of robustness.
Prior private polynomial-time estimators have sub-optimal samples-accuracy-privacy tradeoffs, losing polynomial factors in the dimension $d$ and/or privacy parameter $\log 1/\delta$.

Our methods also yield a polynomial-time algorithm for private mean estimation under a bounded-covariance assumption, recovering the main result of \cite{HopkinsKM22} with slightly improved sample complexity.
We expect them to generalize to other estimation problems where $\Theta$ is convex and nearly-optimal robust SoS algorithms are known -- e.g., linear regression \cite{KKM18} and mean estimation under other bounded-moment assumptions \cite{HL18,KSS18}.

\medskip
    
\noindent \emph{Conclusions on Robust versus Private Estimation: }
Recent work \cite{GeorgievH22} shows that private algorithms \emph{with very high success probabilities} are robust simply by virtue of their privacy guarantees.
%\textcolor{red}{in particular, estimators with optimal privacy-accuracy-samples tradeoffs for problems like mean and covariance estimation automatically have optimal samples-accuracy-robustness tradeoffs as well}\todo{Shyam: I think we can remove this - in fact i think you need to also say optimal failure probability tradeoffs too}.
This complements our results, which show a converse -- from robust estimators with optimal samples-accuracy-robustness tradeoffs we get analogous private estimators (with very high success probabilities).
Together, these hint at a potential \emph{equivalence} between robust and private parameter estimation, which can be made algorithmic in the context of SoS-based algorithms.
Our results show such an equivalence for ``nice enough'' parameter estimation problems, but the broader relationship between privacy and robustness is more subtle; in Section~\ref{sec:techniques} we discuss situations where optimal robust estimators don't necessarily yield optimal private ones, at least in a black-box way.

%In a nutshell, this is because the best possible estimation under $\eta$-corruption for a given parameter estimation problem depends mainly on the \emph{maximal distance} $\|\theta - \theta'\|$ between parameters with corresponding distributions $p_{\theta},p_{\theta'}$ having $d_{TV}(p_{\theta},p_{\theta'}) \leq \eta$, while the best possible privacy depends on the \emph{volume} of $\theta'$ corresponding to some $p_{\theta'}$ with $d_{TV}(p_{\theta},p_{\theta'}) \leq \eta$ for a fixed $\theta$ and a range of $\eta$.

\subsection{Results}

We first recall the definitions of differential privacy and the strong contamination model.

\begin{definition}[Differential Privacy (DP)~\cite{DworkMNS06,DworkKMMN06}]
    Let $\cX$ be a set of \emph{inputs} and $\cX^*$ be all finite-length strings of inputs.
    Let $\cO$ be a set of \emph{outputs}.
    A randomized map (``mechanism'') $M \, : \, \cX^* \rightarrow \cO$ satisfies $(\e,\delta)$-DP if for every \emph{neighboring} $X,X' \in \cX^*$ with Hamming distance $1$ and every subset $S \subseteq \cO$, $\Pr(M(X) \in S) \leq e^{\eps} \Pr(M(X') \in S) + \delta$.
    If $\delta = 0$, we say that $M$ satisfies \emph{pure} DP, otherwise $M$ satisfies \emph{approximate} DP.
    %\todo{Shyam: this sounds really abstract} great i'm glad to see the graph stuff not there
\end{definition}

\begin{definition}[Strong Contamination Model]
    For a probability distribution $D$ and $\eta > 0$, $Y_1,\ldots,Y_n$ are $\eta$-corrupted samples from $D$ if $X_1,\ldots,X_n \overset{i.i.d.}{\sim} D$ and $Y_i = X_i$ for at least $(1-\eta)n$ indices $i$.
\end{definition}

\subsubsection{Learning High-Dimensional Gaussian Distributions in TV Distance}

We begin with our results on learning Gaussians in total variation distance.

\begin{theorem}[Learning Arbitrary Gaussians, Pure DP, Subsection \ref{subsec:main_stuff}]
    \label{thm:pure-dp-tv}
    Assume that $0 < \alpha, \beta, \eps < 1$, $0 < \eta < \eta^*$ for some absolute constant $\eta^*$, and $K, R > 1$.
    There is a polynomial-time $(\e,0)$-DP algorithm with the following guarantees for every $d \in \N$ and every $\mu \in \R^d, \Sigma \in \R^{d \times d}$ such that $\|\mu\| \leq R$ and $\frac{1}{K} \cdot I \preceq \Sigma \preceq K \cdot I$.
    %\blue{There exists a universal constant $\eta_0$, such that given non-negative $\alpha,\beta < 1$, $\eta < \eta_0$, and $R, K$}, and 
    Given $n$ $\eta$-corrupted samples from $\cN(\mu,\Sigma)$, the algorithm returns $\hat{\mu},\hat{\Sigma}$ such that $d_{TV}(\cN(\mu,\Sigma),\cN(\hat{\mu},\hat{\Sigma})) \leq \alpha + \tcO(\eta)$ with probability at least $1-\beta$, if\footnote{With more careful analysis, we expect that the error bound can be tightened to $\alpha + O(\eta \log 1/\eta)$, which is expected to be tight for statistical query algorithms \cite{DiakonikolasKS17}; the same goes for our other results on learning Gaussians.}
    \[
        n \geq \tcO\left(\frac{d^2+\log^2(1/\beta)}{\alpha^2} + \frac{d^2 + \log(1/\beta)}{\alpha \eps} + \frac{d^2 \log K}{\eps} + \frac{d \log R}{\eps}\right).
    \]
\end{theorem}

We are unaware of prior computationally efficient pure-DP algorithms for learning high-dimensional Gaussians in TV distance; we believe that state of the art is based on the techniques of \cite{KamathLSU19},\footnote{replacing the Gaussian mechanism with the Laplace mechanism} which would give an algorithm requiring $n \gg d^3$ samples (and lack robustness).

%bThe privacy guarantees of our algorithm are sufficiently strong that the robustness to an $\eta$-fraction of corrupted samples described in the theorem statement follows from the privacy guarantees alone, by observing that the failure probability $\beta$ of the algorithm can be taken as small as $\exp(- \Omega(\e \alpha n))$ \cite{GH22}.\Snoteinline{confirm that this holds when we have a final theorem statement.}

Pure-DP necessitates the \emph{a priori} upper bounds $R$ and $K$ on $\mu$ and $\Sigma$ in Theorem~\ref{thm:pure-dp-tv}.
Under $(\e,\delta)$-DP these bounds are avoidable.
But, obtaining a polynomial-time $(\e,\delta)$-DP algorithm to learn Gaussians with optimal samples-accuracy-privacy tradeoffs and without assumptions on $\mu,\Sigma$ has been a significant challenge, with progress in several recent works \cite{AshtianiL22, KamathMSSU22, KothariMV22,TsfadiaCKMS22} (see Table~\ref{tab:covariance-estimation}).
These algorithms require a number of samples exceeding the information-theoretic optimum by polynomial factors in either $d$, $\log(1/\delta)$, or both.

We give the first polynomial-time $(\e,\delta)$-DP algorithm for learning an arbitrary high-dimensional Gaussian distribution with nearly-optimal sample complexity with respect to \emph{all} of: dimension, accuracy, privacy, and corruption rate.
Ours is the first $\tilde{O}(d^2)$-sample polynomial-time robust and private estimator; prior works require $\Omega(d^{3.5})$ samples \cite{AshtianiL22, TsfadiaCKMS22}.
%Unlike in the pure-DP case, there is need for an \emph{a priori} upper bound on $\|\mu\|$ or $\|\Sigma\|$.\Gnote{Do we want to call this an a priori bound? Or can we just call this a dependence? In the sense that, do we need to know this ahead of time, or can the algorithm figure it out as it goes.}
%\Gnote{These necessities are from the non-private setting, right? If so, make it explicit.}

\begin{theorem}[Learning Arbitrary Gaussians, $(\e,\delta)$-DP, Subsection \ref{subsec:main_stuff}]
  \label{thm:approx-dp-tv}
  Assume that $0 < \alpha, \beta, \delta, \eps < 1$, and $0 < \eta < \eta^*$ for some absolute constant $\eta^*$.
  %For every $\e,\delta > 0$ t
  There is a polynomial-time $(\e,\delta)$-DP algorithm with the following guarantees for every $d \in \N$, $\mu \in \R^d$, and $\Sigma \in \R^{d \times d}$, $\Sigma \succ 0$.\footnote{We suppress running-time dependence on $\log K$, where $K$ is the condition number of $\Sigma$; logarithmic dependence on the condition number orthogonal to $\ker(\Sigma)$ is necessary for learning Gaussians in TV, regardless of privacy or robustness. Note that the sample complexity has no such dependence on $\log K$.}
 %\blue{There exists a universal constant $\eta_0$, such that given non-negative $\alpha,\beta < 1$ and $\eta < \eta_0$}, 
 Given $n$ $\eta$-corrupted samples from $\cN(\mu,\Sigma)$, the algorithm returns $\hat{\mu},\hat{\Sigma}$ such that $d_{TV}(\cN(\mu,\Sigma),\cN(\hat{\mu},\hat{\Sigma})) \leq \alpha + \tcO(\eta)$ with probability at least $1-\beta$, if
  \[
      n \geq \tcO\left(\frac{d^2+\log^2(1/\beta)}{\alpha^2} + \frac{d^2 + \log(1/\beta)}{\alpha \eps} + \frac{\log (1/\delta)}{\eps}\right).
  \]
\end{theorem}

%\begin{remark}[Guarantees for Parameter Estimation Beyond Gaussians]
%Both Theorems~\ref{thm:pure-dp-tv} and~\ref{thm:approx-dp-tv} proceed by estimating the parameters $\mu,\Sigma$ to small distance in appropriate norms -- roughly, learning $\hat{\Sigma}$ such that $\|I - \Sigma^{-1/2} \hat{\Sigma} \Sigma^{-1/2}\|_F \leq \alpha$ and $\hat{\mu}$ such that $\|\Sigma^{-1/2}(\mu - \hat{\mu})\| \leq \alpha$.
%(Here, $\|\cdot \|_F$ is Frobenius norm.)
%Standard arguments imply that such $\hat{\mu}, \hat{\Sigma}$ have the TV-distance property above.

%However, if we focus on parameter estimation, our algorithms extend beyond the Gaussian setting: we obtain these guarantees on $\hat{\mu},\hat{\Sigma}$ given samples from \emph{any} high-dimensional distribution with \emph{certifiably-hypercontractive $4$th moments}, an increasingly-common assumption which enables polynomial-time algorithms for a number of high-dimensional learning tasks.
%(Departing from the Gaussian setting entails slightly weakening the robustness guarantees of our algorithms; some weakening compared to the Gaussian setting is information-theoretically necessary.)
%See Theorem~\ref{thm:??}.
%\Snote{TODO: do we prove such a theorem?}
%\end{remark}

The sample-complexity guarantees of Theorems~\ref{thm:pure-dp-tv} and~\ref{thm:approx-dp-tv} are information-theoretically tight up to logarithmic factors in $d, \alpha, \e$, and $\log 1/\delta$.
The $\log(1/\beta) / \alpha \epsilon$ term in each is potentially improvable to $\min (\log(1/\beta), \log (1/\delta)) / \alpha \epsilon$, and the $\log^2(1/\beta)$ term is potentially improvable to $\log(1/\beta)$.
However, this still means our algorithms succeed with exponentially small ($e^{-d}$) failure probability, with no blowup in the sample complexity.

\subsubsection{Estimating the Mean of a Subgaussian Distribution}

Mean estimation in high dimensions subject to differential privacy has also received substantial recent attention \cite{KarwaV18,KamathLSU19,BunS19, BunKSW19, KamathSU20,LiuKKO21,BrownGSUZ21,LiuKO22,HopkinsKM22}.
We focus on the following simple problem: given (corrupted) samples from $\cN(\mu,I)$, find $\hat{\mu}$ such that $\|\mu - \hat{\mu}\| \leq \alpha$.
In the pure-DP setting, exponential-time estimators are known which achieve this guarantee using $n \approx \tfrac d {\alpha^2} + \tfrac{d} {\alpha \e}$ samples~\cite{BunKSW19,KamathSU20}.
Existing polynomial-time estimators require $n \gg \min(\tfrac d {\alpha^2 \e}, \tfrac{d^{1.5}}{\e})$ samples or satisfy a weaker privacy guarantee~\cite{KamathLSU19,HopkinsKM22} (see Table~\ref{tab:mean-estimation}).
We give the first nearly-sample-optimal pure-DP algorithm:

\begin{theorem}[Estimating the Mean of a Spherical Subgaussian Distribution, Theorem~\ref{thm:gaussian-mean-main}]
    \label{thm:pure-dp-mean}
  Assume that $0 < \alpha, \beta, \eps < 1$, $0 < \eta < \eta^*$ for some absolute constant $\eta^*$, and $R > 1$.
  %For every $\e > 0$ t
  There is a polynomial-time $(\e,0)$-DP algorithm with the following guarantees for every $d \in \N$, every $\mu \in \R^d$ with $\|\mu\| \leq R$, and every subgaussian distribution $D$ on $\R^d$ with mean $\mu$ and covariance $I$.
  %\blue{
  %There exists a universal constant $\eta_0$, such that given non-negative $\alpha,\beta < 1$, $\eta < \eta_0$, and $R$,
  %}
  Given $n$ $\eta$-corrupted samples from $D$, the algorithm returns $\hat{\mu}$ such that $\|\mu - \hat{\mu}\| \leq \alpha + \tcO(\eta)$ with probability at least $1-\beta$, as long as
  \[
      n \geq \tcO\left(\frac{d+\log(1/\beta)}{\alpha^2} + \frac{d + \log(1/\beta)}{\alpha \eps} + \frac{d \log R}{\eps}\right).
  \]
\end{theorem}

It is natural to ask whether the identity-covariance assumption can be removed from Theorem~\ref{thm:pure-dp-mean}, since information-theoretically the assumption of covariance $\Sigma \preceq I$ is enough to obtain the same guarantees.
%\todo{Possibly need to reword in light of new covariance-aware papers, Gautam can you check?}
Removing this assumption while retaining polynomial running time and high-probability privacy guarantees would improve over state-of-the-art algorithms for robust mean estimation which have withstood significant efforts at improvement \cite{HL19}.

There is also an analogue (\Cref{cor:gaussian-mean-main}) for polynomial-time mean estimation subject to $(\e,\delta)$-DP without the $\|\mu\| \leq R$ assumption, using $\tilde{O}(\tfrac d {\alpha \e} + \tfrac d {\alpha^2} + \tfrac{\log 1/\delta} \e)$ samples.
We obtain this result from our approx-DP framework similar to proving \Cref{thm:approx-dp-tv}: one could alternatively combine Theorem~\ref{thm:pure-dp-mean} with an $(\e,\delta)$-DP procedure that obtains an $O(d)$-accurate estimate, such as \cite{EsfandiariMN22}.
%\Snote{Shyam, am I citing the right thing here, also see table 1?}
%yeah it's right, though technically it's only d-accurate. This is because if all the points are in a ball of radius 1 it gives error sqrt{d}, and gaussians are only in a ball of radius sqrt{d}.

\medskip
Finally, we note that Theorems \ref{thm:pure-dp-tv} and \ref{thm:pure-dp-mean} are known to be near-optimal from standard packing lower bounds~\cite{BunKSW19}, and Theorem~\ref{thm:approx-dp-tv} and \Cref{cor:gaussian-mean-main} are also known to be near-optimal, via the technique of fingerprinting~\cite{KamathLSU19, KamathMS22}, except, as in Theorems~\ref{thm:pure-dp-tv} and~\ref{thm:approx-dp-tv}, that $\log(1/\beta)/\alpha \e$ is potentially improvable to $\min(\log(1/\beta), \log(1/\delta)) / \alpha \e$.
All our algorithmic results are applications of Theorems~\ref{thm:pure_dp_general_main},~\ref{thm:approx_dp_general_main}, which give general tools for turning SoS-based robust estimators into private ones.

\begin{table}[h]
\begin{center}
\begin{tabular}{| c | c | c | c | c |} 
 \hline
 Paper & Sample Complexity & Robust? & Poly-time? & Privacy \\ [0.5ex] 
 \hline\hline
 \cite{KarwaV18} & $\frac{1}{\alpha^2} + \frac{1}{\alpha \eps} + \frac{\min(\log K, \log \delta^{-1})}{\eps}$, $d = 1$ & No & Yes & Pure/Approximate  \\ \hline
 \cite{KamathLSU19} & $\frac{d^2}{\alpha^2} + \frac{d^2 \sqrt{\log \delta^{-1}}}{\alpha \eps} + \frac{d^{3/2} \sqrt{\log K \log \delta^{-1}}}{\eps}$ & No & Yes & Concentrated \\ \hline
  \cite{BunKSW19} & $\frac{d^2}{\alpha^2}+\frac{d^2 \log K}{\alpha \eps}$ & Optimal & No & Pure \\ \hline
  \cite{AdenAliAK21} & $\frac{d^2}{\alpha^2} + \frac{d^2}{\alpha \eps} + \frac{\log \delta^{-1}}{\eps}$ & Optimal & No  & Approximate \\ \hline
  \cite{LiuKO22} & $\frac{d^2}{\alpha^2} + \frac{d^2}{\alpha \eps} + \frac{\log \delta^{-1}}{\alpha \eps}$ & Optimal & No & Approximate \\ \hline
  \cite{KamathMSSU22} & $\frac{d^2}{\alpha^2} + \left(\frac{d^2}{\alpha \eps} + \frac{d^{5/2}}{\eps}\right) \cdot  (\log \delta^{-1})^{O(1)}$ & No & Yes & Approximate \\ \hline
  \cite{KothariMV22} & $\frac{d^8}{\alpha^4} \cdot \left(\frac{\log \delta^{-1}}{\eps}\right)^6$ & Suboptimal & Yes & Approximate \\ \hline
  \cite{AshtianiL22, TsfadiaCKMS22} & $\frac{d^2}{\alpha^2} + \frac{d^2 \sqrt{\log \delta^{-1}}}{\alpha \eps} + \frac{d \log \delta^{-1}}{\eps}$ & No & Yes & Approximate \\ \hline
  \cite{AshtianiL22, TsfadiaCKMS22} & $\frac{d^{3.5}\log \delta^{-1}}{\alpha^3 \eps}$ & Optimal & Yes & Approximate \\ \hline
  Thm \ref{thm:pure-dp-tv} & $\frac{d^2}{\alpha^2} + \frac{d^2}{\alpha \eps} + \frac{d^2 \log K}{\eps}$ & Optimal & Yes & Pure \\ \hline
  Thm \ref{thm:approx-dp-tv} & $\frac{d^2}{\alpha^2} + \frac{d^2}{\alpha \eps} + \frac{\log \delta^{-1}}{\eps}$ & Optimal & Yes & Approximate \\ \hline
\end{tabular}
\caption{\label{tab:covariance-estimation} Private covariance estimation of Gaussians in Mahalanobis distance,
%Along with sample complexity, we include information about whether the algorithm is robust to corruptions, running time, and the type of privacy guarantee. Sample complexities omit
%sorry i had no idea what elide meant, sounds like too hard of a word
omitting
logarithmic factors. Optimal robustness means the algorithm succeeds even with $\tilde{\Omega}(\alpha)$-fraction of corruptions.}
\end{center}
\vspace{-1.2em}
\end{table}

\begin{table}[h]
\begin{center}
\begin{tabular}{| c | c | c | c | c |} 
 \hline
 Paper & Sample Complexity & Robust? & Poly-time? & Privacy \\ [0.5ex] 
 \hline\hline
  \cite{KarwaV18} & $\frac{1}{\alpha^2}+\frac{1}{\alpha \eps} + \frac{\min(\log R, \log \delta^{-1})}{\eps}$, $d=1$ & No & Yes & Pure/Approximate\\ \hline
  \cite{KamathLSU19} & $\frac{d}{\alpha^2}+\frac{d \sqrt{\log \delta^{-1}}}{\alpha \eps} + \frac{\sqrt{d \log R \log \delta^{-1}}}{\eps}$ & No & Yes & Concentrated \\ \hline
  \cite{BunKSW19} & $\frac{d}{\alpha^2}+\frac{d \log R}{\alpha \eps}$ & Optimal
  %\footnote{While the amount of robustness achieved is optimal, both \cite{BunKSW19} and \cite{AdenAliAK21} are proven to only be robust against a weaker model of receiving i.i.d. samples from a distribution close in TV distance to $\mathcal{N}(\mu, I)$, as opposed to against strong contamination.}
  & No & Pure \\ \hline
  \cite{KamathSU20} & $\frac{d}{\alpha^2}+\frac{d}{\alpha \eps} + \frac{d \log R}{\eps}$ & Optimal & No & Pure \\ \hline
  \cite{AdenAliAK21} & $\frac{d}{\alpha^2} + \frac{d}{\alpha \eps} + \frac{\log \delta^{-1}}{\eps}$ & Optimal & No & Approximate \\ \hline
  \cite{LiuKKO21} & $\frac{d}{\alpha^2} + \frac{d^{3/2} \log \delta^{-1}}{\alpha \eps}$ & Optimal & Yes & Approximate \\ \hline
  %\cite{LiuKKO21} & $\frac{d}{\alpha^2} + \frac{d}{\alpha \eps}+\frac{d^{1/2} \log \delta^{-1}}{\alpha \eps}$ & Optimal & No & Approximate \\ \hline
  \cite{BunKSW19, LiuKO22} & $\frac{d}{\alpha^2} + \frac{d}{\alpha \eps}+\frac{\log \delta^{-1}}{\alpha \eps}$ & Optimal & No & Approximate \\ \hline
  \cite{HopkinsKM22} & $\frac{d}{\alpha^2 \eps} + \frac{d \log R}{\eps}$ & Suboptimal & Yes & Pure \\ \hline
  \Cref{thm:pure-dp-mean} & $\frac{d}{\alpha^2} + \frac{d}{\alpha \eps} + \frac{d \log R}{\eps}$ & Optimal & Yes & Pure \\ \hline
  \Cref{cor:gaussian-mean-main} & $\frac{d}{\alpha^2} + \frac{d}{\alpha \eps} + \frac{\log \delta^{-1}}{\eps}$ & Optimal & Yes & Approximate \\ \hline
\end{tabular}
\caption{\label{tab:mean-estimation} Private mean estimation of identity-covariance Gaussians in $\ell_2$-norm, omitting
%Along with sample complexity, we include information about whether the algorithm is robust to corruptions in the data, whether the algorithm runs in polynomial time, and the type of privacy the algorithm guarantees. Sample complexities omit
%sorry i had no idea what elide meant, sounds like too hard of a word
logarithmic factors. Optimal robustness means the algorithm succeeds even with $\tilde{\Omega}(\alpha)$ fraction of corruptions.}
\end{center}
\end{table}

\subsection{Related Work}

Our work joins three bodies of literature too large to survey here: on private and high-dimensional parameter estimation, on high-dimensional statistics via SoS (see \cite{RSS18}), and on high-dimensional algorithmic robust statistics (see \cite{DiakonikolasK22}).
We discuss other works at the intersections of these areas.

\medskip

\noindent \emph{Private and Robust Estimators:}
\cite{DworkL09} first used robust statistics primitives to design private algorithms, a tradition continued by \cite{BunKSW19, KamathSU20, LiuKO22,BrownGSUZ21,RamsayC21,KothariMV22,HopkinsKM22}.
Some of these works do give general-purpose blueprints for converting robust algorithms to private ones~\cite{LiuKO22,KothariMV22}, though none give a black-box reduction as we do in Lemmas~\ref{lem:intro-black-box} and \ref{lem:approx-dp-intro}.
Other works from the Statistics community also investigate connections between robustness and privacy~\cite{AvellaMedina20,AvellaMedina21,RamsayJC22,SlavkovicM22}, including local differential privacy~\cite{LiBY22}.
Our black-box reduction from privacy to robustness can be seen as a generalization of methods of~\cite{BunKSW19,KamathSU20}, which also instantiate the exponential mechanism with a score function counting the minimum point changes to achieve some accuracy guarantee, but for specific robust estimators.
A recent line of work focuses on simultaneously private and robust estimators for high-dimensional statistics \cite{BunKSW19, GhaziKMN21,LiuKKO21,EsfandiariMN22,AshtianiL22,KothariMV22,TsfadiaCKMS22,LiuKO22}; see Tables~\ref{tab:covariance-estimation},~\ref{tab:mean-estimation}.

Recall that \cite{GeorgievH22} observes that pure-DP algorithms which succeed \emph{with sufficiently high probability over the internal coins of the algorithm} are automatically robust to a constant fraction of corrupted inputs.
While optimal inefficient private estimators often satisfy this high-probability requirement, most existing polynomial-time private estimators do not.
Our private estimators have not only (nearly) optimal sample complexity but also (nearly) optimal success probability.

%We note in particular \cite{GeorgievH22}, which observes that pure-DP algorithms which succeed \emph{with sufficiently high probability over the internal coins of the algorithm} are automatically robust to a constant fraction of corrupted inputs.
%For many estimation tasks, optimal private estimators satisfy this high-probability requirement. 
%(Estimators designed using the exponential mechanism~\cite{McSherryT07} often have these properties.)
%Most existing polynomial-time private estimators, however, lack such high success probability.
%In our work, we focus on 
%We develop
%private estimators which have not only (nearly) optimal sample complexity but also (nearly) optimal success probability, making them automatically robust.
%Our privacy-from-robustness reduction also provides an appealing converse to this robustness-from-privacy observation.

\medskip

\noindent \emph{Private Estimators via SoS:}
\cite{HopkinsKM22} and \cite{KothariMV22} pioneer the use of SoS for private algorithm design.
\cite{HopkinsKM22} gives a polynomial-time algorithm for pure-DP mean estimation under a bounded covariance assumption, using $\tfrac d {\alpha^2 \e}$ samples, and \cite{KothariMV22} gives a $\approx d^8$-sample $(\e,\delta)$-DP algorithm for learning $d$-dimensional Gaussians.
\cite{GeorgievH22} uses SoS for private \emph{sparse} mean estimation.

On a technical level, our work most resembles \cite{HopkinsKM22}; we also employ SoS SDPs as score functions and leverage tools from log-concave sampling.
However, there are fundamental roadblocks to using \cite{HopkinsKM22}'s strategy for converting SoS proofs into private algorithms in settings beyond mean estimation under bounded covariance, as we discuss in Section~\ref{sec:techniques}.
We provide a blueprint for converting a much wider range of SoS-based robust algorithms to private ones. 
%
%\cite{HopkinsKM22}'s method of converting SoS proofs into private algorithms is tailored to those involved in mean estimation under a bounded covariance assumption. There are fundamental roadblocks to using it in the context of the SoS proofs needed for optimal sample complexity under Gaussian assumptions, as we discuss in Section~\ref{sec:techniques}.
%At a high level, we provide a blueprint for constructing SoS-based analyses of the exponential-mechanism-based estimators which arise from our black-box reduction. 

\medskip

\noindent
\emph{Inverse Sensitivity Mechanism:}
In \cite{asi2020near,asi2020instance}, Asi and Duchi design private polynomial-time algorithms for statistical problems with an \emph{inverse sensitivity mechanism} which is closely related to our black-box reduction, as described in \eqref{eq:intro-1}.
However, the focus of their work is rather different, as they investigate applications to instance-optimal private estimation, whereas our goal is to understand private estimation through the lens of robustness.
Furthermore, their study is centered on one-dimensional statistics, and their analysis is not black-box.

\medskip

\noindent
\emph{Contemporaneous work:} In independent and simultaneous work, Alabi, Kothari, Tankala, Venkat, and Zhang also design efficient robust and private algorithms for learning high-dimensional Gaussians with nearly-optimal sample complexity with respect to dimension; however, their algorithms require $\poly(1/\e, \log 1/\delta, 1/\alpha)$-factors more samples than those we present \cite{AlabiKTVZ23}.
In another independent and simultaneous work, Asi, Ullman, and Zakynthinou introduce the same black-box transformation from robustness to privacy~\cite{AsiUZ23}.
To contrast the two works: we go beyond this inefficient reduction, and also design efficient algorithms for Gaussian estimation.
On the other hand, they show the transformation gives the optimal error for low-dimensional problems, showing tightness of the robustness-privacy connection in certain settings.
Finally, two works subsequent to ours give computationally-efficient algorithms for mean estimation in Mahalanobis distance while requiring only a near-linear number of samples~\cite{BrownHS23,DuchiHK23}, improving on the exponential time algorithm of~\cite{BrownGSUZ21}.
Both new works are based on ``stable'' estimators for mean and covariance, where stability is a notion of robustness different from the one we consider in this work.

\section{Techniques}
\label{sec:techniques}

\subsection{Black-Box Reduction from Robustness to Privacy} \label{subsec:black-box}
Consider a deterministic\footnote{If we are not concerned with running time, the deterministic assumption is without loss of generality, as any randomized estimator can be converted to a deterministic one with at most a constant-factor loss in accuracy, by enumerating over all choices of the estimator's internal random coins and selecting an output which is contained in a ball which contains at least $50\%$ of the mass of the estimator's output distribution.} robust estimator $\hat{\theta} \, : \, \text{datasets} \rightarrow \Theta$ for a parameter space $\Theta \subset \R^D$, a distribution family $\cP$, and a norm $\|\cdot \|$, with the following guarantee: for a non-decreasing function $\alpha \, : \, [0,1] \rightarrow \R$ and some $n \in \N$, with probability $1-\beta$ over samples $X_1,\ldots,X_n \sim p_\theta \in \cP$, for every $\eta \in [0,1]$, given any $\eta$-corruption of $X_1,\ldots,X_n$, the estimator obtains $\|\hat{\theta} - \theta\| \leq \alpha(\eta)$. 
That is, $\alpha$ is a function that quantifies the error achieved by the estimator for every corruption level $\eta$.
Let $X$ denote an $n$-vector dataset $X_1,\ldots,X_n$, and $d(X,X')$ be the Hamming distance between the datasets $X,X'$.

\emph{Our key conceptual contribution is the following instantiation of the exponential mechanism~\cite{McSherryT07}:} Given $\e > 0$, $X_1,\ldots,X_n$ and a threshold $\eta_0 \in [0,1]$, the mechanism picks a random $\theta \in \Theta + \alpha(\eta_0) \cdot B_{\|\cdot \|}$ with:
\begin{align}
    \label{eq:intro-1}
\Pr(\theta) \propto \exp ( - \e \cdot \text{score}_X(\theta)) \text{ where score}_X(\theta) = \min \{ d(X,X') \, : \, \|\hat{\theta}(X') - \theta\| \leq \alpha(\eta_0) \} \mcom
    \end{align}
    where $B_{\|\cdot \|}$ is the unit ball of $\|\cdot\|$.
    In words: the mechanism assigns each $\theta$ within distance $\alpha (\eta_0)$ of $\Theta$ a score given by the number of input samples which would have to be changed to obtain a dataset $X'$ for which the robust estimator $\hat{\theta}(X')$ is close to $\theta$, and samples $\theta$ with probability  $ \propto \exp(-\e \cdot \text{score}_X(\theta))$.
If $\Theta$ is unbounded these probabilities are not well defined; in that case pure-DP guarantees are not obtainable anyway, due to packing lower bounds~\cite{HardtT10}.
Later, we use a \emph{truncated} version of \eqref{eq:intro-1} to allow unbounded $\Theta$ with $(\e,\delta)$-DP.

The general idea to instantiate the exponential mechanism where the score of some $\theta$ is the number of inputs which must be changed to make some function $\hat{\theta}$ take the value (approximately) $\theta$ appears to be folklore; see for instance the \emph{inverse sensitivity mechanism} of \cite{asi2020near}.
%\footnote{We thank Lydia Zakynthinou and Pasin Manurangsi for making us aware of the the inverse sensitivity mechanism.}
Our contribution is (a) to show that for \eqref{eq:intro-1} to have nontrivial utility guarantees, it suffices for $\hat{\theta}$ to be robust to adversarial corruptions, and (b) to show how to implement variants of \eqref{eq:intro-1} in polynomial time.

To elucidate the role of and how to set the threshold parameter $\eta_0$: if the target bound on the error of our private estimator is some value $\alpha$, we can think of $\eta_0$ as the maximum amount of contamination a robust estimator could tolerate if the goal was to achieve the same error $\alpha$.
This will depend on the distribution class $\cP$; for example, if we consider the class of distributions with bounded covariance $\Sigma \preceq I$, then the appropriate setting is $\eta_0 = \Theta(\alpha^2)$~\cite{DiakonikolasKKLMS17,SteinhardtCV18}.

The exponential mechanism enjoys $(2\e,0)$-DP, but the question of utility remains. Suppose that $X_1,\ldots,X_n \sim p_{\theta^*}$. How small is $\|\theta - \theta^*\|$?
The following lemma bounds this quantity in terms of the robustness of $\hat{\theta}$.
Despite its simplicity, we are not aware of a similar result in the literature.

\begin{lemma}
    \label{lem:intro-black-box}
    Suppose a dataset $X_1,\ldots,X_n \sim p_{\theta^*}$, where the parameter vector $\theta^* \in \Theta \subseteq \R^D$. 
    For any threshold $\eta_0 \in [0,1]$, a random $\theta$ drawn according to \eqref{eq:intro-1} has $\|\theta - \theta^*\| \leq 2\alpha(\eta_0)$ with probability at least $1-2\beta$, if
    \begin{equation}
        \label{eq:intro-2}
        n \geq \max_{\eta_0 \leq \eta \leq 1} \, \frac{D \cdot \log \tfrac{2 \alpha(\eta)}{\alpha(\eta_0)} + \log(1/\beta) + O(\log \eta n)}{\eta \e}\mper
    \end{equation}
\end{lemma}
\noindent Observe that the $O(\log \eta n)$ term in \eqref{eq:intro-2} is negligible compared to $D \log \tfrac{2 \alpha(\eta)}{\alpha(\eta_0)} \geq D \log 2$ if $n \ll 2^D$.

The sample complexity in~\eqref{eq:intro-2} is a maximum over the parameter $\eta$; we pay a cost in samples depending on the underlying robust estimator's robustness profile, taking the worst case over all corruption levels $\eta$.
The price at each $\eta$ scales roughly as the log-volume of the set of  solutions which satisfy the robust estimator's accuracy level under $\eta$-corruptions.
The more robust the estimator is, the smaller this volume will be, matching the intuition that settings which permit more robust estimation also are easier to privatize.

A \emph{robust} analogue of Lemma~\ref{lem:intro-black-box}, in which the dataset $X_1, \dots, X_n$ is a \emph{contamination} of i.i.d.\ samples from $p_{\theta^*}$, follows by a similar proof.

\begin{proof}
%\todo{Honestly this is quite difficult to read, we could move it to after the first 10 pages if we need space?}
  Condition on the $(1-\beta)$-probable event that the robustness guarantees of $\hat{\theta}$ hold with respect
  to $X$.
  Consider $\theta$ with score $\eta n$. 
  By definition, $\|\theta - \hat{\theta}(X')\| \leq \alpha(\eta_0)$ for some $X'$ with $d(X,X') \leq \eta \cdot n$.
  By robustness, $\|\hat{\theta}(X') - \theta^*\| \leq \alpha(\eta)$. 
  Using triangle inequality, $\|\theta - \theta^*\| \leq \alpha(\eta_0) + \alpha(\eta) \leq 2 \alpha(\eta)$, assuming $\eta \ge \eta_0$.
  In summary, any $\theta$ with score $\eta n$ is within distance $2\alpha(\eta)$ of $\theta^*$.

  Let $V_r$ be the volume of a radius $r$ $\|\cdot\|$-ball.
  Any $\theta$ such that $\|\theta - \hat{\theta}(X)\| \leq \alpha(\eta_0)$ has score $0$.
  The normalizing factor implicit in \eqref{eq:intro-1} can be lower bounded by the contribution due to these points, or $V_{\alpha(\eta_0)} \cdot \exp(-\varepsilon \cdot 0) = V_{\alpha(\eta_0)}$.
  Combining this with the argument above, the probability of seeing $\theta$ with score $\eta n$ with $\eta > \eta_0$ in a draw from \eqref{eq:intro-1} is at most $\tfrac{ V_{2\alpha(\eta)}}{V_{\alpha(\eta_0)}} \exp(-\e \eta n)$.
  Summing over all scores $\geq \eta_0 n$, the overall probability of seeing some $\theta$ with score greater than $\eta_0$ is at most
  \[
      \sum_{t = \eta_0 n}^n \frac{V_{2\alpha(t/n)}}{V_{\alpha(\eta_0)}} \cdot \exp(-\e t) = \sum_{t = \eta_0 n}^n \frac{V_{2\alpha(t/n)}}{V_{\alpha(\eta_0)}} \cdot \exp(-\e t) \cdot t^2 \cdot 1/t^2 \leq O(1) \cdot \max_{\eta_0 \leq \eta \leq 1}  \,  \left \{ (\eta n)^2 \cdot \frac{V_{2\alpha(\eta)}}{V_{\alpha(\eta_0)}} \cdot \exp(-\e \eta n) \right \}\mcom
  \]
  where the inequality is H\"older's.
  %\todo{how does this follow? I'm getting $n$ instead of $(\eta n)^2$ and not using Holder}. 
  This quantity is at most $\beta$ for $n$ as in \eqref{eq:intro-2}.
  So, with probability at least $1-\beta$ the random $\theta$ will have score at most $\eta_0 n$, meaning $\|\theta - \theta^*\| \leq 2\alpha(\eta_0)$. At the beginning, we conditioned on a $(1-\beta)$-probable event, so the overall failure probability is at most $2\beta$.
\end{proof}

\medskip

\noindent \emph{Consequences of Lemma~\ref{lem:intro-black-box}:}
Applied to robust mean estimators with optimal error rates under bounded $k$-th moment assumptions, for any $k \geq 2$, Lemma~\ref{lem:intro-black-box} gives optimal pure-DP estimators under those same assumptions, recovering the main results of \cite{KamathSU20}; applied to robust linear regression (with known covariance) \cite{DiakonikolasKS19}, it yields a pure-DP analogue of the nearly-optimal regression result of \cite{LiuKKO21}; and so on.
The same argument can be adapted to perform covariance-aware mean estimation\footnote{a.k.a., mean estimation in Mahalanobis distance} and covariance-aware linear regression, recovering pure-DP versions of the results of \cite{LiuKKO21,BrownGSUZ21}, using a robust estimator of mean and covariance.

    To illustrate, we apply Lemma~\ref{lem:intro-black-box} to Gaussian mean estimation.
    With $n \gg d/\alpha^2$ samples from a $d$-dimensional Gaussian $\cN(\mu,I)$, it is possible to estimate the mean under $\eta$-contamination with error $\|\hat{\mu} - \mu\| \leq O(\alpha + \eta)$, if $\eta < 1/2$.
    For $\e$-DP guarantees, we need to restrict to the case of $\|\mu\| \leq R$ for some (large) $R > 0$; we will assume that even for $\eta \geq 1/2$, $\|\hat{\mu}\| \leq R$. In other words, $\alpha(\eta) = O(\alpha+\eta)$ when $\eta < 1/2$, and $\alpha(\eta) \leq R$ for $\eta \ge 1/2$, where we recall that $\alpha(\eta)$ represents the accuracy of the robust algorithm under $\eta$-contamination.

    Plugging such a robust $\hat{\mu}$ into Lemma~\ref{lem:intro-black-box}, and choosing $\eta_0 = \alpha$, there are two interesting cases: $\eta = O(\eta_0)$ and $\eta = 1$.
    In the former, $\alpha(O(\eta_0)) / \alpha(\eta_0) = O(1)$, so we get the requirement $n \geq O(\tfrac{d + \log(1/\beta)}{\alpha \e})$, and in the latter $\alpha(1) = R$, so we get the additional requirement $n \geq \tfrac{d \log R}{\e}$, meaning that we obtained an $\e$-DP estimator with accuracy $O(\alpha)$ using $n$ samples,
    \[
        n \gg \frac{d + \log(1/\beta)}{\alpha \e} + \frac{d \log R}{\e} + \frac d {\alpha^2}.
    \]
This is tight up to constants \cite{HardtT10, BunKSW19}.
Similarly tight results can be derived for mean estimation under bounded covariance, covariance estimation, linear regression, and more.
We remind the readers that the resulting private algorithms are \emph{not} computationally efficient, though we will see how this approach can be made efficient for several interesting cases.

\medskip

\noindent \emph{When Is Lemma~\ref{lem:intro-black-box} Loose?}
More refined analyses of the construction \eqref{eq:intro-1} are possible.
In particular, if the robust estimator $\hat{\theta}$ enjoys the property that the \emph{volume} of the sets of possible values it assumes under $\eta$-corrupted inputs are substantially smaller than $V_{2\alpha(\eta)}$, the bound in Lemma~\ref{lem:intro-black-box} can be improved accordingly
(at the cost of breaking black-box-ness in the analysis.)

As an example, consider estimating the mean of a Gaussian $\cN(\mu,I)$ to $\ell_\infty$ error $\alpha$.
Using a similar argument as in the $\ell_2$ example above, Lemma~\ref{lem:intro-black-box} gives a sample-complexity upper bound of $\tfrac{\log d}{\alpha^2} + \tfrac d {\alpha \e} + \tfrac{ d \log R}{\e}$.
But, because $d_{TV} (\cN(\mu,I),\cN(\mu',I)) \approx \|\mu - \mu'\|_2$, it's possible to construct a robust estimator $\hat{\mu}$ such that under $\eta$-corruptions, $\|\hat{\mu} - \mu\|_\infty$ can only be as large as $\eta$ if $\|\hat{\mu}- \mu\|_2 \approx \|\hat{\mu} - \mu\|_\infty$; otherwise $\|\hat{\mu} - \mu\|_\infty$ is much smaller.
This affords better control over the volumes of candidate outputs with a given score $\eta n$ than the $\eta$-radius $\ell_\infty$ ball would offer.
Using this, we show in Appendix~\ref{sec:infty-mean} that $\tilde{O}(\tfrac{\log d}{\alpha^2} + \tfrac{d^{2/3}}{\alpha \e^{2/3}} + \tfrac{\sqrt{d}}{\alpha \e} + \tfrac{d \log R}{\e})$ samples are enough, in the pure-DP setting.
%, because we can construct a robust estimator with \emph{simultaneous} accuracy guarantees in $\ell_\infty$ and $\ell_2$, affording us better control over the volumes of $\theta$s with a given score.
%because for a given $\mu$, the only $\mu'$ such that $\|\mu - \mu'\|_\infty = \eta$ and $d_{TV}(\cN(\mu,I),\cN(\mu',I)) \approx \eta$ have $\mu - \mu'$ \emph{sparse} -- these $\mu'$ do not fill up an $\ell_\infty$ ball of radius $\eta$ around $\mu$.
%We show in Appendix~\ref{sec:infty-mean} that $\tfrac{\log d}{\alpha^{O(1)} \e} + \tfrac{d \log R}{\e}$ samples actually suffice for this problem.

\medskip

\noindent \emph{From Robustness to $(\e,\delta)$-DP:} If $\hat{\theta}$ has a nontrivial \emph{breakdown point} -- i.e., a fraction of corruptions $\eta$ beyond which it admits no error guarantees, then Lemma~\ref{lem:intro-black-box} doesn't give a nontrivial private estimator.
For example, in the Gaussian mean estimation setting, if we remove the assumption $\|\mu\| \leq R$, then when $\eta \geq 1/2$ no estimator has a finite accuracy guarantee (i.e., $\alpha(\eta) $ is unbounded for such $\eta$).

By relaxing from pure to $(\e,\delta)$-DP, however, we can design private estimators even from robust estimators $\hat{\theta}$ which have a breakdown point.
Our reduction in this case, however, requires $\hat{\theta}$ to satisfy a \emph{worst-case} robustness property, because we will need to appeal to robustness to ensure not only accuracy, as in Lemma~\ref{lem:intro-black-box}, but also privacy, which is inherently a worst-case guarantee.

Simple adaptations of standard robust estimators of mean and covariance, and robust regression algorithms, have such worst-case robustness guarantees.
This approach gives an alternative to the high-dimensional propose-test-release framework of \cite{LiuKO22}, and the approach of \cite{BrownGSUZ21}, for building approx-DP estimators from robust estimation primitives; we can recover their results on covariance-aware mean estimation and linear regression with $(\e,\delta)$-DP guarantees.
This approach carries the advantages of black-box-ness and potential polynomial-time implementability, since SoS-based robust estimators for mean and covariance have the required worst-case behavior.

Consider again a deterministic robust estimator $\hat{\theta} \, : \, \text{datasets} \rightarrow \Theta \cup \{ \textsc{reject} \}$ for a parameter $\theta \in \R^d$, which takes $n$ inputs and returns either some element of $\Theta$ or \textsc{reject}.
Let $\cP$ be a distribution family, $\|\cdot\|$ be a norm, $\alpha \, : \, [0,1] \rightarrow \R$ be a non-decreasing function, $n \in \N$, and $\eta_0, \eta^* \in [0,1]$.
We continue to employ $\textsc{score}_X(\theta)$ as defined in \eqref{eq:intro-1}.
Suppose as before that with probability $1-\beta$ over samples $X_1,\ldots,X_n \sim p_\theta \in \cP$, for every $\eta < \eta^*$, given any $\eta$-corruption of $X_1,\ldots,X_n$, $\|\hat{\theta} - \theta\| \leq \alpha(\eta)$.
And, suppose that $\hat \theta$ has the following worst-case robustness property: for \emph{any} input $X = X_1,\ldots,X_n$, if $\hat{\theta}(X) \neq \textsc{reject}$, then for every $\eta < \eta^*$, given any $\eta$-corruption $X'$ of $X$, either $\hat{\theta}(X') = \textsc{reject}$, or $\|\hat{\theta}(X') - \hat{\theta}(X)\| \leq \alpha(\eta^*)$.

\begin{lemma}
\label{lem:approx-dp-intro}
   Let $\eta_0 < \eta^* \in [0,1]$ be such that $\eta^* n$ is a sufficiently large constant.
   For every $\e, \delta > 0$, there is an $(O(\e),O(e^{2\e}\delta))$-DP mechanism which, for any $\theta^*$, takes $X_1,\ldots,X_n \sim p_{\theta^*}$ and with probability $1-\beta$ outputs $\theta$ such that $\| \theta - \theta^* \| \leq 2 \alpha(\eta_0)$, if
   \[
     n \geq O \Paren{ \max_{\eta_0 \leq \eta \leq \eta^*} \frac{D \cdot \log \tfrac{2 \alpha(\eta)}{\alpha(\eta_0)} + \log(1/\beta) + \log \eta n}{\eta \e} + \frac{\log(1/\delta)}{\eta^* \e}}\mper
   \]
\end{lemma}

Before proving the lemma, we need a preliminary claim.
\begin{claim}
\label{clm:approx-dp-ball}
Suppose for a dataset $X$ there exists $\theta$ such that $\textsc{score}_X(\theta) < 0.2 \eta^* n$.
Then there exists a ball of radius $2\alpha(\eta^*)$ which contains every $\theta'$ with $\textsc{score}_X(\theta') < 0.4 \eta^* n$.
\end{claim}
\begin{proof}
Since there exists some $\theta$ such that $\textsc{score}_{X}(\theta) < 0.2 \eta^* n$, there's some $Y \sim_{0.2 \eta^*} X$ such that $\hat{\theta}(Y) \neq \textsc{reject}$: this is because we can consider any such $Y$ which has $\textsc{score}_{Y}(\theta) = 0$, and thus $\hat \theta(Y)$ outputs an element of $\Theta$ and not \textsc{reject}.
Similarly, for any other $\theta'$ with $\textsc{score}_{X}(\theta') \leq 0.4 \eta^* n$, there's some $Z \sim_{0.4 \eta^*} X$ such that $\|\theta' - \hat{\theta}(Z)\| \leq \alpha(\eta_0)$.
By triangle inequality, $Z \sim_{0.6 \eta^*} Y$, so by worst-case robustness of $\hat{\theta}$, $\|\theta' - \hat{\theta}(Y)\| \leq \|\theta' - \hat{\theta}(Z)\| + \|\hat{\theta}(Z) - \hat{\theta}(Y)\| \leq \alpha(\eta_0) + \alpha(\eta^*) \leq 2\alpha(\eta^*)$.
\end{proof}

\begin{proof}[Proof of Lemma~\ref{lem:approx-dp-intro}]
  First, let $g \, : \, \mathbb{Z} \rightarrow \mathbb{R}$ be a function with the following properties: for $t < 0.1 \eta^* n$, $g(t) = 1$, for $t > 0.2 \eta^* n$, $g(t) = 0$, and for all $t$, $e^{-\e}g(t+1) - \delta \leq g(t) \leq e^{\e} g(t+1) + \delta$.
  Such a function exists since $n \gg \log \tfrac 1 \delta / \eta^* \e$.
  
  This is not hard to show: one could, for example, consider the function which, for $t$ over the interval $[0.1 \eta^* n, 0.2 \eta^* n]$, first decreases by a multiplicative factor of $e^{-\varepsilon}$ (i.e., $g(t+1) = e^{-\varepsilon} g(t)$) until some point $t^*$ when $g(t^*) \leq \delta$. 
  Then, we set $g(t) = 0$ for all $t > t^*$. 
  This satisfies the requirements on the function for all $t \leq t^*$ with $\delta = 0$, and for $t > t^*$ with $\varepsilon = 0$.
  We need that $\delta \geq \exp(-(t - 0.1 \eta^* n)\varepsilon)$ is satisfied by some $t$ in the interval $[0.1 \eta^* n, 0.2 \eta^* n]$ (roughly speaking, to allow enough multiplicative $e^{-\varepsilon}$ decreases to accumulate in order to cancel out the remainder with a subtractive $\delta$ shift), which we can take to be $t^*$.
  Rearranging the inequality, we get $t \geq \log(1/\delta)/\varepsilon + 0.1 \eta^* n$. 
  But for $t^*$ to lie in the stated interval, we need $\log(1/\delta)/\varepsilon + 0.1 \eta^* n \leq t \leq 0.2 \eta^* n$, which is satisfied as long as $n \gg \log (1/\delta)/\eta^* \varepsilon$, as claimed.
  %Let $\textsc{score}_X(\theta)$ be as in \eqref{eq:intro-1}.
  
  The mechanism is as follows.
  Given $X = X_1,\ldots,X_n$, let $T = \min_{\theta \in \Theta} \textsc{score}_X(\theta)$.
  First, output \textsc{reject} with probability $1-g(T)$.
  If \textsc{reject} is not output, output a sample from the distribution on $\Theta + \alpha(\eta_0) B_{\|\cdot\|}$ where
  \[
  \Pr(\theta) \propto \begin{cases} 
  \textsc{score}_X(\theta) & \text{ if } \textsc{score}_X(\theta) < 0.3 \eta^* n \\
  0 & \text{ otherwise} \end{cases}
  \]
  and $B_{\|\cdot\|}$ is the unit ball for the norm $\|\cdot\|$.

  \emph{Proof of privacy: }
  The \textsc{reject} phase of the mechanism clearly satisfies $(\e,\delta)$-DP, because $\textsc{score}_X(\theta)$ can change by at most $1$ when $X$ is replaced with neighboring $X'$, and based on the definition of $g$.
  
  Now we turn to the sampling phase.
  Let $X,X'$ differ on one sample.
  Let $T,T'$ be the numbers computed in the \textsc{reject} phase of the mechanism; we may assume $T,T' \leq 0.2 \eta^* n$, since otherwise on both $X,X'$ the mechanism outputs \textsc{reject} with probability at least $1-\delta$.
  We show that the mechanism above, conditioned on not rejecting, satisfies $(O(\e),O(e^{2\e}\delta))$-DP; then the overall result follows by composition.
  
  For brevity, we abbreviate $\textsc{score}_X$ to $s_X$.
  For any $S \subseteq \Theta + \alpha(\eta_0) \cdot B_{\|\cdot\|}$, we can bound its associated weight via
  \begin{align*}
  \int_{\theta \in S} & e^{-\e s_X(\theta)} \cdot \mathbf{1}(s_X(\theta) < 0.3 \eta^* n) \leq e^{\e} \int_{\theta \in S} e^{-\e s_{X'} (\theta)} \cdot [\mathbf{1} (s_{X'}(\theta) < 0.3 \eta^* n) + \mathbf{1}(s_{X'}(\theta) \in [0.25 \eta^* n, 0.35 \eta^* n] ].
  \end{align*}
%where we have used repeatedly that for any $\theta$ we have $|s_X(\theta) - s_{X'}(\theta)| \leq 1$, and we have used that $\eta^* n$ is a big-enough constant.
To see why, first note that for any $\theta$ we have $|s_X(\theta) - s_{X'}(\theta)| \leq 1$.
This implies that $e^{-\e s_X(\theta)} \leq e^\varepsilon e^{-\e s_{X'} (\theta)}$.
Similarly, if $s_X(\theta) \leq 0.3\eta^* n$, it also implies that at least one of the following must be true (potentially both): $s_{X'}(\theta) \leq 0.3\eta^* n$ or $s_{X'}(\theta) \in [0.25 \eta^* n, 0.35 \eta^* n]$ (we use the fact that $\eta^* n$ is at least a sufficiently large constant).

Normalizing to get a probability, we have
\begin{align*}
\Pr_X (\theta \in S) 
& \leq e^{\e} \cdot
\frac
{ \int_{\theta \in S} e^{-\e s_{X'} (\theta)} \cdot [\mathbf{1} (s_{X'}(\theta) < 0.3 \eta^* n) + \mathbf{1}(s_{X'}(\theta) \in [0.25 \eta^* n, 0.35 \eta^* n])] }
{ \int_{\theta \in \Theta + \alpha(\eta_0) B_{\|\cdot \|}} e^{-\e s_X(\theta)} \cdot \mathbf{1}(s_X(\theta) < 0.3 \eta^* n) }\\
& \leq e^{\e} \cdot
\frac
{ \int_{\theta \in S} e^{-\e s_{X'} (\theta)} \cdot [\mathbf{1} (s_{X'}(\theta) < 0.3 \eta^* n) + \mathbf{1}(s_{X'}(\theta) \in [0.25 \eta^* n, 0.35 \eta^* n])] }
{e^{-\e} \int_{\theta \in \Theta + \alpha(\eta_0) B_{\|\cdot \|}}  e^{-\e s_{X'}(\theta)} \cdot [\mathbf{1}(s_{X'}(\theta) < 0.3 \eta^* n) - \mathbf{1}(s_{X'}(\theta) \in [0.25 \eta^* n, 0.35 \eta^* n])]}
\end{align*}
The denominator is split into two terms with a similar argument as used for the numerator.

We next simplify the denominator. Because, by assumption, there is
$\theta'$ such that $\textsc{score}_{X'}(\theta') < 0.2 \eta^* n$, there is a ball of radius $\alpha(\eta_0)$, contained in $\Theta + \alpha(\eta_0) \cdot B_{\|\cdot\|}$, of points with score at most $0.2 \eta^* n$; we can hence lower-bound the first term $\int e^{-\e s_{X'}(\theta)} \cdot \mathbf{1}(s_{X'}(\theta) < 0.3 \eta^* n) \geq \exp(- \e \cdot 0.2 \eta^* n) \cdot V_{\alpha(\eta_0)}$,
where $V_{\alpha(\eta_0)}$ is the volume of a $\|\cdot \|$-ball of radius $\alpha(\eta_0)$.

We can use \Cref{clm:approx-dp-ball} to upper-bound the magnitude of the second term in the denominator, $\int e^{-\e  s_{X'}(\theta)} \cdot \mathbf{1}(s_{X'}(\theta) \in [0.25 \eta^* n, 0.35 \eta^* n]) \leq \exp(-\e \cdot 0.25 \eta^* n) \cdot V_{2\alpha(\eta^*)}$, which is at most $\delta$ times the lower bound on the first term, under our hypotheses on the lower bound for $n$.
Overall, we obtain
\begin{align*}
    \Pr_X(\theta \in S) &\leq \frac{e^{2 \eps}}{1-\delta} \cdot \left(\frac{ \int_{\theta \in S} e^{-\e s_{X'} (\theta)} \cdot \mathbf{1} (s_{X'}(\theta) < 0.3 \eta^* n) + \int_{\theta \in S} e^{-\e s_{X'} (\theta)} \cdot \mathbf{1}(s_{X'}(\theta) \in [0.25 \eta^* n, 0.35 \eta^* n])}{\int_{\theta \in \Theta + \alpha(\eta_0) B_{\|\cdot \|}}  e^{-\e s_{X'}(\theta)} \cdot \mathbf{1}(s_{X'}(\theta) < 0.3 \eta^* n)}\right) \\
    &= \frac{e^{2\e}}{1-\delta} \cdot \big(\Pr_{X'}(\theta \in S) + \Pr_{X'}(s_{X'}(\theta) \in [0.25 \eta^* n, 0.35 \eta^* n])\big).
\end{align*}
Using \Cref{clm:approx-dp-ball} in the same fashion to bound the last term, this is at most $e^{2\e} \Pr_{X'}(\theta \in S) + O(e^{2\e}\delta)$, which completes the privacy proof.

\emph{Proof of accuracy:} Observe that with probability at least $1-\beta$ over samples $X_1,\ldots,X_n$, the \textsc{reject} phase of the mechanism accepts with probability $1$.
Conditioned on it doing so, the remainder of the accuracy proof parallels the proof of Lemma~\ref{lem:intro-black-box}, except instead of allowing $\eta \in [\eta_0, 1]$ we can now limit it to $\eta \in [\eta_0, \eta^*]$.
\end{proof}

\subsection{Algorithms}
Even if the robust estimator $\hat{\theta}$ can be computed in polynomial time, the sampling problem in \eqref{eq:intro-1} lacks an obvious polynomial-time algorithm, for two reasons.
First, computing the score of a single $\theta \in \Theta$ given an input dataset $X$ appears to require solving a minimization problem over all other datasets $X'$.
Second, even if computing the scores were somehow made efficient, the resulting sampling problem might still be computationally hard.
Our main technical contribution is to overcome both of these hurdles in the context of learning high-dimensional Gaussian distributions.

%\subsubsection{Background: Sum of Squares and Robust Estimation}
The Sum of Squares method (\emph{SoS}) uses convex programming to solve multivariate systems of polynomial inequalities.
It is extremely useful for designing polynomial-time robust estimators.
%\todo{Shyam: Maybe to save space we can move the SoS proof stuff to the preliminaries, and just discuss pseudoexpectations? Since we should be able to phrase everything in terms of just pseudoexpectations}

\begin{definition}[SoS Proof]
    Let $p_1(x) \geq 0,\ldots,p_m(x) \geq 0$ be a system of polynomial inequalities in variables $x_1,\ldots,x_n$.
    An inequality $q(x) \geq 0$ has a \emph{degree $d$ SoS proof} from $p_1 \geq 0 ,\ldots,p_m \geq 0$, written $\{ p_1 \geq 0, \ldots, p_m \geq 0\} \proves_d^x q \geq 0$, if for each multiset $S \subseteq [m]$ there exists a sum of squares polynomial $q_S(x)$, such that $\deg (q_S(x) \cdot \prod_{i \in S} p_i(x)) \leq d$ and such that
    \[
        q(x) = \sum_{S \subseteq [m]} q_S(x) \cdot \prod_{i \in S} p_i(x)\mper
    \]
\end{definition}

SoS proofs form a convex set described by a semidefinite program (SDP), so they have duals:

\begin{definition}[Pseudoexpectation]
    Let $\R[x]_{\leq d}$ be the set of degree at most $d$ polynomials in variables $x_1,\ldots,x_n$.
    A linear operator $\pE \, : \, \R[x]_{\leq d} \rightarrow \R$ is a degree $d$ \emph{pseudoexpectation} if $\pE 1 = 1$ and $\pE p^2 \geq 0$ for any $p$ of degree at most $d/2$.
    A pseudoexpectation $\pE$ \emph{satisfies} a system of polynomial inequalities $p_1 \geq 0,\ldots,p_m \geq 0$, written $\pE \models p_1 \geq 0,\ldots,p_m \geq 0$, if for every $S \subseteq [m]$ and every $p$, we have $\pE \prod_{i \in S} p_i \cdot p^2 \geq 0$ when the degree of this polynomial is at most $d$.
    %, where $\|p\|$ is the $\ell_2$-norm of the vector of coefficients of $p$ in the monomial basis.
\end{definition}

%The set of pseudoexpectations satisfying $p_1 \geq 0,\ldots,p_m \geq 0$ can also be described using an SDP.
%The algorithmic guarantees associated with this SDP can be somewhat delicate, depending on the polynomials $p_1,\ldots,p_m$.
%For robust estimation, the following algorithmic result is usually enough, but we will need to go beyond it when we design a private algorithm.
%
%\begin{fact}
%    \label{fact:sos-alg}
%    If $p_1,\ldots,p_m$ have bit-complexity at most $B$ and $\exists x$ such that $p_1(x) \geq 0,\ldots,p_m(x) \geq 0$, then in $(Bn \log(1/\delta))^{O(d)}$ time it is possible to find  a degree-$d$ $\pE$ such that $\pE \models^{\delta} p_1 \geq 0,\ldots,p_m \geq 0$.
%\end{fact}
%
The by-now standard approach to use SoS to robustly estimate a $D$-dimensional parameter $\theta$ in a norm $\|\cdot \|$ works as follows.
For $\eta$-corrupted $X = X_1,\ldots,X_n$ from $p_{\theta^*}$, define a degree-$O(1)$ system of polynomial inequalities $\cA(X,\theta,z)$ where $\theta = \theta_1,\ldots,\theta_D, z = z_1,\ldots,z_{(nD)^{O(1)}}$ are some indeterminates.
With high probability, $\cA(X,\theta,z)$ should (a) be satisfied by some choice of $z$ when $\theta = \theta^*$, and (b) should have $\cA(X,\theta,z) \proves_{O(1)} \iprod{\theta - \theta^*,v} \leq \alpha$ for every $v$ in the dual ball of $\|\cdot \|$.

%$\cA(X,X',w,z) = \{p_1 \geq 0,\ldots,p_m \geq 0\}$, where $X' = X_1',\ldots,X_n', w = w_1,\ldots,w_n,$ and $z = z_1,\ldots,z_{n^{O(1)}}$ are some indeterminates.
%$\cA$ should include the equations $w_i^2 = w_i$ (enforcing that $w_i$ is $0/1$-valued) and $w_i X_i = w_i X_i'$ (enforcing that if $w_i = 1$ then $X_i' = X_i$) and $\sum_{i \leq n} w_i \geq (1-\eta)n$ (enforcing that at least $(1-\eta)n$ $w_i$s are set to $1$).

%$\cA$ should have the following properties: (1) when $X_i'$ are set equal to the un-corrupted samples underlying $X_1,\ldots,X_n$ and $w$ is the indicator for $X_i = X_i'$, there exists $z$ such that $\cA$ is satisfied, and (2) for some polynomial $\hat{\theta}(w,X',z)$ and every unit $v$ in the dual ball of $\| \cdot \|$, $\cA \proves_{O(1)} \iprod{\hat{\theta} - \theta^*,v} \leq \alpha$.

To give a robust estimation algorithm, on input $\eta$-corrupted $X$, we can obtain $\pE$ which satisfies $\cA(X,\theta,z)$ using semidefinite programming,\footnote{This ignores some issues of numerical accuracy which turn out to be important; see below.} and then output $\hat{\theta} = \pE \theta$.
Applying $\pE$ to the SoS proofs $\cA \proves_{O(1)}^{\theta, z} \iprod{\theta - \theta^*,v} \leq \alpha$, we get $\| \pE \theta - \theta^* \| \leq \alpha$.

%For robust estimation of the mean and covariance of a Gaussian distribution, this strategy is carried out by \cite{KMS22}.
%Focusing on mean estimation and suppressing logarithmic factors:

\begin{lemma}[Informal, implicit in \cite{KothariMZ22}]
    There exists $\cA$ with the above properties with respect to $n \gg d/ \eta^2$ $\eta$-corrupted samples from $\cN(\theta^*,I)$, for any $\theta^* \in \R^d$, where $\|\cdot \| = \ell_2$, and $\alpha = \tilde{O}(\eta)$.
\end{lemma}

\subsubsection{Robustness to Privacy, Algorithmically}
For this technical overview, we focus on mean estimation in the pure-DP setting; similar ideas extend to covariance estimation and $(\e,\delta)$-DP.
Even for the SoS-based robust mean estimation algorithm described above, which we call $\textsc{kmz}$, given $X$ we do not know how to efficiently compute
\begin{align}
    \label{eq:intro-3}
    \text{score}_X(\theta) = \min \{d(X,X') \, : \, \| \textsc{kmz}(Y) - \theta \| \leq \alpha \}\mcom
\end{align}
much less sample from the distribution \eqref{eq:intro-1}.
At a very high level, will tackle these challenges by using the polynomial system $\cA(X,\theta,z)$ underlying $\textsc{kmz}$ to design an SoS-based relaxation of the above score function, $\text{SoS-score}_X(\theta)$, which has favorable enough convexity properties that we will be able to both efficiently compute it and sample from the distribution it induces (both up to small error).
The SoS robustness proofs which $\cA$ enjoys will be enough for us to apply an argument like Lemma~\ref{lem:intro-black-box} to prove accuracy of the resulting estimator, and it will be private by construction.

First, we describe an attempt at an SoS relaxation of SoS-score, which will have several flaws we'll fix later.
We can introduce more indeterminates $X'_1,\ldots,X'_n$, $w_1,\ldots,w_n$, $\theta'$, and consider
\begin{align}
    \label{eq:intro-4}
    \cB_t = \left \{ w_i^2 = w_i, \, \sum_{i=1}^n w_i = n-t, \, w_i X_i = w_i X_i', \right \} \cup \cA(X',\theta',z)\mcom
\end{align}
which is satisfied when $X'$ is a dataset with $d(X,X') \leq t$ and $\cA(X',\theta',z)$ is satisfied.
Let
\begin{align}
    \label{eq:intro-5}
    \text{SoS-score}_X(\theta) = \min \, t \, \text{ s.t. } \exists \text{ degree $O(1)$ } \pE \text{ in variables } X',w,\theta',z, \, \pE \models \cB_t,  \,   \|\pE \theta' - \theta\| \leq \alpha\mper
\end{align}

Before describing the analysis of the SoS-score and discussing the flaws in $\cB_t$, we first explain some intuition behind the first set of constraints in $\cB_t$, which resemble common choices in SoS-based robust statistics. Intuitively, one should think of each $w_i$ as an indicator for the data point $X_i$ being corrupted, and $X_i'$ as the ``uncorrupted'' data points. If $\{X_i'\}$ were actual points and $w_i$ were actual real numbers, the constraint $w_i^2 = w_i$ enforces that every $w_i$ is either $0$ or $1$ (which should hold for indicator variables). The constraint $\sum w_i = n-t$ enforces that $n-t$ of the variables are $1$, and the constraint $w_i X_i = w_i X_i'$ enforces that $X_i = X_i'$ whenever $w_i = 1$, which means at least $n-t$ indices $i$ satisfy $X_i = X_i'$. Together, these constraints enforce that the datatsets $X$ and $X'$ differ by at most $t$ points. Finally, the constraints $\mathcal{A}(X', \theta', z)$, while not yet defined, will try to enforce the data points $X'$ to ``look like'' uncorrupted samples from a Gaussian distribution.

\medskip

\noindent \emph{Privacy and Accuracy for SoS-score:}
Suppose for a moment that SoS-score solves our computational problems.
Does it lead to a good private estimator, when we sample from the distribution $\Pr(\theta) \propto \exp(-\e \cdot \text{SoS-score}_X(\theta))$?
Standard arguments show privacy; the main question is accuracy.

It turns out the relaxation is tight enough that the proof of Lemma~\ref{lem:intro-black-box} still applies!
The key step in that proof is to argue via robustness that if $\theta$ has low score, then $\|\theta^* - \theta\|$ is small.
To establish the corresponding statement for SoS-score, we need to show that if $X_1,\ldots,X_n \sim \cN(\theta^*,I)$ and $\pE \models \cB_t$ for $t = \eta n$, then $\|\pE \theta' - \theta^* \| \leq \tilde{O}(\eta)$.
This is slightly stronger than what we already know from the SoS proofs associated to $\cA$, because now we have \emph{indeterminates} $X'$ which represent $\eta$-corrupted samples, rather than a fixed collection of $\eta$-corrupted samples, and we need $\cB_t \proves_{O(1)}^{X',\theta',w,z} \iprod{\theta' - \theta^*,v} \leq \tilde{O}(\eta)$.
Luckily, the SoS proofs of \cite{KothariMZ22} readily generalize to show this.

In fact, \cite{KothariMZ22}'s SoS proofs already show this in part because within the ``auxiliary'' indeterminates $z$ they already use variables like our $X'$ and $w$.
This means that \eqref{eq:intro-4}, \eqref{eq:intro-5}, while closely following our black-box reduction strategy, contain an unnecessary layer of indirection.
When we implement this strategy in detail in Sections~\ref{sec:mean}, \ref{sec:preconditioning}, and \ref{sec:tv}, we remove this indirection for simplicity.

\medskip

\noindent \emph{On ``Satisfies'':}
An important technical difference between our score function and that of \cite{HopkinsKM22} is that the $\pE$s it involves must have $\pE \models \sum_{i=1}^n w_i = n-t$, rather than something weaker, like $\pE \sum_{i=1}^n w_i = n-t$.
While in some applications of SoS this ``satisfies'' versus ``in expectation'' distinction is minor, it is actually crucial for our accuracy guarantees -- if we only required $\pE \sum_{i=1}^n w_i = n-t$, we could have $\pE$ which satisfies the rest of $\cB_t$ but has $\| \pE \theta' - \theta^* \| \geq \Omega(R)$, just by taking $\pE$ to be the moments of a distribution which has all $w_i = 0$ with probability $1/t$.

However, this creates two significant technical challenges.
First, for bit-complexity reasons, no polynomial-time algorithm to check if there exists $\pE$ satisfying a given system of polynomials is known -- existing techniques to find $\pE$s work best in the context of \emph{satisfiable} polynomial systems \cite{RaghavendraW17}.
We sidestep this challenge by generalizing a technique from the robust statistics literature, which searches for $\pE$ which \emph{approximately} satisfies a system of polynomials, to the setting where those polynomials may be unsatisfiable -- see Appendix~\ref{sec:computing-score-functions}.
Ultimately, we find a further-relaxed score function $\text{SoS-score}_X'$, which we evaluate to error $\tau$ in $(nd \log 1/\tau)^{O(1)}$ time.

\medskip
\noindent \emph{Quasi-Convexity, Sampling, and Weak Membership:}
The second challenge is that $\text{SoS-score}_X(\theta)$ need not be convex in $\theta$ -- if it were, we could sample from $\Pr(\theta) \propto \exp(-\e \cdot \text{SoS-score}_X(\theta))$ with log-concave sampling techniques, as in \cite{HopkinsKM22}.
Indeed, consider $\theta_0$ and $\theta_1$ with corresponding scores $t_0,t_1$ witnessed by $\pE_0, \pE_1$.
The problem is that $\tfrac 12 (\pE_0 + \pE_1)$ need not satisfy $\sum_{i=1}^n w_i \geq n - \tfrac 12 (t_0 + t_1)$, even though it does have $\tfrac 12(\pE_0 + \pE_1)[\sum_{i=1}^n w_i ] \geq n - \tfrac 12 (t_0 + t_1)$.

$\text{SoS-score}_X(\theta)$ is \emph{quasi-convex} in $\theta$, meaning that its sub-level sets $S_t = \{ \theta \, : \, \text{SoS-score}_X(\theta) \leq t\}$ are convex for all $t$.
This is good news: if we discretize the range of possible scores $[0,n]$ into $t_1,\ldots,t_{n^{O(1)}}$ (replacing SoS-score with a version rounded to the nearest $t_i$), we can hope to compute the \emph{volumes} $V_i = \text{Vol}(S_{t_i})$, as well as sample uniformly from the $S_{t_i}$s, using standard techniques for sampling from a convex body.
Then, we could sample $\theta$ by first sampling a score $t_i$ with probability proportional to $e^{-\e t_i}(1- e^{-\e (t_{i+1} - t_i)}) V_i$, then drawing uniformly from $S_{t_i}$.

Approximate sampling and volume algorithms for convex bodies typically access the body via a \emph{weak membership oracle}, meaning that the oracle is allowed to give incorrect answers to query points very near the body's boundary.\footnote{It seems to be folklore that volume computation algorithms, e.g. the seminal \cite{DFK91}, work given only weak membership oracles, as opposed to e.g. weak separation oracles. For completeness, in Appendix~\ref{appendix:sampling}, we analyze a hit-and-run sampling algorithm which uses a weak membership oracle, tracking the numerical errors this creates.}
We have access to an oracle which computes $\text{SoS-score}_X(\theta)$ up to exponentially-small errors.
Ideally, we'd create a weak membership oracle by answering a query about $S_{t_i}$ by checking if $\text{SoS-score}_X(\theta) \leq t_i$, but if $\text{SoS-score}_X$ is not Lipschitz, a small error in computing this value may translate to answering a query incorrectly about some $\theta$ far from the boundary of $S_{t_i}$.
That is, we may not notice if $S_{t_i + 2^{-n}}$ is much larger than $S_{t_i}$. 

However, because $\text{SoS-score}_X$ is bounded in $[0,n]$ and the sublevel sets are convex, we are able to show that $S_{t_i + 2^{-n}}$ could only be much larger than $S_{t_i}$ at a small-measure set of $t_i$s.
Thus, if we choose our discretization $t_1,\ldots,t_{n^{O(1)}}$ \emph{randomly}, with very high probability our approximate score oracle for $\text{SoS-score}_X$ translates to a weak membership oracle for the $S_{t_i}$s (Lemma~\ref{lem:volume_perfect_sampling}).

\medskip

\noindent \emph{Putting it Together:}
Thus, by modifying $\text{SoS-score}_X$ by (a) rounding to the nearest threshold $t_i$, thresholds chosen randomly, and (b) accounting for some numerical errors, we obtain a polynomial-time-samplable proxy for \eqref{eq:intro-1}.
Theorems~\ref{thm:pure_dp_general_main} and~\ref{thm:approx_dp_general_main} capture this strategy formally.

%\pagebreak
\section{Preliminaries}

First, we note a few notational conventions. We will use $\textbf{0}$ to denote the origin in $\BR^d$ (or in Euclidean space generally). For $x \in \BR^d$ and $r \ge 0$, we define $B(x, r)$ to be the $\ell_2$-ball of radius $r$ around $x$.

We will use $\E$ to denote expectation. We also use $\E_i$ to denote the formal average over indices: for instance, given variables $x_1', \dots, x_n'$, $\E_i x_i'$ means the polynomial $\frac{x_1' + \cdots + x_n'}{n}.$

We note a series of important definitions that we will use in our analysis.

\begin{definition}[sensitivity] \label{def:sensitivity}
    We say that a function $f(\theta, \cX)$ has sensitivity $\Delta$ with respect to $\cX$ if for all $\theta$ and all neighboring datasets $\cX, \cX'$ (i.e., datasets that differ in exactly one data point), $|f(\theta, \cX) - f(\theta, \cX')| \le \Delta$. We will implicitly assume that sensitivity is with respect to the dataset.
\end{definition}

\begin{definition}[quasi-convexity] \label{def:quasi-convexity}
A function $f: S \to \R$, defined on a convex subset $S$ of a real vector space is quasi-convex if for all $x, y \in S$ and $\lambda \in \brac{0,1}$ we have 
\begin{equation*}
f\paren{\lambda x + \paren{1-\lambda} y} \le \max \set{f\paren{x}, f\paren{y}}.
\end{equation*} 
\end{definition}

\smallskip

Next, we note some important distance metrics for mean vectors and covariance matrices. We will use $\|\cdot\|_F$ to denote Frobenius norm and $\|\cdot\|_{op}$ to denote the operator norm (a.k.a.\ spectral norm) of a matrix.

\begin{definition}[Mahalanobis distance]
    Given two vectors $\mu, \mu' \in \BR^d$ and a positive definite covariance matrix $\Sigma \in \BR^{d \times d}$, we define the \emph{Mahalanobis} distance between $\mu$ and $\mu'$ with respect to $\Sigma$, written as $\|\mu-\mu'\|_{\Sigma}$, to equal $\|\Sigma^{-1/2}(\mu-\mu')\|_2$.

    In addition, given two covariance matrices $\Sigma, \Sigma' \in \BR^{d \times d}$, we define the \emph{Mahalanobis} distance between $\Sigma$ and $\Sigma'$ to equal $\|\Sigma^{-1/2} \Sigma' \Sigma^{-1/2} - I\|_F$.
\end{definition}

Note that there are two different definitions of Mahalanobis distance, though which definition we are using will be clear from context.

It is well known that Mahalanobis distance captures total variation distance. Namely, if $\|\mu-\mu'\|_{\Sigma} = \alpha \le 1$, then $\dtv(\cN(\mu, \Sigma), \cN(\mu', \Sigma)) = \Theta(\alpha)$, and if $\Sigma, \Sigma'$ have Mahalanobis distance $\alpha \le 1$, then  $\dtv(\cN(\textbf{0}, \Sigma), \cN(\textbf{0}, \Sigma')) = \Theta(\alpha)$.

It is well-known that Mahalanobis distance between covariance matrices is roughly symmetric: namely, $\|\Sigma^{-1/2} \Sigma' \Sigma^{-1/2} - I\|_F = \Theta(\|\Sigma'^{-1/2} \Sigma \Sigma'^{-1/2} - I\|_F)$ if either is at most $0.5$. In addition, $\|\Sigma^{-1/2} \Sigma' \Sigma^{-1/2} - I\|_F = \|\Sigma'^{1/2} \Sigma^{-1} \Sigma'^{1/2} - I\|_F$, and $\|\Sigma'^{-1/2} \Sigma \Sigma'^{-1/2} - I\|_F = \|\Sigma^{1/2} \Sigma'^{-1} \Sigma^{1/2} - I\|_F$.

\begin{definition}[Spectral distance]
    Given two covariance matrices $\Sigma, \Sigma' \in \BR^{d \times d}$, we define the \emph{spectral} distance between $\Sigma$ and $\Sigma'$ to equal $\|\Sigma^{-1/2} \Sigma' \Sigma^{-1/2} - I\|_{op}$.
\end{definition}

Similarly, we have $\|\Sigma^{-1/2} \Sigma' \Sigma^{-1/2} - I\|_{op} = \|\Sigma'^{1/2} \Sigma^{-1} \Sigma'^{1/2} - I\|_{op}$, which are asymptotically equal to $\|\Sigma'^{-1/2} \Sigma \Sigma'^{-1/2} - I\|_{op} = \|\Sigma^{1/2} \Sigma'^{-1} \Sigma^{1/2} - I\|_{op}$ if either is at most $0.5$.

\medskip

Next, we define the notions of flattening and tensor powers.

\begin{definition}[Tensor power]
    Given two vectors $x \in \BR^d, y \in \BR^{d'}$, the tensor product $x \otimes y$ is the vector in $\BR^{d \cdot d'}$, with entries indexed by $(i, j) \in [d] \times [d']$, such that $(x \otimes y)_{ij} = x_i \cdot y_j$.

    We also will define $x^{\otimes 2} := x \otimes x$.
\end{definition}

\begin{definition}[Flattening]
    Given a matrix $M \in \BR^{d \times d'},$ we define the \emph{flattening} $M^\flat$ to be the vector in $\BR^{d \cdot d'}$ with $(M^\flat)_{ij} = M_{i, j}$.
\end{definition}

Note that for any vectors $x, y$, $x \otimes y$ equals $(x y^T)^\flat$.

To represent linear functionals and polynomials, we look at the value of the linear functional over monomials.
\begin{definition}[monomial vector] 
    A monomial vector of degree $d$ is a 
    $n^{O(d)}$-dimensional vector $v_d\paren{x}$ indexed by multisets $S \subseteq \brac{n}$, $\Abs{S} \le d$, where the entry $v_d\paren{x}_S$ is the monomial 
    \begin{equation*}
        v_d\paren{x}_S := \prod_{i \in S} x_i .
    \end{equation*}
    \end{definition}
%note: dimension is technically not n choose <= d because of the multisets.
    
\begin{remark}
    The definition of $n$ for number of variables and $d$ for degree is a slight abuse of notation, as in the rest of the paper $n$ represents the number of data points and $d$ is the dimension of the data points. We will only use the former definition here and in \Cref{sec:computing-score-functions}.
\end{remark}
    
Linear functionals over the set of polynomials of up to degree $d$ over $\R^n$ form an $n^{O(d)}$-dimensional vector space and we can represent them numerically as follows.

\begin{definition}[numerical representation of linear functionals and polynomials]
Suppose $\cL$ is a linear functional over polynomials of up to degree $d$ over $\R^n$. 
We define the representation of $\cL$, $\cR\paren{\cL} \in \R^{n^{O(d)}}$  indexed by multisets $S \subseteq \brac{n}, \abs{S} \le d$, as
\begin{equation*}
\cR\paren{\cL}_S = \cL\paren{v_d\paren{x}_S}.
\end{equation*}
Similarly, for a polynomial $q$, we define its representation $\cR(q) \in \R^{n^{O(d)}}$ to be 
\begin{equation*}
\cR\paren{q}_S = \text{coefficient of $x^S$ in $q$}.
\end{equation*}
In a slight abuse of notation, we may use $\|q\|_2 = \|\cR(q)\|_2$ to denote the norm of the vector of coefficients.
\end{definition}

% define frobenius norm distance and spectral norm distance, or rephrase what i mean when i say these words
% define flattening, tensor power
% define mahalanobis distance
% define quasiconvexity and remove it from other sections
% define sensitivty
%\input{content/private-robust-equiv.tex}
%\input{content/sos-robust-to-private.tex}
\section{A General Private Sampling Algorithm}

In this section, we prove two general theorems showing that if one has a score function corresponding to a robust algorithm for parameter estimation from samples, with a few important properties, then one can construct a differentially private algorithm. The results can either generate a pure-DP algorithm (\Cref{thm:pure_dp_general_main}), or an approx-DP algorithm (\Cref{thm:approx_dp_general_main}), depending on the properties we assume about the robust algorithm.

Assuming the robust algorithm and score function can be computed efficiently, and we have another property that we call \emph{quasi-convexity}, the private algorithms also run in polynomial time. One can also generate analogous statements by removing these assumptions, but the algorithm no longer runs in polynomial time. To avoid rewriting, we color certain parts of Theorems \ref{thm:pure_dp_general_main} and \ref{thm:approx_dp_general_main} in \blue{blue}: one can read the same theorems and ignore what is written in blue to obtain an inefficient private algorithm arising from an inefficient robust algorithm.

We first state our theorem for creating a pure-DP algorithm.

\begin{theorem} \label{thm:pure_dp_general_main}
    Let $0 < \eta, r < 1 < R$ be fixed parameters.
    Suppose we have a score function $\cS(\theta, \cY) \in [0, n]$ that takes as input a dataset $\cY = \{y_1, \dots, y_n\}$ and a parameter $\theta \in \Theta \subset \BR^d$ \blue{(where $\Theta$ is convex and contained in a ball of radius $R$)}, with the following properties:
\begin{itemize}
    \item (Bounded Sensitivity) For any two adjacent datasets $\cY, \cY'$ and any $\theta \in \Theta$, $|\cS(\theta, \cY)-\cS(\theta, \cY')| \le 1.$
    \item \blue{(Quasi-Convexity) For any fixed dataset $\cY$, any $\theta, \theta' \in \Theta$, and any $0 \le \lambda \le 1$, we have that $\cS(\lambda \theta + (1-\lambda) \theta', \cY) \le \max(\cS(\theta, \cY), \cS(\theta', \cY))$.}
    \item \blue{(Efficiently Computable) For any given $\theta \in \Theta$ and dataset $\cY$, we can compute $\cS(\theta, \cY)$ up to error $\gamma$ in $\poly(n, d, \log \frac{R}{r}, \log \gamma^{-1})$ time for any $\gamma > 0$.}
    %\item (Robust algorithm never guesses far away point) Suppose $\cY$ is an $\eta$-corrupted sample of $n$ points drawn with parameter $\theta_0 \in \BR^d$. Then, all points $\theta'$ with $\cS(\theta', \cY) \ge (1 - 3 \eta) n$ are distance at most $\alpha$ from $\theta_0$.
    \item \blue{(Robust algorithm finds low-scoring point) For a given dataset $\cY$, let $T = \min_{\theta_0 \in \Theta} \cS(\theta_0, \cY)$. Then, we can find some point $\theta$ such that for all $\theta'$ within distance $r$ of $\theta$, $\cS(\theta', \cY) \le T+1$, in time $\poly(n, d, \log \frac{R}{r})$.}
    %\item (Volume) Suppose $\cY$ is an $\alpha$-corrupted sample of $n$ points drawn with parameter $\theta_0 \in \BR^d$. Then, there exists a point $\theta$ such that all points $\theta'$ within distance $\rho_2$ from $\mu$ satisfy $\cS(\mu', \cY) \ge (1 - 2 \eta) n$.\todo{This volume condition may be unnecessary because it might be possible to redefine score to maximum score for $\mu'$ close to $\mu$.}
    \item (Volume) For any given dataset $\cY$ and $\eta' \ge \eta$, let $V_{\eta'}(\cY)$ represent the $d$-dimensional volume of points $\theta \in \Theta \subset \BR^d$ with score at most $\eta' n$. (Note that $V_1(\cY)$ is the full volume of $\Theta$).
    %\item (Separation Oracle) If $\cS(\mu, \cY) < T$ there exists a poly-time computable separation oracle that separates $\mu$ from the set of $\mu'$ such that $\cS(\mu', \cY) \ge T$.
\end{itemize}
    Then, 
    %\todo{Might need some additional $d \log R/\eps$ term or so.} 
    we have a pure $\eps$-DP algorithm $\mathcal{A}$ on datasets of size $n$, \blue{that runs in $\poly(n, d, \log \frac{R}{r})$ time,} with the following property. For any dataset $\cY$, if there exists $\theta$ with $\cS(\theta, \cY) \le \eta n$ and if $n \ge \Omega\left(\max\limits_{\eta': \eta \le \eta' \le 1}\frac{\log(V_{\eta'}(\cY)/V_{\eta}(\cY)) + \log (1/(\beta \cdot \eta))}{\eps \cdot \eta'}\right),$ then $\mathcal{A}(\cY)$ outputs some $\theta \in \Theta$ of score at most $2\eta n$ with probability $1-\beta$.
\end{theorem}

We remark that this theorem has several important conditions. The bounded sensitivity of the score is important as it ensures that if we sample according to the exponential mechanism, the sampling probability of any $\theta$ does not change significantly between adjacent datasets. The conditions of quasi-convexity, computability, and finding a low-scoring point are only required for the algorithm to run in polynomial time. Indeed, the latter two of these conditions are important for the robust algorithm to succeed, and the quasi-convexity assumption generalizes a convexity assumption on the score, which roughly corresponds to sampling from log-concave distributions. Finally, the sample complexity is dictated both by the number of samples needed for the robust algorithm to succeed and by bounds on the volume of low versus high scoring points.

%We briefly justify why most Sum-of-Sqaures based robust algorithms fit well in this framework. At a high level, most such algorithms define some auxiliary variables $w_1, \dots, w_n$, where each $w_i \in \{0, 1\}$ represents an indicator of our belief that the sample is uncorrupted. If our hypothesis for the uncorrupted data points is $\cX'$, one desires some natural conditions, such as $w_i = 1$ implies $x_i' = y_i$, and the number of points we change should not be at most $\eta \cdot n$ (if our algorithm can deal with $\eta$-fraction of corrupted samples), meaning $\sum w_i \ge (1-\eta) \cdot n$. The score function alluded to in \Cref{subsec:} would be precisely the minimum number of points that we need to change to obtain a ``reasonable''-looking dataset $\cX'$, such that the algorithm believes $\cX'$ is drawn from a distribution with parameter $\theta$. Note that between adjacent datasets $\cY, \cY'$, their distance to any $\cX'$ changes by at most $1$, so the score function also has sensitivity $1$.
%
%The SoS-based approach produces a convex programming relaxation of the facts that each $w_i \in \{0, 1\}$, that $w_i = 1$ implies $x_i' = y_i$, and that $\cX'$ looks ``reasonable.'' 

\medskip

Along with a general result for pure-DP algorithms, we also prove a similar result for approx-DP algorithms, which we now state.

\begin{theorem} \label{thm:approx_dp_general_main}
    Let $0 < \eta < 0.1$ and $r < 1 < R$ be fixed parameters.
    Suppose we have a score function $\cS(\theta, \cY) \in [0, \infty)$ that takes as input a dataset $\cY = \{y_1, \dots, y_n\}$ and a parameter $\theta \in \Theta \subset \BR^d$ \blue{(where $\Theta$ is convex and contained in a ball of radius $R$)}, with the same properties as in Theorem \ref{thm:pure_dp_general_main}.
    
    %In addition, suppose that there exists some parameter $\eta^* > 10\eta$ such that for any dataset $\mathcal{Z},$ the $d$-dimensional volume of points with score at most $\eta^* n$ with respect to $\mathcal{Z}$ is at most some $W$, and if there exists some point with score at most $0.7 \eta^* n$, then the volume of points with score at most $0.8 \eta^* n$ is at least some $W'$.

    In addition, fix some parameter $\eta^* \in [10 \eta, 1]$.
    Suppose that $n \ge \Omega\left(\frac{\log(1/\delta) + \log (V_{\eta^*}(\cY)/V_{0.8 \eta^*}(\cY))}{\eps \cdot \eta^*}\right)$ for all $\cY$ such that there exists $\theta$ with $\cS(\theta, \cY) \le 0.7 \eta^* n$. Then,  we have an $(\eps, \delta)$-DP algorithm $\mathcal{A}$ \blue{that runs in $\poly(n, d, \log \frac{R}{r})$ time}, such that for any dataset $\cY$, if there exists $\theta$ with $\cS(\theta, \cY) \le \eta n$ and if $n \ge \Omega\left(\max\limits_{\eta': \eta \le \eta' \le \eta^*}\frac{\log(V_{\eta'}(\cY)/V_{\eta}(\cY)) + \log (1/(\beta \cdot \eta))}{\eps \cdot \eta'}\right)$, then $\mathcal{A}(\cY)$ outputs some $\theta \in \Theta$ of score at most $2 \eta n$ with probability $1-\beta$.
    %\todo{Figure out correct sample complexity}
\end{theorem}

The main difference in the approx-DP setting is that we set some threshold $\eta^*$, and only consider volumes of points of score up to $\eta^* \cdot n$. This is because, roughly speaking, we will sample using a truncated exponential mechanism until score roughly $\eta^* n$. (In reality, we need to be more careful about how we truncate.) But because of this truncation, the volume bound will be crucial for not only bounding sample complexity but also ensuring privacy, to make sure the probability of sampling a point near the threshold score is low.

We will only prove Theorems \ref{thm:pure_dp_general_main} and \ref{thm:approx_dp_general_main} for the efficient case. In the proofs, one can verify that the requirements of quasi-convexity, efficient computability, and efficiently finding a low-scoring point, as well as the promise that $\Theta$ is convex and bounded, are only needed for our sampling algorithms to run in polynomial time. Hence, the inefficient algorithm results also follow.

\subsection{Sampling and volume computation with an imperfect oracle}

To prove the main results of this section, we heavily rely on the theory of sampling and volume computation for convex bodies, given only membership oracle access (as opposed to membership and separation oracle access). While one may wish to directly apply these techniques, we cannot afford to do so, because, to the best of our knowledge, all such results have been written assuming infinite-precision arithmetic and perfect membership oracles. In our setting, we must show such results are possible even if we only have bounded precision arithmetic and imperfect membership oracles. 
This will be crucial because we assume we cannot perfectly compute the score function, but can only approximately compute it.
We now formally define approximate membership oracles.

\begin{definition}
    Given two nested convex bodies $K_1 \subset K_2$, a $(K_1, K_2)$\textbf{-membership oracle} $\mathcal{O}$ is an oracle that, if given an input $x \in K_1$, outputs YES; if given an input $x \not\in K_2$, outputs NO; and if given an input $x \in K_2 \backslash K_1$, may output either YES or NO.
\end{definition}

In addition, we will wish for multiplicative approximations for the sake of pure-DP, meaning each point (in a sufficiently fine net) in the convex body should be sampled in a way that is point-wise close to uniform, as opposed to close to uniform in total variation distance. While one could use the techniques of~\cite{MangoubiV22} to achieve the point-wise guarantee, they still make an assumption of using perfect membership oracles and infinite-precision arithmetic.

To deal with the issues of precision and imperfect oracles, we apply the known analyses of hit-and-run sampling, made discrete in an appropriate fashion, and make sure that the probability of ever being near the boundary of the convex body, where the membership oracle may be incorrect, is low. To ensure the multiplicative approximation, we make a final step where we slightly perturb and then discretize the sample further, and show that this is sufficient. Since most of the analysis derives from known results, we defer the proofs to \Cref{appendix:sampling}, and here we simply state the results we need.

\begin{lemma}{\textbf{(Main convex body sampling lemma)}} \label{lem:sampling_main}
    Fix any parameter $\gamma_6 \le d^{-100}$ and $r < 1 < R$.
    Let $K_1, K_2$ be convex bodies such that $B(\textbf{0}, r) \subset K_1 \subset K_2 \subset B(\textbf{0}, R),$ and $\vol(K_2)-\vol(K_1) \le \left(\frac{\gamma_1 \cdot r}{6d}\right)^d,$ for some $\gamma_1$ such that $\log \gamma_1^{-1} = \poly(d, \log \frac{R}{r}, \log \gamma_6^{-1})$. Suppose we have a $(K_1, K_2)$-membership oracle $\mathcal{O}$. Then, in $\poly(d, \log \frac{R}{r}, \log \gamma_6^{-1})$ time and queries to $\mathcal{O}$, we can output a point $z$ that is $(1 \pm \gamma_6)$-pointwise close to uniform on the set of points in $\BR^d$ with all coordinates integer multiples of $\gamma_5$ that are accepted by $\mathcal{O}$, for $\gamma_5 = \frac{r \cdot \gamma_6}{d^3}.$
\end{lemma}

\begin{lemma}{\textbf{(Volume computation)}} \label{lem:volume_computation}
    Set $\gamma_6 = \frac{\eps}{d^{100} \log (R/r)}$, and set $\gamma_1, \gamma_5$, along with $r, R, K_1, K_2, \mathcal{O}$, as in Lemma \ref{lem:sampling_main}. Fix any $\eps < 0.5$. Then, for any $\gamma < 1$, in $\poly(d, \log \frac{R}{r}, \frac{1}{\eps}, \log \gamma^{-1})$ time and oracle accesses, we can approximate the number of points in $\BR^d$ with all coordinates integer multiples of $\gamma_5$ that are accepted by $\mathcal{O}$, up to a $1 \pm \eps$ multiplicative factor, with failure probability $\gamma$.
\end{lemma}

\begin{remark}
    Our parameters skip to $\gamma_5$ and $\gamma_6$ since we define auxiliary parameters $\gamma_2, \gamma_3, \gamma_4$ in the proofs of Lemmas \ref{lem:sampling_main} and \ref{lem:volume_computation}.
\end{remark}

\subsection{Proof of \Cref{thm:pure_dp_general_main}}

Our algorithm will roughly sample each $\theta$ based on the exponential mechanism, where each $\theta$ is sampled proportional to $e^{-\eps \cdot \cS(\theta, \cY)}$. In the following lemma, we apply Lemmas \ref{lem:sampling_main} and \ref{lem:volume_computation} to obtain a desired sampling procedure.

We note the following fact, which we prove in \Cref{appendix:sampling} (see \Cref{fact:count_volume_equivalence_appendix}).

\begin{fact} \label{fact:count_volume_equivalence}
    %Suppose $K \subset \BR^d$ is a convex body that contains a ball of radius $r$. Suppose $\gamma$ is a parameter which is at most $\frac{r}{d^3}$. Then, the number of points in $K$ that have all coordinates integral multiples of $\gamma$ is $(1 \pm O(1/d)) \cdot \vol(K)/\gamma^d$.
    Suppose $K \subset \BR^d$ is a convex body that contains a ball of radius $r$. Suppose $\gamma$ is a parameter which is at most $\frac{r}{2 d^3}$. Then, the number of points $N$ in $K$ that have all coordinates integral multiples of $\gamma$ is $(1 \pm O(\frac{\gamma}{r} \cdot d^2)) \cdot \vol(K)/\gamma^d$.
\end{fact}
%\todo{Check proof of Lemma A.12, extract as separate lemma}

\begin{lemma} \label{lem:volume_perfect_sampling}
    Set $\gamma_6 = \frac{\eps}{d^{100} \log (R/r)}$ and set $\gamma_1, \gamma_5$ as in Lemma \ref{lem:volume_computation}. Let $\tilde{\Theta}$ be the set of points in $\Theta$ with all coordinates integral multiples of $\gamma_5$.
    Then, in $\poly(n, d, \frac{1}{\eps}, \log \frac{R}{r})$ time, we can sample from each $\theta \in \tilde{\Theta}$ with probability proportional to $e^{-\eps \cdot \cS(\theta, \cY)} \cdot e^{\pm O(\eps)}$.
\end{lemma}

\begin{proof}
    First, we define $\gamma_8 := \frac{\eps}{2} \cdot e^{-n} \cdot (\gamma_5/2R)^d$ and $\gamma_7 := \gamma_8 \cdot \left(\frac{\gamma_1 \cdot r}{6d}\right)^d/(2(2R)^d)$.
    
    We now describe our algorithm. Define $T := \min_{\theta_0 \in \Theta} \cS(\theta_0, \cY)$. Even if $T$ is unknown, in time $\poly(n, d, \log \frac{R}{r})$, we can find some point $\theta$ such that $\cS(\theta', \cY) \le T+1$ for all $\theta'$ within $r$ of $\theta$. We can also get some estimate $T'$ between $T$ and $T+1$.
    We pick a uniformly random number $\hat{T}$ between $T'+1$ and $T'+2$ that is an integral multiple of $\gamma_7$. Note that $\hat{T} \le \min_{\theta_0 \in \Theta} \cS(\theta_0, \cY)+3.$ Now, for any point $\theta \in \tilde{\Theta}$, let $t(\theta)$ be the smallest nonnegative integer $t$ such that the \textbf{estimate} (where the estimate has accuracy $\gamma_7$) of the score $\cS(\theta, \cY)$ is at most $\hat{T}+t$. (Note that if the estimate is less than $\hat{T}$, then $t = 0$.)
    
    Our goal will be to produce a sample that is $e^{\pm O(\eps)}$-pointwise close to the distribution proportional to $e^{-\eps \cdot t(\theta)}$.
    For each integer $t \ge 0$, define $K_1^{(t)}$ to be the convex body of points in $\Theta$ with (true) score at most $\hat{T}+t-\gamma_7$, and $K_2^{(t)}$ to be the convex body of points in $\Theta$ with (true) score at most $\hat{T}+t+\gamma_7$. Note that $K_1^{(t)}$ and $K_2^{(t)}$ are convex by the assumption that the score is quasi-convex.
    We will apply Lemmas \ref{lem:sampling_main} and \ref{lem:volume_computation}, with $\cO$ as the $(K_1^{(t)}, K_2^{(t)})$-oracle that accepts if the estimate of the score is at most $\hat{T}+t$, i.e., if $t(\theta) \le t$.
    (Note that while $K_1^{(t)}$ may not contain $\textbf{0}$, it contains a ball of radius $r$ around an efficiently computable point $\theta$, which is sufficient.)
    Also, let $S^{(t)}, N^{(t)}$ be the set of and number of points in $\tilde{\Theta}$, respectively, such that $t(\theta) \le t$.
    Since $t(\theta) \in \{0, 1, \dots,n\}$, we can write $\sum_{\theta \in \tilde{\Theta}} e^{-\eps t(\theta)} = N^{(0)} + \sum_{t = 1}^{n} e^{-\eps t} (N^{(t)}-N^{(t-1)}) = \sum_{t = 0}^{n-1} \left(e^{-\eps t} (1-e^{-\eps}) N^{(t)}\right) + e^{-\eps n} N^{(n)}$. Assuming that $\vol(K_2^{(t)})-\vol(K_1^{(t)}) \le \left(\frac{\gamma_1 \cdot r}{6d}\right)^d$ for all $t$, then we can provide a $e^{\pm \eps}$-factor approximation $\tilde{N}^{(t)}$ for each $N^{(t)}$, with failure probability at most $\gamma_8$, in time $\poly(d, \log \frac{R}{r}, \frac{1}{\eps}, \log \gamma_8^{-1})$, by \Cref{lem:volume_computation}.
    
    %For each nonnegative integer $0 \le t \le n$, let $K_1^{(t)}, K_2^{(t)}$ represent the set of points with score at most $T+t-\gamma_7$ and $T+t+\gamma_7$, respectively, and let $\mathcal{O}_t$ represent the oracle that accepts a point $\theta$ if the algorithm computing $\cS(\theta, \cY)$ up to error $\gamma_7$.

    Our final algorithm will sample each number $t \in \{0, 1, \dots, n-1\}$ with probability proportional to $e^{-\eps t} (1-e^{-\eps}) \tilde{N}^{(t)}$ and $t = n$ with probability proportional to $e^{-\eps n} \tilde{N}^{(n)}$. Then, we use \Cref{lem:sampling_main} to sample ($1 \pm \gamma_6)$-pointwise close to uniform from the set $S^{(t)}$ in time $\poly(d, \log \frac{R}{r}, \log \gamma_6^{-1})$.
    
    Overall, assuming that $\vol(K_2^{(t)})-\vol(K_1^{(t)}) \le \left(\frac{\gamma_1 \cdot r}{6d}\right)^d$ for all $t$, since $\gamma_6 < \eps$, we obtain an $e^{\pm 2\eps}$-pointwise approximation to sampling from the distribution proportional to $e^{-\eps \cdot t(\theta)}$ for $\theta \in \tilde{\Theta}$, with failure probability at most $\gamma_8$. In addition, note that $t(\theta) = \cS(\theta, \cY) - T \pm O(1)$ holds for all $\theta \in \tilde{\Theta}$, which means in fact we are sampling proportional to $e^{-\eps \cdot \cS(\theta, \cY)}$ up to a $e^{\pm O(\eps)}$ pointwise approximation. There are two ways for this to fail: if either there is some $t$ with $\vol(K_2^{(t)})-\vol(K_1^{(t)}) > \left(\frac{\gamma_1 \cdot r}{6d}\right)^d$ or in the $\gamma_8$ probability event that some estimate $\tilde{N}^{(t)}$ is incorrect. Note however, that this volume represents the set of points with score between $T'+t+1+u-\gamma_7$ and $T'+t+1+u+\gamma_7$, where $u \in [0, 1)$ is chosen at random to be an integer multiple of $\gamma_7$. Therefore, the expectation $\E_u[\vol(K_2^{(t)})-\vol(K_1^{(t)})]$ is at most $2 \cdot \gamma_7$ times the volume difference of points with score at least $T'+t+2$ and $T'+t+1,$ which is at most $\vol(\Theta) \le (2R)^d$. So, by Markov's inequality, $\vol(K_2^{(t)})-\vol(K_1^{(t)}) > \left(\frac{\gamma_1 \cdot r}{6d}\right)^d$ with probability at most $\left(2 \gamma_7 \cdot (2R)^d\right)/\left(\frac{\gamma_1 \cdot r}{6d}\right)^d = \gamma_8$. 
    
    Therefore, with probability at least $1-2\gamma_8$, we are sampling $\theta \in \tilde{\Theta}$ from a distribution proportional to $e^{-\eps \cdot \cS(\theta, \cY)} \cdot e^{\pm O(\eps)}$.
    However, note that the number of points in $\tilde{\Theta}$ is at most $\vol(\Theta)/(\gamma_5)^d \cdot (1+o(1)) \le (2R/\gamma_5)^d$ by \Cref{fact:count_volume_equivalence}, so each point in $\tilde{\Theta}$ is selected with probability at least $\Omega(e^{-n} \cdot (\gamma_5/2R)^d)$. So, since we set $\gamma_8 = \frac{\eps}{2} \cdot e^{-n} \cdot (\gamma_5/2R)^d$, we are still sampling each element with probability proportional to $e^{-\eps \cdot \cS(\theta, \cY)} \cdot e^{\pm O(\eps)}$.
\end{proof}

\begin{proof}[Proof of Theorem \ref{thm:pure_dp_general_main}]
    The algorithm is the same as in Lemma \ref{lem:volume_perfect_sampling}. 
    %$e^{\pm O(\eps)}$-pointwise approximate sample from the distribution that is proportional to $e^{-\eps \cdot \cS(\theta, \cY)}$ for each $\theta \in \tilde{\Theta}$. 
    To see why this implies a private algorithm, for any two adjacent datasets $\cY, \cY',$ the score of any point changes by at most $1$, which means the distribution does not change by more than a $e^{\pm O(\eps)}$ factor multiplicatively for any fixed $\theta$ between $\cY$ and $\cY'$. So, if we could approximately sample from this distribution, the distribution still does not change by more than a $e^{\pm O(\eps)}$ factor multiplicatively. This ensures the algorithm will be $O(\eps)$-DP.
    
    The runtime has already been verified, with the fact that $n \ge \Omega(1/\eps)$ is already known, so we can ignore polynomial runtime dependencies on $\frac{1}{\eps}$.
    
    Finally, we check accuracy. Assume there exists a $\theta \in \Theta$ with score at most $\eta n$. By Fact \ref{fact:count_volume_equivalence}, if $\gamma_5 \le \frac{r \cdot \gamma_6}{d^3},$ then for any convex body $K$ containing a ball of radius $r$, $\vol(K) = (1 \pm o(1)) \cdot (\gamma_5)^d \cdot N_K$ if $N_K$ is the number of points in $K \cap \tilde{\Theta}$. Now, for any $j \ge 1$, we bound the probability that we select a $\theta \in \tilde{\Theta}$ with score between $2^j \cdot \eta n$ and $2^{j+1} \cdot \eta n$. If we consider $K_j$ to be the convex body of points in $\Theta$ with score at most $2^{j+1} \cdot \eta n$, then the probability of sampling a point with a score between $2^j \cdot \eta n$ and $2^{j+1} \cdot \eta n$ is proportional to at most $e^{-\eps \cdot 2^j \cdot \eta n + O(\eps)} \cdot \vol(K_j)/(\gamma_5)^d \cdot (1+o(1))$. However, the set of points with score at most $\eta n + 1$ contains a ball of radius $r$, so the probability of sampling such a point is proportional to at least $e^{-\eps \cdot (\eta n + 1) - O(\eps)} \cdot V_{\eta}/(\gamma_5)^d \cdot (1-o(1))$.
    
    So, to select a point with score at most $2 \eta n$ with probability $1-O(\beta)$, it suffices to check that $\sum_{j = 1}^{\lceil \log_2 (1/\eta) \rceil} e^{-\eps \cdot (2^j-1) \cdot \eta n} \cdot \vol(K_j)/V_\eta \le \beta$. Now, by setting $\eta' = 2^{j+1} \cdot \eta$, we have that $\vol(K_j)/V_{\eta} = V_{\eta'}/V_{\eta},$ and $e^{-\eps \cdot (2^j-1) \cdot \eta n} \le e^{-\eps \cdot \eta' \cdot n/4}$. Thus, if $n \ge 8 \log(V_{\eta'}/V_{\eta}) \cdot \frac{1}{\eta' \cdot \eps}$, then $e^{-\eps \cdot \eta' \cdot n/8} \le \frac{V_{\eta}}{V_{\eta'}}.$ Also, if $n \ge \frac{8\log (1/(\beta \cdot \eta))}{\eps \cdot \eta'}$, then $e^{-\eps \cdot \eta' \cdot n/8} \le \beta \cdot \eta$. Therefore, $e^{-\eps \cdot \eta' \cdot n/4} \le \frac{V_{\eta}}{V_{\eta'}} \cdot \beta \cdot \eta$, which means $\sum_{j = 1}^{\lceil \log_2 (1/\eta) \rceil} e^{-\eps \cdot (2^j-1) \cdot \eta n} \cdot \frac{\vol(K_j)}{V_{\eta}} \le \sum_{j = 1}^{\lceil \log_2 (1/\eta) \rceil} \frac{V_{\eta}}{\vol(K_j)} \cdot \beta \cdot \eta \cdot \frac{\vol(K_j)}{V_{\eta}} \le \beta$. Thus, the algorithm is accurate with $1-\beta$ probability.
\end{proof}

\subsection{Proof of \Cref{thm:approx_dp_general_main}}

In this subsection, we prove Theorem \ref{thm:approx_dp_general_main}.
We start by describing the algorithm. 

First, we define the function $g: \BZ \to [0, 1]$ as follows. First, for $t < 0.3 \eta^* n$, we let $g(t) = 1$, and for $t \ge 0.7 \eta^* n$, we let $g(n) = 0$. For $0.3 \cdot \eta^* n \le t \le 0.5 \eta^* n$, we let $g(t) = \max\left(\frac{1}{2}, 1-\delta \cdot e^{\eps(t - 0.3 \eta^* n)}\right)$, and for $0.5 \eta^* n \le t \le 0.7 \eta^* n$, we let $g(t) = \min\left(\frac{1}{2}, \delta \cdot e^{\eps(0.7 \eta^* n - t)}\right)$. The first step of the algorithm is to compute some $\hat{T}$, which equals $\min_{\theta \in \Theta} \cS(\theta, \cY)$ up to additive error $1$. The first part of the algorithm, which we call $\mathcal{A}_1$, will \emph{accept} the dataset $\cY$ with probability $g(\hat{T})$.

If $\mathcal{A}_1$ rejects $\cY$, the overall algorithm $\mathcal{A}$ outputs nothing. If $\mathcal{A}_1$ accepts $\cY,$ the algorithm proceeds to the second phase. The second phase, at a high level, attempts to sample a $\theta$ proportional to $e^{-\eps \cdot \cS(\theta, \cY)}$ as long as $\cS(\theta, \cY) \le 0.9 \eta^* n$. This may be impossible as we cannot perfectly compute $\theta$. Instead, if we define the function $h(t)$ to equal $e^{-\eps t}$ for $t \le 0.9 \eta^* n$ and $0$ for $t > 0.9 \eta^* n$, we prove the following.

\begin{lemma} \label{lem:approx_dp_sampling}
    Let $\tilde{\Theta}$ be as in Lemma \ref{lem:volume_perfect_sampling}. Then, in time $\poly(n, d, \frac{1}{\eps}, \log \frac{R}{r})$ time and with failure probability at most $\min(\beta, \delta)$, we can sample from each $\theta \in \tilde{\Theta}$ with probability proportional to $h'(\theta)$, where $h'(\theta)$ is a function satisfying $h(\cS(\theta, \cY)+O(1)) \cdot e^{-O(\eps)} \le h'(\theta) \le h(\cS(\theta, \cY)-O(1)) \cdot e^{O(\eps)}$.
\end{lemma}

\begin{proof}
    The proof is nearly identical to that of Lemma \ref{lem:volume_perfect_sampling}. We compute $\hat{T}$ the same way as in Lemma \ref{lem:volume_perfect_sampling}, and again we define $t(\theta)$ to be the smallest nonnegative $t$ such that our estimate of $\cS(\theta, \cY)$ is at most $\hat{T}+t$. This time, rather than approximately sampling with probability proportional to $e^{-\eps t(\theta)}$, we approximately sample proportional to $h(t(\theta)+\hat{T})$. (Note that $t(\theta)+\hat{T} = \cS(\theta, \cY) \pm O(1)$.) As in Lemma \ref{lem:volume_perfect_sampling}, we define $S^{(t)}, N^{(t)}$ to be the set of and number of points in $\tilde{\Theta}$, respectively, such that $t(\theta) \le t$.
    We can write $\sum_{\theta \in \tilde{\Theta}} h(t(\theta)+\hat{T}) = h(\hat{T}) \cdot N^{(0)} + \sum_{t = 1}^{n} h(\hat{T}+t) (N^{(t)}-N^{(t-1)}) = \sum_{t = 0}^{n-1} \left(h(\hat{T}+t)-h(\hat{T}+t+1)\right) N^{(t)},$ since $h(\hat{T}+n) = 0$.
    
    Again, we can compute each $N^{(t)}$ up to a $e^{\pm \eps}$ multiplicative factor (to get estimates $\tilde{N}^{(t)}$, choose $t$ proportional to $\tilde{N}^{(t)} \cdot (h(\hat{T}+t)-h(\hat{T}+t+1))$, and then sample $(1 \pm \gamma_6)$-pointwise close to uniform on $S^{(t)}$, where we ensure $\gamma_6 \le \eps$. The algorithm fails with probability $2 \gamma_8$. This time we cannot charge this error to multiplicative error (since some points may have large enough score that they will be sampled with probability $0$), so we additionally ensure that $\gamma_8 \le \frac{\min(\beta, \delta)}{2}$ as well. (I.e., we set $\gamma_8 := \min\left(\frac{\beta}{2}, \frac{\delta}{2}, \frac{\eps}{2} \cdot e^{-n} \cdot (\gamma_5/2R)^d\right)$.) Note that as long as $n \ge \log \delta^{-1} + \log \beta^{-1}$, the runtime is still $\poly(n, d, \frac{1}{\eps}, \log \frac{R}{r})$.
\end{proof}

\begin{proof}[Proof of Theorem \ref{thm:approx_dp_general_main}]
    The algorithm is as described, where the second phase samples (with failure probability at most $\min(\beta, \delta)$) proportional to $h'(\theta)$. We recall that by \Cref{fact:count_volume_equivalence}, for the convex body $N^{(t)}$ of points with score at most $t+\hat{T}$, $\vol(K^{(t)}) = (1\pm o(1)) \cdot \gamma_5^d \cdot N^{(t)}$.
    
    First, we check privacy. It is clear that the first phase of the algorithm is $(O(\eps), O(\delta))$-DP as long as $n \gg \frac{\log \delta^{-1}}{\eps \cdot \eta^*}$, since for any two adjacent datasets, $\max_{\theta \in \Theta} \cS(\theta, \cY)$ changes by at most $1$, and our estimate for this maximum is accurate up to error $1$. So, $\hat{T}$ changes by at most $O(1)$ between adjacent datasets $\cY$ and $\cY'$, which is sufficient. For the second phase, we sample each $\theta$ proportional to $e^{-\eps \cdot (\cS(\theta, \cY) \pm O(1))}$ if $\cS(\theta, \cY) \le 0.9 \eta^* n - O(1)$, and proportional to $0$ if $\cS(\theta, \cY) \ge 0.9 \eta^* n + O(1)$. So, the sampling probability stays proportional between adjacent datasets, unless $\cS(\theta, \cY) = 0.9 \eta^* n \pm O(1)$. So, we need to make sure the probability of sampling such a dataset is at most $O(\delta)$, so that the overall algorithm is $(O(\eps), O(\delta))$-DP.
    
    To see why this is true, the probability of sampling a point $\theta \in \tilde{\Theta}$ with score in the range $0.9 \eta^* n \pm O(1)$ is proportional to at most $e^{-\eps \cdot (0.9 \eta^* n - O(1))} (1+o(1)) \cdot \gamma_5^{-d} \cdot V_{0.9 \eta^* + O(1/n)}(\cY) \le O(1) \cdot e^{-\eps \cdot 0.9 \eta^* n} \cdot \gamma_5^{-d} \cdot V_{\eta^*}(\cY)$. Conversely, since we didn't reject, we know that $\hat{T} \le 0.7 \eta^* n$, which means the probability that we sample a point $\theta \in \tilde{\Theta}$ with score at most $0.85 \eta^* n$ is proportional to at least $e^{-\eps \cdot (0.85 \eta^* n + O(1))} (1-o(1)) \cdot \gamma_5^{-d} \cdot V_{0.85 \eta^* - O(1/n)}(\cY) \ge \Omega(1) \cdot e^{-\eps \cdot 0.85 \eta^* n} \cdot \gamma_5^{-d} \cdot V_{0.8 \eta^*}(\cY)$. So, it suffices to show that $\frac{e^{-\eps \cdot 0.9 \eta^* n} \cdot \gamma_5^{-d} \cdot V_{\eta^*}(\cY)}{e^{-\eps \cdot 0.85 \eta^* n} \cdot \gamma_5^{-d} \cdot V_{0.8 \eta^*}(\cY)} \le \delta$, which is true as long as $n \ge \frac{\log (1/\delta) \cdot \log(V_{\eta^*}(\cY)/V_{0.8 \eta^*}(\cY))}{\eps \cdot \eta^*}$. Note that this only has to be true for datasets $\cY$ such that $\min_{\theta} \cS(\theta, \cY) \le 0.7 \eta^* n$, since otherwise the algorithm would have already rejected in the first phase.
    
    To check efficiency, note that we can compute $\hat{T}$ in $\poly(n, d, \log \frac{R}{r})$ time, using the condition that the robust algorithm can find a low-scoring point up to error $1$. Then, in time $\poly(n, d, \log \frac{R}{r})$, since $n \ge \Omega(\frac{1}{\eps})$, we can sample proportional to $h'(\theta)$, using Lemma \ref{lem:approx_dp_sampling}.
    
    Checking accuracy will be very similar to as in the proof of Theorem \ref{thm:pure_dp_general_main}. Suppose there exists $\theta$ such that $\cS(\theta, \cY) \le \eta n$. Then, $\hat{T} \le \eta n + O(1) \le 0.1 \eta^* n + O(1)$, which means the first part of the algorithm will succeed. For the second phase, we sample each $\theta \in \tilde{\Theta}$ with probability proportional to $h'(\theta)$, with failure probability at most $\beta$. The probability of sampling a point with score at most $\eta n + 1$ is proportional to at least $e^{-\eps \cdot (\eta n + O(1))} \cdot V_\eta(\cY)/(\gamma_5)^d \cdot (1-o(1))$. Also, if we define $K_j$ to be the convex body of points in $\Theta$ with score at most $2^{j+1} \cdot \eta n$, then the probability of selecting a $\theta \in \tilde{\Theta}$ with score between $2^j \cdot \eta n$ and $2^{j+1} \cdot \eta n$ is proportional to at most $e^{-\eps \cdot (2^j \cdot \eta n - O(1)))} \cdot \min(\vol(K_j), V_{\eta^*}(\cY))/(\gamma_5)^d \cdot (1+o(1)),$ by Lemma \ref{lem:approx_dp_sampling} and the definition of $h$.
    
    Hence, we wish to check that $\sum_{j=1}^{\lceil \log_2(\eta^*/\eta) \rceil} e^{-\eps \cdot (2^j - 1) \cdot \eta n} \cdot \min(\vol(K_j), V_{\eta^*}(\cY))/V_\eta(\cY)$ is at most $\beta$: it suffices to show that $e^{-\eps \cdot 2^j \cdot \eta n/2} \cdot \min(\vol(K_j), V_{\eta^*}(\cY))/V_\eta(\cY) \le \beta \cdot \eta$ for any $1 \le j \le \lceil \log_2(\eta^*/\eta) \rceil$. If $2^{j+1} \cdot \eta \le \eta^*$, then by setting $\eta' = 2^{j+1}$, we are assuming that $n \ge \Omega\left(\frac{\log(V_{\eta'}(\cY)/V_{\eta}(\cY)) + \log (1/(\beta \cdot \eta))}{\eps \cdot \eta'}\right)$. This means $e^{-\eps \cdot 2^j \cdot \eta n/2} \cdot \min(\vol(K_j), V_{\eta^*}(\cY))/V_\eta(\cY) \le e^{-\eps \cdot \eta' n/4}\cdot \vol(K_j)/V_\eta(\cY) \le \frac{V_\eta(\cY)}{V_{\eta'}(\cY)} \cdot \beta \cdot \eta \cdot \frac{\vol(K_j)}{V_\eta(\cY)} = \beta \cdot \eta$. If $2^{j+1} \cdot \eta > \eta^*$, then by setting $\eta' = \eta^*$, we are assuming that $n \ge \Omega\left(\frac{\log(V_{\eta^*}(\cY)/V_{\eta}(\cY)) + \log (1/(\beta \cdot \eta))}{\eps \cdot \eta^*}\right)$. This means $e^{-\eps \cdot 2^j \cdot \eta n/2} \cdot \min(\vol(K_j), V_{\eta^*}(\cY))/V_\eta(\cY) \le e^{-\eps \cdot \eta^* n/4}\cdot V_{\eta^*}(\cY)/V_\eta(\cY) \le \frac{V_\eta(\cY)}{V_{\eta^*}(\cY)} \cdot \beta \cdot \eta \cdot \frac{V_{\eta^*}(\cY)}{V_\eta(\cY)} = \beta \cdot \eta$. Hence, the algorithm is accurate.
\end{proof}
\section{Estimating the Mean of a Gaussian}
\label{sec:mean}
\subsection{Main Theorem}

Our main theorem in this section is a polynomial time and pure-DP algorithm for private mean estimation of an identity-covariance Gaussian, with optimal sample complexity.

\begin{theorem}[Private Mean Estimation of a (Sub-)Gaussian] 
\label{thm:gaussian-mean-main}
Assume that $0 < \alpha, \beta, \eps < 1$ and $R > 0$.
Let $\mu \in \R^d$, where $\normt{\mu} \le R$, be unknown.  There is an $\epsilon$-DP algorithm that takes $n$ \iid\ samples from $\cN\paren{\mu, I}$ (or in general, a subgaussian distribution with mean $\mu$ and covariance $I$) and with probability $1-\beta$ outputs $\mu$ such that $\normt{\mu - \muhat} \le \alpha$, where
\begin{equation*}
n =
\tcO \Paren{
\frac{d + \log (1/\beta)}{\alpha^2}
+
\frac{d + \log (1/\beta)}{\alpha \eps}
+
\frac{d \log R}{\eps}
}.
\end{equation*}
Here, $\tcO$ only hides logarithmic
factors in $1/\alpha$.
Moreover, this algorithm runs in time $\poly(n, d)$, and succeeds with the same accuracy even if $\eta = \tilde{\Omega}(\alpha)$ fraction of the samples are adversarially corrupted, assuming $\eta \le \eta^*$ for some universal constant $\eta^*$.
\end{theorem}

For pure-DP algorithms, the $\frac{d \log R}{\eps}$ term is required by a standard packing lower bound.
However, in the approximate-DP setting, we can replace this term with $\frac{\log (1/\delta)}{\eps}$, as we now state.
%Indeed, Theorem \ref{thm:gaussian-mean-main}, along with a result of \cite{EsfandiariMN22} (alternatively, one could use the results of \cite{GhaziKM21, TsfadiaCKMS22}), implies such an approximate-DP result also with optimal sample complexity, that we now state.

\begin{theorem}[Private Mean Estimation of a (Sub-)Gaussian with Approx-DP] \label{cor:gaussian-mean-main}
Let $\mu \in \R^d$, where $\normt{\mu} \le R$, be unknown. There is an $(\epsilon, \delta)$-DP algorithm that takes $n$ \iid\ samples from $\cN\paren{\mu, I}$ (or a subgaussian distribution with mean $\mu$ and covariance $I$) and with probability $1-\beta$ outputs $\mu$ such that $\normt{\mu - \muhat} \le \alpha$, where
\begin{equation*}
n =
\tcO \Paren{
\frac{d + \log (1/\beta)}{\alpha^2}
+
\frac{d + \log (1/\beta)}{\alpha \eps}
+
\frac{\log (1/\delta)}{\eps}
}.
\end{equation*}
Moreover, this algorithm runs in time $\poly(n, d, \log R)$, and still succeeds with the same accuracy even if $\eta = \tilde{\Omega}(\alpha)$ fraction of the samples are adversarially corrupted.
\end{theorem}

Note that the runtime dependence on $\log R$ is required as even reading the input up to $O(1)$-precision requires $\log R$ time.
% However, we remark that this is the only place that this dependence is required, and that the sample complexity has no dependence on $\log R$.
%
%To explain why \Cref{cor:gaussian-mean-main} follows from \Cref{thm:gaussian-mean-main}, we use \cite[Corollary 5]{EsfandiariMN22}, which uses $\log \frac{1}{\delta} \log \frac{1}{\beta}$ samples and polynomial time to learn $\mu$ up to error $O(d)$, assuming at least $2/3$ of the samples are in an unknown ball of radius $O(\sqrt{d})$. From here, we just need to learn $\mu$ given that it is in a known ball of radius $R$, so we may apply \Cref{thm:gaussian-mean-main} with $R = d$ to obtain our desired result. We could alternatively use our approx-DP technique (Theorem \ref{thm:approx_dp_general_main}) to improve the final term to $\frac{\log (1/\delta)}{\eps}$ (i.e., removing the $\log (1/\beta)$), though for simplicity we do not perform this analysis.
%
%The rest of this section is entirely devoted to proving \Cref{thm:gaussian-mean-main}.

We note that one could alternatively prove \Cref{cor:gaussian-mean-main} by combining \Cref{thm:gaussian-mean-main} with \cite[Corollary 5]{EsfandiariMN22} (or alternatively \cite{GhaziKM21, TsfadiaCKMS22}), which allows us to learn $\mu$ up to radius $O(d)$ first. However, this method is slightly suboptimal in that the final term would be $\frac{\log(1/\delta) \cdot \log(1/\beta)}{\eps}$.

The rest of this section is devoted to proving \Cref{thm:gaussian-mean-main} and \Cref{cor:gaussian-mean-main}.

\subsection{Resilience of First and Second Moments}

In this subsection, we note some known concentration inequalities for subgaussian random variables (commonly known as resilience or stability conditions) that will be crucial for our analysis.

\begin{lemma}[Resilience of First and Second Moments, Proposition 3.3 in \cite{DiakonikolasK22}]
\label{lem:resilience-of-moments}
Let 
$n \ge 
O\paren{\paren{d + \log \paren{1/\beta}}/{\alpha^2}}$, for some $\alpha = \tcO(\eta)$.
Let $\set{x_i}_{i=1}^{n} \overset{i.i.d.}{\sim} \cD$, where $\cD$ is a subgaussian random variable with mean $\mu \in \R^d$ and covariance $I$.
Then, with probability $1-\beta$, for all vectors $b \in \brac{0, 1}^n$ such that $\E_i b_i \ge 1-\eta$ and all unit vectors $v \in \R^d$, we have
\begin{equation*}
\left|\E_i b_i \langle v, x_i - \mu \rangle\right| \le \alpha.
\end{equation*}
In addition,
\begin{equation*}
\left|\E_i b_i \langle v, x_i - \mu \rangle^2 - 1 \right| \le \alpha.
\end{equation*}
\end{lemma}

\begin{corollary} \label{cor:resilience}
    Let $\mu, \cD, \{x_i\}, \alpha, \eta$ be as in Lemma \ref{lem:resilience-of-moments}. Then, with probability $1-\beta$, the following all hold for all unit vectors $v$ simultaneously.
\begin{enumerate}
    \item $\left|\E_i \langle x_i-\mu, v \rangle\right| \le \alpha$.
    \item $\left|\E_i \langle x_i-\mu, v \rangle^2 - 1\right| \le \alpha$.
    \item For any $a_1, \dots, a_n \in [0, 1]$ such that $\E a_i \le \eta$, $\left|\E_i a_i \langle x_i-\mu, v \rangle\right| \le \alpha$ and $\left|\E_i a_i \langle x_i-\mu, v \rangle^2\right| \le \alpha$.
    \item $\E_i \left|\langle x_i-\mu, v \rangle\right| \le O(1)$.
\end{enumerate}
\end{corollary}

\begin{proof}
    Fix a vector $v$ and let $z_i := \langle x_i-\mu, v \rangle$. Suppose the events of \Cref{lem:resilience-of-moments} hold.

    Parts 1 and 2 are immediate from \Cref{lem:resilience-of-moments}, by setting $b_i = 1$ for all $i$. Part 3 follows by setting $a_i = 1-b_i$, and then noticing that $\left|\frac{1}{n} \sum_{i=1}^{n} a_i z_i\right| \le \left|\frac{1}{n} \sum_{i=1}^{n} z_i\right| + \left|\frac{1}{n} \sum_{i=1}^{n} b_i z_i\right| \le \tcO(\eta)$.
    
    Finally, to check part 4, we may consider $\eta = 0.1$ and then apply part 3, to obtain that $\left|\frac{1}{n} \sum_{i=1}^{n} a_i z_i\right| \le O(1)$ for any $a \in [0, 1]^n$ with $\sum a_i \le 0.1 n$. Since every vector in $[-1, 1]^n$ can be written as a sum and difference of at most $20$ vectors $a \in [0, 1]^n$ with $\sum a_i \le 0.1 n$, we thus have that $\left|\frac{1}{n} \sum_{i=1}^{n} c_i z_i\right| \le O(1)$ for all choices of $c_i \in \{-1, 1\}^n$ simultaneously. Thus, $\frac{1}{n} \sum_{i=1}^{n} |z_i| \le O(1)$.
\end{proof}

\begin{remark}
    The conditions in \Cref{cor:resilience} will be the only conditions we will require about the samples we draw. So in fact, our algorithm will output a point close to $\mu$ if given an $\eta$-corrupted version of $\cX$ for any $\cX$ satisfying \Cref{cor:resilience}.
    
    We also note that if $\mu, \{x_i\}$ satisfy \Cref{cor:resilience}, then for all symmetric $H$ with $\|H-I\|_{op} \le \alpha$, $H \mu, \{H x_i\}$ also satisfies \Cref{cor:resilience} (up to replacing $\alpha$ with $O(\alpha)$). To see why, assume without loss of generality that $\mu = \textbf{0}$. Then, using Condition 1, for all unit vectors $v$, $|\E_i \langle H x_i, v \rangle| = |\E_i \langle x_i, H v \rangle| \le \alpha \cdot \|H v\|_2 \le \alpha \cdot (1+\alpha) \le 2 \alpha$. We can repeat the same argument for the 2nd, 3rd, and 4th conditions.
\end{remark}

\subsection{Robust Algorithm} \label{subsec:mean-robust}
Here, we describe the robust algorithm that will inspire our score function to generate a differentially private algorithm. The robust algorithm, as well as the algorithms used in the covariance settings, are essentially the same as in~\cite{KothariMZ22}.

Suppose $\{x_i\}_{i=1}^n$ are samples from $\cN\paren{\mu, I}$ (or a subgaussian distribution with mean $\mu$ and covariance $I$).
Let $\{y_i\}$ be an arbitrary $\eta$-corruption of the $\{x_i\}$. 
Consider the following pseudo-expectation program with input points $\{y_i\}$ and domain the degree-$4$ pseudo-expectations with $\set{w_i}, \set{x_i'},\set{M_{i, j}}$ as indeterminates. ($M = \set{M_{i, j}}$ will represent a $d \times d$-matrix of indeterminates.)
\begin{align*}
\text{ find $\pE$ } & \\
\text{ such that } & \pE \text{ satisfies } w_i^2 = w_i, \\
& \pE \text{ satisfies } \sum w_i \ge (1- \eta) n, \\
& \pE \text{ satisfies } w_i x_i' = w_i y_i, \\
& \pE \text{ satisfies } \frac{1}{n} \sum \paren{x_i' - \mu'} \transpose{\paren{x_i' - \mu'}} + M\transpose{M} = \paren{1+ \tcO\paren{\eta}} I, \text{where $\mu' = \E_i x_i'$}
\end{align*}
It can be proven that if $n$ is as in \Cref{lem:resilience-of-moments}, with probability $1-\beta$ over the choice of $\{x_i\}$ and for any $\eta$-corruption $\{y_i\}$ of $\{x_i\}$, then $\normt{\pE{\mu'} - \mu} = \tcO\paren{\eta}$ for \emph{any} feasible pseudo-expectation $\pE$.

\subsection{Score Function and its Properties}
Our goal is to use 
\Cref{thm:pure_dp_general_main}, but to do so, we need to design a suitable score function. Our score function will be very similar to the robust algorithm, but modified to deal with precision issues.

Before we define our score function, we make a definition of \emph{certifiable means}, which modifies the pseudoexpectation program in \Cref{subsec:mean-robust} to deal with approximate pseudoexpectations.

\begin{definition}[Certifiable Mean]
\label{def:certifable-mean}
Let $\alpha, \tau, \phi, T \in \R^{\ge 0}$, $y_1, \dots y_n \in \R^d$ (with $\cY := \{y_1, \dots, y_n\}$), and $\tmu \in \R^d$.
We call the point $\tmu$ an $\paren{\alpha, \tau, \phi, T}$-\emph{certifiable mean} for $\cY$ if and only if there exists a linear functional $\cL$ over the set of polynomials in indeterminates $\set{w_i}, \set{x_{i, j}'}, \set{M_{j, k}}$ of degree at most $6$ such that 
\begin{enumerate}
    \item $\cL 1 = 1$,
    \item for every polynomial $p$, where $\normt{\cR(p)} \le 1$ (where we recall that $\cR(p)$ is the vector of monomial coefficients of $p$):
    %TODO: degree of polynomial p
    \begin{enumerate}
        \item $\cL p^2 \ge -\tau \cdot T$,
        \item $\forall i, \cL \paren{w_i^2 - w_i} p^2 \in \brac{-\tau \cdot T, \tau \cdot T}$,
        \item $\cL \paren{\sum w_i -n + T} p^2 \ge -5 \tau \cdot T \cdot n$, %\todo{had to change}
        \item $\forall i, j, \cL w_i \paren{x_{i, j}' - y_{i, j}}p^2 \in \brac{-\tau \cdot T, \tau \cdot T}$,
        \item $\forall j, k:\cL \Paren{\Brac{\frac{1}{n}\sum_i \paren{x_i' - \mu'}\transpose{\paren{x_i' - \mu'}}+M\transpose{M} - \paren{1+\alpha} I}_{j, k}p^2} \in \brac{-\tau \cdot T, \tau \cdot T}$, where $x_i' = \{x_{i, j}'\}_{1 \le j \le d}$, and $\mu' = \E_i x_i'$.
        Note that here $\brac{\dots}_{j,k}$ denotes the $(j,k)$ entry of a matrix, which is a polynomial in indeterminates $\set{w_i}, \set{x_i'}, \set{M_{j, k}}$. We write in this format for the sake of conciseness.
    \end{enumerate}
    \item  $\forall i, |\cL \mu_i' - \tmu_i| \le \phi + \tau \cdot T$ and $|\cL \mu_i'| \le 2 R + T \cdot \tau$.
    \item $\|\cR(\cL)\|_2 \le R' + T \cdot \tau$ for some sufficiently large $R' = \poly(n, d, R)$.\footnote{See \Cref{lem:vol-lb-mean} for more details on how large we require $R'$ to be.}
    % \item  $\|\cL \mu' - \tmu\|_2 \le \alpha + \tau \cdot T$.
\end{enumerate}
    For such $\cL$, we also say that $\cL$ is an $(\alpha, \tau, \phi, T)$-certificate for $\cY$.
    %In addition we will require $\|\cR(\cL)\|_2 \le R' + T \cdot \tau$ for some sufficiently large $R' = \poly(n, d, R)$,\footnote{See \Cref{lem:vol-lb-mean} for more details on how large we require $R'$ to be.} where we recall that $\cR(\cL)$ is the vector which represents the value of $\cL$ applied to each monomial of degree at most $6$.
    For the final constraint, recall that $\cR(\cL)$ is the vector which represents the value of $\cL$ applied to each monomial of degree at most $6$. Also, note that $\cR(\cL)$ has dimension polynomial in the number of variables, which is polynomial in $n, d$. This final constraint is only needed for computability purposes.

Note that one may think of $\cL$ as an approximate pseudo-expectation. In addition, for each constraint 2a) to 2e) we implicitly assume a bound on the degree of $p$ so that $\cL$ is applied to a polynomial of degree at most $6$.
\end{definition}

For our purposes, we will end up setting $\phi = \frac{\alpha}{\sqrt{d}}$ and $\tau = 1/(n \cdot d \cdot R)^{O(1)}$, for a large enough $O(1)$.

Now we use this definition to define a score function.
\begin{definition}[Score Function]
    \label{def:score-mean}
    Let $\mathbb{B}_{\infty}^d(2 R + n \tau + \phi)$ denote the $\ell_\infty$-ball of radius $2 R + n \tau + \phi$ in $\R^d$ centered at the origin.
    Let $\alpha, \tau, \phi \in \R^{\ge 0}$, $y_1, \dots y_n \in \R^d$ (with $\cY = \{y_1, \dots, y_n\}$) and $\tmu \in \R^d$.
    We define the score function $\cS : \mathbb{B}_{\infty}^d(2 R + n \tau + \phi) \to \R$ (viewed as a function of $\tmu$) as
    \begin{equation*}
        \cS\paren{\tmu, \cY; \alpha, \tau, \phi} = \min_{T \ge 0} \text{ such that $\tmu$ is a $\paren{\alpha, \tau, \phi, T}$ certifiable mean for $\cY = \{y_1, \dots, y_n\}$}.
    \end{equation*}
\end{definition}

In the rest of this section we will prove the following properties for this score function. This will allow us to use \Cref{thm:pure_dp_general_main}.

\begin{enumerate}
    \item Bounded Sensitivity: Score has sensitivity $1$ with respect to $\cY$.
    \item Quasi-Convexity: Score is quasi-convex as a function of $\tmu$.
    \item Accuracy: All points $\tmu$ that have score at most $\eta \cdot n$
    %(where $\alpha \sqrt{\log (1/\alpha)} \le \eta \le 0.01$.
    have distance at most $\alpha = \tcO(\eta)$ away from $\mu$. (Robustness for volume/accuracy purposes).
    \item Volume: The volume of points that have score at most $\eta \cdot n$ is sufficiently large, and the volume of points with score at most $\eta' \cdot n$ for $\eta' > \eta$ is not too large.
    \item Efficient Computability: Score is efficiently computable for any fixed $\tmu, \cY$.
    \item Robust algorithm finds low-scoring point: Finding $\tmu$ that minimizes score (up to error $1$) for any fixed $\cY$ can be done efficiently.
\end{enumerate}

For simplicity, we may ignore the arguments $\alpha, \tau, \phi$ in the score $\cS$, and write $\cS(\tmu, \cY)$.

\subsubsection{Sensitivity}
Before proving sensitivity we need to prove the prove that the score function is always at most $n$. In \Cref{def:score-mean}, we take the minimum over $T \ge 0$, so the score function is automatically at least $0$.

\begin{lemma}[score function upper bound]
\label{lem:mean-score-function-upper}
For any $\cY$ and any $\tmu \in \mathbb{B}_{\infty}^d(2 R + n \tau + \phi)$, the value $\cS(\tmu, \cY),$ as defined in \Cref{def:score-mean}, is and most $n$. 
\end{lemma}
\begin{proof}
It suffices to show that for $T = n$, and for any $\tmu \in \mathbb{B}_{\infty}^d(2 R + n \tau + \phi)$, there exists a linear functional $\cL$ such that the constraints of \Cref{def:certifable-mean} are satisfied.

Let's define $\cL$. For any monomial $p$ we should assign a value to $\cL p$.
To begin, let $\cL 1 = 1$. 
If $p$ contains $w_i$ or $M_{j, k}$ where $j \neq k$ let $\cL p = 0$. Now we need to define $\cL p$ for monomials that only contain $x_{i, j}'$ and $M_{j, j}$.
To do so, first set $\mu_j = \min(2 R + n \tau, \max(-(2 R + n \tau), \tmu_j))$, i.e., we clip $\mu_j$ to have magnitude at most $2 R + n \tau$.
For such monomials $p$, let $\cL p$ be equal to $\paren{1+\alpha}^{\paren{b/2}} \cdot \prod_{j=1}^d \mu_j^{a_j}$, where $a_j$ is equal to the sum of the number of the factors of the form $x_{i, j}'$ over all $i$ in $p$, and $b$ is equal to the number of $M_{j,j}$ factors in $p$ over all $j$.
Basically, when applying $\cL$ to a polynomial we are treating the indeterminates in the problem as if they were scalars and had the assignment $w_i = 0$, $x_i' = \mu$, and $M = \sqrt{\paren{1+\alpha}I}$, where $\mu$ is the clipped version of $\tmu$. In the non-relaxed version of the problem, this assignment would correspond to changing every point to $\mu$.

It is easy to check that all of the constraints would be satisfied under this choice of $\cL$, even if $\tau = 0$, as long as $R'$ is a sufficiently large polynomial in $R, n, d$.
Moreover, $|\mu_i - \tmu_i| \le \phi$ for all $i$.
Therefore, the value of the score function $\cS$ defined in \Cref{def:score-mean}
is at most $n$.
\end{proof}

\begin{lemma}[sensitivity] \label{lem:sensitivity}
The score function $\cS$ as defined in \Cref{def:score-mean} has sensitivity $1$ with respect to $\cY$.
\end{lemma}
\begin{proof}
Suppose that $\cY$, $\cY'$ are two neighboring datasets, and $\tmu \in B_\infty^d(2 R + n \tau + \phi)$. Moreover, assume $\cS\paren{\tmu, \cY} = T$.
If we show that $\cS\paren{\tmu, \cY'} \le \cS\paren{\tmu, \cY}+1 = T+1$, by symmetry we are done.
Since $\cS\paren{\tmu, \cY} = T$, we know that there exists some functional $\cL$ such that the constraints of \Cref{def:certifable-mean} are satisfied for $\cL$, $\cY$, and $T$. If we construct a new functional $\cL'$ such that the constraints of \Cref{def:certifable-mean} are satisfied for $\cL'$, $\cY'$ and $T+1$, we have shown that $\cS\paren{\tmu, \cY'} \le T+1$ and we are done.

Without loss of generality assume $\cY$ and $\cY'$ differ on index $j$.
In order to construct $\cL'$, for any monomial $p$, let
\begin{equation*}
\cL' p = \begin{cases}
0 &\text{if $p$ has a $w_j$ factor},\\
\cL p &\text{otherwise}
\end{cases}.
\end{equation*}
Now let's go through all of the constraints and verify them. The first condition holds since by definition, $\cL' 1 = \cL 1 = 1$.
Now let's prove the conditions in the second set of conditions.
Suppose $\|p\|_2 \le 1$ and $p = q + w_j r$, where $q$ does not contain a monomial containing $w_j$.
\begin{itemize}
    \item $\cL' p^2 \ge -\tau \cdot \paren{T+1}$.
     
    \begin{equation*}\cL' p^2 = \cL' \paren{q + w_j r}^2 = \cL'q^2 = \cL q^2 \ge  -\tau \cdot T \ge - \tau \cdot \paren{T + 1}
    \end{equation*}
    as desired, where we used the fact that $\|q\|_2 \le \|p\|_2 \le 1$.
    \item $\forall i: \cL'\paren{w_i^2 - w_i} p^2 \in \brac{-\tau \cdot \paren{T+1}, \tau \cdot \paren{T+1}}$. 
    If $i = j$, this would be zero; if not, then we can write $p$ as $q + w_j r$ similar to the previous part and get the desired bounds.
    \item $\cL'\paren{\sum w_i - n + \paren{T+1}}p^2 \ge -5 \tau \cdot \paren{T+1} \cdot n$. 
    \begin{align*}
    \cL'\paren{\sum w_i - n + \paren{T+1}} p^2 
    &= \cL'\paren{\sum w_i - n + \paren{T+1}} q^2 \\
    &= \cL' \paren{\sum_{i\neq j} w_i - n + \paren{T+1}} q^2 \\
    &= \cL \paren{\sum_{i\neq j} w_i - n + \paren{T+1}} q^2 \\
    &= \cL \paren{\sum w_i - n + T} q^2 - w_j q^2 + q^2
    \end{align*}
To bound the first term, we have that $\cL \paren{\sum w_i -n + T}q^2 \ge -5 \tau \cdot T \cdot n$.
To bound the second and third terms, we have $\cL[-w_j q^2 + q^2] = \cL[(1-w_j) q^2] = \cL[(1-w_j)^2 q^2] + \cL[(w_j-w_j^2) q^2].$ We know that $\|(1-w_j) q\|_2 \le 2 \|q\|_2 \le 2$, so $\cL[(1-w_j)^2 q^2] \ge -4 \tau \cdot T$, and $\cL[(w_j-w_j^2) q^2] \ge -\tau \cdot T$. So together, we have a bound of at least $-5 \tau \cdot T \cdot n -5 \tau \cdot T$.
Therefore it remains to  prove that
$-5 \tau \cdot T \cdot n -5 \tau \cdot T \ge -5\tau \cdot (T+1) \cdot n$, which is trivial by \Cref{lem:mean-score-function-upper}.

\item $\forall j, k:\cL' \Paren{\Brac{\frac{1}{n}\sum_i \paren{x_i' - \mu'}\transpose{\paren{x_i' - \mu'}}+M\transpose{M} - \paren{1+\alpha} I}_{j, k}p^2} \in \brac{-\tau \cdot T, \tau \cdot T}$, where $\mu' = \E_i x_i'$. Similar to previous parts, we just need to plug in $p = q + w_j r$, and we get the desired inequality.
\end{itemize}

The third condition holds because $\cL' \mu_i' - \tmu_i = \cL \mu_i' - \tmu_i$, and because $\|\cL \mu_i'\|_\infty \le 2 R + T \cdot \tau \le 2 R + (T+1) \cdot \tau$.
The last condition holds because $\|\cR(\cL')\|_2 \le \|\cR(\cL)\|_2$ clearly holds. Therefore, there exists a linear functional $\cL'$ which satisfies the constraints of \Cref{def:certifable-mean} for $\cY'$ and $T+1$, so the score function $\cS$ has sensitivity $1$ with respect to $\cY$.
\end{proof}

\subsubsection{Quasi-convexity}

\begin{lemma}[quasi-convexity] \label{lem:quasi-convex-mean}
The score function $\cS$ as defined in \Cref{def:score-mean} is quasi-convex in $\tmu$.
\end{lemma}
\begin{proof}
Suppose $\cS\paren{\tmu_1, \cY} = T_1, \cS\paren{\tmu_2, \cY} = T_2$, and suppose there exists $\cL_1$ and $\cL_2$ that satisfy the constraints in \Cref{def:certifable-mean} with $\tmu_1, T_1$, and $\tmu_2, T_2$ respectively.
If we can construct a functional $\cL_3$ such that the constraints in \Cref{def:certifable-mean} are satisfied with $\tmu_3 = \lambda \tmu_1 + \paren{1- \lambda} \tmu_2$, and $T_3 = \max\set{T_1, T_2}$, we are done.
Let $\cL_3 = \lambda \cL_1 + \paren{1 - \lambda} \cL_2$.
Then, all of the constraints in \Cref{def:certifable-mean} are satisfied trivially except for $\cL_3 \paren{\sum w_i -n + T_3} p^2 \ge -5 \tau \cdot T_3 \cdot n$.
Let's verify this constraint.
Without loss of generality suppose $T_3 = T_2 \ge T_1$, then
\begin{align*}
\cL_3 \Paren{\sum w_i -n + T_3} p^2 
&=
\paren{\lambda \cL_1 + \paren{1- \lambda}\cL_2} \Paren{\sum w_i -n + T_2} p^2 \\ 
&=
\lambda \cL_1 \Paren{\sum w_i -n + T_1} p^2 
+ 
\paren{1- \lambda} \cL_2 \Paren{\sum w_i -n + T_2} p^2 +
\lambda \paren{T_2 - T_1}\cL_1 p^2
\\
& 
\ge
-5 \tau \cdot n \paren{\lambda T_1 + \paren{1 - \lambda} T_2} 
- \lambda \paren{T_2 - T_1} \cdot \tau \cdot T_1\\
&\ge
-5 \tau \cdot n \paren{\lambda T_1 + \paren{1 - \lambda} T_2 + \lambda \paren{T_2 - T_1}}
\tag{$n \ge T_1$, \Cref{lem:mean-score-function-upper}} \\
&= 
-5 \tau \cdot T_3 \cdot n,
\end{align*}
as desired.
\end{proof}

\subsubsection{Accuracy}

We show that any point $\tmu$ of low score with respect to i.i.d. samples from $\cN(\mu, I)$ must be close to $\mu$. We remark that because of our sensitivity bound, this will also imply a similar result for corrupted samples.

\begin{lemma} \label{lem:mean-accuracy}
    Let $\alpha = \widetilde{O}(\eta)$ and suppose $\alpha, \eta$ are bounded by a sufficiently small constant. Let $n \ge \frac{d+\log (1/\beta)}{\alpha^2}$, and $\cX = \{x_1, \dots, x_n\} \sim \cN(\mu, I)$, for $\mu \in \BR^d$.

    Then, for any $\alpha^* \le \alpha$, and assuming $\tau \ll 1/(ndR)^{O(1)}$, with probability at least $1-\beta$, every point $\widetilde{\mu} \in \BR^d$ that is $(\alpha^*, \tau, \phi, T)$-certifiable for $\cX$ with $T = \eta n$ and $\phi \le \alpha/\sqrt{d}$ must satisfy $\|\widetilde{\mu} - \mu\|_2 \le O(\alpha)$.
\end{lemma}

The proof of \Cref{lem:mean-accuracy} essentially follows from the same argument as in \cite{KothariMZ22}, with slight modifications to deal with our modified score function. Hence, we defer the proof to \Cref{appendix:sos}.

\subsubsection{Volume of Good Points}

\begin{lemma} \label{lem:vol-lb-mean}
    Let $n \ge O\paren{\paren{d + \log \paren{1/\beta}}/{\alpha^2}}$, for some $\alpha = \tcO(\eta)$.
    Let $\cX = \{x_1, \dots, x_n\} \sim \cN(\mu, I)$, where $\|\mu\|_2 \le R$, and let $\cY = \{y_1, \dots, y_n\}$ represent an $\eta$-corruption of $\cX$. Then, for any $0 \le \tau, \phi \le \frac{\alpha}{\sqrt{d}}$ and $T = \eta \cdot n$, with probability at least $1-\beta$, there exists $\hat{\mu}$ such that every $\tmu$ such that $\|\tmu-\hat{\mu}\|_\infty \le \phi$ is an $(\alpha, \tau, \phi, T)$-certifiable mean for $\cY$.
\end{lemma}

\begin{proof}
    Our linear operator $\cL$ generalizes pseudo-expectations $\pE$, which in turn generalizes expectations over a single point mass. So, it suffices to find values for $\{w_i\}, \{x_i'\}, \{M_{i, j}\}$ that satisfy the constraints of the robust algorithm. If so, then by setting $\hat{\mu} = \mu' = \frac{1}{n} \sum x_i'$, we have that for all $\tmu$ such that $\|\tmu- \hat{\mu}\|_\infty \le \phi$, $\tmu$ is an $(\alpha, \tau, \phi, T)$-certifiable mean.
    
    Indeed, finding such a pseudo-expectation is quite simple to do.
    %: it will actually just be an expectation over a single point.
    We just set every $w_i = 1$ if $y_i = x_i$ and $0$ otherwise, and set every $x_i' = x_i$, so $\mu' = \frac{1}{n} \sum_i x_i$. By \Cref{lem:resilience-of-moments}, we have that $\frac{1}{n} \sum_{i=1}^{n} \langle x_i-\mu, v \rangle^2 \le 1+\tcO(\eta)$ for all unit vectors $v$. In addition,
\[\frac{1}{n} \sum_{i=1}^{n} \langle x_i-\mu, v \rangle^2 = \frac{1}{n} \sum_{i=1}^{n} \langle (x_i-\mu')+(\mu'-\mu), v \rangle^2 = \langle \mu'-\mu, v \rangle^2 + \frac{1}{n} \sum_{i=1}^{n} \langle x_i-\mu', v \rangle^2 \ge \frac{1}{n} \sum_{i=1}^{n} \langle x_i-\mu', v \rangle^2.\]
    So, $\frac{1}{n} \sum_{i=1}^{n} \langle x_i-\mu', v \rangle^2 \le 1+\tcO(\eta)$ for all unit vectors $v$, which means $\frac{1}{n} \sum_{i=1}^{n} (x_i-\mu) (x_i-\mu)^\top \preccurlyeq (1+\tcO(\eta)) I$. Therefore, there exists a $d \times d$ matrix $M$ such that $\frac{1}{n} \sum_{i=1}^{n} (x_i-\mu) (x_i-\mu)^\top + MM^\top = (1+\tcO(\eta)) I$.

    Finally, we remark that every $w_i, x_{i, j}$, and $M_{j, k}$ is bounded by $R \cdot n$. Therefore, the corresponding linear operator $\cL$ satisfies $\|\cR(\cL)\|_2 \le (Rnd)^{O(1)}$. Moreover, with $1-\beta$ probability, $\|\mu'-\mu\|_\infty \le \|\mu'-\mu\|_2 \le \alpha,$ which means that $\|\cL[\mu_i']\|_\infty = \|\mu_i'\|_\infty \le \|\mu\|_\infty + \|\mu'-\mu\|_\infty \le 2 R$.
\end{proof}

\begin{lemma} \label{lem:vol-ub-mean}
    Let $n \ge O\paren{\paren{d + \log \paren{1/\beta}}/{\alpha^2}}$, for some $\alpha = \tcO(\eta)$.
    Let $\cX = \{x_1, \dots, x_n\} \sim \cN(\mu, I)$, and let $\cY = \{y_1, \dots, y_n\}$ represent an $\eta$-corruption of $\cX$. Set $\phi = \alpha/\sqrt{d}$. Then, for every integer $T \in [\eta \cdot n, \eta^* \cdot n]$ for some fixed constant $\eta^* < 1$, with probability at least $1-\beta$, every $(\alpha, \tau, \phi, T)$-certifiable mean with respect to $\cY$ has distance at most $\tcO(T/n)$ from $\mu$.
\end{lemma}

\begin{proof}
    Since the score function has sensitivity at most 1 (\Cref{lem:sensitivity}),
    this means that any $(\alpha, \tau, \phi, T)$-certifiable mean with respect to $\cY$ is an $(\alpha, \tau, \phi, T+\eta n)$-certifiable mean with respect to $\cX$.
    
    Now, define $\eta' := \frac{T+\eta n}{n} = O(\frac{T}{n})$. In this case, by setting $\alpha' = \tcO(\eta')$ and since $\alpha = \tcO(\eta) \le \alpha'$, we have that by \Cref{lem:mean-accuracy} that any $(\alpha, \tau, \frac{\alpha'}{\sqrt{d}}, T+\eta n)$-certifiable mean $\widetilde{\mu}$ must satisfy $\|\widetilde{\mu}-\mu\|_2 \le O(\alpha') \le \tcO(T/n)$. Since $\alpha' \ge \alpha$, any $(\alpha, \tau, \phi, T+\eta n)$-certifiable mean is also a $(\alpha, \tau, \frac{\alpha'}{\sqrt{d}}, T+\eta n)$-certifiable mean, which completes the proof.
\end{proof}

If we set $\phi = \alpha/\sqrt{d}$ and $\tau \ll 1/(nd)^{O(1)}$, this means the volume of $(\alpha, \tau, \phi, T)$-certifiable means for $T = \eta n$ is at least $(\alpha/\sqrt{d})^d$. However, for any $T = \eta' n$ for $\eta \le \eta' \le \eta^*$, the volume of $(\alpha, \tau, \phi, T)$-certifiable means is at most $(\tilde{O}(\eta'))^d$ times the volume of a $d$-dimensional sphere, which is $(\tilde{O}(\eta'))^d/\sqrt{d}^d$. Finally, for $T = \eta' n$ with $\eta' > \eta^*$, the volume of $\Theta$, the set of all candidate means $\tmu$ with $\|\tmu\|_2 \le R$, is at most $O(R/\sqrt{d})^d$.

\subsubsection{Efficient Computability}

Verifying that we can efficiently compute the score roughly follows from the ellipsoid method used in semidefinite programming. We had to modify the score accordingly (relaxing constraints using $\tau$) -- however, we show in \Cref{thm:computing-score}, deferred to \Cref{sec:computing-score-functions}, that for the score in \Cref{def:score-mean}, defined by the constraints in \Cref{def:certifable-mean}, we can compute it up to error $\gamma$ in time $\poly(n, d, \log R, \log \gamma^{-1})$. Hence, this verifies the ``efficiently computable'' criterion for \Cref{thm:pure_dp_general_main}.

\subsubsection{Efficient Finding of Low-Scoring Point}

Verifying the ``robust algorithm finds low-scoring point'' criterion is also direct from \Cref{thm:computing-score}. We simply remove the first half of the third constraint, i.e., $|\cL \mu_i' - \tmu_i| \le \phi + \tau \cdot T$. We can apply \Cref{thm:computing-score} in the same way to find some linear operator $\cL$ with score at most $\min_{\tmu} \cS(\tmu, \cY) + 1$. Then, we can compute $\cL[\mu']$ and set $r = \phi$, and obtain that every point within $\ell_2$ distance $\phi$ of $\cL[\mu']$ has score at most $\min_{\tmu} \cS(\tmu, \cY) + 1$.

\subsection{Proof of Theorem \ref{thm:gaussian-mean-main}}

We apply Theorem \ref{thm:pure_dp_general_main}, using the score function defined in \Cref{def:score-mean} and with $\Theta = B_\infty^d(2R + n \tau + \phi)$. Indeed, for $r = \phi = \alpha/\sqrt{d}$, we have verified all conditions, as long as $n \ge O((d+\log (1/\beta))/\alpha^2)$. Therefore, we have an $\eps$-DP algorithm running in time $\poly(n, d, \log \frac{R \sqrt{d}}{\alpha}) = \poly(n, \log R)$ that finds a candidate mean $\tmu$ of score at most $2 \eta n$, as long as
\[n \ge O\left(\max_{\eta': \eta \le \eta' \le 1} \frac{\log (V_{\eta'}(\cY)/V_\eta(\cY)) + \log (1/(\beta \cdot \eta'))}{\eps \cdot \eta'}\right).\]

Using Lemmas \ref{lem:vol-lb-mean} and \ref{lem:vol-ub-mean}, we have that for $\eta' \le \eta^*$ for some $\eta^* = \Omega(1)$, $V_{\eta'}(\cY)/V_{\eta}(\cY) = (\tcO(\eta')/\eta)^d \le (O(1/\eta))^d$. For $\eta' > \eta^*$, we have that $V_{\eta'}(\cY)/V_{\eta}(\cY) \le (O(R/\eta))^d$. So overall, it suffices for 
\begin{align*}
    n &\ge O\left(\frac{d+\log (1/\beta)}{\alpha^2}\right) + O\left(\max_{\eta \le \eta' \le \eta^*} \frac{d \log (1/\eta) + \log(1/(\beta \cdot \eta))}{\eps \cdot \eta'}  + \max_{\eta^* \le \eta' \le 1} \frac{d \log (R/\eta) + \log(1/(\beta \cdot \eta))}{\eps \cdot \eta'}\right) \\
    &= \tcO\left(\frac{d + \log(1/\beta)}{\alpha^2} + \frac{d + \log (1/\beta)}{\eps \cdot \alpha} + \frac{d \log R}{\eps}\right).
\end{align*}

Hence, our algorithm, using this many samples, can find a point $\tmu$ of score at most $2 \eta n$. Finally, by replacing $\eta$ with $2 \eta$ and applying \Cref{lem:mean-accuracy}, we have that any point $\tmu$ with score at most $2 \eta n$ is within $O(\alpha)$ of $\mu$. While we did not verify \Cref{lem:mean-accuracy} for corrupted points, by our bound on sensitivity, we know that for any $\cY$ which is an $\eta$-corruption of $\cX$, any point with score at most $2 \eta n$ with respect to $\cY$ has score at most $3 \eta n$ with respect to $\cX$, and therefore is within $O(\alpha)$ of $\mu$. This completes the proof.

\subsection{The approx-DP setting}

In this subsection, we prove Theorem \ref{cor:gaussian-mean-main}.
In this setting, the score function is identical, but we can afford fewer samples as we apply the algorithm of \Cref{thm:approx_dp_general_main} instead of \Cref{thm:pure_dp_general_main}. The main additional thing we must check is that for \emph{any} dataset $\cY$, if $\cS(\mu, \cY) \le 0.7 \eta^* n$ for some $\mu$, then the volume ratio $V_{\eta^*}(\cY)/V_{0.8 \eta^*}(\cY)$ is not too high.

Before proving our main result of this subsection, we must first establish the following ``worst-case robustness'' guarantee, which is important for ensuring privacy. We defer the proof to \Cref{appendix:sos}.

\begin{lemma} \label{lem:approx-mean-accuracy}
    Fix $\eta^*$ to be a sufficiently small constant, and $T = \eta^* n$. Also, suppose $\phi \le \alpha/\sqrt{d}$. Then, for a dataset $\cY$ with every $y_i$ bounded in $\ell_2$ norm by $R \cdot d^{100}$, if there exist $\tmu_1, \tmu_2 \in \BR^d$ that are both $(\alpha, \tau, \phi, T)$-certifiable means with respect to $\cY$, then $\|\tmu_1-\tmu_2\|_2 \le O(1)$.
\end{lemma}

As a corollary of \Cref{lem:approx-mean-accuracy}, we have the following result.

\begin{corollary} \label{cor:approx-mean-volume}
    Suppose that $\cY$ is a dataset with every $y_i$ bounded in $\ell_2$ norm by $R \cdot d^{100}$ that has an $(\alpha, \tau, \phi, 0.7 \eta^* n)$-certifiable mean, and let $\hat{\mu} = \cL[\mu']$ where $\cL$ is an $(\alpha, \tau, \phi, 0.7 \eta^* n)$-certificate. Also, suppose $\phi \le \alpha/\sqrt{d}$. Then, the set of $(\alpha, \tau, \phi, 0.8 \eta^* n)$-certifiable means contains all $\tmu$ such that $\|\tmu-\hat{\mu}\|_\infty \le \phi$, and any $(\alpha, \tau, \phi, \eta^* n)$-certifiable mean $\tmu$ must satisfy $\|\tmu-\hat{\mu}\| \le O(1)$.
\end{corollary}

\begin{proof}
    If $\cL$ is an $(\alpha, \tau, \phi, 0.7 \eta^* n)$-certificate, it is also an $(\alpha, \tau, \phi, 0.8 \eta^* n)$-certificate. This means that for $\hat{\mu} := \cL[\mu_i']$, every $\tmu$ such that $\|\hat{\mu} - \tmu\|_\infty \le \phi$ is $(\alpha, \tau, 0.8 \eta^* n)$-certifiable.
    To see why, note that for a $(\alpha, \tau, 0.8 \eta^* n)$-certificate $\cL$ of $\cY$, Constraint 3 (which is the only constraint that deals with $\tmu$, which we recall is not indeterminate) just requires that $\|\cL[\mu'] - \tmu\|_\infty \le \phi+\tau \cdot T$, where $T = 0.8 \eta^* n$. Finally, we need $\tmu \in \Omega$, which requires $\|\tmu\|_\infty \le 2R+n\tau+\phi$. However, since we require $\|\cL[\mu']\|_\infty \le 2R + \tau \cdot T \le 2R + \tau \cdot n$, any $\tmu$ with $\|\tmu - \cL[\mu']\|_\infty \le \phi$ is an $(\alpha, \tau, 0.8 \eta^* n)$-certifiable covariance.
    %For any choice of $\tSigma$ such that $(1-\alpha) \hat{\Sigma} \preccurlyeq \tSigma \preccurlyeq (1+\alpha) \hat{\Sigma}$, we will choose $A, B$ so that $AA^\top = (1+\alpha) \hat{\Sigma} - \tSigma$ and $BB^\top = \tSigma - (1-\alpha) \hat{\Sigma}$.
    %We can modify the linear operator $\cL$ to some $\cL'$ so that for any monomial of the form $p \cdot p'$ where $p$ is only in variables $\{w_i\}$, $\{x'_{i, j}\}$, and $\{M_{\{j, j'\}, \{k, k'\}}\}$ and $p'$ is only in variables $\{A_{j, k}\}$ and $\{B_{j, k}\}$, $\cL'[p \cdot p'] = \cL[p] \cdot p'$, where we think of $p'$ as constant 
    
    The second part is immediate by Lemma \ref{lem:approx-mean-accuracy}.
    %\todo{Deal with fact that $y_i$'s are bounded by $K \cdot d^{100}$}
\end{proof}

Therefore, by setting $\phi := \alpha/\sqrt{d}$, the set of $(\alpha, \tau, \phi, \eta^* n)$-certifiable means has volume at most $O(1/\sqrt{d})^{d}$, since the volume of a unit sphere is $O(1/\sqrt{d})^{d}$. The set of $(\alpha, \tau, \phi, 0.8 \eta^* n)$-certifiable means has volume at least $\phi^d \ge \Omega(\alpha/\sqrt{d})^{d}$. So, the ratio $V_{\eta^*}(\cY)/V_{0.8 \eta^*}(\cY) \le O(1/\alpha)^{d}$.

\medskip

We now prove \Cref{cor:gaussian-mean-main}, by applying \Cref{thm:approx_dp_general_main}. First, note that we may truncate the samples so that no $y_i \in \cY$ has norm more than $R \cdot d^{100}$. Since we are promised $\|\mu\|_2 \le R$, the probability that any uncorrupted sample has this norm is at most $e^{-d^{100}}$. We will set $\eta^*$ to be a sufficiently small constant (such as $0.01$). We just showed, using \Cref{cor:approx-mean-volume}, that for all $\cY$ such that $\min_{\tmu} \cS(\tmu, \cY) \le 0.7 \eta^* n$, $V_{\eta^*}(\cY)/V_{0.8 \eta^*}(\cY) \le O(1/\alpha)^{d}$. So, as long as $n \ge O\left(\frac{\log (1/\delta) + d \log (1/\alpha)}{\eps}\right),$ the algorithm of \Cref{thm:approx_dp_general_main} is $(\eps, \delta)$-differentially private. In addition, we have already verified all of the conditions, so the algorithm is accurate as long as we additionally have $n \ge \tcO((d + \log (1/\beta))/\eta^2)$ and 
\[n \ge O\left(\max_{\eta': \eta \le \eta' \le \eta^*} \frac{\log (V_{\eta'}(\cY)/V_\eta(\cY)) + \log (1/(\beta \cdot \eta'))}{\eps \cdot \eta'}\right).\]

By our volume bounds, this means it suffices for
\begin{align*}
    n &\ge \tcO\left(\frac{d+\log (1/\beta)}{\alpha^2}\right) + O\left(\max_{\eta \le \eta' \le \eta^*} \frac{d \log (1/\eta) + \log (1/(\beta \cdot \eta))}{\eps \cdot \eta'}\right) + O\left(\frac{\log (1/\delta) + d \log(1/\alpha)}{\eps}\right) \\
    &= \tcO\left(\frac{d +\log (1/\beta)}{\alpha^2} + \frac{d + \log (1/\beta)}{\eps \cdot \alpha} + \frac{\log (1/\delta)}{\eps}\right).
\end{align*}
This concludes the proof of \Cref{cor:gaussian-mean-main}.

\section{Preconditioning the Gaussian}
\label{sec:preconditioning}

\subsection{Main Theorems}

Our goal is to obtain polynomial time algorithms for private covariance estimation of a unknown Gaussian, with optimal sample complexity. Before achieving this, an important step is preconditioning the Gaussian so that the samples come from a near-isotropic Gaussian. This requires approximately learning the covariance up to spectral distance, which we focus on in this section.

We prove both a pure-DP and approx-DP result in this section, showing that one can privately (and robustly) learn the covariance of a Gaussian up to spectral distance using roughly $d^2$ samples. In addition, in the approx-DP setting, our sample complexity has no dependence on the parameter $K$, which describes the ratio between a priori upper and lower bounds on the true covariance matrix, though the runtime depends on $\log K$.

\begin{theorem}[Private Preconditioning of a Gaussian, Pure-DP] \label{thm:gaussian-covariance-pure-conditioning}
Let $\Sigma \in \R^{d \times d}$ be such that $K^{-1} I \preccurlyeq \Sigma \preccurlyeq K \cdot I$. Then, there exists an $\epsilon$-differentially private algorithm that takes $n$ \iid\ samples from $\cN\paren{\textbf{0}, \Sigma}$ and with probability $1-\beta$ outputs $\tSigma$ such that $\|\Sigma^{-1/2} \tSigma \Sigma^{-1/2} - I\|_{op} \le \alpha$, for
\begin{equation*}
n =
\tcO \Paren{
\frac{d^2 + \log^2 (1/\beta)}{\alpha^2}
+
\frac{d^2 + \log (1/\beta)}{\alpha \eps}
+
\frac{d^2 \log K}{\eps}
}.
\end{equation*}
Here $\tcO$ is hiding %TODO
factors.
Moreover, this algorithm runs in time $\poly(n, d)$, and succeeds with the same accuracy even if $\eta = \tilde{\Omega}(\alpha)$ fraction of the points are adversarially corrupted.
\end{theorem}

\begin{theorem}[Private Preconditioning of a Gaussian, Approx-DP] \label{thm:gaussian-covariance-approx-conditioning}
Let $\Sigma \in \R^{d \times d}$ be such that $K^{-1} I \preccurlyeq \Sigma \preccurlyeq K \cdot I$. Then, there exists an $(\epsilon, \delta)$-differentially private algorithm that takes $n$ \iid\ samples from $\cN\paren{\textbf{0}, \Sigma}$ and with probability $1-\beta$ outputs $\tSigma$ such that $\|\Sigma^{-1/2} \tSigma \Sigma^{-1/2} - I\|_{op} \le \alpha$, where
\begin{equation*}
n =
\tcO \Paren{
\frac{d^2 + \log^2 (1/\beta)}{\alpha^2}
+
\frac{d^2 + \log (1/\beta)}{\alpha \eps}
+
\frac{\log (1/\delta)}{\eps}
}.
\end{equation*}
Here $\tcO$ is hiding %TODO
factors.
Moreover, this algorithm runs in time $\poly(n, d, \log K)$, and succeeds with the same accuracy even if $\eta = \tilde{\Omega}(\alpha)$ fraction of the points are adversarially corrupted.
\end{theorem}

\subsection{Resilience of Moments}

\iffalse
    \begin{lemma}[Resilience of First and Second Moments, Lemmas 4.3 and 4.4 from \cite{DiakonikolasKKLMS16}]
    \label{lem:resilience-of-moments-covariance}
    Let $\set{x_i} \sim \cN \paren{\textbf{0}, I}$ for $\mu \in \R^d$.
    Let 
    $n \ge 
    \tcO\paren{\paren{d + \log \paren{1/\beta}}^2/{\eta^2}}$.
    Then, with probability $1-\beta$, for all vectors $b \in \brac{0, 1}^n$ such that $\E_i b_i \ge 1-\eta$ and all matrices $P \in \R^{d \times d}$ with $\|P\|_F = 1$, we have
    \begin{equation*}
    \left|\E_i b_i \left\langle P, \frac{x_ix_i^\top - I}{\sqrt{2}} \right\rangle\right| \le \tcO(\eta).
    \end{equation*}
    In addition,
    \begin{equation*}
    \left|\E_i b_i \left\langle P, \frac{x_ix_i^\top - I}{\sqrt{2}} \right\rangle^2 - 1 \right| \le \tcO(\eta).
    \end{equation*}
    \end{lemma}
\fi

Similar to the mean estimation case, we will also require higher-order moment bounds, and stability conditions that imply the top roughly $\eta$ fraction of samples in any ``covariance'' direction cannot be too large.

\begin{lemma} \label{lem:resilience-of-moments-covariance}
    Let $\{x_i\} \sim \cN(\textbf{0}, I)$ and $n \ge \tcO((d^2+\log^2 (1/\beta))/\eta^2)$. Then, with probability $1-\beta$, the following all hold for all symmetric $P \in \BR^{d \times d}$ with $\|P\|_F = 1$ simultaneously, for some $\alpha = \tilde{O}(\eta)$.
\begin{enumerate}
    \item $\left|\E_i \langle (x_ix_i^\top-I)/\sqrt{2}, P \rangle\right| \le \alpha$.
    \item $\left|\E_i \langle (x_ix_i^\top-I)/\sqrt{2}, P \rangle^2 - 1\right| \le \alpha$.
    \item For any real values $a_1, \dots, a_n \in [0, 1]$ such that $\E_i a_i \le \eta$, $\left|\E_i a_i \langle (x_ix_i^\top-I)/\sqrt{2}, P \rangle\right| \le \alpha$ and $\left|\E_i a_i \langle (x_ix_i^\top-I)/\sqrt{2}, P \rangle^2\right| \le \alpha$.
    \item $\E_i \left|\langle (x_ix_i^\top-I)/\sqrt{2}, P \rangle\right| \le O(1)$.
\end{enumerate}
\end{lemma}

To our knowledge, such a result is not known with this number of samples. The best-known result we know of can obtain the same bounds but requires $\tcO(d^2 \log^5 (1/\beta)/\eta^2)$ samples~\cite{DiakonikolasKKLMS19}, which means the number of samples required is $d^{2+\Omega(1)}$ if we want exponentially small failure probability. We prove \Cref{lem:resilience-of-moments-covariance} in \Cref{sec:high-prob-stability}.

\begin{remark}
    As in the mean estimation case, \Cref{lem:resilience-of-moments-covariance} will be the only conditions we will require about the samples we draw. (Or if $x_i \sim \cN(\textbf{0}, \Sigma)$, then $\{\Sigma^{-1/2} x_i\}$ are resilient.)
\end{remark}

We also note the following corollary.

\begin{corollary} \label{cor:empirical-covariance}
    Let $\alpha, \eta$ be as in \Cref{lem:resilience-of-moments-covariance}.
    Let $\{x_i\} \sim \cN(\textbf{0}, \Sigma)$ and $n \ge \tcO((d^2+\log^2 (1/\beta))/\eta^2)$. Define $\hat{\Sigma} = \frac{1}{n} \sum x_i x_i^\top$. Then, with probability $1-\beta$, $\|\Sigma^{-1/2} \hat{\Sigma} \Sigma^{-1/2} - I\|_F \le \sqrt{2} \cdot \alpha$.
\end{corollary}

\begin{proof}
    We can write $x_i = \Sigma^{1/2} y_i$ where $y_i \sim \cN(\textbf{0}, I)$. So, $\Sigma^{-1/2} \hat{\Sigma} \Sigma^{-1/2} - I = \frac{1}{n} \sum_{i=1}^n (y_i y_i^\top - I)$. Therefore, by part 1 of \Cref{lem:resilience-of-moments-covariance}, $|\langle\Sigma^{-1/2} \hat{\Sigma} \Sigma^{-1/2} - I, P \rangle| \le \sqrt{2} \alpha$ for any symmetric $P$ with $\|P\|_F$. However, note that for any symmetric matrix $M$, $\|M\|_F = \langle M, \frac{M}{\|M\|_F} \rangle,$ and $\frac{M}{\|M\|_F}$ is symmetric with Frobenius norm $1$. Thus, by setting $M = \langle\Sigma^{-1/2} \hat{\Sigma} \Sigma^{-1/2} - I$ and $P = \frac{M}{\|M\|_F}$, we have that $\|M\|_F \le \sqrt{2} \alpha$.
\end{proof}

\subsection{Robust Algorithm}
Suppose $\set{x_i}$ are samples from $\cN\paren{\textbf{0}, \Sigma}$.
Let $\set{y_i}$ be an arbitrarily $\eta$-corruption of $\set{x_i}$. 
Consider the following pseudo-expectation program, where $\set{y_i}$ are the input points and the domain is the degree-$12$ pseudo-expectations with $\set{w_i}, \set{x_i}$ as indeterminates.
%\set{M_{\{i, i'\}, \{j, j'\}}}
\begin{align*}
\text{ find $\pE$ } & \\
\text{ such that } & \pE \text{ satisfies } w_i^2 = w_i, \\
& \pE \text{ satisfies } \sum w_i \ge (1- \eta) n, \\
& \pE \text{ satisfies } w_i x_i' = w_i y_i, \\
& \pE \paren{2+ \tcO\paren{\eta}} \cdot (v^\top \Sigma' v)^2 - \frac{1}{n} \sum (\langle v, x' \rangle^2 - v^\top \Sigma' v)^2 \text{ has a degree $4$-SoS proof of} \\ & \hspace{4cm} \text{nonnegativity in $v \in \BR^d$, where $\Sigma' = \frac{1}{n} \sum_{i=1}^{n} (x_i')(x_i')^\top$.}
\end{align*}

To explain the last condition further, note that $\pE\left[\paren{2+ \tcO\paren{\eta}} \cdot (v^\top \Sigma' v)^2 - \frac{1}{n} \sum (\langle v, x' \rangle^2 - v^\top \Sigma' v)^2\right]$ is a degree $4$ polynomial in $v \in \BR^d$: the claim is that this polynomial has a degree-4 sum of squares certificate of being nonnegative.

It can be proven that if $n$ is as in \Cref{lem:resilience-of-moments-covariance}, with probability $1-\beta$ over the choice of $\set{x_i}$, if we output $\pE{\Sigma'}$, then $\|\Sigma^{-1/2} (\pE{\Sigma'}) \Sigma^{-1/2} - I\|_{op} = \tcO\paren{\eta}$.

\subsection{Score Function and its Properties}
Our goal is to use 
\Cref{thm:pure_dp_general_main}, so we relax the pseudo-expectation from the robust algorithm to a linear operator that behaves as an approximate pseudoexpectation.

\begin{definition}[Certifiable Covariance]
\label{def:certifiable-covariance}
Let $\alpha, \tau, T \in \R^{\ge 0}$, $y_1, \dots y_n \in \R^d$ and let $\widetilde{\Sigma} \in \R^{d \times d}$ be PSD.
We call the point $\widetilde{\Sigma}$ an $\paren{\alpha, \tau, T}$-certifiable covariance for $\set{y_i}$ if and only if there exists a linear functional $\cL$ over the set of polynomials in indeterminates $\set{w_i}, \set{x_{i, j}'}, \set{M_{\{j, j'\}, \{k, k'\}}}$ 
%\set{A_{j, k}}, \set{B_{j, k}}
of degree at most $12$ such that 
\begin{enumerate}
    \item $\cL 1 = 1$
    \item for every polynomial $p$, where $\normt{\cR(p)} \le 1$
    %TODO: degree of polynomial p
    \begin{enumerate}
        \item $\cL p^2 \ge -\tau \cdot T$,
        \item $\forall i, \cL \paren{w_i^2 - w_i} p^2 \in \brac{-\tau \cdot T, \tau \cdot T}$,
        \item $\cL \paren{\sum w_i -n + T} p^2 \ge -5 \tau \cdot T \cdot n$,
        \item $\forall i, \cL w_i \paren{x_i' - y_i}p^2 \in \brac{-\tau \cdot T, \tau \cdot T}$,
    \end{enumerate}
    \item $\cL \brac{\frac{1}{n}\sum_{i=1}^n \left(\langle v, x_i'\rangle^2 - v^\top \Sigma' v\right)^2 + (v^{\otimes 2})^\top M^\top M v^{\otimes 2} - (2 + \alpha) (v^\top \Sigma' v)^2}$, as a degree-4 polynomial in $v = (v_1, \dots, v_d),$ has all coefficients between $[-\tau \cdot T, \tau \cdot T]$, where $\Sigma' := \E_i (x_i') (x_i')^\top$.
    \item  $\cL \brac{(1+\alpha) \Sigma' - \widetilde{\Sigma}} \succcurlyeq -\tau \cdot T \cdot I,$ and $\cL \brac{\widetilde{\Sigma} - (1-\alpha) \Sigma'} \succcurlyeq -\tau \cdot T \cdot I,$ where $\cL$ applied to a matrix is applied entrywise,
    \item $(\frac{1-\alpha/10}{K} - \tau \cdot T) \cdot I \preccurlyeq \cL[\Sigma'] \preccurlyeq ((1+\alpha/10) K + \tau \cdot T) \cdot I$.
    \item $\|\cR(\cL)\|_2 \le R'+ T \cdot \tau$, for some sufficiently large $R' = \poly(n, d, K)$. As in the mean estimation case, this requirement is only needed for computability purposes.
\end{enumerate}
%We also require $\|\cR(\cL)\|_2 \le R'+ T \cdot \tau$ for some sufficiently large $R' = \poly(n, d, K)$.
%, where $\cR(\cL)$ is the vector which represents the value of $\cL$ applied to each monomial of degree at most $12$. 
We will also say that $\cL$ is an $(\alpha, \tau, T)$-certificate for $\cY$.

Note that one may think of $\cL$ as an approximate pseudo-expectation, and it is clear that $\cL$ generalizes pseudo-expectations. In addition, for each constraint 2a) to 2d) we implicitly assume a bound on the degree of $p$ so that $\cL$ is applied to a polynomial of degree at most $12$.
\end{definition}

For our purposes, we will end up setting $\tau = 1/(K \cdot n \cdot d)^{O(1)}$, for a large enough $O(1)$.

Now we use this definition to define a score function.
\begin{definition}[Score Function]
    \label{def:score-covariance}
    Let $\mathbb{K}_{\alpha/2, K}^d$ denote the set of symmetric positive definite matrices $\Sigma \in \BR^{d \times d}$ with all eigenvalues between $\frac{1-\alpha/2}{K}$ and $(1+\alpha/2)K$.
    (Note that $\mathbb{K}_{\alpha/2, K}^d$ is convex.)
    Let $\alpha, \tau \in \R^{\ge 0}$, $y_1, \dots, y_n \in \R^d$ (with $\cY := \{y_1,\dots,y_n\}$), and 
    $\tSigma \in \mathbb{K}_{\alpha/2, K}^d$.
    %$\tSigma \in \R^{d \times d}$.
    We define the score function $\cS : \mathbb{K}_{\alpha/2, K}^d \to \R$ (viewed as a function of $\tSigma$) as
    \begin{equation*}
        \cS\paren{\tSigma, \cY; \alpha, \tau} = \min_{T \ge 0} \text{ such that $\tSigma$ is a $\paren{\alpha, \tau, T}$ certifiable covariance for $\set{y_i}$ }.
    \end{equation*}
\end{definition}

In the rest of this subsection we will prove the following properties for this score function. This will allow us to use \Cref{thm:pure_dp_general_main}.

\begin{enumerate}
    \item Score has sensitivity $1$.
    \item Score is quasi-convex as a function of $\tSigma$.
    \item All points $\tSigma$ that have score at most $\eta \cdot n$ have spectral distance at most $\tilde{O}(\eta)$ away from $\Sigma$. (Robustness for volume/accuracy purposes).
    \item The volume of points that have score at most $\eta \cdot n$ is sufficiently large, and the volume of points with score at most $\eta' \cdot n$ for $\eta' > \eta$ is not too large.
    \item Score is efficiently computable.
    \item We can approximately minimize score efficiently.
\end{enumerate}

%In the rest of this section, we assume $n, \alpha, \eta$ are fixed, with $n \ge \tcO\left(\frac{d^2+\log^2 (1/\beta)}{\eta^2}\right)$ and $\alpha = \tcO(\eta)$.
For simplicity, we may ignore the arguments $\alpha, \tau$ in the score $\cS$, and write $\cS(\tSigma, \cY)$.

\subsubsection{Existence of Low-Scoring $\tilde{\Sigma}$}

Before verifying the desired conditions of our score functions, we prove that for data points drawn from $\cN(\textbf{0}, \Sigma)$, with high probability every $\tilde{\Sigma}$ which is reasonably close to $\Sigma$ has low score. This will be important both for sensitivity and for volume bounds. 
While such results are already known in the literature~\cite{KothariS17} for \emph{certifiable} fourth moment bounds, which we will need to verify Condition 3 of \Cref{def:certifiable-covariance}, the previous result requires $n = \tilde{O}(d^2 \log^2(1/\beta)/\alpha^2)$, as opposed to our goal of $n = \tilde{O}((d^2 +\log^2(1/\beta))/\alpha^2)$. As a result, we reprove some known results to establish a low-scoring $\tSigma$, but with better failure probability bounds.

First, we note the following basic proposition which is immediate by Cauchy-Schwarz.

\begin{proposition} \label{prop:prod_trace_F}
    For any two matrices $A, B \in \BR^{d \times d}$, $|\Tr(AB)| \le \|A\|_F \cdot \|B\|_F$.
\end{proposition}

The following proposition is also well-known.

\begin{proposition} \label{prop:F_trace_F}
    For any two matrices $A, B \in \BR^{d \times d}$, $\|AB\|_F \le \|A\|_{op} \cdot \|B\|_F, \|B\|_{op} \cdot \|A\|_F \le \|A\|_F \cdot \|B\|_F$.
\end{proposition}

\begin{proposition} \label{prop:frobenius_replacement}
    Let $M \in \BR^{d \times d}$ be a real symmetric matrix, and let $J \in \BR^{d \times d}$ be any real-valued matrix (possibly not symmetric) such that $\|JJ^\top - I\|_{op} \le \alpha$. Then, $\|J^\top M J\|_F^2 = (1 \pm 3 \alpha) \cdot \|M\|_F^2$.
\end{proposition}

\begin{proof}
    Start by writing $$\|J^\top M J\|_F^2 = \Tr((J^\top M J)(J^\top M J)^\top) = \Tr(J^\top M JJ^\top M J) = \Tr(M JJ^\top M JJ^\top).$$ Now, write $JJ^\top = I + H$ for some symmetric matrix $H$ such that $\|H\|_{op} \le \alpha$. Therefore,
\begin{align*}
    \Tr(M JJ^\top M JJ^\top) &= \Tr(M (I+H) M (I+H)) \\
    &= \Tr(M^2) + \Tr(MHM) + \Tr(MMH) + \Tr(MHMH) \\
    &= \Tr(M^2) + 2 \Tr(M M H) + \Tr((MH)^2).
\end{align*}
    Now, by Propositions \ref{prop:prod_trace_F} and \ref{prop:F_trace_F}, we have that $|\Tr((MH)^2)| \le \|MH\|_F^2 \le \|M\|_F^2 \cdot \|H\|_{op}^2$. In addition, $|\Tr(MMH)| \le \|M\|_F \cdot \|MH\|_F \le \|M\|_F^2 \cdot \|H\|_{op}$. Since $\|H\|_{op} \le \alpha$, this implies that $\|J^\top M J\|_F^2 = \Tr(M^2) \pm 3 \alpha \cdot \|M\|_F^2 = (1 \pm 3 \alpha) \cdot \|M\|_F^2$.
\end{proof}

\begin{lemma} \label{lem:vol-lb-covariance}
    Let $n \ge O\paren{\paren{d^2 + \log^2 \paren{1/\beta}}/{\alpha^2}}$, for some $\alpha = \tcO(\eta)$, where $\alpha$ is sufficiently small.
    Also, suppose $\frac{1}{K} \cdot I \preccurlyeq \Sigma \preccurlyeq K \cdot I$. Say $\cX = \{x_1, \dots, x_n\} \sim \cN(\textbf{0}, \Sigma)$, and $\cY = \{y_1, \dots, y_n\}$ is an $\eta$-corruption of $\cX$. Then, for any $\tau \ge 0$ and for $T = \eta \cdot n$, with probability at least $1-\beta$, every $\tSigma$ of spectral distance at most $\alpha/2$ from $\Sigma$ (i.e., $\|(\Sigma)^{-1/2} \tSigma (\Sigma)^{-1/2} - I\|_{op} \le \alpha/2$) is an $(\alpha, \tau, T)$-certifiable covariance for $\cY$.
\end{lemma}

\begin{proof}
    As in the case for mean estimation, we use the fact that our linear operators generalize pseudo-expectations, which in turn generalize expectations over a single point. Again, we set $w_i = 1$ if $y_i = x_i$ and $0$ otherwise, and $x_i' = x_i$ for all $i.$ This also means that $\Sigma' = \frac{1}{n} \sum_{i=1}^n x_i x_i^\top$. For $T = \eta n$, it is clear that Constraints 1 and 2a-2d are all satisfied in \Cref{def:certifiable-covariance}.
    
    %To verify Constraint 3 in \Cref{def:certifiable-covariance}, we remark from \cite[Lemma 2.9, part 5]{KothariMZ22} that, assuming $\alpha = \tcO(\eta)$ $\frac{1}{n} \sum_i \langle x_i x_i^\top - \Sigma', P \rangle^2 \le (2+\alpha) \cdot \|\Sigma^{1/2} P \Sigma^{1/2}\|_F^2$. We remark that in their proof, they replace $x_i$ with $x_i - \bar{x}$ (for $\bar{x} = \frac{1}{n} \sum x_i$) and $\Sigma'$ with $\frac{1}{n} \sum (x_i-\bar{x})(x_i-\bar{x})^\top$, but the same proof follows 
    
    To verify Constraint 3 in \Cref{def:certifiable-covariance}, first note that $\Sigma^{-1/2} x_1, \dots, \Sigma^{-1/2} x_n \overset{i.i.d.}{\sim} \cN(\textbf{0}, I)$.
    Now, by part 2 of \Cref{lem:resilience-of-moments-covariance}, where we replace $\alpha$ with $\alpha/4$, we have $\frac{1}{n} \sum_{i=1}^{n} \langle \Sigma^{-1/2} x_i x_i^\top \Sigma^{-1/2} - I, P\rangle^2 \le (2+\alpha/2) \cdot \|P\|_F^2$ with probability at least $1-\beta$, for all $d \times d$ symmetric matrices $P$.
    We can write 
\begin{align*}
    \langle \Sigma^{-1/2} x_i x_i^\top \Sigma^{-1/2} - I, P \rangle &= \Tr [(\Sigma^{-1/2} x_i x_i^\top \Sigma^{-1/2} - I) \cdot P] \\
    &= \Tr[x_i x_i^\top \cdot \Sigma^{-1/2} P \Sigma^{-1/2} - P] \\
    &= \Tr[(x_i x_i^\top - \Sigma) \cdot (\Sigma^{-1/2} P \Sigma^{-1/2})] \\
    &= \langle x_i x_i^\top - \Sigma, \Sigma^{-1/2} P \Sigma^{-1/2}\rangle.
\end{align*}
    So, by replacing $P$ with $\Sigma^{1/2} P \Sigma^{1/2}$, we have that for all symmetric matrices $P$,
\begin{equation} \label{eq:rescaled-stability-1}
    \frac{1}{n} \sum_{i=1}^n \langle x_i x_i^\top - \Sigma, P \rangle^2 \le (2+\alpha/2) \cdot \|\Sigma^{1/2} P \Sigma^{1/2}\|_F^2.
\end{equation}

    Now, note that $\Sigma' = \frac{1}{n} \sum_{i=1}^{n} x_i x_i^\top$ satisfies $\|\Sigma^{-1/2} \Sigma' \Sigma^{-1/2} - I\|_F \le \alpha/100$ with probability at least $1-\beta$, by \Cref{cor:empirical-covariance} (replacing $\alpha$ with $\alpha/200$). Therefore, by setting $J = \Sigma^{-1/2} (\Sigma')^{1/2}$, we have $\|JJ^\top - I\|_F \le \alpha/100$, which means $\|(\Sigma')^{1/2} P (\Sigma')^{1/2}\|_F^2 = \|J^\top \Sigma^{1/2} P \Sigma^{1/2} J\|_F^2 \ge (1-3 \alpha/100) \cdot \|\Sigma^{1/2} P \Sigma^{1/2}\|_F^2$ by Proposition \ref{prop:frobenius_replacement}. In addition, since $\Sigma'$ is the empirical average of $x_i x_i^\top$, this means for any symmetric $P$,
\begin{align*}
    \frac{1}{n} \sum_{i=1}^n \langle x_i x_i^\top - \Sigma', P \rangle^2 &\le \frac{1}{n} \sum_{i=1}^n \langle x_i x_i^\top - \Sigma, P \rangle^2 \\
    &\le \frac{2+\alpha/2}{(1-3 \alpha/100)} \cdot \|(\Sigma')^{1/2} P (\Sigma')^{1/2}\|_F^2 \\
    &\le (2+\alpha) \cdot \|(\Sigma')^{1/2} P (\Sigma')^{1/2}\|_F^2. 
\end{align*}

    For fixed $\{x_i\}$ (and thus fixed $\Sigma'$), note that for a symmetric matrix $P$, $\langle x_i x_i^\top - \Sigma', P \rangle$ is a linear functional mapping $P$ to $\BR$, and $(\Sigma')^{1/2} P (\Sigma')^{1/2}$ is a linear map sending symmetric matrices $P$ to symmetric matrices. 
    For a symmetric matrix $P \in \BR^{d \times d}$, let $P^\flat \in \BR^{d^2}$ be the vector $\{P_{ij}\}_{i, j \le d}$, and let $(P^\flat)' \in \BR^{d(d+1)/2}$ be the vector $\{P_{ij}\}_{i \le j}$.
    So, if we consider the embedding $P \to (P^\flat)'$, there exist vectors $v_1, \dots, v_n \in \BR^{d(d+1)/2}$ (corresponding to taking inner product with $x_i x_i^\top - \Sigma'$ for each $i$) and a $\frac{d(d+1)}{2} \times \frac{d(d+1)}{2}$ matrix $J$ (corresponding to left- and right- multiplication by $(\Sigma')^{-1/2}$), such that $\frac{1}{n} \sum_{i=1}^{n} \langle v_i, (P^\flat)' \rangle^2 \le (2+\alpha) \cdot \|J \cdot (P^\flat)'\|_2^2$. Therefore, there is some other matrix $J'$ such that $\frac{1}{n} \sum_{i=1}^{n} \langle v_i, (P^\flat)' \rangle^2 + \|J' \cdot (P^\flat)'\|_2^2 = (2+\alpha) \cdot \|J \cdot (P^\flat)'\|_2^2$, meaning that
\[\frac{1}{n} \sum_{i=1}^{n} \langle x_i x_i^\top - \Sigma', P \rangle^2 + \|J' \cdot (P^\flat)'\|_2^2 = (2+\alpha) \cdot \|(\Sigma')^{1/2} P (\Sigma')^{1/2}\|_F^2.\]
    We can convert $J' \in \BR^{d(d+1)/2 \times d(d+1)/2}$ into a matrix $M \in \BR^{d(d+1)/2 \times d^2},$ by replacing any column in $J'$ corresponding to entry $(i,j)$ for $i < j$ with two copies for $(i, j)$ and $(j, i)$, each divided by $2$.
    Importantly, $J' \cdot (P^\flat)' = M \cdot P^\flat$.
    Therefore, for any $P = vv^\top$, since $P^\flat = \{v_i v_j\}_{i, j \le n} = v^{\otimes 2}$ and $(P^\flat)' = \{v_i v_j\}_{i \le j}$, there exists a matrix $M \in \BR^{d(d+1)/2 \times d^2}$ such that
\[\frac{1}{n} \sum_{i=1}^n \left(\langle v, x_i \rangle^2 - v^\top \Sigma' v \right)^2 + (v^{\otimes 2})^\top M^\top M v^{\otimes 2} = (2+\alpha) (v^\top \Sigma' v)^2.\]
    While $M$ is lacking in rows (it should have $d^2$ rows and columns), we can simply add additional $0$ rows.

    Therefore, the first 3 constraints are satisfied, and moreover, $\Sigma'$ has spectral distance at most $\alpha/100$ from $\Sigma$, which means Constraint 5 is also satisfied since $\frac{1}{K} \cdot I \preccurlyeq \Sigma \preccurlyeq K \cdot I$. We can choose any $\tSigma$ such that $(1-\alpha) \Sigma' \preccurlyeq \tSigma \preccurlyeq (1+\alpha) \Sigma'$ and $\tSigma \in \mathbb{K}_{\alpha/2, K}^d$, since then $(1+\alpha) \Sigma' - \tSigma$ and $\tSigma - (1-\alpha) \Sigma'$ are both PSD, so Constraint 4 is satisfied. Since with probability $1-\beta$ we have that $\|\Sigma^{-1/2} \Sigma' \Sigma^{-1/2} - I\|_F \le \alpha/100$, every $\tSigma$ such that $(1-\frac{\alpha}{2}) \Sigma \preccurlyeq \tSigma \preccurlyeq (1+\frac{\alpha}{2}) \Sigma$ also satisfies $(1-\alpha) \Sigma' \preccurlyeq \tSigma \preccurlyeq (1+\alpha) \Sigma'$. Moreover, any such $\tSigma$ is also in $\mathbb{K}_{\alpha/2, K}^d$, since $\Sigma$ has all eigenvalues between $\frac{1}{K}$ and $K$.
    %Hence, for every $\tSigma$ of spectral distance at most $\alpha/2$ from $\Sigma$, Constraint 4 is also satisfied.
    %, and thus can be written as $AA^\top, BB^\top$, respectively.

    Finally, we remark that every $w_i, x_{i, j}$, and $M_{\{j, k\}, \{j', k'\}}$ is bounded by $\poly(n, d, K)$. Therefore, the corresponding linear operator $\cL$ satisfies $\|\cR(\cL)\|_2 \le (Knd)^{O(1)}$, so Constraint 6 is satisfied.
\end{proof}

\subsubsection{Sensitivity}

The proof of sensitivity is similar to the mean estimation case.
We again have an upper bound of $n$ on the value of the score function. This time we can essentially use \Cref{lem:vol-lb-covariance}.

\begin{lemma}[score function upper bound]
\label{lem:covariance-score-function-upper}
Let $n \ge O\paren{\paren{d^2 + \log^2 \paren{1/\beta}}/{\alpha^2}}$, for some $\alpha = \tcO(\eta)$, and assume $\alpha, \eta$ are sufficiently small.
For any $\tau \ge 0$, any $y_1, \dots, y_n \in \BR^d$, and any $\tSigma$ with $\frac{1-\alpha/2}{K} \cdot I \preccurlyeq \tSigma \preccurlyeq K (1+\alpha/2) \cdot I$, the score function $\cS(\tSigma, \cY; \alpha, \tau)$, defined in \Cref{def:score-covariance}, is less than or equal to $n$. 
\end{lemma}
\begin{proof}
%It suffices to show that in \Cref{def:certifiable-covariance} for $T= n$, there exists a linear functional $\cL$ such that the constraints of \Cref{def:certifiable-covariance} are satisfied.
%
%Let's define $\cL$. We set $\cL 1 = 1$, and for any monomial $p$ that contains any $w_i$, or any $A_{j, k}$ or $B_{j, k}$, we set $\cL p = 0$.
%Clearly Constraint 1 in \Cref{def:certifiable-covariance} holds, and for $T = n$, the left hand side of Constraints 2b, 2c, 2d, 3, and 4 all evaluate to $0$. Finally, 2a holds because $\cL p^2 = \cL a^2 \ge 0$ if $a$ is the value of the $1$ monomial in $p$.
%
We use the fact that our linear operators generalize pseudo-expectations, which generalize expectations over a single point mass.
We will define a covariance $\Sigma$ as follows.
%Let $(1+\frac{\alpha}{10}) \cdot K^{-1} \cdot I \preccurlyeq \Sigma \preccurlyeq (1-\frac{\alpha}{10}) \cdot K \cdot I$ be such that $(1-\alpha) \Sigma \preccurlyeq \tSigma \preccurlyeq (1+\alpha) \Sigma$. Indeed, such a $\Sigma$ exists: 
If we write $\tSigma = U \Lambda U^\top$ for the diagonal matrix $\Lambda$ of eigenvalues, note that every $\Lambda_{ii}$ is between $\frac{1-\alpha/2}{K}$ and $K(1+\alpha/2)$. 
Let $\Lambda'$ equal the diagonal matrix where every $\Lambda'_{ii} = \min\left(K(1-\alpha/10), \max\left(\frac{1+\alpha/10}{K}, \Lambda_{ii}\right)\right)$, and define $\Sigma = U \Lambda' U^\top$. Note that $\frac{1+\alpha/10}{K} \cdot I \preccurlyeq \Sigma \preccurlyeq K(1-\alpha/10) \cdot I$ and $\frac{1-\alpha/10}{1+\alpha/2} \cdot \tSigma \preccurlyeq \Sigma \preccurlyeq \frac{1+\alpha/10}{1-\alpha/2} \cdot \tSigma$.

In \Cref{lem:vol-lb-covariance}, we showed that for $\cX = \{x_1, \dots, x_n\} \overset{i.i.d.}{\sim} \cN(\textbf{0}, \Sigma)$, we can set $x_i' = x_i$ and set $M$ to satisfy all of the constraints, with probability at least $1-\beta$. Moreover, by \Cref{cor:empirical-covariance} (replacing $\alpha$ with $\alpha/200$), $\|\Sigma^{-1/2} \Sigma' \Sigma^{-1/2}-I\|_F \le \alpha/100$, so $(1-\alpha/100) \Sigma \preccurlyeq \Sigma' \preccurlyeq (1+\alpha/100) \Sigma$. Thus, $K^{-1} \cdot I \preccurlyeq \Sigma' \preccurlyeq K \cdot I$. Moreover, $(1+\alpha) \Sigma' \succcurlyeq (1+\alpha)(1-\alpha/100) \Sigma \succcurlyeq \frac{1+\alpha/2}{1-\alpha/10} \Sigma \succcurlyeq \tSigma$ and $(1-\alpha) \Sigma' \preccurlyeq (1-\alpha)(1+\alpha/100) \Sigma \preccurlyeq \frac{1-\alpha/2}{1+\alpha/10} \Sigma \preccurlyeq \tSigma$. So, there exists a set $\cX$ that satisfies the constraints, which means for a general set of data points $\cY = \{y_1, \dots, y_n\}$, the score is at most $n$, since we can set $w_i = 0$ and $x_i' = x_i$ for all $i$.
%Thus, for $T = n$, it is clear that all constraints are satisfied.
\end{proof}

\begin{lemma}[sensitivity] \label{lem:sensitivity-covariance}
For any $\tSigma \in \mathbb{K}_{\alpha/2, K}^d$, i.e., $\frac{1-\alpha/2}{K} \cdot I \preccurlyeq \tSigma \preccurlyeq K (1+\alpha/2) \cdot I$, $\cS(\tSigma, \cY)$ has sensitivity $1$ with respect to $\cY$.
\end{lemma}
\begin{proof}
Suppose that $\cY$, $\cY'$ are two neighboring datasets, and $\tSigma \in \R^{d \times d}$. Moreover, assume $\cS\paren{\tSigma, \cY} = T$.
If we show that $\cS\paren{\tSigma, \cY'} \le \cS\paren{\tSigma, \cY} = T+1$, by symmetry we are done.

Without loss of generality assume $\cY$ and $\cY'$ differ on index $j$.
In order to construct $\cL'$, for any monomial $p$, let
\begin{equation*}
\cL' p = \begin{cases}
0 &\text{if $p$ has a $w_j$ factor},\\
\cL p &\text{otherwise}
\end{cases}.
\end{equation*}
To verify the constraints, Constraints 1 and 2a-2d are identical to in the mean estimation case (where checking Constraint 2c applies \Cref{lem:covariance-score-function-upper}). Also, $\|\cR(\cL')\|_2 \le \|\cR(\cL)\|_2$ clearly holds. So, we just need to verify Constraints 3, 4, and 5 in \Cref{def:certifiable-covariance}.

However, note that these three constraints do not involve $w_j$ at all, so in fact their evaluation is the same regardless of $\cL$ or $\cL'$. The only difference is we are allowing the values $\cL[\cdot]$ to have a greater range, which makes it easier.
\end{proof}

\subsubsection{Quasi-convexity}

\begin{lemma}[quasi-convexity] \label{lem:quasi-convex-covariance-spectral}
The score function $\cS$ is quasi-convex in $\tSigma$.
\end{lemma}
\begin{proof}
Suppose $\cS\paren{\tSigma_1, \cY} = T_1, \cS\paren{\tSigma_2, \cY} = T_2$, and suppose there exists $\cL_1$ and $\cL_2$ that satisfy the constraints in \Cref{def:certifiable-covariance} with $\tSigma_1, T_1$,and $\tSigma_2, T_2$ respectively.
If we can construct a functional $\cL_3$ such that the constraints in \Cref{def:certifiable-covariance}, are satisfied with $\tSigma_3 = \lambda \tSigma_1 + \paren{1- \lambda} \tSigma_2$, and $T_3 = \max\set{T_1, T_2}$, we are done.
Let $\cL_3 = \lambda \cL_1 + \paren{1 - \lambda} \cL_2$.
As in the mean estimation case, all of the constraints in \Cref{def:certifiable-covariance} will be satisfied trivially except for Constraints 2c and 5, and Constraint 2c is the same as in the mean estimation case. So, the same verification implies that this constraint is also satisfied. Constraint 5 is also straightforward, since if $(\frac{1-\alpha/10}{K} - \tau \cdot T_1) \cdot I \preccurlyeq \cL_1[\Sigma'] \preccurlyeq ((1+\alpha/10)K + \tau \cdot T_1) \cdot I$ and $(\frac{1-\alpha/10}{K} - \tau \cdot T_2) \cdot I \preccurlyeq \cL_2[\Sigma'] \preccurlyeq ((1+\alpha/10)K + \tau \cdot T_2) \cdot I$, then $(\frac{1-\alpha/10}{K} - \tau \cdot \max\set{T_1, T_2}) \cdot I \preccurlyeq \lambda \cdot \cL_1[\Sigma'] + (1-\lambda) \cdot \cL_2[\Sigma'] \preccurlyeq ((1+\alpha/10)K + \tau \cdot \max\set{T_1, T_2}) \cdot I$. 
\end{proof}

\subsubsection{Accuracy}

We show that any point $\tSigma$ of low score with respect to i.i.d.\ samples from $\cN(\textbf{0}, \Sigma)$ must be close to $\Sigma$ in spectral distance, i.e., $\|\Sigma^{-1/2} \tSigma \Sigma^{-1/2} - I\|_{op} \le O(\alpha)$.

\begin{lemma} \label{lem:covariance-accuracy}
    Let $\alpha = \widetilde{O}(\eta)$ and suppose $\alpha, \eta$ are bounded by a sufficiently small constant. Let $n \ge \tcO\left(\frac{d^2 + \log^2 (1/\beta)}{\alpha^2}\right)$, and $\cX = \{x_1, \dots, x_n\} \sim \cN(\textbf{0}, \Sigma)$, for $K^{-1} I \preccurlyeq \Sigma \preccurlyeq K \cdot I$.
    Also, suppose $\tau \ll (nd K/\eps)^{-O(1)}$.

    Then, for any $\alpha^* \le \alpha$, with probability at least $1-\beta$, every symmetric matrix $\tSigma \in \BR^{d \times d}$ that is $(\alpha^*, \tau, T)$-certifiable for $\cX$ with $T = \eta n$ must satisfy $\|\Sigma^{-1/2} \tSigma \Sigma^{-1/2} - I\|_{op} \le O(\alpha)$.
\end{lemma}

As in the mean estimation case, the proof follows the same approach as~\cite{KothariMZ22}, so we defer this to \Cref{appendix:sos}.

\subsubsection{Volume of Good Points}

\begin{lemma} \label{lem:vol-ub-covariance}
    Let $\cX = \{x_1, \dots, x_n\} \sim \cN(\textbf{0}, \Sigma)$, and let $\cY = \{y_1, \dots, y_n\}$ represent an $\eta$-corruption of $\cX$. Then, for every integer $T \in [\eta \cdot n, \eta^* \cdot n]$ for some fixed constant $\eta^* < 1$, with probability at least $1-\beta$, every $(\alpha, \tau, T)$-certifiable covariance with respect to $\cY$ has spectral distance at most $\tcO(T/n)$ from $\Sigma$.
\end{lemma}

\begin{proof}
    Since the score function has sensitivity at most 1 (\Cref{lem:sensitivity-covariance}),
    this means that any $(\alpha, \tau, T)$-certifiable mean with respect to $\cY$ is an $(\alpha, \tau, T+\eta n)$-certifiable mean with respect to $\cX$.
    
    Now, define $\eta' := \frac{T+\eta n}{n} = O(\frac{T}{n})$. In this case, by setting $\alpha' = \tcO(\eta')$ and since $\alpha = \tcO(\eta) \le \alpha'$, we have that by \Cref{lem:covariance-accuracy} that any $(\alpha, \tau, T+\eta n)$-certifiable covariance $\tSigma$ must satisfy $\|\Sigma^{-1/2} \tSigma \Sigma^{-1/2} -I\|_{op} \le O(\alpha') \le \tcO(T/n)$. 
\end{proof}

So, for any $\eta' \in [\eta, \eta^*]$, any $\tSigma$ with score at most $\eta' \cdot n$ with respect to the $\cY$ in \Cref{lem:vol-ub-covariance} must satisfy $\|\Sigma^{-1/2} \tSigma \Sigma^{-1/2}\|_{op} \le \tilde{O}(\eta') \le 1/2$.
%be of the form $(1-\alpha') \cdot \Sigma + 2 \alpha' \cdot R$ where $0 \preccurlyeq R \preccurlyeq \Sigma$ and $\alpha' = \tcO(\eta')$.
So, if we define $V_\Sigma$ to the set of PSD matrices spectrally bounded by $\Sigma$ (where we think of symmetric matrices as vectors in $\BR^{d(d+1)/2}$), the set of $(\alpha, \tau, \eta' n)$-certifiable covariance matrices with respect to $\cY$ has volume at most $e^{O(d^2)} \cdot V_\Sigma$.
%, meaning $V_{\eta'} \le O(1)^{d^2} \cdot V_\Sigma$ for $\eta' \in [\eta, \eta^*]$. 
%In addition, by Constraint 5 of \Cref{def:certifiable-covariance}, we know that $\cL[\Sigma']$ is always spectrally bounded between $\frac{1}{4 K} \cdot I$ and $4 K \cdot I$, and so $\tSigma$ is spectrally bounded between $0$ and $8 K \cdot I \preccurlyeq 8 K^2 \cdot \Sigma$. Thus, for any $\eta' \in [\eta^*, 1]$, $V_{\eta'} \le O(K^2)^{d^2} \cdot V_{\Sigma} = e^{O(d^2 \log K)} \cdot V_{\Sigma}$.
Next, by \Cref{lem:vol-lb-covariance}, with probability at least $1-\beta$ every $\tSigma$ within spectral distance $\alpha/2$ of $\Sigma$ has score $\cS(\tSigma, \cY) \le \eta \cdot n$, and moreover is in $\mathbb{K}_{\alpha/2, K}^d$. 
Thus, the set of $(\alpha, \tau, \eta n)$-certifiable covariance matrices in $\mathbb{K}_{\alpha/2, K}^d$ with respect to $\cY$ has volume at least $V_{\Sigma} \cdot (\alpha/2)^{d^2}$.
Finally, for $T > \eta^* n$, any $(\alpha, \tau, T)$-certifiable covariance matrix in $\mathbb{K}_{\alpha/2, K}^d$ must have operator norm bounded by $K$, so the volume of such certifiable covariances at most $K^{O(d^2)} \cdot V_{\Sigma}$ since every eigenvalue of $\Sigma$ is at least $\frac{1}{2K}.$
%Thus, $V_{\eta} \ge V_{\Sigma} \cdot (\alpha/2)^{d^2}$.

\subsubsection{Efficient Computability}

As in the mean estimation case, we apply \Cref{thm:computing-score} in \Cref{sec:computing-score-functions}. This time, there are constraints where we wish to spectrally bound $\cL$ applied to a matrix. However, this constraint is also captured by \Cref{thm:computing-score}. So, we have efficient computability.

\subsubsection{Efficient Finding of Low-Scoring Point}

To verify that the ``robust algorithm finds low-scoring point'', we simply remove the constraint that $\cL \brac{(1+\alpha) \Sigma' - \widetilde{\Sigma}} \succcurlyeq -\tau \cdot T \cdot I,$ and $\cL \brac{\widetilde{\Sigma} - (1-\alpha) \Sigma'} \succcurlyeq -\tau \cdot T \cdot I$. We can apply \Cref{thm:computing-score} in the same way to find some linear operator $\cL$ with score at most $\min_{\tSigma} \cS(\tSigma, \cY) + 1$. Then, we can compute $\cL[\Sigma']$ and set $r \le \tau$, and obtain that every matrix $\tSigma$ such that $\|\tSigma - \cL[\Sigma']\|_F \le r$ satisfies $\frac{1-\alpha/2}{K} \cdot I \preccurlyeq \tSigma \preccurlyeq K (1+\alpha/2) \cdot I$, and has score at most $\min_{\tSigma} \cS(\tSigma, \cY) + 1$.

\subsection{Proof of Theorem \ref{thm:gaussian-covariance-pure-conditioning}}

We apply Theorem \ref{thm:pure_dp_general_main}, using the score function defined in \Cref{def:score-covariance} and thinking of the candidate parameters $\Sigma$ as lying in $\Theta := \mathbb{K}_{\alpha/2, K}^d \subset \BR^{d(d+1)/2}$. Indeed, for $r = \alpha/(d \cdot K)^{O(1)}$ and $R = K^{O(1)}$, we have verified all conditions, as long as $n \ge \tcO((d^2+\log^2 (1/\beta))/\eta^2)$. Therefore, we have an $\eps$-DP algorithm running in time $\poly(n, d, \log \frac{K d}{\alpha}) = \poly(n, \log K)$ that finds a candidate covariance $\tSigma$ of score at most $2 \eta n$, as long as
\[n \ge O\left(\max_{\eta': \eta \le \eta' \le 1} \frac{\log (V_{\eta'}(\cY)/V_\eta(\cY)) + \log (1/(\beta \cdot \eta'))}{\eps \cdot \eta'}\right).\]

By our volume bounds, this means it suffices for
\begin{align*}
    n &\ge \tcO\left(\frac{d^2+\log^2 (1/\beta)}{\alpha^2}\right) + O\left(\max_{\eta \le \eta' \le \eta^*} \frac{d^2 \log (1/\alpha) + \log (1/(\beta \cdot \eta))}{\eps \cdot \eta'} + \max_{\eta^* \le \eta' \le 1} \frac{d^2 \log (K/\alpha) + \log (1/(\beta \cdot \eta))}{\eps \cdot \eta'}\right) \\
    &= \tcO\left(\frac{d^2+\log^2 (1/\beta)}{\alpha^2} + \frac{d^2 + \log (1/\beta)}{\eps \cdot \alpha} + \frac{d^2 \log K}{\eps}\right).
\end{align*}

Hence, our algorithm, using this many samples, can find a point $\tSigma$ of score at most $2 \eta n$ with respect to $\cY$, which means it has score at most $3 \eta n$ with respect to the uncorrupted samples $\cX$. Finally, by replacing $\eta$ with $3 \eta$ and applying \Cref{lem:covariance-accuracy}, we have that any point $\tSigma$ with score at most $3 \eta n$ with respect to $\cX$ is within $O(\alpha)$ spectral distance of $\Sigma$. This completes the proof.

\subsection{The approx-DP setting}

In this subsection, we prove Theorem \ref{thm:gaussian-covariance-approx-conditioning}.
In this setting, the score function is identical, but we can afford fewer samples as we apply the algorithm of \Cref{thm:approx_dp_general_main} instead of \Cref{thm:pure_dp_general_main}. The main additional thing we must check is that for \emph{any} dataset $\cY$, if $\cS(\Sigma, \cY) \le 0.7 \eta^* n$ for some $\Sigma$, then the volume ratio $V_{\eta^*}(\cY)/V_{0.8 \eta^*}(\cY)$ is not too high.

Before proving our main result of this subsection, we must first establish the following lemma, which is important for ensuring privacy. We defer the proof to \Cref{appendix:sos}.

\begin{lemma} \label{lem:approx-spectral-accuracy}
    Fix $\eta^*$ to be a sufficiently small constant, and $T = \eta^* n$. Then, for a dataset $\cY$ with every $y_i$ bounded in $\ell_2$ norm by $K \cdot d^{100}$, if there exist linear operators $\cL_1, \cL_2$ that are both $(\alpha, \tau, T)$-certificates for $\cY$, then $\tSigma_1 \preccurlyeq O(1) \cdot \tSigma_2$ and $\tSigma_2 \preccurlyeq O(1) \cdot \tSigma_1$.
\end{lemma}

As a corollary of \Cref{lem:approx-spectral-accuracy}, we have the following result.

\begin{corollary} \label{cor:approx-spectral-volume}
    Suppose that $\cY$ is a dataset with every $y_i$ bounded in $\ell_2$ norm by $K \cdot d^{100}$ that has an $(\alpha, \tau, 0.7 \eta^* n)$-certifiable covariance, and let $\hat{\Sigma} = \cL[\Sigma']$ where $\cL$ is an $(\alpha, \tau, 0.7 \eta^* n)$-certificate. Then, the set of $(\alpha, \tau, 0.8 \eta^* n)$-certifiable covariance matrices $\tSigma \subset \mathbb{K}_{\alpha/2, K}^d$ contains all matrices spectrally bounded between $(1-\alpha/5) \hat{\Sigma}$ and $(1+\alpha/5) \hat{\Sigma}$, and the set of $(\alpha, \tau, \eta^* n)$-certifiable covariance matrices is spectrally bounded between $\frac{1}{C} \cdot \hat{\Sigma}$ and $C \cdot \hat{\Sigma}$ for some constant $C = O(1)$.
\end{corollary}

\begin{proof}
    If $\cL$ is an $(\alpha, \tau, 0.7 \eta^* n)$-certificate, it is also an $(\alpha, \tau, 0.8 \eta^* n)$-certificate. This means every $\tSigma$ such that $(1-\alpha/5) \hat{\Sigma} \preccurlyeq \tSigma \preccurlyeq (1+\alpha/5) \hat{\Sigma}$ is $(\alpha, \tau, 0.8 \eta^* n)$-certifiable.
    To see why, note that for a $(\alpha, \tau, 0.8 \eta^* n)$-certificate $\cL$ of $\cY$, Constraint 4 (which is the only constraint that deals with $\tSigma$, which we recall is not indeterminate) just requires that $\cL[(1+\alpha) \Sigma' - \tSigma] \succcurlyeq -\tau \cdot T \cdot I$ and $\cL[\tSigma - (1-\alpha) \Sigma'] \succcurlyeq -\tau \cdot T \cdot I$. So, any $\tSigma$ spectrally bounded between $(1-\alpha/5) \hat{\Sigma}$ and $(1+\alpha/5) \hat{\Sigma}$ is an $(\alpha, \tau, 0.8 \eta^* n)$-certifiable covariance.
    Moreover, $\cL[\Sigma']$ must have all eigenvalues between $\frac{1-\alpha/10}{K} - \tau \cdot T \ge \frac{1-\alpha/5}{K}$ and $(1+\alpha/10)K+\tau \cdot T \le (1+\alpha/5) K$, by Constraint 5. Thus, any such $\tSigma$ has all eigenvalues between $\frac{1-\alpha/2}{K}$ and $(1+\alpha/2)K$, which means it is in $\mathbb{K}_{\alpha/2, K}^d$.
    %For any choice of $\tSigma$ such that $(1-\alpha) \hat{\Sigma} \preccurlyeq \tSigma \preccurlyeq (1+\alpha) \hat{\Sigma}$, we will choose $A, B$ so that $AA^\top = (1+\alpha) \hat{\Sigma} - \tSigma$ and $BB^\top = \tSigma - (1-\alpha) \hat{\Sigma}$.
    %We can modify the linear operator $\cL$ to some $\cL'$ so that for any monomial of the form $p \cdot p'$ where $p$ is only in variables $\{w_i\}$, $\{x'_{i, j}\}$, and $\{M_{\{j, j'\}, \{k, k'\}}\}$ and $p'$ is only in variables $\{A_{j, k}\}$ and $\{B_{j, k}\}$, $\cL'[p \cdot p'] = \cL[p] \cdot p'$, where we think of $p'$ as constant 
    
    The second part is immediate by Lemma \ref{lem:approx-spectral-accuracy}.
    %\todo{Deal with fact that $y_i$'s are bounded by $K \cdot d^{100}$}
\end{proof}

Therefore, if we let $V_{\hat{\Sigma}}$ represent the volume of PSD matrices spectrally bounded above by $\hat{\Sigma}$ (where we think of symmetric matrices as vectors in $\BR^{d(d+1)/2}$), the set of $(\alpha, \tau, \eta^* n)$-certifiable covariance matrices has volume at most $O(1)^{d^2} \cdot V_{\hat{\Sigma}}$ and the set of $(\alpha, \tau, 0.8 \eta^* n)$-certifiable covariance matrices has volume at least $\alpha^{d^2} \cdot V_{\hat{\Sigma}}$. So, the ratio $V_{\eta^*}(\cY)/V_{0.8 \eta^*}(\cY) \le O(1/\alpha)^{d^2}$.

\medskip

We now prove \Cref{thm:gaussian-covariance-approx-conditioning}, by applying \Cref{thm:approx_dp_general_main}. First, note that we may truncate the samples so that no $y_i \in \cY$ has norm more than $K \cdot d^{100}$. Since we are promised $\|\Sigma\|_{op} \le K$, the probability that any uncorrupted sample has this norm is at most $e^{-d^{100}}$. We will set $\eta^*$ to be a sufficiently small constant (such as $0.01$). We just showed, using \Cref{cor:approx-spectral-volume}, that for all $\cY$ such that $\min_{\tSigma} \cS(\tSigma, \cY) \le 0.7 \eta^* n$, $V_{\eta^*}(\cY)/V_{0.8 \eta^*}(\cY) \le O(1/\alpha)^{d^2}$. So, as long as $n \ge O\left(\frac{\log (1/\delta) + d^2 \log (1/\alpha)}{\eps}\right),$ the algorithm of \Cref{thm:approx_dp_general_main} is $(\eps, \delta)$-differentially private. In addition, we have already verified all of the conditions, so the algorithm is accurate as long as $n \ge \tcO((d^2 + \log^2 (1/\beta))/\eta^2)$ and 
\[n \ge O\left(\max_{\eta': \eta \le \eta' \le \eta^*} \frac{\log (V_{\eta'}(\cY)/V_\eta(\cY)) + \log (1/(\beta \cdot \eta'))}{\eps \cdot \eta'}\right).\]

By our volume bounds, this means it suffices for
\begin{align*}
    n &\ge \tcO\left(\frac{d^2+\log^2 (1/\beta)}{\alpha^2}\right) + O\left(\frac{\log (1/\delta) + d^2 \log(1/\alpha)}{\eps}\right) + O\left(\max_{\eta \le \eta' \le \eta^*} \frac{d^2 \log (1/\eta) + \log (1/(\beta \cdot \eta))}{\eps \cdot \eta'}\right) \\
    &= \tcO\left(\frac{d^2 +\log^2 (1/\beta)}{\alpha^2} + \frac{d^2 + \log (1/\beta)}{\eps \cdot \alpha} + \frac{\log (1/\delta)}{\eps}\right).
\end{align*}
This concludes the proof of Theorem \ref{thm:gaussian-covariance-approx-conditioning}.

\section{Learning a Gaussian in Total Variation Distance}
\label{sec:tv}

The main result we prove in this section is is to privately learn the covariance $\Sigma$ of a Gaussian up to low Frobenius norm error, if we are promised all eigenvalues of $\Sigma$ are between $(1-\alpha)$ and $(1+\alpha)$.

\begin{theorem}[Privately Learning a Preconditioned Gaussian] \label{thm:gaussian-tv-1}
Let $\Sigma \in \R^{d \times d}$ where $(1-\alpha) \cdot I \preccurlyeq \Sigma \preccurlyeq (1+\alpha) \cdot I$. There exists an $\epsilon$-differentially private algorithm that takes $n$ \iid\ samples from $\cN\paren{\textbf{0}, \Sigma}$ and with probability $1-\beta$ outputs $\tSigma$ such that $\|\tSigma - \Sigma\|_F \le O(\alpha)$, where
\begin{equation*}
n =
\tcO \Paren{
\frac{(d + \log (1/\beta))^2}{\alpha^2}
+
\frac{d^2 + \log (1/\beta)}{\alpha \eps}
}.
\end{equation*}
%Here $\tcO$ is hiding %TODO 
%factors.
Moreover, this algorithm runs in time $\poly(n, d)$, and succeeds with the same accuracy even if $\eta = \tilde{\Omega}(\alpha)$ fraction of the points are adversarially corrupted.
\end{theorem}

By combining \Cref{thm:gaussian-tv-1} with \Cref{thm:gaussian-mean-main} and \Cref{thm:gaussian-covariance-pure-conditioning} (or \Cref{cor:gaussian-mean-main} and \Cref{thm:gaussian-covariance-approx-conditioning}), we will be able to prove our main results on privately learning Gaussians up to low total variation distance, namely Theorems \ref{thm:pure-dp-tv} and \ref{thm:approx-dp-tv}. We prove these theorems in \Cref{subsec:main_stuff}.

\subsection{Robust Algorithm}
Suppose $\{x_i\}$ is a set of samples from $\cN\paren{\textbf{0}, \Sigma},$ where $(1-\alpha) \cdot I \preccurlyeq \Sigma \preccurlyeq (1+\alpha) \cdot I$.
Let $\{y_i\}$ be an arbitrary $\eta$-corruption of $\{x_i\}$. 
Consider the following pseudo-expectation program, where $\{y_i\}$ are the input points and the domain is the degree-$12$ pseudo-expectations with $\set{w_i}, \set{x_i},\set{M_{\{j, j'\}, \{k, k'\}}}$ as indeterminates.
\begin{align*}
\text{ find $\pE$ } & \\
\text{ such that } & \pE \text{ satisfies } w_i^2 = w_i, \\
& \pE \text{ satisfies } \sum w_i \ge (1- \eta) n, \\
& \pE \text{ satisfies } w_i x_i' = w_i y_i, \\
& \pE \left[\frac{1}{n} \sum_i \paren{x_i' \otimes x_i' - S'} \transpose{\paren{x_i' \otimes x_i'- S'}} + M\transpose{M}\right] = \paren{2+ \tcO\paren{\eta}} I, \text{where $S' = \frac{1}{n} \sum_i x_i' \otimes x_i'$}. 
\end{align*}
We use $\Sigma'$ to represent $\frac{1}{n} \sum_i (x_i') (x_i')^\top$: note that $S'$ is the flattening of $\Sigma'$.
It can be proven that if $n$ is as in \Cref{lem:resilience-of-moments-covariance}, with probability $1-\beta$ over the choice of $\set{x_i}$, if we output $\pE{\Sigma'}$, then $\|\pE{\Sigma'} - \Sigma\|_F = \tcO\paren{\eta}$.

\subsection{Score Function and its Properties}
Again, we need design a suitable score function based on the robust algorithm, this time for learning covariance up to low Frobenius norm error.

\begin{definition}[Certifiable Covariance]
\label{def:certifiable-TV}
Let $\alpha, \tau, \phi, T \in \R^{\ge 0}$, $y_1, \dots y_n \in \R^d$ ($\cY := \{y_1, \dots, y_n\}$), and $\tSigma \in \R^d$.
We call the point $\tSigma$ an $\paren{\alpha, \tau, \phi, T}$-\emph{certifiable covariance} for $\cY$ if and only if there exists a linear functional $\cL$ over the set of polynomials in indeterminates $\set{w_i}, \set{x_{i, j}'}, \set{M_{\{j, k\}, \{j', k'\}}}$ of degree at most $12$ such that 
\begin{enumerate}
    \item $\cL 1 = 1$
    \item for every polynomial $p$, where $\normt{\cR(p)} \le 1$:
    %TODO: degree of polynomial p
    \begin{enumerate}
        \item $\cL p^2 \ge -\tau \cdot T$,
        \item $\forall i, \cL \paren{w_i^2 - w_i} p^2 \in \brac{-\tau \cdot T, \tau \cdot T}$,
        \item $\cL \paren{\sum w_i -n + T} p^2 \ge -5 \tau \cdot T \cdot n$,
        \item $\forall i, j, \cL w_i \paren{x_{i, j}' - y_{i, j}}p^2 \in \brac{-\tau \cdot T, \tau \cdot T}$,
    \end{enumerate}
    \item $\forall j, k:\cL \Paren{\frac{1}{n}\sum_i \Brac{ \paren{(x_i')^{\otimes 2} - S'}\transpose{\paren{(x_i')^{\otimes 2} - S'}}+M\transpose{M} - \paren{2+24\alpha} I_{d^2}}_{\{j, j'\}, \{k, k'\}}} \in \brac{-\tau \cdot T, \tau \cdot T}$, where $x_i' = \{x_{i, j}'\}_{1 \le j \le d}$, and $S' = \E_i x_i'^{\otimes 2}$.
    %Note that here $\brac{\dots}_{j,k}$ denotes the $(j,k)$ entry of a matrix, which is a polynomial in indeterminates $\set{w_i}, \set{x_i'}, \set{M_{j, k}}$. We write in this format for the sake of conciseness.
    \item $(1-2\alpha - \phi/2 - \tau \cdot T) \cdot I \preccurlyeq \cL[\Sigma'] \preccurlyeq (1+2\alpha + \phi/2 + \tau \cdot T) \cdot I$ and $-(\phi + \tau \cdot T) \cdot I \preccurlyeq \cL \brac{\Sigma' - \widetilde{\Sigma}} \preccurlyeq (\phi + \tau \cdot T) \cdot I,$ where $\Sigma' = \E_i (x_i') (x_i')^\top$ and where $\cL$ applied to a matrix is applied entrywise.
    %\item  $\forall j, k \in [d], \|\cL \Sigma_{j, k}' - \tSigma_{j, k}\| \le \phi + \tau \cdot T$, where $\Sigma' = \frac{1}{n} \sum_{i=1}^n x_i' (x_i')^\top$.
    \item $\|\cR(\cL)\|_2 \le R' + T \cdot \tau$ for some sufficiently large $R' = \poly(n, d)$. As in the mean estimation case, this requirement is only needed for computability purposes.
\end{enumerate}
%, where $\cR(\cL)$ is the vector which represents the value of $\cL$ applied to each monomial of degree at most $12$.
We will also say that $\cL$ is an $(\alpha, \tau, \phi, T)$-certificate for $\cY$.

Again, we may think of $\cL$ as an approximate pseudo-expectation. In addition, for each constraint 2a) to 2d) we implicitly assume a bound on the degree of $p$ so that $\cL$ is applied to a polynomial of degree at most $12$.
\end{definition}

For our purposes, we will end up setting $\tau = 1/(n \cdot d)^{O(1)}$, for a large enough $O(1)$. From now on, we also assume $\phi := \alpha/\sqrt{d}$.

Now we use this definition to define a score function.
\begin{definition}[Score Function]
    \label{def:score-TV}
    Let $\mathbb{K}_{2\alpha+\phi}^d$ denote the set of covariance matrices with all eigenvalues between $1-(2\alpha+\phi)$ and $1+(2\alpha+\phi)$.
    Let $\alpha, \tau, \phi, T \in \R^{\ge 0}$, $y_1, \dots y_n \in \R^d$ (with $\cY = \{y_1, \dots, y_n\}$) and $\tSigma \in \mathbb{K}_{2\alpha+\phi}^d$.
    
    We define the score function $\cS : \mathbb{K}_{2\alpha+\phi}^d \to \R$ (viewed as a function of $\tSigma$) as
    \begin{equation*}
        \cS\paren{\tSigma, \cY; \alpha, \tau, \phi} = \min_T \text{ such that $\tSigma$ is a $\paren{\alpha, \tau, \phi, T}$ certifiable covariance for $\cY = \{y_1, \dots, y_n\}$}.
    \end{equation*}
\end{definition}

In the rest of this section we will prove the following properties for this score function. This will allow us to use \Cref{thm:pure_dp_general_main}.

\begin{enumerate}
    \item Score has sensitivity $1$.
    \item Score is quasi-convex as a function of $\tSigma$.
    \item All points $\tSigma$ that have score at most $\eta \cdot n$ have $\|\tSigma-\Sigma\|_F \le \tcO(\eta)$. (Robustness for volume/accuracy purposes).
    \item The volume of points that have score at most $\eta \cdot n$ is sufficiently large, and the volume of points with score at most $\eta' \cdot n$ for $\eta' > \eta$ is not too large.
    \item Score is efficiently computable.
    \item We can approximately minimize score efficiently.
\end{enumerate}

Checking these constraints will, for the most part, be identical to the cases for mean estimation and covariance estimation in spectral distance. So for the sake of brevity, we omit any details that are essentially identical to these cases.

\subsubsection{Existence of Low-Scoring $\Sigma'$}

As in the case of covariance estimation, we must show that for for $\eta$-corrupted samples $\cY$, there exists a range of $\tSigma$ where $\cS(\cY, \tSigma)$ is small.
%for data points drawn from $\cN(\textbf{0}, \Sigma)$, that some $\Sigma'$ close to $\Sigma$ has low score. In this setting, it actually turns out to be easier, because the dataset has already been well-conditioned and since the robust algorithm/score function are slightly easier to work with.

\begin{lemma} \label{lem:vol-lb-covariance-tv}
    Suppose that $n \ge \tcO\left(\frac{d^2+\log^2(1/\beta)}{\eta^2}\right)$ and $\alpha = \tcO(\eta)$.
    Let $\cX = \{x_1, \dots, x_n\} \overset{i.i.d.}{\sim} \cN(\textbf{0}, \Sigma)$, where $\|\Sigma - I\|_{op} \le \alpha$, and let $\cY = \{y_1, \dots, y_n\}$ represent an $\eta$-corruption of $\cX$. Then, with probability at least $1-\beta$, for $\Sigma' = \E_i x_i x_i^\top$, every $\tSigma$ such that $\|\tSigma - \Sigma'\|_{op} \le \phi$ is $(\alpha, \tau, \phi, \eta n)$-certifiable with respect to $\cY$.
\end{lemma}

\begin{proof}
    Again, we use the fact that $\cL$ generalizes pseudoexpectations, which generalize expectations over a single data point. We will set $w_i = 1$ if $x_i = y_i$ and $0$ otherwise, and $x_i' = x_i$ for all $i$. By part 2 of \Cref{lem:resilience-of-moments-covariance}, we know that for all $d \times d$ symmetric matrices $P$ with $\|P\|_F = 1$, if $x_i \overset{i.i.d.}{\sim} \cN(\textbf{0}, I)$, then $\left|\E_i \langle x_i x_i^\top - I, P \rangle^2 - 2\right| \le \alpha$. 
    %This implies for any real (possibly non-symmetric matrix) $P$, we still have $\left|\frac{1}{n} \langle x_i x_i^\top - I, P \rangle^2\right| \le (2+O(\alpha)) \cdot \|P\|_F^2$, if $x_i \overset{i.i.d.}{\sim} \cN(\textbf{0}, I)$. 
    But in our case, $x_i \overset{i.i.d.}{\sim} \cN(\textbf{0}, \Sigma),$ but by Equation~\eqref{eq:rescaled-stability-1} (as in the proof of Lemma \ref{lem:vol-lb-covariance}), we have
\[\E_i \langle x_i x_i^\top - \Sigma, P \rangle^2 \le (2+\alpha/2) \cdot \|\Sigma^{1/2} P \Sigma^{1/2}\|_F^2\]
    for any symmetric matrix $P$. Also, by Proposition \ref{prop:frobenius_replacement},
\begin{equation*}
    \|\Sigma^{1/2} P \Sigma^{1/2}\|_F \le (1 + 3 \alpha) \cdot \|P\|_F.
\end{equation*}
    So, by setting $\Sigma'$ to be the empirical average of $x_i x_i^\top$, this means
\begin{equation*}
    \frac{1}{n} \sum_{i=1}^{n} \langle x_i x_i^\top - \Sigma', P \rangle^2 \le \frac{1}{n} \sum_{i=1}^{n} \langle x_i x_i^\top - \Sigma, P \rangle^2 \le (2+\alpha/2) \cdot \|\Sigma^{1/2} P \Sigma^{1/2}\|_F^2 \le (2+8 \alpha) \cdot \|P\|_F^2
\end{equation*}
    for any symmetric matrix $P$. Note, this is also true for non-symmetric matrices, because if $P$ is nonsymmetric, then $\frac{P+P^\top}{2}$ has smaller Frobenius norm but $\langle x_i x_i^\top - \Sigma', P \rangle = \langle x_i x_i^\top - \Sigma', \frac{P+P^\top}{2} \rangle$.
    
    By flattening and defining $S' = (\Sigma')^\flat$, we have that 
\[\frac{1}{n} \sum_{i=1}^n \langle x_i^{\otimes 2} - S', v \rangle^2 \le 2+8\alpha\]
    for all unit vectors $v \in \BR^{d^2}$. Hence, $\frac{1}{n} \sum_{i=1}^n (x_i^{\otimes 2} - S')(x_i^{\otimes 2} - S')^\top$ has all eigenvalues at most $2+8\alpha$, and thus we can find some positive semidefinite $MM^\top$ such that $\frac{1}{n} \sum_{i=1}^n (x_i^{\otimes 2} - S')(x_i^{\otimes 2} - S')^\top+MM^\top = (2+24\alpha) \cdot I_{d^2}.$

    Thus, the first three conditions are verified. For the fourth condition, note that $\Sigma' = \cL[\Sigma']$ (since $\cL$ is just an expectation over a single data point), by \Cref{cor:empirical-covariance}, satisfies $\|\Sigma'-\Sigma\|_{op} \le \alpha$ as long as $\Sigma \preccurlyeq (1+\alpha) \cdot I$. So, $(1-2 \alpha) \cdot I \preccurlyeq \cL[\Sigma'] \preccurlyeq (1+2 \alpha) \cdot I$. Thus, as long as $\tSigma \in \mathbb{K}_{2\alpha+\phi}^d$ satisfies $\|\Sigma'-\tSigma\|_{op} \le \phi \cdot I,$ the fourth condition is satisfied. But note that every $\tSigma$ with $\|\Sigma'-\tSigma\|_{op} \le \phi \cdot I$ must be in $\mathbb{K}_{2\alpha+\phi}^d$, by triangle inequality on the bounds $\|\Sigma' - \Sigma\|_{op} \le \alpha$, $\|\Sigma' - \tSigma\|_{op} \le \phi \cdot I$, and $\|\Sigma - I\|_{op} \le \alpha$.
    
    Finally, we remark that every $w_i, x_{i, j}$, and $M_{\{j, k\}, \{j', k'\}}$ is bounded by $\poly(n, d)$. Therefore, the corresponding linear operator $\cL$ satisfies $\|\cR(\cL)\|_2 \le (nd)^{O(1)}$.
\end{proof}

\subsubsection{Sensitivity}
Before proving sensitivity we need to prove the following upper bound on the value of the score function.

\begin{lemma}[score function upper bound]
\label{lem:TV-score-function-upper}
For any dataset $\cY$ and any $\tSigma \in \mathbb{K}_{2\alpha+\phi}^d$, the value $\cS(\tSigma, \cY)$, as defined in \Cref{def:score-TV}, is less than or equal to $n$. 
\end{lemma}
\begin{proof}
The proof is almost identical to \Cref{lem:covariance-score-function-upper}. First, set $\tSigma = U \Lambda U^\top$, where every $\Lambda_{ii}$ is between $1-2\alpha-\phi$ and $1+2\alpha+\phi$. Let $\Lambda'$ be the diagonal matrix with $\Lambda'_{ii} = \min\left(1+2\alpha+\phi/4, \max(1-2\alpha-\phi/4, \Lambda_{ii})\right)$, i.e., we truncate every eigenvalue to be in the range $[1-2\alpha-\phi/4, 1+2\alpha+\phi/4]$, and define $\Sigma = U \Lambda' U^\top$. Now, sample $x_1, \dots, x_n \sim \cN(0, \Sigma)$. As in \Cref{lem:covariance-score-function-upper}, we again set $w_i = 0$ and $x_i' = x_i$ for all $i$, with $T = n$.

Clearly, the first two conditions hold, and by \Cref{lem:vol-lb-covariance-tv}, with high probability the third condition holds because $\Sigma$ has all eigenvalues between $1-2\alpha-\phi/4$ and $1+2\alpha+\phi/4$ (so we just apply the same bounds but replace $\alpha$ with $2 \alpha + \phi/4 \le 3 \alpha$). The fifth condition also clearly holds. For the fourth condition, with high probability, as long as $n \ge \tilde{O}\left(\frac{d}{\phi^2}\right) = \tilde{O}\left(\frac{d^2}{\alpha^2}\right)$ (by our assumption that $\phi = \alpha/\sqrt{d}$), we have $\|\Sigma'-\Sigma\|_{op} \le \phi/4$. Thus, $\Sigma' = \cL[\Sigma']$ has all eigenvalues between $1-\alpha-\phi/2$ and $1+\alpha+\phi/2$. Moreover, $\|\tSigma-\Sigma\|_{op} \le 3\phi/4$ by the way we truncate the eigenvalues of $\tSigma$ (and since $\tSigma$ has all eigenvalues between $1-2\alpha-\phi$ and $1+2\alpha+\phi$). So, by Triangle inequality, $\|\tSigma - \Sigma'\|_{op} \le \phi$.

Therefore, for any $\cY$ and $\tSigma$, with high probability over $\cX$ we obtain a $(\alpha, \tau, \phi, n)$-certificate $\cL$ for $\cY$. Thus, there must exist such a certificate, so the score is at most $n$.
\end{proof}

\begin{lemma}[sensitivity] \label{lem:sensitivity-tv}
The score function $\cS$ as defined in \Cref{def:score-TV} has sensitivity $1$ with respect to its first input.
\end{lemma}
\begin{proof}
The proof is nearly identical to \Cref{lem:sensitivity-covariance}.
Suppose that $\cY$, $\cY'$ are two neighboring datasets, and $\tSigma \in \R^{d \times d}$. Moreover, assume $\cS\paren{\tSigma, \cY} = T$.
If we show that $\cS\paren{\tSigma, \cY'} \le \cS\paren{\tSigma, \cY} = T+1$, by symmetry we are done.

The only constraints that are different in our setting from \Cref{lem:sensitivity-covariance} are Constraints 3 and 4 in \Cref{def:certifiable-TV}.
However, note that these three constraints do not involve $w_j$ at all, so in fact their evaluation is the same regardless of $\cL$ or $\cL'$. The only difference is we are allowing the values $\cL[\cdot]$ to have a greater range, which makes it easier.
\end{proof}

\subsubsection{Quasi-convexity}
\begin{lemma}[quasi-convexity]
The score function $\cS$ as defined in \Cref{def:score-TV} is quasi-convex in its second input, $\tSigma$.
\end{lemma}
\begin{proof}
    Again, all of the constraints are satisfied trivially except 2c), and the same proof as in \Cref{lem:quasi-convex-mean} and \Cref{lem:quasi-convex-covariance-spectral} works for this case.
\end{proof}

\subsubsection{Accuracy}

We now show accuracy, meaning that any point $\tSigma$ of low score with respect to i.i.d. samples from $\cN(\textbf{0}, \Sigma)$ must be close to $\Sigma$. Because of our sensitivity bound, this will also imply a similar result for corrupted samples. Like for Lemma \ref{lem:mean-accuracy} and \ref{lem:covariance-accuracy}, we defer the proof to \Cref{appendix:sos}.

\begin{lemma} \label{lem:tv-accuracy}
    Let $\alpha = \widetilde{O}(\eta)$ and suppose $\alpha, \eta$ are bounded by a sufficiently small constant. Let $n \ge \frac{d^2+\log^2 (1/\beta)}{\alpha^2}$, and $\cX = \{x_1, \dots, x_n\} \sim \cN(\textbf{0}, \Sigma)$, for $\Sigma \in \BR^{d \times d}$ with $(1-\alpha) I \preccurlyeq \Sigma \preccurlyeq (1+\alpha) I$.

    Then, for any $\alpha^* \le \alpha$, and assuming $\tau \ll 1/(nd)^{O(1)}$, with probability at least $1-\beta$, any covariance matrix $\tSigma \in \BR^{d \times d}$ that is $(\alpha^*, \tau, \phi, T)$-certifiable for $\cX$ with $T = \eta n$ and $\phi \le \alpha/\sqrt{d}$ must satisfy $\|\tSigma - \Sigma\|_F \le O(\alpha)$.
\end{lemma}

\subsubsection{Volume of Good Points}

Finally, we use our accuracy bounds to get an upper bound for the volumes of $V_\eta$. We can also obtain a lower bound from \Cref{lem:vol-lb-covariance-tv}.

\begin{lemma} \label{lem:vol-ub-covariance-tv}
    Let $\cX = \{x_1, \dots, x_n\} \sim \cN(\textbf{0}, \Sigma)$ (where $(1-\alpha) I \preccurlyeq \Sigma \preccurlyeq (1+\alpha) I$), and let $\cY = \{y_1, \dots, y_n\}$ represent an $\eta$-corruption of $\cX$. Then, for every integer $T \in [\eta \cdot n, \eta^* \cdot n]$ for some fixed constant $\eta^* < 1$, with probability at least $1-\beta$, every $(\alpha, \tau, \phi, T)$-certifiable covariance $\tSigma$ with respect to $\cY$, for $\phi = \alpha/\sqrt{d}$, satisfies $\|\Sigma - \tSigma\|_F \le \tcO(T/n)$.
\end{lemma}

\begin{proof}
    Since the score function has sensitivity at most 1 (\Cref{lem:sensitivity-tv}), any $(\alpha, \tau, \phi, T)$-certifiable covariance with respect to $\cY$ is an $(\alpha, \tau, \phi, T+\eta n)$-certifiable covariance with respect to $\cX$.
    
    Now, define $\eta' := \frac{T+\eta n}{n} = O(\frac{T}{n})$. In this case, by setting $\alpha' = \tcO(\eta')$ and since $\alpha = \tcO(\eta) \le \alpha'$, we have that by \Cref{lem:tv-accuracy} that any $(\alpha, \tau, \phi, T+\eta n)$-certifiable covariance $\tSigma$ must satisfy $\|\tSigma-\Sigma\|_F \le O(\alpha') \le \tcO(T/n)$. 
\end{proof}

We think of the set of potential covariances $\tSigma`$ as lying in $\BR^{d(d+1)/2}$, by taking the upper-diagonal entries.
If we set $\phi = \alpha/\sqrt{d}$ and $\tau \ll 1/(nd)^{O(1)}$, by \Cref{lem:vol-lb-covariance-tv}, the set of $(\alpha, \tau, \phi, T)$-certifiable covariances $\tSigma \in \mathbb{K}_{2\alpha+\phi}^d$, for $T = \eta n$, contains all $\tSigma$ within an operator norm ball of radius $\alpha/\sqrt{d}$ around $\frac{1}{n} \sum x_i x_i^\top$, with $1-\beta$ probability. This contains a Frobenius norm ball of radius $\alpha/d$, and thus has volume at least $(\alpha/d)^{d(d+1)/2}$ times the volume of a $\frac{d(d+1)}{2}$-dimensional sphere, which is at least $(\alpha/d^2)^{d(d+1)/2}$.
However, for any $T = \eta' n$ for $\eta \le \eta' \le \eta^*$, the volume of $(\alpha, \tau, \phi, T)$-certifiable covariances is at most $(\tilde{O}(\eta'))^{d(d+1)/2}$ times the volume of a $\frac{d(d+1)}{2}$-dimensional sphere, which is $(\tilde{O}(\eta')/d)^{d(d+1)/2}$.
Finally, for $T = \eta' n$ with $\eta' > \eta^*$, the volume of $\Theta$, the set of all candidate covariances $\tSigma \in \mathbb{K}_{2\alpha+\phi}^d$, is at most $d^{O(d^2)}$.
% In addition, the set $\mathbb{K}_{2\alpha+\phi}^d$, viewed as lying in $\BR^{d(d+1)/2}$, is contained in a $\ell_2$-norm ball of radius $O(d)$ around the origin.

\subsubsection{Efficient Computability}

As in the mean estimation case, we apply \Cref{thm:computing-score} in \Cref{sec:computing-score-functions}: the proof is identical to verify ``efficient computability''.

\subsubsection{Efficient Finding of Low-Scoring Point}

To verify that the ``robust algorithm finds low-scoring point'', we remove the constraint that $-(\phi+\tau \cdot T) \cdot I \preccurlyeq \cL[\Sigma'-\tSigma] \preccurlyeq (\phi + \tau \cdot T) \cdot I$. We can apply \Cref{thm:computing-score} in the same way to find some linear operator $\cL$ with score at most $\min_{\tSigma} \cS(\tSigma, \cY) + 1$. Then, we can compute $\cL[\Sigma']$ and set $r \le \phi$, and obtain that every matrix $\tSigma$ with $\|\tSigma - \cL[\Sigma']\|_F \le r$ has score at most $\min_{\tSigma} \cS(\tSigma, \cY) + 1$.

\subsection{Proof of \Cref{thm:gaussian-tv-1}}

We apply \Cref{thm:pure_dp_general_main}, using the score function defined in \Cref{def:score-TV}. Indeed, for $r = \phi = \alpha/\sqrt{d}$, we have verified all conditions, as long as $n \ge \tcO((d^2+\log^2(1/\beta))/\alpha^2)$. Therefore, we have an $\eps$-DP algorithm running in time $\poly(n, d, \log \frac{R}{\alpha})$ that finds a candidate covariance $\tSigma$ of score at most $2 \eta n$, as long as
\[n \ge O\left(\max_{\eta': \eta \le \eta' \le 1} \frac{\log (V_{\eta'}(\cY)/V_\eta(\cY)) + \log (1/(\beta \cdot \eta'))}{\eps \cdot \eta'}\right).\]

Using Lemmas \ref{lem:vol-lb-covariance-tv} and \ref{lem:vol-ub-covariance-tv}, and by the commentary after \Cref{lem:vol-ub-covariance-tv}, we have that for $\eta' \le \eta^*$ for some $\eta^* = \Omega(1)$, then $V_{\eta'}(\cY)/V_{\eta}(\cY) \le (\tcO(\eta') \cdot d/\alpha)^{d(d+1)/2} \le (O(d/\alpha))^{d(d+1)/2}$. For $\eta' > \eta^*$, we have that $V_{\eta'}(\cY)/V_{\eta}(\cY) \le O(d/\alpha)^{O(d^2)}$. So overall, it suffices for 
\begin{align*}
    n &\ge \tcO\left(\frac{d^2+\log^2 (1/\beta)}{\alpha^2}\right) + O\left(\max_{\eta \le \eta' \le 1} \frac{d^2 \log (d/\alpha) + \log(1/(\beta \cdot \eta'))}{\eps \cdot \eta'}\right)  \\
    &= \tcO\left(\frac{d^2 + \log^2(1/\beta)}{\alpha^2} + \frac{d^2 + \log (1/\beta)}{\eps \cdot \alpha}\right).
\end{align*}

Hence, our algorithm, using this many samples, is $\eps$-DP, and can find a point $\tSigma$ of score at most $2 \eta n$ with respect to $\cY$.
%, which means it has score at most $3 \eta n$ with respect to an $\eta$-corrupted $\cX$ by \Cref{lem:sensitivity-tv}. 
So, by replacing $\eta$ with $2 \eta$ and applying \Cref{lem:tv-accuracy}, we have that any point $\tSigma$ with score at most $2 \eta n$ with respect to $\cY$ must have $\|\tSigma - \Sigma\|_F \le O(\alpha)$. This completes the proof.

\subsection{Proof of Theorems \ref{thm:pure-dp-tv} and \ref{thm:approx-dp-tv}} \label{subsec:main_stuff}

By combining Theorems \ref{thm:gaussian-covariance-pure-conditioning}, \ref{thm:gaussian-tv-1}, and \ref{thm:gaussian-mean-main}, we are able to prove Theorem \ref{thm:pure-dp-tv}.

\begin{proof}[Proof of Theorem \ref{thm:pure-dp-tv}]
    %We will assume that we have $3n$ samples $y_1, dots, y_n, y_{n+1}, \dots, y_{2n}, y_{2n+1}, \dots, y_{3n}$. Let $x_1, \dots, x_{3n}$ represent the corresponding uncorrupted samples. We assume the corruption fraction is at most $\eta/3$ (so at most $\eta n$ points in total are corrupted).
    Let the corrupted samples be $y_1, \dots, y_n$, and let the uncorrupted samples be $x_1, \dots, x_{n}$.
    
    We may assume without loss of generality that $\alpha = \tcO(\eta)$ (either by raising $\alpha$ or $\eta$ appropriately). Via the standard method of pairing samples and subtracting them, we may first assume that the mean is $\textbf{0}$, and we will attempt to learn covariance. By \Cref{thm:gaussian-covariance-pure-conditioning}, we can thus privately learn a $\tSigma_1$ such that $\|\Sigma^{-1/2} \tSigma_1 \Sigma^{-1/2} - I\|_{op} \le \alpha$, given samples $y_1, \dots, y_n$.

    Next, we learn $\Sigma$ up to Frobenius distance rather than just spectral distance. Let $\hat{y}_i = \tSigma_1^{-1/2} y_i$, and let $\hat{x}_i = \tSigma_1^{-1/2} x_i$ and $x_i^* = \Sigma^{-1/2} x_i$. Note that $x_i^* \overset{i.i.d.}{\sim} \cN(\textbf{0}, I)$, and $\hat{x}_i = J \cdot x_i^*$ for $J = \tSigma_1^{-1/2} \Sigma^{1/2}$. However, $J$ may be adversarially dependent on the data points, as we chose $\tSigma_1$ based on the samples $\cY$\footnote{One may attempt to remove this issue by using different samples for this step, but due to the adversarial nature of the strong contamination model, previous samples may affect how later samples are corrupted!}. Nevertheless, we may still apply \Cref{thm:gaussian-tv-1}, because it will turn out that the $\{\hat{x}_i\}$ samples will have the desired resilience conditions for \emph{every} choice of $J$ with $\|JJ^\top - I\|_{op} \le \alpha$. 
    
    Indeed, note that $\langle \hat{x}_i \hat{x}_i^\top - JJ^\top, P \rangle = \langle J x_i^* (x_i^*)^\top J^\top - J J^\top , P \rangle = \langle (x_i^*) (x_i^*)^\top - I , J^\top PJ \rangle$ for all $J$, and $\|J^\top P J\|_F = (1 \pm 3 \alpha) \cdot \|P\|_F$ by Proposition \ref{prop:frobenius_replacement}, since $JJ^\top = \tSigma_1^{-1/2} \Sigma \tSigma_1^{-1/2}$. Thus, assuming $\{x_i^*\}$ satisfy the resilience properties (Lemma \ref{lem:resilience-of-moments-covariance}), $\frac{1}{n} \sum \langle \hat{x}_i \hat{x}_i^\top - JJ^\top, P \rangle^2 \le (2+O(\alpha)) \cdot \|P\|_F^2$ for all symmetric matrices $P$. This is sufficient to ensure Lemma \ref{lem:vol-lb-covariance-tv} holds, if we replace $\cX$ with $\{\hat{x}_1, \dots, \hat{x}_n\}$ and $\Sigma$ with $JJ^\top = \tSigma_1^{-1/2} \Sigma \tSigma_1^{-1/2}$. Likewise, Lemma \ref{lem:tv-accuracy} will also work in the same way, replacing each $x_i$ with $\hat{x}_i$, and replacing $\Sigma$ with $JJ^\top$. The rest of the conditions also clearly hold (as they either do not depend on the dataset or follow from Lemmas \ref{lem:vol-lb-covariance-tv} and \ref{lem:tv-accuracy}). Therefore, we can apply \Cref{thm:gaussian-tv-1} to privately and robustly find $\tSigma_2$ such that $\|\tSigma_2 - JJ^\top\|_F \le O(\alpha)$, by applying the algorithm on $\hat{y}_1, \dots, \hat{y}_n$. Since both $\Sigma_2$ and $JJ^\top$ are spectrally bounded between $1 \pm \alpha$, this implies $\|I - \tSigma_2^{-1/2} JJ^\top \tSigma_2^{-1/2}\|_F \le O(\alpha)$, which means $\|I - J^\top \tSigma_2^{-1} J\|_F \le \alpha$. Note, however, that we can write this as $\|I - \Sigma^{1/2} (\tSigma_1^{-1/2} \tSigma_2^{-1} \tSigma_1^{-1/2}) \Sigma^{1/2}\|_F \le \alpha$, which implies that $\Sigma$ and $\tSigma_1^{1/2} \tSigma_2 \tSigma_1^{1/2}$ are $\alpha$-close in Frobenius distance. So, we can output $\hat{\Sigma} = \tSigma_1^{1/2} \tSigma_2 \tSigma_1^{1/2}$.

    Finally, we must decide on $\hat{\mu}$. To do so, we return to our original samples $y_1, \dots, y_n$ (where we did not do sample pairing and subtraction), and redefine $\hat{y}_i = \hat{\Sigma}^{-1/2} y_i$, $\hat{x}_i = \hat{\Sigma}^{-1/2} x_i$. Also, redefine $x_i^* = \Sigma^{-1/2} x_i$. Now, $x_i^* \sim \cN(\Sigma^{-1/2} \mu, I)$, and $\hat{x}_i = J \cdot x_i^*$ for some new choice of $J = \hat{\Sigma}^{-1/2} \Sigma^{1/2}$, and note $\|JJ^\top - I\|_F \le \alpha$, but $J$ may be adversarial. However, this is sufficient to satisfy all resilience conditions by the remark after \Cref{cor:resilience}. Hence, using \Cref{thm:gaussian-mean-main} on the corrupted samples $\hat{y}_i$, we learn $\hat{\Sigma}^{-1/2} \mu$ up to $\ell_2$ error $O(\alpha)$.
    Multiplying this by $\hat{\Sigma}^{1/2}$, we find $\hat{\mu}$ such that $\|\mu-\hat{\mu}\|_{\hat{\Sigma}} \le O(\alpha)$, which implies
    $\dtv(\cN(\hat{\mu}, \hat{\Sigma}), \cN(\mu,  \hat{\Sigma})) \le O(\alpha)$. But since $\Sigma$ and $\hat{\Sigma}$ have Frobenius distance at most $O(\alpha)$, this means $\dtv(\cN(\mu, \hat{\Sigma}), \cN(\mu,  \Sigma)) \le \alpha$. So, by the Triangle inequality, we have $\dtv(\cN(\hat{\mu}, \hat{\Sigma}), \cN(\mu, \Sigma)) \le O(\alpha)$, which completes the proof.

    The privacy factor and increases by a factor of $3$ via basic composition of privacy, the failure probability also increases by a factor of $3$, and the sample complexity is simply the maximum of the sample complexities required by Theorems \ref{thm:gaussian-covariance-pure-conditioning}, \ref{thm:gaussian-tv-1}, and \ref{thm:gaussian-mean-main}.
\end{proof}

The proof of Theorem \ref{thm:approx-dp-tv} is very similar: this time, we combine Theorems \ref{thm:gaussian-covariance-approx-conditioning}, \ref{thm:gaussian-tv-1}, and \ref{cor:gaussian-mean-main}.

\begin{proof}[Proof of \Cref{thm:approx-dp-tv}]
    The proof is identical to the proof of \Cref{thm:pure-dp-tv}. First, we privately learn $\tSigma_1$ such that $\|\Sigma^{-1/2} \tSigma_1 \Sigma^{-1/2} - I\|_{op} \le \alpha$, using \Cref{thm:gaussian-covariance-approx-conditioning}. We then replace each $y_i$ with $\hat{y}_i = \tSigma_1^{-1/2} y_i$, and via the same procedure we learn some $\hat{\Sigma}$ such that $\Sigma, \hat{\Sigma}$ are close in Frobenius distance. Finally, we redefine $\hat{y}_i = \hat{\Sigma}^{-1/2} y_i$, and learn $\hat{\mu}$ such that $\dtv(\cN(\hat{\mu}, \hat{\Sigma}), \cN(\mu, \Sigma)) \le O(\alpha)$, using \Cref{cor:gaussian-mean-main}, in the same way as we applied \Cref{thm:gaussian-mean-main}, to prove \Cref{thm:pure-dp-tv}.
    %This is enough to ensure $\dtv(\cN(\hat{\mu}, \hat{\Sigma}), \cN(\mu, \Sigma)) \le O(\alpha)$.

    The privacy factor and failure probability increase by a factor of $3$, and the sample complexity is the maximum of the sample complexities required by Theorems \ref{thm:gaussian-covariance-approx-conditioning}, \ref{thm:gaussian-tv-1}, and \ref{cor:gaussian-mean-main}.
\end{proof}
\section*{Acknowledgements}

We thank Xiyang Liu, Weihao Kong, and Sewoong Oh for helpful conversations at the beginning of this project.
We also thank Lydia Zakynthinou and Pasin Manurangsi for making us aware of prior work on the inverse sensitivity mechanism.
%add anyone else
\bibliographystyle{alpha}
%\bibliography{biblio,gbiblio,hopkins-lifetime-refs}
\newcommand{\etalchar}[1]{$^{#1}$}

\appendix

\section{Omitted proofs for Private Sampling} \label{appendix:sampling}

In this section, we prove Lemmas \ref{lem:sampling_main} and \ref{lem:volume_computation}.

\subsection{Preliminaries}

In this subsection, we note a few miscellaneous results that will be very important in proving Theorems \ref{thm:pure_dp_general_main} and \ref{thm:approx_dp_general_main}.

We will use certain facts about dilated regions in $\BR^d$, which we define now.

\begin{definition}[Dilation about the origin]
    Suppose that $K \subset \BR^d$. Then, for any positive real $a > 0$, we define $a K := \{a \cdot x: x \in K\}$. In words, $a K$ is ``$K$ dilated by a factor $a$ about the origin.''
\end{definition}

We need the following ellipsoid theorem, showing that convex bodies are contained in reasonably small ellipsoids but also contain reasonably large ellipsoids.

\begin{theorem} \label{thm:ellipsoid_covariance} \cite[Theorem 4.1]{KannanLS95}
    Let $K \subset \BR^d$ be a convex body in isotropic position, meaning that if $X$ is uniformly drawn from $K$, $\E[X] = \textbf{0}$ and $Cov(X) = I$. Then, for $B$ the unit ball of radius $1$, $\sqrt{\frac{d+2}{d}} \cdot B \subset K \subset \sqrt{d(d+2)} \cdot B$. 
\end{theorem}

Next, we need the following basic proposition.

\begin{proposition} \label{prop:multiplicative_vs_additive_volume}
    Suppose that $K, K'$ are convex bodies such that $B(\textbf{0}, r) \subset K \subset K'$, and suppose that $K' \not\subset (1+\gamma_1) K$, for some $\gamma_1 > 0$.
    %, where $(1+\gamma_1) K$ represents $K$ dilated by a factor $1+\gamma_1$ about the origin. 
    Then, $\vol(K')-\vol(K) \ge (\frac{\gamma_1 \cdot r}{6d})^d$.
\end{proposition}

\begin{proof}
    Since $K' \not\subset (1+\gamma_1) K$, there exists a vector $u \in K'$ but not in $(1+\gamma_1) K$. So, if $v := \frac{u}{1+\gamma_1},$ then $(1+\gamma_1) v \in K'$ but $v \not\in K$. Since $\textbf{0} \in K'$ and $\gamma_1 > 0$, convexity implies that $v, (1+\frac{\gamma_1}{2}) v \in K'$. However, $v \not\in K$ and $(1+\frac{\gamma_1}{2}) v \not\in K$ also because otherwise, convexity and the fact that $\textbf{0} \in K$ implies that $v \in K$. In summary, there exists a vector $v$ such that $v, (1+\frac{\gamma_1}{2})v$ are both contained in $K' \backslash K$.
    
    Now, let's consider the ball of radius $\rho$ centered at $(1+\frac{\gamma_1}{4})v$ for some small value $\rho$. We will show that for $\rho$ appropriately chosen, this ball is contained in $K'$ but is disjoint from $K$.
    
    For any such point in the ball, we can write it as $v + (\frac{\gamma_1}{4}v+w)$ for some $w$ with $\|w\|_2 \le \rho$. If this point were in $K$, then since $v \not\in K$, this means by convexity $v - \lambda (\frac{\gamma_1}{4}v+w)$ is not in $K$ for all $\lambda \ge 0$. By choosing $\lambda = \frac{4}{\gamma_1},$ we have that $-\frac{4}{\gamma_1} \cdot w$ is not in $K$. This is a contradiction if we choose $\rho \le \frac{\gamma_1 \cdot r}{4},$ since this implies $-\frac{4}{\gamma_1} \cdot w$ has norm at most $r$ so it must be in $K$. Thus, if $\rho \le \frac{\gamma_1 \cdot r}{4},$ the ball of radius $\rho$ around $(1+\frac{\gamma_1}{4})v$ is disjoint from $K$.
    
    Next, we alternatively write the point as $(1+\frac{\gamma_1}{2}) v - (\frac{\gamma_1}{4}v-w)$. To show it is in $K'$, note that $(1+\frac{\gamma_1}{2}) v$ is in $K'$, so by convexity it suffices to show that $(1+\frac{\gamma_1}{2}) v - \lambda (\frac{\gamma_1}{4}v-w)$ is in $K'$ for some $\lambda \ge 1$. By setting $\lambda = (1+\frac{\gamma_1}{2}) \cdot \frac{4}{\gamma_1},$ it suffices to show that $(1+\frac{\gamma_1}{2}) \cdot \frac{4}{\gamma_1} \cdot w = \left(\frac{4}{\gamma_1}+2\right) \cdot w$ is in $K'$. But this is similarly true as long as $\rho \le r/(\frac{4}{\gamma_1}+2)$, which holds as long as $\rho \le \frac{\gamma_1 \cdot r}{6}.$
    
    Therefore, $K' \backslash K$ contains a ball of radius $\frac{\gamma_1 \cdot r}{6},$ which has volume at least $(\frac{\gamma_1 \cdot r}{6d})^d$.
\end{proof}

\subsection{Sampling from a well-rounded convex body with an imperfect oracle}

In this subsection, our goal is to sample uniformly from a convex body, but we wish to do this even if we only can afford polynomial bit precision and do not have a perfect membership oracle. To do this, we will apply the well-known hit-and-run Markov chain, but with some minor adjustments to avoid the issue of requiring infinite precision arithmetic.

First, we describe the standard hit-and-run Markov chain assuming infinite precision arithmetic. Given a convex body $K$ for which we have a membership oracle, the hit-and-run algorithm starts with a point $x_0 \in K$. At each step $t$, we move from $x_{t-1} \in K$ to $x_t \in K$ as follows. We first pick a vector $v$ at random from the unit sphere. We then let $x_t$ be uniformly chosen on the line segment $\{x_{t-1} + \lambda \cdot v\}_{\lambda \in \BR} \cap K$, i.e., the line segment parallel to $v$ that goes through $x_{t-1}$, but restricted by $K$ since we cannot sample outside $K$.

The main result of the hit-and-run algorithm we apply is the following, due to Lov\'{a}sz and Vempala.

\begin{theorem}~\cite{LovaszV04} \label{thm:hit_and_run}
    Let $K$ be a $d$-dimensional convex body that contains the ball $B(\textbf{0}, 1)$ and is contained in the ball $B(\textbf{0}, D)$. Then, for a sufficiently large constant $C$ and for any $0 < \gamma < \frac{1}{2}$, after $m \ge C d^2 D^2 \log \gamma^{-1}$ steps of hit-and-run starting from the origin (i.e., setting $x_0$ to be the origin), the distribution of the final point $x_m$ has total variation distance at most $\gamma$ from the uniform distribution over $K$.
\end{theorem}

In our setting, we cannot directly use the hit-and-run algorithm for two reasons. The first reason is that we cannot pick a truly uniform direction and sample truly uniformly along that direction from a starting point. The second reason is that we don't have a perfect membership oracle. That being said, we will be able to make minor modifications to the algorithm and show that we output a distribution that is ``close'' to uniform on $K$.

We assume we are given two unknown convex bodies $K_1, K_2$, such that $B(\textbf{0}, 1) \subset K_1 \subset K_2 \subset (1+\gamma_1) K_1 \subset B(\textbf{0}, D)$. One should think of $D$ as polynomially large (we will later improve this to being exponentially large) and $\gamma_1$ as exponentially small. We also assume we have an $(K_1, K_2)$-membership oracle $\mathcal{O}$.

For some small parameter $\gamma>0$, we define the hit-and-run algorithm with $\gamma$ precision as follows. Given a point $x_{t-1}$, we select a random unit vector $v$ and round the coordinates of $v$ to multiples of $\gamma$. Next, we attempt to sample along the line $x_{t-1} + \lambda \cdot v$ for $\lambda \in \BR$, restricted to $K_1$. To do this with our oracle $\mathcal{O}$, we perform a binary search to find a positive integer $a_1$ such that $\mathcal{O}$ accepts $x_{t-1} + a_1 \cdot \gamma \cdot v$ but rejects $x_{t-1} + (a_1+1) \cdot \gamma \cdot v$. Likewise we find a negative integer $-a_2$ such that oracle accepts $x_{t-1} - a_2 \cdot \gamma \cdot v$ but rejects $x_{t-1} - (a_2+1) \cdot \gamma \cdot v$. Finally, we compute $x_t := x_{t-1} + a \cdot \gamma \cdot v$, where $a$ is an integer chosen uniformly at random betwen $-a_2$ and $a_1$, inclusive. We note that it may be possible to choose $x_t$ that the oracle rejects, but we know that $x_t$ is always in $K_2$.

\medskip

Before analyzing the modified hit-and-run algorithm, we first show the following proposition.

\begin{proposition} \label{prop:cone}
    Let $K$ be any convex body. Suppose that $x$ is a point and $\gamma > 0$ is a parameter such that the ball $B(x, \gamma)$ is contained in $K$, and let $\Lambda$ be an arbitrary line passing through $x$. Define $L$ to be the length of $\Lambda \cap K$.
    Then, for any parameter $0 < \lambda < 1$, the length of points $x'$ on $\Lambda \cap K$ such that $B(x', \lambda \cdot \gamma)$ is contained in $K$ is at least $(1-\lambda) \cdot L$.
\end{proposition}

\begin{proof}
    Let $\Lambda'$ represent the segment of $\Lambda$ that is contained in $K$, with endpoints $y$ and $z$. Since the ball $B(x, \gamma)$ is contained in $K$, we can consider the convex hull of this ball and the points $y$ and $z$. Note that the ball of radius $\lambda' \cdot r$ around $\lambda' x + (1-\lambda') y$ or around $\lambda' x + (1-\lambda') z$ is contained in this convex hull. So, all points $x'$ on $\Lambda'$ such that $B(x', \lambda \cdot \gamma)$ is not contained in $K$ cannot be between $\lambda x + (1-\lambda) y$ and $\lambda x + (1-\lambda) z$, so the length of the interval of such points is at least $(1-\lambda) \cdot L$.
\end{proof}

Next, we show that the hit-and-run algorithm with $\gamma$ precision, assuming $\gamma$ is sufficiently small, always stays within $K_1$ up to a small margin of error.

\begin{proposition} \label{prop:hitandrun_bounded}
    Let $K_1, K_2$ be convex bodies such that $B(\textbf{0}, 1) \subset K_1 \subset K_2 \subset (1+\gamma_1) K_1 \subset B(\textbf{0}, D)$. Consider running $m$ steps of hit-and-run from the origin with $\gamma_1$ precision, with $x_t$ being the point chosen after the $t^{\text{th}}$ step for all $0 \le t \le m$. Then, for any $0 < \tau < 1$ such that $(\tau/2)^{m+1} \ge D \cdot \gamma_1$, we have that with probability at least $1-O(m \cdot \tau)$, all the points $x_t$ satisfy the $B(x_t, (\tau/2)^m) \subset K_1$.
\end{proposition}

\begin{proof}
    Suppose that after $t$ steps of hit and run, the point selected is $x_t$. Suppose that $B(x_t, \gamma^{(t)})$ is contained in $K_1$, for some positive real $\gamma^{(t)}$, which also means $B(x_t, \gamma^{(t)}) \subset K_2$. Let $\Lambda$ represent an arbitrary line through $x_t$. By making oracle calls to $\mathcal{O}$ using the binary search procedure, we obtain some line segment $\Lambda' \subset \Lambda$ that goes entirely through $K_1$ but is contained in $K_2$. Let $L$ represent the length of the line segment we found, and $L_1$ represent the length of $\Lambda \cap K_1$. Also, let $L_2$ represent the length of $\Lambda \cap K_2$, so $L_1 \le L \le L_2$.
    
    Recall that $B(x_t, \gamma^{(t)}) \subset K_2$, but note that for any point $x'$ outside $K_1$, $B(x', \gamma_1 \cdot D) \not\subset K_2$. Therefore, by Proposition \ref{prop:cone}, the value of $L_2 - L_1$ is at most $\frac{\gamma_1 \cdot D}{\gamma^{(t)}} \cdot L_2$, which assuming $\gamma^{(t)} \ge 2 \gamma_1 D$ is at most $\frac{2 \gamma_1 \cdot D}{\gamma^{(t)}} \cdot L$. In addition, the length of points $x'$ in $L_1$ such that $B(x', \tau \cdot \gamma^{(t)}) \not\subset K_1$ is at most $\tau \cdot L_1 \le \tau \cdot L$. So, if we sample randomly from $L$ even after discretizing by rounding coordinates to the nearest multiples of $\gamma_1$, the probability of selecting a point $x'$ such that $B(x', \tau \cdot \gamma^{(t)}-\gamma_1) \not\subset K_1$ is at most $\tau + O\left(\frac{\gamma_1 \cdot D}{\gamma^{(t)}}\right)$.
    
    For $t = 0$, we assume $x_0$ is the origin, so we can set $\gamma^{(0)} = 1$. In general, we fix some parameter $\tau$, and let $\gamma^{(t+1)} := \tau \cdot \gamma^{(t)}-\gamma_1$. If $(\tau/2)^{m+1} \ge D \cdot \gamma_1$, then we will inductively have that $\gamma^{(t)} \ge (\tau/2)^t$ for all $0 \le t \le m$, and so $\frac{\gamma_1 \cdot D}{\gamma^{(t)}} \le \tau/2$. Therefore, by a union bound over all $0 \le t \le m$, we have that with probability at least $1-O(m \cdot \tau)$, every $x_t$ selected satisfies $B(x_t, \gamma^{(t)}) \subset K_1$, so $B(x_t, (\tau/2)^m) \subset K_1$.
    %
    %, the length of points $x'$ on this line such that $x' \in K_2$ but $B(x', \tau \cdot \gamma^{(t)}) \not\subset K_2$ is at most $2 \tau \cdot D$. Since $K_1 \subset K_2 \subset (1+\gamma_1) K_1 \subset B(0, D)$, this implies that the length of points $x'$ on the line segment $\Lambda' := \Lambda \cap K_2$ but such that $B(x', \tau \cdot \gamma^{(t)}-\gamma_1 \cdot D)$ is not contained in $K_1$ is at most $2 \tau \cdot D$. Also, the length of points even in $\Lambda \cap K_1$ is at least $2 \gamma^{(t)}$, since $K_1$ contains the ball of radius $\gamma^{(t)}$ around $x_t$ and $\Lambda$ passes through $x_t$. 
    %
    %By making oracle calls to the convex bodies (which may be faulty on $K_2 \backslash K_1$), we find a line segment passing through $x$ that goes entirely through $K_1$ but is contained in $K_2$. So the length is at least $2 \gamma^{(t)}$. Also, the length of points on this line segment that are within $\tau \cdot \gamma^{(t)}-\gamma_1 \cdot D$ of being outside $K_1$ is at most $2 \tau \cdot D$. If we discretize by rounding points to nearest multiples of $\gamma_1$ and pick a random discretized point $x_{t+1}$ on this line segment, we still have that with probability at least $1 - \frac{2 \tau \cdot D + O(\gamma_1)}{2 \gamma^{(t)}},$ that $B(x_{t+1}, \gamma^{(t+1)})$ is contained in $K_1$, for $\gamma_{t+1} = \tau \cdot \gamma^{(t)}-\gamma_1 \cdot D$.
\end{proof}

We show that the hit-and-run algorithm with limited precision outputs a distribution that is ``close'' to uniform on the convex body $K_1$. We will use the following formal definition of closeness.

\begin{definition}
    We define two distributions $\mathcal{D}, \mathcal{D}'$ over Euclidean space $\BR^d$ to be $(\gamma, \gamma')$\textbf{-close} if there exists a coupling of $\mathcal{D}, \mathcal{D}'$ such that $\BP_{(a, c) \sim (\mathcal{D}, \mathcal{D}')}(\|a-c\|_2 \ge \gamma) \le \gamma'$.
\end{definition}

\begin{lemma} \label{lem:hitandrun_good}
    Given parameters $D, \gamma_2, \gamma_3$, there exists $\gamma_1$ such that $\log \gamma_1^{-1} = \poly(d, D, \log \gamma_2^{-1}, \log \gamma_3^{-1})$, and the following holds.
    If $K_1, K_2$ are convex bodies such that $B(\textbf{0}, 1) \subset K_1 \subset K_2 \subset (1+\gamma_1) K_1 \subset B(\textbf{0}, D)$, then after $m \ge O(d^2 D^2 \log \gamma_3^{-1})$ steps of hit-and-run starting from the origin with $\gamma_1$ precision, the final point is $(\gamma_2, \gamma_3)$-close to the uniform distribution over $K_1$.
\end{lemma}

\begin{proof}
    We create a coupling between running hit-and-run with perfect precision and running hit-and-run with $\gamma_1$ precision. After $t$ steps, let $x_t$ be the point we sampled for hit-and-run with perfect precision, and let $x_t'$ be the point we sampled for hit-and-run with $\gamma_1$ precision. We start with $x_0 = x_0'$ as the origin.
    
    Let $\Lambda$ be the random line drawn through $x_t$, and let $\Lambda'$ be the rounded random line drawn through $x_t'$. We will couple the lines so that with $1-\gamma_1$ probability, the lines are essentially parallel up to $\gamma_1$ error. Let's write $\Lambda = \{x_t + \lambda \cdot v_t\}_{\lambda \in \BR}$ and $\Lambda' = \{x_t' + \lambda \cdot v_t'\}_{\lambda \in \BR}$, where $v_t, v_t'$ are unit vectors with $\|v_t-v_t'\|_2 \le \sqrt{d} \cdot \gamma_1$ with probability at least $1-\gamma_1$. 
    If $B(x_t, (\tau/2)^m) \subset K_1$, then the probability that a random point $x'$ on $\Lambda \cap K_1$ satisfies $B(x', (\tau/2)^{m+1}) \not\subset K_1$ is at most $\tau+O\left(\frac{\gamma_1 \cdot D}{(\tau/2)^m}\right)$, by the argument of Proposition \ref{prop:hitandrun_bounded}. Likewise, if $B(x_t', (\tau/2)^m) \subset K_1$, then the probability that a random point $x'$ on $\Lambda'$ in the segment selected by step $t+1$ of the algorithm satisfies $B(x', (\tau/2)^{m+1}) \not\subset K_1$ is at most $\tau+O\left(\frac{\gamma_1 \cdot D}{(\tau/2)^m}\right)$. 
    
    Now, suppose that $B(x_t, (\tau/2)^m), B(x_t', (\tau/2)^m) \subset K_1$, and $\|x_t - x_t'\|_2 \le \tau^{(t)}$ for some parameter $\tau^{(t)} \le (\tau/2)^m-2 D \cdot \gamma_1$. Then, if we selected $\lambda$ uniformly such that $x_t+\lambda v_t \in K_1$, then $x_t' + \lambda v_t'$ has distance at most $\tau^{(t)}+2 D \cdot \gamma_1$ from $x_t+\lambda v_t$ (even after rounding off $\lambda$ to the nearest multiple of $\gamma_1$). This means that with probability at most $\tau+O\left(\frac{\gamma_1 \cdot D}{(\tau/2)^m}\right)$, $x_t'+\lambda v_t' \in K_1$. Likewise, if we selected $\lambda'$ according to the distribution from hit-and-run on $x_t'$ with $\gamma_1$-precision, then $x_t + \lambda' v_t \in K_1$ with probability at most $\tau+O\left(\frac{\gamma_1 \cdot D}{(\tau/2)^m}\right)$ (even after replacing $\lambda'$ with a uniform real in $[\lambda'-\gamma_1/2, \lambda'+\gamma_1/2]$ to ``un-round'' it). So, by keeping $\lambda$ the same (up to rounding) whenever possible (which can happen with at most $O\left(\tau + \frac{\gamma_1 \cdot D}{(\tau/2)^m}\right)$ failure probability, we have that $\|x_{t+1}-x_{t+1}'\|_2 \le \tau^{(t)}+2D \cdot \gamma_1$ if $B(x_t, (\tau/2)^m), B(x_t', (\tau/2)^m) \subset K_1$, $\|x_{t}-x_{t}'\|_2 \le \tau^{(t)}$, and $\tau^{(t)} \le (\tau/2)^m-2D \cdot \gamma_1$.
    
     To finish the proof, we set $\tau^{(t)} = 2D \gamma_1 \cdot t$. We assume that $\tau^{(t)} \le (\tau/2)^m - 2 D \gamma_1$, so it suffices for $4 D m \gamma_1 \le (\tau/2)^m$. Let $\mathcal{E}_t$ be the event that $B(x_t, (\tau/2)^m), B(x_t', (\tau/2)^m) \subset K_1$, and $\|x_{t}-x_{t}'\|_2 \le \tau^{(t)}$. Then, if $\mathcal{E}_t$ holds, the probability that $\|x_{t+1}-x_{t+1}'\|_2 \le \tau^{(t+1)}$ does not hold is at most $O\left(\tau + \frac{\gamma_1 \cdot D}{(\tau/2)^m}\right)$. In addition, the probability that $B(x_{t+1}, (\tau/2)^m), B(x_{t+1}', (\tau/2)^m) \subset K_1$ does not hold for any choice of $t+1$ is at most $O(m \cdot \tau)$ if $(\tau/2)^{m+1} \ge D \cdot \gamma_1$, by Proposition \ref{prop:hitandrun_bounded}. So, $\BP(\mathcal{E}_{t} \backslash \mathcal{E}_{t+1}) \le O(m \cdot \tau + \tau + \frac{\gamma_1 \cdot D}{(\tau/2)^m}),$ which means that the probability that $\mathcal{E}_m$ doesn't hold is at most $O\left(m^2 \cdot \tau + \frac{m \cdot \gamma_1 \cdot D}{(\tau/2)^m}\right)$, as long as $4 D m \gamma_1 \le (\tau/2)^{m}$. Assuming $\mathcal{E}_m$, we have that $\|x_m-x_m'\|_2 \le 2Dm \gamma_1$, and by Theorem \ref{thm:hit_and_run}, if $m \ge C d^2 D^2 \log \gamma^{-1}$ then the distribution of $x_m$ is $\gamma$-far from uniform over $K_1$.
     
     To summarize, we have that there exists a coupling of $x_m'$ (which is the random walk after $m$ steps of hit-and-run with $\gamma_1$-precision) with a uniform distribution $x$ over $K_1$ such that $\BP\left((\|x-x_m'\|_2 \le 2D m \gamma_1\right) \le O\left(\gamma + m^2 \tau + \frac{m \gamma_1 \cdot D}{(\tau/2)^m}\right)$, as long as $4D m \gamma_1 \le (\tau/2)^{m},$ $D \gamma_1 \le (\tau/2)^{m+1}$, and $m \ge C d^2 D^2 \log \gamma^{-1}$. Given some small parameters $\gamma_2, \gamma_3$, we set $\gamma = c \gamma_3$, $m = C d^2 D^2 \log \gamma^{-1}$, and $\tau = \frac{\gamma}{m^2}$ for some small constant $c$. Finally, we set $\gamma_1 = \min\left(\frac{\gamma_2}{2 D m}, \frac{\gamma (\tau/2)^m}{m D}, \frac{(\tau/2)^{m+1}}{4D m}\right)$ so that the conditions are satisfied and $\BP\left(\|x-x_m'\|_2 \ge \gamma_2\right) \le O(\gamma) \le \gamma_3$.
\end{proof}

Next, we must show that, rather than having $K_1, K_2 \subset B(\textbf{0}, D)$ for some polynomially sized $D$, we can have $K_1, K_2 \subset B(\textbf{0}, R)$ for $R$ exponentially large. In other words, one can avoid issues when the convex body is poorly conditioned.

\begin{lemma} \label{lem:shrink_condition_number}
    Let $\gamma_2, \gamma_3$ be as in Lemma \ref{lem:hitandrun_good}, and let $\gamma_1$ be defined as in the end of Lemma \ref{lem:hitandrun_good}, assuming $D := 2d^3$. For some $r < 1 < R$, Let $K_1, K_2$ be convex bodies with a $(K_1, K_2)$-membership oracle, such that $B(\textbf{0}, r) \subset K_1 \subset K_2 \subset B(\textbf{0}, R)$, and $\vol(K_2)-\vol(K_1) \le \left(\frac{\gamma_1 \cdot r}{6d}\right)^d$. Then, there exists a $\poly(d, \log \frac{R}{r}, \log \gamma_2^{-1}, \log \gamma_3^{-1})$-time algorithm that can find an affine transformation $\mathbb{A}$ such that $\mathbb{A}(K_1)$ is contained in $B(\textbf{0}, 2d^3)$ but contains $B(\textbf{0}, 1)$.
\end{lemma}

\begin{proof}
    The proof is an ellipsoid method, modified to deal with the fact that we do not have a perfect membership oracle and that we do not have a separation oracle.
    We will keep track of an interior ellipsoid $E_1 \subset K_1$ and an exterior ellipsoid $E_2 \supset K_2$, and keep trying to either grow $E_1$ or shrink $E_2$.
    The way we do this will be inspired by some recent work on sampling and volume computation of convex bodies~\cite{JiaLLV21}.
    
    Given current interior ellipsoid $E_1$ and exterior ellipsoid $E_2$, let $\mathbb{A}$ be some affine transformation that sends $E_1$ to the ball $B(\textbf{0}, 1) = \mathbb{A} E_1$, and where the largest axis of $\mathbb{A} E_2$ is parallel to the first coordinate direction. (Note that $\mathbb{A} E_2$ may not have center as the origin.) At every step, we will only increase the volume of $E_1$ and decrease the volume of $E_2$, and since $E_1$ started out as $B(\textbf{0}, r)$, the affine transformation $\mathbb{A}$ multiplies the volume by at most $\left(\frac{1}{r}\right)^d$. Therefore, $\vol(\mathbb{A} K_2) - \vol(\mathbb{A} K_1) \le \left(\frac{\gamma_1}{6d}\right)^d$.
    
    %Assume that we have a and that we have made some affine transformation so 
    %For now, assume
    %that $E_1$ is the ball $B(0, 1)$, that the largest axis of the ellipsoid $E_2$ is parallel to the first coordinate direction, and that $\vol(K_2)-\vol(K_1) \le \left(\frac{\gamma_1}{6d}\right)^d$. 
    %In slight abuse of notation, we replace $E_1, E_2, K_1, K_2$ with the affine transformations $\mathbb{A} E_1, \mathbb{A} E_2, \mathbb{A} K_1, \mathbb{A} K_2$.
    We may assume that this largest axis of $\mathbb{A} E_2$ has length at least $D := 2d^3$, or else we are already done. Define $K_1' := \mathbb{A} K_1 \cap B(\textbf{0}, D)$ and $K_2' := \mathbb{A} K_2 \cap B(\textbf{0}, D)$. Note that $\vol(K_2')-\vol(K_1') \le \vol(\mathbb{A} K_2)-\vol(\mathbb{A} K_1) \le \left(\frac{\gamma_1}{6d}\right)^d$, and $B(\textbf{0}, 1) \subset K_1'$, so by Proposition \ref{prop:multiplicative_vs_additive_volume}, $\mathbb{A} K_2 \subset (1+\gamma_1) \mathbb{A} K_1$ and $K_2' \subset (1+\gamma_1) K_1'$. Given $(K_1, K_2)$-membership oracle access, it is simple to obtain $(K_1', K_2')$-membership oracle access. Therefore, by Lemma \ref{lem:hitandrun_good}, we can produce a sample from a distribution that is $(\gamma_2, \gamma_3)$-close to uniform over $K_1'$ in $\poly(d, D, \log \gamma_2^{-1}, \log \gamma_3^{-1})$ time.
    
    Assuming without loss of generality that $\gamma_2, \gamma_3 \le d^{-100}$, we can repeat the sampling $\poly(d, D) = d^{O(1)}$ times and approximately learn the mean $\mu_1$ and covariance $\Sigma_1$ of the uniform distribution with respect to $K_1'$, up to $\ell_2$ norm (resp., Frobenius norm) error $1$ by using the empirical mean $\hat{\mu}_1$ and empirical covariance $\hat{\Sigma}_1$ as our estimates.
    
    First, suppose that one of our (approximate) samples from $K_1'$ was a point $x$ with $\ell_2$ norm at least $25d$. Then, if we define $y = (1-\gamma_1) x$, then $y \in \mathbb{A} K_1$ and $\|y\|_2 \ge 20 d$. If we rotate the space $\BR^d$ and assume $y = (y_1, 0, 0, \dots, 0) \in \BR^d$ for $y_1 \ge 20 d$, then the ellipse $E = \left\{z: (\frac{z_1-10}{10})^2 + \sum_{i=2}^{d} \left(\frac{z_i}{(1-1/d)}\right)^2 \le 1\right\}$ is contained in the convex hull of $B(\textbf{0}, 1)$ and $y$. The volume ratio $\vol(E)/\vol(B(\textbf{0}, 1))$ is $10 \cdot (1-\frac{1}{d})^{d-1} \ge \frac{10}{e} \ge 2$, so we can replace $\mathbb{A} E_1 = B(\textbf{0}, 1)$ with the larger ellipsoid $E \subset \mathbb{A} K_1$.
    
    Alternatively, every sample we drew has $\ell_2$ norm at most $25d$, which means that the empirical covariance $\hat{\Sigma}_1$ has operator norm at most $25d$. Thus, $\Sigma_1$ has operator norm at most $30 d$, meaning $x^\top \Sigma_1 x \le 30 d$ for all $\|x\|_2 = 1$.
    Now, by Theorem \ref{thm:ellipsoid_covariance}, $K_1' \subset \left\{x: (x-\mu_1)^\top \Sigma_1^{-1} (x-\mu_1) \le d(d+2)\right\}$. So if $x \subset K_1'$, then $(x-\mu_1)^\top \Sigma_1^{-1} (x-\mu_1) \le d(d+2)$, and since the minimum eigenvalue of $\Sigma_1^{-1}$ is at least $\frac{1}{25 d + 1},$ this means that $\|x-\mu_1\|^2 \le d(d+2)(25d+1)$ for all $x \in K_1'$. Since the origin is in $K_1'$, this implies that every point in $K_1'$ has norm bounded by $O(d^{3/2})$.
    
    Recall that the original convex bodies $K_1, K_2$ are known to be in $E_2$, and $\mathbb{A} E_2$ has largest axis parallel to the first coordinate direction. If the major radius of $\mathbb{A} E_2$ is some $F$, then we claim that all points in $\mathbb{A} K_1$ or $\mathbb{A} K_2$ have first coordinate bounded in magnitude by $O\left(\frac{F}{d^{3/2}}\right)$. To see why, for any $x \in \mathbb{A} K_1$, $x \cdot \frac{D}{F}$ is in $K_1$ by convexity. Moreover, since $\|x\|_2 \le F$, this means that $\|x \cdot \frac{D}{F}\|_2 \le D$ so $x \cdot \frac{D}{F} \in K_1'$. Therefore, we actually have $\left\|x \cdot \frac{D}{F}\right\| \le O(d^{3/2}),$ which means that $\|x\| \le O(F \cdot d^{3/2}/D) = O(F/d^{3/2})$. This implies that $|x_1|$ is at most $O(F/d^{3/2})$ for all $x \in \mathbb{A} K_1$: this therefore is also true for all $x \in \mathbb{A} K_2$. The intersection of the ellipsoid $\mathbb{A} E_2$ with the set of points with first coordinate at most $O(F/d^{3/2})$ is contained in the ellipsoid $E$ which shrinks the first axis of $\mathbb{A} E_2$ by a factor of $10$ and grows all other directions by $1 + \frac{1}{d}$. So, we can replace $\mathbb{A} E_2$ with another ellipsoid $E \supset \mathbb{A} K_2$ with volume at most $\frac{e}{10} \le 0.5$ times the volume of $\mathbb{A} E_2$.
    
    Therefore, unless $2d^3 \cdot \mathbb{A} E_1 \supset \mathbb{A} E_2$, we can find either a new larger $E_1$ or a new smaller $E_2$ in polynomial time. %While we made the assumption that $E_1 = B(0, 1)$ and the largest axis of $E_2$ is parallel to the first coordinate direction, this is WLOG by applying some affine transformation to make this true and then undoing the transformation. The affine transformation 
    Each time this takes $\poly(d, D, \log \gamma_2^{-1}, \log \gamma_3^{-1}) = \poly(d, \log \gamma_2^{-1}, \log \gamma_3^{-1})$ time. However, the volume ratio of the original ellipsoids $B(\textbf{0}, r)$ and $B(\textbf{0}, R)$ is $(R/r)^d$, so we can only repeat this process at most $O(d \log \frac{R}{r})$ times. 
    %Eventually, we can take an affine transformation that sends the final $E_1$ to the ball $B(0, 1)$, and it will send $E_2$ to an ellipsoid contained in $B(0, 2d^3)$.
\end{proof}

We combine Lemma \ref{lem:hitandrun_good} and Lemma \ref{lem:shrink_condition_number} to obtain the following corollary.

\begin{corollary} \label{cor:hitandrun_good}
    For any parameters $r, \gamma_2, \gamma_3 < 1 < R$, there exists some $\gamma_1$ such that $\log \gamma_1^{-1} = \poly(d, \log \frac{R}{r}, \log \gamma_2^{-1}, \log \gamma_3^{-1})$ and the following holds. If $K_1, K_2$ are convex bodies such that $B(\textbf{0}, r) \subset K_1 \subset K_2 \subset B(\textbf{0}, R)$ and $\vol(K_2)-\vol(K_1) \le \left(\frac{\gamma_1 \cdot r}{6d}\right)^d,$ then there is a $\poly(d, \log \frac{R}{r}, \log \gamma_2^{-1}, \log \gamma_3^{-1})$-time algorithm that can sample from a distribution that is $(\gamma_2, \gamma_3)$-close to uniform on $K_1$.
\end{corollary}

\begin{proof}
    First, use Lemma \ref{lem:shrink_condition_number} to find an affine transformation $\mathbb{A}$ such that $\mathbb{A}$ applied to $K_1$ contains $B(\textbf{0}, 1)$ but is contained in $B(\textbf{0}, 2d^3)$. Then, if we define $\gamma_2' = \frac{\gamma_2}{(R/r) \cdot d^3}$, we can produce a sample that is $(\gamma_2', \gamma_3)$-close to uniform on $\mathbb{A} K_1$. Finally, undo the affine transformation and the sample will still be $(\gamma_2, \gamma_3)$-close to uniform. 
\end{proof}

Unfortunately, being $(\gamma_2, \gamma_3)$-close to uniform does not necessarily ensure privacy. This is because one may extract information about the data based on minor perturbations of the generated sample. To fix this, we convert this version of closeness to pointwise closeness to uniform on a fine grid of points.

\begin{lemma}{\textbf{(Convex body sampling, \Cref{lem:sampling_main})}} \label{lem:sampling_main_appendix}
    Fix any parameters $\gamma_6 \le d^{-100}$ and $r < 1 < R$.
    Let $K_1, K_2$ be convex bodies such that $B(\textbf{0}, r) \subset K_1 \subset K_2 \subset B(\textbf{0}, R),$ and $\vol(K_2)-\vol(K_1) \le \left(\frac{\gamma_1 \cdot r}{6d}\right)^d,$ for $\gamma_1$ that will be defined in terms of $\gamma_6$. Suppose we have a $(K_1, K_2)$-membership oracle $\mathcal{O}$. Then, in $\text{poly}(d, \log \frac{R}{r}, \log \gamma_6^{-1})$ time and queries to $\mathcal{O}$, we can output a point $z$ that is $(1 \pm \gamma_6)$-pointwise close to uniform on the set of points in $\BR^d$ with all coordinates integer multiples of $\gamma_5$ that are accepted by $\mathcal{O}$, for $\gamma_5 = \frac{r \cdot \gamma_6}{d^3}.$
\end{lemma}

\begin{proof}
    First, we will define parameters $\gamma_1$ through $\gamma_5$ based on $r, R,$ and $\gamma_6$. Define $\gamma_4 := \frac{r}{d^2}$ and $\gamma_5 := \frac{\gamma_4 \cdot \gamma_6}{d}.$ Next, define $\gamma_2 := \frac{\gamma_5 \gamma_6}{d^2}$ and $\gamma_3 := \left(\frac{\gamma_5}{2R}\right)^d \cdot \frac{\gamma_6}{d}$. Finally, define $\gamma_1$ to be the value for $\gamma_1$ that appears when applying Corollary \ref{cor:hitandrun_good} with $\gamma_2$ and $\gamma_3$.
    %as done at the end of Lemma \ref{lem:hitandrun_good}.

    Let $K_1' = (1 + \frac{1}{d}) K_1$ and $K_2' = (1 + \frac{1}{d}) K_2$. Let $\mathcal{D}_1$ be the uniform distribution over $K_1'$. By applying Corollary \ref{cor:hitandrun_good} on $(K_1, K_2)$ and then scaling the point by $1+\frac{1}{d}$, we obtain a point $c \sim \mathcal{D}_2$, where $\mathcal{D}_2$ is $(\gamma_2, \gamma_3)$-close to $\mathcal{D}_1$. 
    
    Our algorithm works as follows. First, replace $c$ with $c + y$, where each coordinate $y_i$ was uniformly chosen from $[-\gamma_4, \gamma_4]$ with precision $\gamma_1$. Then, round each coordinate of $c+y$ to the nearest multiple of $\gamma_5$ to get a point $z$.
    Finally, we run a rejection sampling algorithm by checking whether the $(K_1, K_2)$-membership oracle accepts $z$. If so, we return $z$. If not, we restart the procedure until we accept some $z$. It will be simple to see that each step of the rejection sampling algorithm succeeds with probability $(1-\Omega(\frac{1}{d}))^d \ge \Omega(1)$ because $z$ will be in $K_1$ with this probability, so we can stop the rejection sampling after $O(\log \gamma_3^{-1})$ steps, to incur additional additive error $\gamma_3$.
    
    We now analyze the accuracy. Let $\mathcal{D}_3$ be a distribution so that we have a coupling between $\mathcal{D}_1, \mathcal{D}_3, \mathcal{D}_2$ such that if $(a, b, c) \sim \mathcal{D}_1, \mathcal{D}_3, \mathcal{D}_2,$ then $\|a-b\|_2 \le \gamma_2$ with probability $1$, and $\BP(b \neq c) \le \gamma_3$.
    Now, for any point $z$ with all coordinates multiples of $\gamma_5$ such that the $(K_1, K_2)$-membership oracle accepts $z$, we compute the probability of sampling $z$. In order to sample $z$, we must have sampled $c$ and $y$ so that $c+y$ rounds to $z$. This probability is the same as the probability that $b+y$ rounds to $z$, up to additive error $\gamma_3$. If we condition on choosing $a$ such that $\|a-z\|_\infty \le \gamma_4 - (\gamma_1+\gamma_2+\gamma_5)$, then $\|b-z\|_\infty \le \gamma_4-(\gamma_1+\gamma_5)$, so if each coordinate $y_i$ were chosen uniformly from $[-\gamma_4, \gamma_4]$ with perfect precision, the probability that $b+y$ rounds to $z$ will exactly be $(\gamma_5/(2 \gamma_4))^d$. Due to precision issues, the actual probability that $b+y$ rounds to $z$ is $((\gamma_5 \pm O(\gamma_1))/(2 \gamma_4))^d$. Likewise, if we choose $a$ such that $\|a-z\|_\infty \ge \gamma_4+(\gamma_1+\gamma_2+\gamma_5)$, then $\|b-z\|_\infty \ge \gamma_4+(\gamma_1+\gamma_5)$, so we will never select $b+y$ to round to $z$. Finally, if $\gamma_4-(\gamma_1+\gamma_2+\gamma_5) \le \|a-z\|_\infty \le \gamma_4+(\gamma_1+\gamma_2+\gamma_5)$, then the probability that $b+y$ rounds to $z$ is between $0$ and $((\gamma_5 + O(\gamma_1))/(2 \gamma_4))^d$.
    
    Since $a$ is truly uniform from $K_1' = (1+\frac{1}{d}) K_1$, we claim that the probability of selecting an $a$ with $\|a-z\|_\infty \le \gamma_4+(\gamma_1+\gamma_2+\gamma_5)$ is $(2(\gamma_4+(\gamma_1+\gamma_2+\gamma_5)))^d/\vol(K_1')$. For this to be true, we need every point $a$ with $\|a-z\|_\infty \le \gamma_4+(\gamma_1+\gamma_2+\gamma_5)$ to be in $K_1'$. Since $\mathcal{O}$ accepts $z$, this means $z \in K_2 \subset (1+\gamma_1) K_1$, so every point within $\ell_2$ distance $(\frac{1}{d}-\gamma_1) \cdot r$ of $z$ is contained in $(1+\frac{1}{d}) K_1 = K_1'$. So, it suffices for $\sqrt{d} \cdot (\gamma_1+\gamma_2+\gamma_4+\gamma_5) \le (\frac{1}{d}-\gamma_1) \cdot r$.
    Likewise, the probability of selecting an $a$ with $\|a-z\|_\infty \le \gamma_4-\gamma_5-\gamma_2$ is $(2(\gamma_4-\gamma_5-\gamma_2))^d/\vol(K_1')$. 
    
    So, the overall probability that $b+y$ rounds to $z$ is at least $(2(\gamma_4-(\gamma_1+\gamma_2+\gamma_5)))^d/\vol(K_1') \cdot ((\gamma_5-O(\gamma_1))/(2 \gamma_4))^d$ and at most $(2(\gamma_4+(\gamma_1+\gamma_2+\gamma_5)))^d/\vol(K_1') \cdot ((\gamma_5+O(\gamma_1))/(2 \gamma_4))^d$. Assuming that $d \cdot \gamma_1, \gamma_2 \ll \gamma_5$ and $d \cdot \gamma_5 \ll \gamma_4$, these bounds equal $\frac{\gamma_5^d}{\vol(K_1')} \cdot \left(1 \pm O\left(\frac{d \cdot \gamma_5}{\gamma_4} + \frac{d \cdot \gamma_1}{\gamma_5}\right)\right).$ We also need that $\gamma_4 \ll \frac{r}{d \sqrt{d}},$ so that $\sqrt{d} \cdot (\gamma_1+\gamma_2+\gamma_4+\gamma_5) \le (\frac{1}{d}-\gamma_1) \cdot r$. Finally, we had an additive error of $\gamma_3$ due to the coupling of the points $b$ and $c$, as well as another $\gamma_3$ for the rejection algorithm failing. So, the final probability of choosing some point $z$ with all coordinates integer multiples of $\gamma_5$ that is accepted by the $(K_1, K_2)$-membership oracle is $\left(\frac{d}{d+1}\right)^d \cdot \frac{\gamma_5^d}{\vol(K_1)} \cdot \left(1 \pm O\left(\frac{d \cdot \gamma_5}{\gamma_4} + \frac{d \cdot \gamma_1}{\gamma_5}\right)\right) \pm O(\gamma_3)$, where we used the fact that $\vol(K_1') = \vol(K_1) \cdot \left(1+\frac{1}{d}\right)^d$.
    
    Based on how we set $\gamma_1, \dots, \gamma_5$, all the conditions hold, and we can simplify the probability as $\left(\frac{d}{d+1}\right)^d \cdot \frac{\gamma_5^d}{\vol(K_1)} \cdot \left(1 \pm O\left(\frac{\gamma_6}{d}\right)\right) \pm \left(\frac{\gamma_5}{2R}\right)^d \cdot \gamma_6$. However, since $\vol(K_1) \le (2 R)^d$ and $\left(\frac{d}{d+1}\right)^d \ge \frac{1}{e},$ in total this equals $\left(\frac{d}{d+1}\right)^d \cdot \frac{\gamma_5^d}{\vol(K_1)} \cdot \left(1 \pm O\left(\frac{\gamma_6}{d}\right)\right)$. So, our sampling algorithm is pointwise accurate up to a $1 \pm O\left(\frac{\gamma_6}{d}\right)$ factor.
\end{proof}

Finally, we show that our sampling algorithm can also allow us to approximately compute the volume of points accepted by the oracle $\mathcal{O}$. More accurately, we can approximate the number of points in the grid of precision $\gamma_5$ that are accepted by $\mathcal{O}$.

We first note the following auxiliary fact.

\begin{fact} \label{fact:count_volume_equivalence_appendix}
    Suppose $K \subset \BR^d$ is a convex body that contains a ball of radius $r$. Suppose $\gamma$ is a parameter which is at most $\frac{r}{2 d^3}$. Then, the number of points $N$ in $K$ that have all coordinates integral multiples of $\gamma$ is $(1 \pm O(\frac{\gamma}{r} \cdot d^2)) \cdot \vol(K)/\gamma^d$.
\end{fact}

\begin{proof}
    Suppose the ball of radius $r$ is centered at some point $p \in \BR^d$. Now, define $K'$ to be the set of points $x \in \BR^d$ such that after rounding each coordinate of $x$ to the nearest multiple of $\gamma$, the point is still in $K$. Note that $\vol(K')$ precisely equals $\gamma^d \cdot N$.
    
    Let $\kappa := \frac{\gamma}{r} \cdot d \le \frac{1}{2d^2}$. If we define $K_{\text{big}}$ as the dilation of $K$ about $p$ by a factor $1+\kappa$ and $K_{\text{small}}$ as the dilation of $K$ about $p$ by a factor $\frac{1}{1+2\kappa}$, note that $K_{\text{small}} \subset K \subset K_{\text{big}}$, and $\vol(K_{\text{small}}), \vol(K_{\text{big}})$ are both $(1 \pm O(\kappa \cdot d)) \cdot \vol(K)$. Thus, if we prove that $K_{\text{small}} \subset K' \subset K_{\text{big}}$, the proof is complete.

    Note that if $x \in K'$, then there is a point $y \in K$ of distance at most $\gamma \sqrt{d} \le \kappa \cdot r$ from $x$. Thus, for some point $z$ with $\|z-p\|_2 \le r$ (which means $z \in K$ by our assumption), we can write $x = y + \kappa \cdot (z-p)$. We can further rewrite this as
\begin{align*}
    x &= p + (y-p) + \kappa \cdot (z-p) \\
    &= p + \left(1 + \kappa\right) \cdot \left(\frac{1}{1+\kappa} \cdot (y-p) + \frac{\kappa}{1+\kappa} \cdot (z-p)\right) \\
    &= \frac{1}{1+\kappa} \cdot \left(p + \left(1+\kappa\right) \cdot (y-p)\right) + \frac{\kappa}{1+\kappa} \cdot \left(p + \left(1+\kappa\right) \cdot (z-p)\right).
\end{align*}
    Hence, $x$ is a convex combination of $p+\left(1+\kappa\right) \cdot (y-p)$ and $p+\left(1+\kappa\right) \cdot (z-p)$, which are both in $K_{\text{big}}$ by definition, since both $y, z \in K$. Since $K$ is convex, so is $K_{\text{big}}$, and thus $x \in K_{\text{big}}$.

    %Thus, we proved $K' \subset K_{\text{big}}.$ To prove that $K' \supset K_{\text{small}}$, assume that $d \ge 2$, or else $K_{\text{small}}$ just contains $p$, which is clearly in $K'$.
    If $x \not\in K'$. then there is a point $y \not\in K$ of distance at most $\gamma \sqrt{d} \le \kappa \cdot r$ from $x$. We can write $y = x + 2\kappa \cdot (z-p)$, where $\|z-p\|_2 \le \frac{r}{2}$. The same equations as above, but switching $x$ with $y$ and replacing $\kappa$ with $2\kappa$, gives us
\[y = \frac{1}{1+2\kappa} \cdot \left(p + \left(1+2\kappa\right) \cdot (x-p)\right) + \frac{2\kappa}{1+2\kappa} \cdot \left(p + \left(1+2\kappa\right) \cdot (z-p)\right).\]
    Thus, $y$ is a linear combination of $p + \left(1+2\kappa\right) \cdot (x-p)$ and $p + \left(1+2\kappa\right) \cdot (z-p)$, which means one of these points cannot be in $K$. However, $\|z-p\|_2 \le \frac{r}{2}$, so $p + \left(1+2\kappa\right) \cdot (z-p)$ is in a ball of radius $\left(1+2\kappa\right) \cdot \frac{r}{2} \le r$ around $p$ (using the fact that $\kappa \le \frac{1}{2d^2} \le \frac{1}{2}$), and thus must be in $K$. Thus, $p + \left(1+2\kappa\right) \cdot (x-p)$ is not in $K$, which means $x \not\in K_{\text{small}},$ because if we dilate $K_{\text{small}}$ by $1+2\kappa$ about $p$, it is still contained in $K$.
\end{proof}

\begin{lemma}{\textbf{(Volume sampling, \Cref{lem:volume_computation})}} \label{lem:volume_computation_appendix}
    Let all notation be as in Lemma \ref{lem:sampling_main}. Fix any $\eps < 0.5$, and set $\gamma_6 \le \frac{\eps}{d \log \frac{R}{r}}$ and $\gamma_1, \dots, \gamma_5$ in terms of $\gamma_6$ as in Lemma \ref{lem:sampling_main}. Then, for any $\gamma < 1$, in $\poly(d, \log \frac{R}{r}, \frac{1}{\eps}, \log \gamma^{-1})$ time and oracle accesses, we can approximate the number of points in $\BR^d$ with all coordinates integer multiples of $\gamma_5$ that are accepted by $\mathcal{O}$, up to a $1 \pm \eps$ multiplicative factor, with failure probability $\gamma$.
\end{lemma}

\begin{proof}
    For some $\rho \in [r, R]$, let $K_1^{(\rho)} = K_1 \cap B(\textbf{0}, \rho)$ and $K_2^{(\rho)} = K_2 \cap B(\textbf{0}, \rho)$. Clearly, $B(\textbf{0}, r) \subset K_1^{(\rho)} \subset K_2^{(\rho)} \subset B(\textbf{0}, R)$, and $\vol(K_2^{(\rho)})-\vol(K_1^{(\rho)}) \le \left(\frac{\gamma_1 \cdot r}{6d}\right)^d$. Also, let $S^{(\rho)}$ be the set of points in $K_2^{(\rho)}$ with all coordinates multiples of $\gamma_5$ that are accepted by the oracle, and let $N^{(\rho)} := |S^{(\rho)}|$.
    
    Since $\gamma_2 \le \gamma_5 \le \frac{r \cdot \gamma_6}{d^3},$ and since $\vol(K_1^{(\rho)})$ and $\vol(K_2^{(\rho)})$ are within a $1 \pm o(1)$ multiplicative factor of each other, we have that $\vol(K_1^{(\rho)}) = \vol(K_2^{(\rho)}) = (1 \pm o(1)) \cdot N^{(\rho)} \cdot (\gamma_5)^d$,
    by \Cref{fact:count_volume_equivalence_appendix}.
    %To see why, suppose $x$ is a point that, after rounding each coordinate to the nearest multiple of $\gamma_5$, is in $(1-\frac{\gamma_5 \sqrt{d}}{r}) \cdot K_1^{(\rho)}$. Then, since $x$ moved by at most $\gamma_5 \sqrt{d}$ in absolute value, and since $B(\textbf{0}, r) \subset K_1^{(\rho)}$, $x$ must be in $K_1^{(\rho)}$ and so is accepted by the oracle.
    %Therefore, $(\gamma_5)^d \cdot N^{(\rho)} \ge (1-\frac{\gamma_5 \sqrt{d}}{r})^d \cdot \vol(K_1^{(\rho)})$ which means $\vol(K_1^{(\rho)}) \le (1+o(1)) \cdot N^{(\rho)} \cdot (\gamma_5)^d$. For the other direction, any point that is in $K_2^{(\rho)} \subset (1+\gamma_2) K_1^{(\rho)}$, if we change each coordinate by up to $\gamma_5$, is still in $(1+\gamma_2) \cdot (1+\frac{\gamma_5 \sqrt{d}}{r}) \cdot K_1^{(\rho)}$. Therefore $(\gamma_5)^d \cdot N^{(\rho)} \le \vol(K_1^{(\rho)}) \cdot (1+o(1))$.
    
    Now, if $\rho' \le (1+\frac{1}{d}) \rho$, note that $K_2^{(\rho')} \subset (1+\frac{1}{d}) K_2^{(\rho)}$. Therefore, this means $\vol(K_1^{(\rho')}) \le (e+o(1)) \cdot \vol(K_1^{(\rho)}),$ which means that $N^{(\rho')} \le (e+o(1)) \cdot N^{(\rho)}$. Given this, by Lemma \ref{lem:sampling_main}, we can generate $1 \pm \gamma_6$-pointwise random samples from $S^{(\rho')}$ and check the fraction of them that are in $S^{(\rho)}$ by determining for each sample if its $\ell_2$ norm is at most $\rho$. By Hoeffding's inequality, for any $\eps' < 1$ we can compute $\frac{N^{(\rho)}}{N^{(\rho')}}$ with failure probability $\gamma$ up to an additive error of $\pm O(\eps' + \gamma_6)$ in $O((\eps')^{-2} \log \gamma^{-1})$ random samples, and since $1 \le \frac{N^{(\rho')}}{N^{(\rho)}} \le e+o(1),$ this also implies we can compute the ratio up to a multiplicative factor of $1 \pm O(\eps'),$ assuming $\gamma_6 \le \eps'$. 
    
    Now, consider $r = \rho_0, \rho_1, \rho_2, \dots, \rho_M = R$, where $\frac{\rho_{t+1}}{\rho_t} \le 1 + \frac{1}{d}$. We can let $M = O(d \log \frac{R}{r})$. Then, by setting $\eps' = \frac{\eps}{M}$ we can compute $\frac{N^{(\rho_{t+1})}}{N^{(\rho_t)}}$ up to multiplicative error $e^{O(\eps/M)}$ in $\poly\left(d, \log \frac{R}{r}, \log \gamma_6^{-1}, \log \gamma_1^{-1}, \frac{1}{\eps}\right)$ time, with failure probability $\frac{\gamma_1}{M}$. By multiplying all of our estimates to form a telescoping product, we can compute $\frac{N^{(R)}}{N^{(r)}}$ up to a multiplicative factor $e^{\pm O(\eps)}$ with failure probability $\gamma_1$. Our goal is precisely to compute $N^{(R)}$, so it suffices to compute $N^{(r)}$. But since $B(\textbf{0}, r) \subset K_1$, this is just the number of points with all coordinates integral multiples of $\gamma_5$ that are in $B(\textbf{0}, r)$.
    By the argument of the above paragraph, this is just $\gamma_5^{-d} \cdot \vol(B(\textbf{0}, r)) \cdot \left(1 \pm \frac{\gamma_5 \sqrt{d}}{r}\right)^{d} = \left(\frac{r}{\gamma_5}\right)^d \cdot e^{\pm \gamma_5 \cdot d^{3/2}/r} \cdot \vol(B(\textbf{0}, 1))$. Since $\gamma_5 = \frac{r \cdot \gamma_6}{d^3},$ $\gamma_5 \cdot d^{3/2}/r \le \gamma_6 \le \eps$.
    Therefore, since the volume of a $d$-dimensional sphere has an explicit representation, we can compute $N^{(R)}$ up to multiplicative error $e^{\pm O(\eps)}$ in time $\poly(d, \log \frac{R}{r}, \log \gamma_6^{-1}, \frac{1}{\eps}, \log \gamma^{-1}) = \poly(d, \log \frac{R}{r}, \frac{1}{\eps}, \log \gamma^{-1})$ time.
\end{proof}

\section{Sum-of-squares proofs} \label{appendix:sos}

In this section, we prove sum-of-squares proofs that are crucial in establishing accuracy of our algorithms, as well as privacy in the approx-DP setting.
These include both results both when the data points are sampled from a Gaussian, and for worst-case results.
Due to precision issues when solving a semidefinite program, our bounds must hold with respect to not only all pseudo-expectations but also with respect to linear operators that are ``approximate pseudoexpectations''.
The exponentially-small numerical errors this introduces are manageable by observing that the coefficients in the SoS proofs we use to analyze these approximate pseudoexpectations are at most some fixed polynomial in the bit-representation of the input; see e.g. the discussion in \cite{HopkinsL18}.

In \Cref{subsec:accuracy-proofs}, we recall the sum-of-squares results from \cite{KothariMZ22}, and use these to establish accuracy for Gaussian data. Namely, we show that a low-scoring point with respect to samples drawn from a Gaussian (or more generally for samples with the required resilience samples) must be a good estimate for the mean/covariance of the Gaussian. 
Next, we prove two sum-of-squares results showing that any set of data points, no matter how corrupted, cannot have a very large volume of potential means (or covariances) which all have low scores. This differs from accuracy results proven in prior work, which assume that a large fraction of the points come from some distribution. This establishes a ``worst-case accuracy'' result, which is crucial to establishing privacy in our approx-DP algorithms.
%To ease the notation, we prove $1$-dimensional equivalents of the proofs, and then apply the $1$-dimensional proofs directly to obtain high-dimensional SoS proofs of accuracy.
We prove a result for covariance estimation in \Cref{subsec:covariance_sos_arbitrary} and a result for mean estimation in \Cref{subsec:mean_sos_arbitrary}.
%Finally, we apply this result to arbitrary high-dimensional datasets in \Cref{subsec:approx-accuracy-proofs}.

\subsection{Proofs of Accuracy Lemmas} \label{subsec:accuracy-proofs}

In this subsection, we prove the accuracy results for mean and covariance estimation (Lemmas \ref{lem:mean-accuracy}, \ref{lem:covariance-accuracy}, and \ref{lem:tv-accuracy}).

The main sum-of-squares result that we apply is the following lemma due to Kothari, Manohar, and Zhang.

\begin{lemma} \cite[Lemma 4.1, restated]{KothariMZ22} \label{lem:kmz_sos_main}
    Let $x_1, \dots, x_n \in \BR^d$, and let $\mu_0 = \frac{1}{n} \sum_{i=1}^n x_i$ be the sample mean. Let $V(\mu, v)$ for $v \in \BR^d$ be a degree at most 2 polynomial in $\mu$, that is always nonnegative for all $\mu \in \BR^d$ and $v$ in some fixed subset $\cS \subset \BR^d$. Suppose that for all vectors $a \in [0, 1]^n$ with $\sum_{i=1}^n a_i \ge (1-\eta) n$, and for all $v \in \cS$, we have
\[\left|\frac{1}{n} \sum_{i=1}^n a_i \langle x_i-\mu_0, v \rangle\right| \le \tcO(\eta) \cdot \sqrt{V(\mu_0, v)} \hspace{0.2cm} \text{and} \hspace{0.2cm} \left|\frac{1}{n} \sum_{i=1}^n a_i [\langle x_i-\mu_0, v \rangle^2 - V(\mu_0, v)]\right| \le \tcO(\eta) \cdot V(\mu_0, v).\]

    Let $\pE$ be a degree-6 pseudoexpectation on $\{x_i'\}_{i=1}^n$ and $\{w_i\}_{i=1}^n$ such that 
\begin{enumerate}
    \item $\forall i \in [n]$, $\pE$ satisfies $w_i^2 = w_i$,
    \item $\forall i \in [n]$, $\pE$ satisfies $w_i x_i' = w_i x_i$,
    \item $\pE$ satisfies $\sum_{i=1}^n w_i \ge (1-\eta) n,$
    \item For all $v \in \cS$, $\pE\left[\frac{1}{n} \sum_{i=1}^n \langle x_i'-\mu', v\rangle^2\right] \le (1+\tcO(\eta)) \cdot \pE[V(\mu', v)]$, where $\mu' := \frac{1}{n} \sum_{i=1}^n x_i'$.
\end{enumerate}

    Then, for every unit vector $v \in \cS$, the following two inequalities hold:
\begin{align*}
    \pE\left[\langle \mu'-\mu_0, v\rangle^2\right] &\le O(\eta) \cdot (\pE[V(\mu', v)]+V(\mu_0, v)). \\
    |\langle \pE[\mu']-\mu_0, v \rangle| &\le \tcO(\eta) \cdot \sqrt{V(\mu_0, v)+\pE[V(\mu', v)]} + \sqrt{\tcO(\eta) \cdot (\pE[V(\mu', v)]-V(\mu_0, v))}.
\end{align*}
\end{lemma}

We first prove \Cref{lem:mean-accuracy}.

\begin{proof}[Proof of \Cref{lem:mean-accuracy}]
    Since $\phi \le \alpha/\sqrt{d}$, it suffices to show that any $(\alpha^*, \tau, \phi, T)$-certificate $\cL$ for $\cX$ satisfies $\|\cL[\mu']-\mu\|\le O(\alpha)$. If we assume $\tau = 0$ and $\alpha = \tcO(\eta)$, then in fact $\cL$ is a degree-6 pseudoexpectation that precisely satisfies the four required conditions of \Cref{lem:kmz_sos_main}, if we set $V(\mu', v) := 1$ and define $\cS$ to be the set of unit vectors in $\BR^d$. In addition, by Corollary \ref{cor:resilience}, the required conditions on $x_i$ hold up to replacing $\alpha$ with $2 \alpha$ and the sample mean $\mu_0$ with the true mean $\mu$. However, Corollary \ref{cor:resilience} implies that $\|\mu-\mu_0\|_2 \le \alpha$, so $|\langle x_i-\mu_0, v \rangle - \langle x_i-\mu, v \rangle| \le \alpha$ and $|\langle x_i-\mu_0, v \rangle^2 - \langle x_i-\mu, v \rangle^2| \le \alpha (1+|\langle x_i-\mu, v \rangle|)$. But $\frac{1}{n} \sum_{i=1}^n |\langle x_i-\mu, v \rangle| \le O(1)$. Together this means that the conditions of Lemma \ref{lem:kmz_sos_main} hold up to replacing $\alpha$ with $O(\alpha)$.

    Therefore, for any unit vector $v$, $|\langle \cL[\mu']-\mu, v \rangle| \le \tcO(\eta) \le O(\alpha)$ by \Cref{lem:kmz_sos_main}, as desired.

    While our proof was for exact pseudoexpectations since we set $\tau = 0$, as mentioned in \cite{KothariMZ22}, the proof also extends to approximate pseudoexpectations for small $\tau$, since the coefficients at each step in the sum-of-squares proof are polynomially bounded (see, e.g., the discussion in~\cite{HopkinsL18} or~\cite{KothariMZ22}). Here, we must make the assumption that every $x_i$ and $\|\cR(\cL)\|_2$ have magnitude bounded by $(nd R)^{O(1)}$, which holds automatically assuming the resilience properties and Condition 4 in \Cref{def:certifable-mean}.
\end{proof}

Next, we prove \Cref{lem:covariance-accuracy}.

\begin{proof}[Proof of \Cref{lem:covariance-accuracy}]
    Given samples $x_1, \dots, x_n$ and $d$-dimensional indeterminates $x_1', \dots, x_n'$, we define the indeterminates $z_1', \dots, z_n'$ as $z_i' := (x_i')(x_i')^\top$ (note that each $z_i'$ is a $d \times d$-dimensional matrix), and $\Sigma' := \frac{1}{n} \sum z_i'$. We also define $z_i := x_i x_i^\top$ and $\Sigma_0 := \frac{1}{n} \sum z_i$.

    We will apply Lemma \ref{lem:kmz_sos_main}, but replacing $d$ with $d^2$, $x_i$ with $z_i$, $x_i'$ with $z_i'$, $\mu_0$ with $\Sigma_0$, and $\mu'$ with $\Sigma'$. We also define $\cS$ to be the subset of vectors of the form $vv^\top$ where $v$ is a $d$-dimensional unit vector. (Note that $vv^\top$ is $d^2$-dimensional and has $\ell_2$ norm $1$ when flattened). Finally, for $\Sigma, M \in \BR^{d \times d}$, we define $V(\Sigma, M) := 2 \cdot \langle \Sigma, M \rangle^2$.

    Now, for any $(\alpha^*, \tau, T)$-certificate $\cL$ with $\alpha^* \le \alpha$, it suffices to show that for any unit vector $v \in \BR^d$, $(1-O(\alpha)) v^\top \Sigma v \le \cL[v^\top \Sigma' v] \le (1+O(\alpha)) v^\top \Sigma v$. This would imply that $(1-O(\alpha)) \Sigma \preccurlyeq \cL[\Sigma'] \preccurlyeq (1+O(\alpha)) \Sigma$, which means for $\tau \ll 1/\poly(n, d, K)$, $(1-O(\alpha)) \Sigma \preccurlyeq \tSigma \preccurlyeq (1+O(\alpha)) \Sigma$.

    We start by assuming $\tau = 0$, so $\cL$ is actually a degree-12 pseudoexpectation. Then, $\cL$ satisfies $w_i^2 = w_i$, $w_i (x_i')(x_i')^\top = w_i x_i x_i^\top$, and $\sum w_i \ge (1-\eta) n$. In addition, since $\cL\left[\|(v^{\otimes 2})^\top M^\top M v^{\otimes 2}\|_2^2\right] = \cL\left[\|M v^{\otimes 2}\|_2^2\right] \ge 0$, this means
\[\cL\left[\frac{1}{n} \sum_{i=1}^n \left(\langle x_i', v \rangle^2 - v^\top \Sigma' v\right)^2\right] \le (2+\tcO(\eta)) \cdot \cL\left[(v^\top \Sigma' v)^2\right].\]
    But note that $V(\mu', v)$ is precisely replaced with $2 \cdot \langle \Sigma', vv^\top \rangle^2 = 2 (v^\top \Sigma' v)^2$. In addition, $\langle x_i', v \rangle^2-v^\top \Sigma' v = \langle z_i'-\Sigma', vv^\top \rangle$. Hence, the 4 conditions in Lemma \ref{lem:kmz_sos_main} are satisfied.

    In addition, to apply Lemma \ref{lem:kmz_sos_main} we need to verify the desired conditions for $z_i$. By the resilience properties (Lemma \ref{lem:resilience-of-moments-covariance}) of $\{\Sigma^{-1/2} x_i\}$, we have that $(1-\alpha) \|v\|_2^2 \le \frac{1}{n} \sum v^\top \Sigma^{-1/2} x_i x_i^\top \Sigma^{-1/2} v \le (1+\alpha) \|v\|_2^2,$ which means by replacing $v$ with $\Sigma^{1/2} v$, we have 
\begin{equation}
    (1-\alpha) (v^\top \Sigma v) \le v^\top \Sigma_0 v \le (1+\alpha) v^\top \Sigma v. \label{eq:empirical_cov_good}
\end{equation}
    Now, for any unit vector $v$ and $a_1, \dots, a_n \in [0, 1]$ with $\sum a_i \ge (1-\eta) n$, $\left|\frac{1}{n} \sum_{i=1}^n a_i (\langle v, \Sigma^{-1/2} x_i \rangle^2 - 1)\right| = \left|\frac{1}{n} \sum_{i=1}^n a_i \langle (\Sigma^{-1/2} x_i) (\Sigma^{-1/2} x_i)^\top-I, vv^\top\rangle\right| \le \tcO(\eta)$ if $\{\Sigma^{-1/2} x_i\}$ satisfy the resilience properties. By scaling, for general vectors $v$, $\left|\frac{1}{n} \sum_{i=1}^n a_i (\langle v, \Sigma^{-1/2} x_i \rangle^2 - \|v\|_2^2)\right| \le \tcO(\eta) \cdot \|v\|_2^2$, which means by replacing $v$ with $\Sigma^{1/2} v$, we have $\left|\frac{1}{n} \sum_{i=1}^n a_i (\langle v, x_i \rangle^2 - v^\top \Sigma v)\right| \le \tcO(\eta) \cdot v^\top \Sigma v$. By \eqref{eq:empirical_cov_good}, this implies 
\[\left|\frac{1}{n} \sum_{i=1}^n a_i \langle vv^\top, z_i - \Sigma_0 \rangle\right| \le \tcO(\eta) \cdot \langle vv^\top, \Sigma_0 \rangle.\]

    Next, note that
\begin{align*}
    \langle z_i-\Sigma_0, vv^\top\rangle^2 &= \langle z_i-\Sigma, vv^\top\rangle^2 + 2 \langle \Sigma-\Sigma_0, vv^\top \rangle \cdot \langle z_i-\Sigma, vv^\top\rangle + \langle \Sigma-\Sigma_0, v^\top \rangle^2 \\
    &= \langle z_i-\Sigma, vv^\top\rangle^2 \pm O(\alpha) \cdot (v^\top \Sigma v) \cdot |\langle z_i-\Sigma, vv^\top \rangle| \pm O(\alpha^2) \cdot (v^\top \Sigma v)^2.
\end{align*}
    We can rewrite $\langle z_i-\Sigma, vv^\top \rangle = \langle \Sigma^{-1/2} x_i, \Sigma^{1/2} v \rangle^2 - \|\Sigma^{1/2}v\|_2^2$. This means by applying Lemma \ref{lem:resilience-of-moments-covariance} with $P = (\Sigma^{1/2} v)(\Sigma^{1/2} v)^\top$, we have that
\[\frac{1}{n}\sum_{i=1}^{n} \langle z_i-\Sigma, vv^\top\rangle^2 = (2 \pm O(\alpha)) \cdot \|\Sigma^{1/2} v\|_2^4 = (2 \pm O(\alpha)) \cdot (v^\top \Sigma v)^2.\]
    and
\[\frac{1}{n}\sum_{i=1}^{n} \left|\langle z_i-\Sigma, vv^\top\rangle\right| \le O(1) \cdot \|\Sigma^{1/2} v\|_2^2 = O(1) \cdot (v^\top \Sigma v).\]
    Together, this implies that
\[\left|\frac{1}{n} \sum_{i=1}^n a_i \langle vv^\top, z_i-\Sigma_0\rangle^2\right| = (2 \pm O(\alpha)) \cdot (v^\top \Sigma v)^2 = (1 \pm O(\alpha)) \cdot V(\Sigma_0, vv^\top).\]
    Since $\sum a_i \ge (1-\eta) n$, this completes the verification of the conditions.

    Now, we may apply Lemma \ref{lem:kmz_sos_main}. We first have that 
\[\cL[(v^\top (\Sigma' - \Sigma_0) v)^2] \le O(\eta) \cdot \left(\cL[(v^\top \Sigma' v)^2] + \cL[(v^\top \Sigma_0 v)^2]\right).\]
    By Cauchy-Schwarz, we know that 
\[\cL[(v^\top \Sigma' v)^2] \le 2 \cdot \left(\cL[(v^\top \Sigma_0 v)^2] + \cL[(v^\top (\Sigma'-\Sigma_0) v)^2]\right),\]
    which means 
\[\cL[(v^\top (\Sigma'-\Sigma_0) v)^2] \le O(\eta) \cdot \left(\cL[(v^\top (\Sigma'-\Sigma_0) v)^2]+\cL[(v^\top \Sigma_0 v)^2]\right)\]
    and therefore, 
\[\cL[(v^\top (\Sigma'-\Sigma_0) v)^2] \le O(\eta) \cdot (v^\top \Sigma_0 v)^2\]
    since $v, \Sigma_0$ are fixed determinates. This also implies that $\cL[(v^\top \Sigma' v)^2] \le O(1) \cdot (v^\top \Sigma_0 v)^2.$

    Let $A := v^\top \Sigma_0 v$, and $B := \cL[v^\top (\Sigma'-\Sigma_0) v]$. Then, $V(\Sigma_0, vv^\top) = 2 A^2$, $\cL[V(\Sigma', vv^\top)] = O(A^2)$, and $\cL[V(\Sigma', vv^\top)]-V(\Sigma_0, vv^\top) = 2 \cdot \cL[(v^\top \Sigma' v)^2-(v^\top \Sigma_0 v)^2] = 2 \cdot \cL[(v^\top (\Sigma'-\Sigma_0) v)^2] + 4 A \cdot B$. In addition, we know that $\cL[(v^\top (\Sigma'-\Sigma_0) v)^2] \le O(\eta) \cdot A^2$. Hence, Lemma \ref{lem:kmz_sos_main} implies that 
\[|B| \le \tcO(\eta) \cdot A + \sqrt{\tcO(\eta) \cdot \tcO(\eta) \cdot A^2 + \tcO(\eta) \cdot A \cdot B} \le \tcO(\eta) \cdot A + \sqrt{\tcO(\eta)} \cdot \sqrt{A \cdot |B|}.\]
    This means that $|B| \le \tcO(\eta) \cdot A$, which means that $\cL[v^\top \Sigma' v] = (1 \pm \tcO(\eta)) \cdot v^\top \Sigma_0 v = (1 \pm O(\alpha)) \cdot v^\top \Sigma v$, where the last equation follows by \eqref{eq:empirical_cov_good}.

    This completes the proof for true pseudoexpectations. Again, the proof extends to approximate pseudoexpectations, since the coefficients at each step in the sum-of-squares proof are polynomially bounded.
\end{proof}

Finally, we prove \Cref{lem:tv-accuracy}.

\begin{proof}[Proof of \Cref{lem:tv-accuracy}]
    %We define $z_1', \dots, z_n'$, $z_1, \dots, z_n$, $\Sigma'$, and $\Sigma_0$ as in \Cref{}. 
    As in \Cref{lem:covariance-accuracy}, we apply \Cref{lem:kmz_sos_main} with some replacements. This time, we replace $d$ with $d^2$, $x_i$ with $z_i = x_i^{\otimes 2}$, $x_i'$ with $z_i' = (x_i')^{\otimes 2}$, $\mu_0$ with $S_0 = \frac{1}{n} \sum_i z_i$, and $\mu'$ with $S' = \frac{1}{n} \sum_i z_i'$. In addition, the set $\cS \subset \BR^{d^2}$ will represent all vectors $P$ of norm $1$ such that the $d \times d$ matrix $M$ such that $M^\flat = P$ is symmetric. Finally, we define $V(S, P) := 2$.

    For any $(\alpha^*, \tau, \phi, T)$-certificate $\cL$ with $\phi \le \alpha/\sqrt{d}$ and $\alpha^* \le \alpha$, it suffices to show that $\|\cL[\Sigma']-\Sigma\|_F \le \alpha$, since $\|\cL[\Sigma'] - \tSigma\|_F \le \sqrt{d} \cdot \|\cL[\Sigma'] - \tSigma\|_{op} \le O(\alpha)$.

    We again assume $\tau = 0$, so $\cL$ is actually a degree-12 pseudoexpectation. We will also replace the $24 \alpha$ in Condition 3 with $\alpha$ for convenience (this only causes the error to multiply by an $O(1)$ factor). Then, $\cL$ satisfies $w_i^2 = w_i,$ $\sum w_i \ge (1-\eta) n$, and $w_i(x_i')^{\otimes 2} = w_i x_i^{\otimes 2}$. Next, for any $P \in \cS$, $\langle z_i'-S', P\rangle^2 = P^\top ((x_i')^{\otimes 2} - S')((x_i')^{\otimes 2} - S')^\top P,$ and we are assuming $\cL[((x_i')^{\otimes 2} - S')((x_i')^{\otimes 2} - S')^\top] \preccurlyeq (2+\alpha) \cdot \cL[I] = (2+\alpha) \cdot I$, where $I$ refers to the $d^2 \times d^2$-identity matrix. Hence, $\cL[\langle z_i'-S', P\rangle^2] \le 2+\alpha \le (1+\tcO(\eta)) \cdot \cL[V(S', P)]$ since $V \equiv 2$, so the 4 conditions in \Cref{lem:kmz_sos_main} are satisfied.

    Next, we must verify the desired conditions for $z_i$. Note that $\langle x_i^{\otimes 2} - S_0, P \rangle = \langle x_i x_i^\top - \Sigma_0, P^\sharp \rangle$ (where $P^\sharp$ is the symmetric matrix that flattens to $P$). Also, note that $\langle x_i x_i^\top - \Sigma, P^\sharp \rangle = \langle  \Sigma^{-1/2} x_i x_i^\top \Sigma^{-1/2} - I, \Sigma^{1/2} P^\sharp \Sigma^{1/2} \rangle$. Writing $Q = \Sigma^{1/2} P^\sharp \Sigma^{1/2},$ by Proposition \ref{prop:frobenius_replacement} we have that $\|Q\|_F = 1 \pm O(\alpha)$. This implies, using the resilience of $\{\Sigma^{-1/2} x_i\}$ (Lemma \ref{lem:resilience-of-moments-covariance}) that
\begin{alignat*}{3}
    \left|\frac{1}{n} \sum_{i=1}^n a_i \langle x_i x_i^\top - \Sigma, P^\sharp \rangle\right| &= \left|\frac{1}{n} \sum_{i=1}^n a_i \langle \Sigma^{-1/2} x_i x_i^\top \Sigma^{-1/2} - I, Q \rangle\right| &&\le O(\alpha), \\
    \frac{1}{n} \sum_{i=1}^n a_i \langle x_i x_i^\top - \Sigma, P^\sharp \rangle^2 &= \frac{1}{n} \sum_{i=1}^n a_i \langle \Sigma^{-1/2} x_i x_i^\top \Sigma^{-1/2} - I, Q \rangle^2 &&= 2 \pm O(\alpha).
\end{alignat*}

    In addition, note that $\|\Sigma-\Sigma_0\|_F \le \alpha$ due to the resilience guarantees (Lemma \ref{lem:resilience-of-moments-covariance}), which means $\langle x_i x_i^\top - \Sigma_0, P^\sharp \rangle = \langle x_i x_i^\top - \Sigma, P^\sharp\rangle \pm \alpha$. In addition, Lemma \ref{lem:resilience-of-moments-covariance} implies that $\frac{1}{n} \sum \left|\langle x_i x_i^\top - \Sigma, P^\sharp \rangle\right| = \frac{1}{n} \sum \left|\langle \Sigma^{-1/2} x_i x_i^\top \Sigma^{-1/2}-I, Q \rangle\right| \le O(1)$. This immediately implies that
\begin{alignat*}{3}
    \left|\frac{1}{n} \sum_{i=1}^n a_i \langle z_i - S_0, P \rangle\right| &= \left|\frac{1}{n} \sum_{i=1}^n a_i \langle x_i x_i^\top - \Sigma_0, P^\sharp \rangle\right| &&\le O(\alpha), \\
    \frac{1}{n} \sum_{i=1}^n a_i \langle z_i - S_0, P \rangle^2 &= \frac{1}{n} \sum_{i=1}^n a_i \langle x_i x_i^\top - \Sigma_0, P^\sharp \rangle^2 &&= 2 \pm O(\alpha).
\end{alignat*}

    Since $V \equiv 2$, this immediately implies we can apply \Cref{lem:kmz_sos_main}. Doing so, we obtain $|\langle \cL[S']-S_0, P \rangle| = |\langle \cL[\Sigma']-\Sigma_0, P^\sharp \rangle| \le \tcO(\eta)$ for all symmetric $P^\sharp$ with $\|P^\sharp\|_F = 1$. Hence, $\|\cL[\Sigma']-\Sigma_0\|_F \le O(\alpha)$, which means $\|\cL[\Sigma']-\Sigma\|_F \le O(\alpha)$ as well.

    This completes the proof for true pseudoexpectations. Again, the proof extends to approximate pseudoexpectations, since the coefficients at each step in the sum-of-squares proof are polynomially bounded.
\end{proof}
    
\subsection{SoS bounds for arbitrary samples: Covariance estimation} \label{subsec:covariance_sos_arbitrary}

In this subsection, we prove \Cref{lem:approx-spectral-accuracy}, which is our worst-case robustness result for covariance estimation. First, we establish a 1-dimensional Sum-of-Squares result that will be crucial in proving \Cref{lem:approx-spectral-accuracy}.

\begin{lemma} \label{lem:1d_covariance_arbitrary}
Let $z_1, \dots, z_n$ be a set of $n$ reals, such that the $95$th percentile of the $z_i^2$ values is $1$.
Suppose that there exists a degree-$6$ pseudoexpectation $\pE$ on the variables $\{w_i\}, \{z_i'\}$ such that:
\begin{enumerate}
    \item $\forall i,$ $\pE$ \emph{satisfies} $w_i^2 - w_i = 0$, \label{const2:w_pm_1}
    \item $\pE$ \emph{satisfies} $\sum w_i - 0.99 n \ge 0$, \label{const2:sum_w_large}
    \item $\forall i,$ $\pE$ \emph{satisfies} $w_i (z_i' - z_i)  = 0$, \label{const2:zi_match}
    \item $\pE \left[\frac{1}{n}\sum_i ((z_i')^2 - \sigma')^2\right] \le 3 \cdot \pE\left[(\sigma')^2\right]$, where we define $\sigma' := \frac{1}{n} \sum (z_i')^2$. \label{const2:variance_bound}
\end{enumerate}
Then, $\pE[\sigma'] = \Theta(1)$. Moreover, if the 95th percentile of $z_i^2$ is less than $1$, we still have $\pE[\sigma'] \le O(1)$.
\end{lemma}

\begin{proof}
    First, let's show that $\pE[\sigma'] \ge \Omega(1)$. To prove this, note that by Constraint \ref{const2:w_pm_1}, $\pE$ satisfies $w_i = w_i^2 \ge 0$ and $(1-w_i) = (1-w_i)^2 \ge 0$. So,
\begin{alignat*}{3}
    \pE[\sigma'] &= \frac{1}{n} \sum_{i=1}^{n} \pE[w_i(z_i')^2 + (1-w_i) (z_i')^2] &&\hspace{1cm} \text{(Definition of $\sigma'$)} \\
    &\ge \frac{1}{n} \sum_{i=1}^{n} \pE[w_i(z_i')^2] &&\hspace{1cm} \text{(Positivity of $1-w_i$)} \\
    &= \frac{1}{n} \sum_{i=1}^{n} \pE[w_iz_i^2] &&\hspace{1cm} \text{(Constraint \ref{const2:zi_match})} \\
    &= \frac{1}{n} \sum_{i=1}^{n} z_i^2 \pE[w_i] &&\hspace{1cm} \text{(Linearity)}
\end{alignat*}

    Since $\pE[w_i]$ is bounded between $0$ and $1$ (as $\pE w_i = \pE w_i^2$ and $\pE [1-w_i] = \pE[(1-w_i)^2]$), and since $\sum \pE[w_i] \ge 0.99 n$, the minimum possible value of $\sum z_i^2 \pE[w_i]$ is the sum of $z_i^2$ over the $0.99 n$ smallest values of $z_i^2$. Since the $95$th percentile of the $z_i^2$ values is $1$, this means $\sum z_i^2 \pE[w_i] \ge 0.04 n$. Thus, $\pE[\sigma'] \ge 0.04$.
    
    Next, we must show that $\pE[\sigma'] \le O(1)$, if the 95th percentile of the $z_i^2$ values is at most $1$. To do this, we consider restricting $\pE$ to the (at least) $0.95 n$ indices $S$ where $z_i^2 \le 1$ (note that $z_i$ are fixed real numbers, not variables). More formally, we define $\pE'$ to be a pseudoexpectation where on any monomial $p$, $\pE' p = 0$ if $p$ has a positive power of some $w_i$ for $i \not\in S$, and $\pE' p = \pE p$ otherwise. It is clear that $\pE'$ is still a degree-6 pseudoexpectation, since $\pE' 1 = \pE 1 = 1$, and $\pE' [p^2] = \pE[(p')^2]$ where $p'$ is the polynomial that removes all monomials containing some $w_i$ for $i \not\in S$. In addition, if we replace $\pE$ with $\pE'$, Constraint \ref{const2:variance_bound} is unchanged. Constraints \ref{const2:w_pm_1} and \ref{const2:zi_match} are unchanged for $i \in S$, and trivially hold for $i \not\in S$. Finally, since $\pE$ \emph{satisfies} $w_i \le 1$ for all $i$ (since  $\pE'$ satisfies $1-w_i = (1-w_i)^2 \ge 0$), we thus have that $\pE$ satisfies $\sum_{i \in S} w_i + \sum_{i \not\in S} 1 - 0.99 n \ge 0$, which means $\sum_{i \in S} w_i \ge 0.94 n$. So, $\pE'$ satisfies $\sum w_i - 0.94 n \ge 0$.
    Overall, by replacing $\pE$ with $\pE'$, we have the constraints are unchanged except \ref{const2:sum_w_large}, and the goal of showing $\pE'[\sigma'] \le O(1)$ is sufficient.
    
    We also remark that we can rewrite Constraint \ref{const2:variance_bound} (now with $\pE'$) as
\begin{equation} \label{eq:cons_3_rewritten}
    \pE'\left[\frac{1}{n} \sum_{i=1}^n (z_i')^4\right] \le 4 \cdot \pE'\left[(\sigma')^2\right] = 4 \cdot \pE'\left[\left(\frac{1}{n} \sum_{i=1}^{n} (z_i')^2\right)^2\right].
\end{equation}
    
    Now, note that
\begin{alignat*}{3}
    \hspace{0.5cm} \frac{1}{n} \cdot \underbrace{\pE'\left[\sum_{i=1}^n (1-w_i) (z_i')^4\right]}_{A}
    &\le \pE'\left[\frac{1}{n} \sum_{i=1}^n (z_i')^4\right] && \hspace{0.5cm} \text{(Constraint \ref{const2:w_pm_1})} \\
    &\le 4 \cdot \pE'\left[\left(\frac{1}{n} \sum_{i=1}^n (z_i')^2\right)^2\right] && \hspace{0.5cm} \text{(Equation \eqref{eq:cons_3_rewritten})}\\
    &= 4 \cdot \pE'\left[\left(\frac{1}{n} \sum_{i=1}^{n} (1-w_i) (z_i')^2 + \frac{1}{n} \sum_{i=1}^{n} w_i (z_i')^2\right)^2\right]  &&  \\
    &\le 8 \cdot \pE'\left[\left(\frac{1}{n} \sum_{i=1}^{n} (1-w_i) (z_i')^2\right)^2 + \left(\frac{1}{n} \sum_{i=1}^{n} w_i (z_i')^2\right)^2\right] && \hspace{0.5cm} \text{(Cauchy-Schwarz)} \\
    &= \frac{8}{n^2} \cdot \Biggr(\underbrace{\pE'\left[\left(\sum_{i=1}^{n} (1-w_i) (z_i')^2\right)^2\right]}_{B} + \pE'\left[\left(\sum_{i=1}^{n} w_i z_i^2\right)^2\right]\Biggr) && \hspace{0.5cm} \text{(Constraint \ref{const2:zi_match})}.
\end{alignat*}    

    Note that $\pE' \left[\left(\sum_{i=1}^{n} w_i z_i^2\right)^2\right] = \sum_{i, j \in S} \pE [w_i w_j] z_i^2 z_j^2$. In addition, for any $i, j$, $\pE[w_i w_j] \le \frac{1}{2} \left(\pE[w_i^2]+\pE[w_j^2]\right) \le 1$.
    So, since $0 \le z_i^2 \le 1$ for all $i \in S$, we have that $\pE' \left[\left(\sum_{i=1}^{n} w_i z_i^2\right)^2\right] \le n^2$. Therefore, 
\begin{equation} \label{eq:qwerty}
    \frac{1}{n} \cdot A \le \frac{8}{n^2} \cdot B + 8.
\end{equation}
Also,
\begin{alignat*}{3}
    0.06 n \cdot A
    &= \pE'\left[\left(\sum_{i=1}^n (1-w_i) (z_i')^4\right) \cdot 0.06n\right] && \hspace{0.5cm} \text{(Definition of $A$)} \\
    &= \pE'\left[\left(\sum_{i=1}^n (1-w_i)^2 (z_i')^4 \right) \cdot 0.06n \right] && \hspace{0.5cm} \text{(Constraint \ref{const2:w_pm_1})} \\
    &\ge \pE'\left[\left(\sum_{i=1}^n (1-w_i)^2 (z_i')^4\right) \cdot \left(\sum_{i=1}^n (1-w_i)^2\right) \right] && \hspace{0.5cm} \text{(Constraints \ref{const2:w_pm_1} and \ref{const2:sum_w_large})} \\
    &\ge \pE'\left[\left(\sum_{i=1}^n (1-w_i)^2 (z_i')^2\right)^2\right] && \hspace{0.5cm} \text{(Cauchy-Schwarz)},
\end{alignat*}
    
    Therefore, $0.06 n \cdot A \ge B$, but \eqref{eq:qwerty} tells us that $n \cdot A \le 8 (B + n^2)$. So, $B \le 0.06 n \cdot A \le 0.48 \cdot (B+n^2),$ which means $B \le n^2$. Therefore, by Cauchy-Schwarz, $\pE'\left[\frac{1}{n} \sum_{i=1}^n (1-w_i) (z_i')^2\right]^2 \le \pE'\left[\left(\frac{1}{n} \sum_{i=1}^n (1-w_i) (z_i')^2\right)^2\right] = \frac{1}{n^2} \cdot B \le 1$, which means $\pE'\left[\frac{1}{n} \sum_{i=1}^n (1-w_i) (z_i')^2\right] \le 1$. In addition, we know that $\pE'\left[\frac{1}{n} \sum_{i=1}^{n} w_i (z_i')^2\right] \le \pE'\left[\frac{1}{n} \sum_{i=1}^{n} w_i z_i^2\right] \le 1$. So overall, since $\sigma'$ has no coefficients with $w_i$, we obtain 
\[\pE[\sigma'] = \pE'[\sigma'] = \pE'\left[\frac{1}{n} \sum (z_i')^2\right] = \pE'\left[\frac{1}{n} \sum_{i=1}^n (1-w_i) (z_i')^2\right] + \pE'\left[\frac{1}{n} \sum_{i=1}^n w_i (z_i')^2\right] \le 2. \qedhere\]
\end{proof}

%\subsection{Proof of \Cref{lem:approx-spectral-accuracy}} \label{subsec:approx-accuracy-proofs}

\begin{proof}[Proof of \Cref{lem:approx-spectral-accuracy}]
    %Let $\hat{\Sigma}_1 = \cL_1[\Sigma']$ and $\hat{\Sigma}_2 = \cL_2[\Sigma']$. 
    Our main goal will be to show that $\cL_1[\Sigma'], \cL_2[\Sigma']$ are close in spectral distance.
    To do so, we show that for any unit vector $v$, $v^\top \hat{\Sigma}_1 v$ and $v^\top \hat{\Sigma}_2 v$ are equal up to an $O(1)$ multiplicative factor. This will imply that $\cL_1[\Sigma'] \preccurlyeq O(1) \cdot \cL_2[\Sigma']$ and $\cL_2[\Sigma'] \preccurlyeq O(1) \cdot \cL_1[\Sigma']$.

    Assume first that $\cL_1, \cL_2$ are actual pseudoexpectations (i.e., if $\tau = 0$).
    We define $z_i := \langle x_i, v \rangle$ and $z_i' := \langle x_i', v \rangle$. If $\cL_1, \cL_2$ are $(\alpha, \tau, T)$-certificates for $\tau = 0$ and $T \le 0.01 n$, then it is clear that $\cL_1$ and $\cL_2$ satisfy Constraints \ref{const2:w_pm_1}, \ref{const2:sum_w_large}, and \ref{const2:zi_match} of Lemma \ref{lem:1d_covariance_arbitrary}. To check Constraint \ref{const2:variance_bound}, note that by Constraint 3 of \Cref{def:certifiable-covariance},
\begin{align*}
    \cL\left[\frac{1}{n} \sum_{i=1}^n ((z_i')^2-\sigma')^2 - (2+\alpha) (\sigma')^2\right] &= \cL\left[\frac{1}{n} \sum_{i=1}^n \left(\langle x_i', v\rangle^2- v^\top \Sigma' v\right)^2 - (2+\alpha) \cdot (v^\top \Sigma' v)^2\right] \\
    &= -\cL\left[\|M v^{\otimes 2}\|_2^2\right] \\
    &\le 0,
\end{align*}
    for either $\cL = \cL_1$ or $\cL = \cL_2$, where $\Sigma' := \frac{1}{n} \sum (x_i') (x_i')^\top$ and $\sigma' := \frac{1}{n} \sum (z_i')^2$.
    
    Hence, both we can apply Lemma \ref{lem:1d_covariance_arbitrary} for both $\cL_1$ and $\cL_2$. If the $95$th percentile of $\langle y_i, v \rangle^2$ is equal to $1$, this implies that $\cL_1[v^\top \Sigma' v], \cL_2[v^\top \Sigma' v]$ are both $\Theta(1)$. If the $95$th percentile of $\langle y_i, v \rangle^2$ is some value $G$, we may rescale and use linearity to say that $\cL_1[v^\top \Sigma' v], \cL_2[v^\top \Sigma' v]$ are both $\Theta(G^2)$.

    Hence, this implies that $\cL_1[\Sigma']$ and $\cL_2[\Sigma']$ are within $O(1)$ spectral distance of each other, at least when $\tau = 0$. For general $\tau$, we note that again the coefficients at each step in the sum-of-squares proof are bounded by $\poly(n, d, K)$. The only possible issue is the rescaling, if $G \gg (ndK)^{O(1)}$ or $G \ll (ndK)^{-O(1)}$. We avoid the former case because we are assuming that every sample is bounded by $\poly(n, d, K)$ in magnitude, using truncation. In the latter case, we use the fact that if the 95th percentile of $z_i^2$ is less than $1$, then $\cL[\sigma'] \le O(1)$ in \Cref{lem:1d_covariance_arbitrary}. In this case, by scaling by $\frac{1}{K^2}$, we have $\cL[v^\top \Sigma' v] \le \frac{O(1)}{K^2}$, which violates Constraint 5 of \Cref{def:certifiable-covariance}. 

    In summary, we have that $\cL_1[\Sigma'] \preccurlyeq O(1) \cdot \cL_2[\Sigma']$ and $\cL_2[\Sigma'] \preccurlyeq O(1) \cdot \cL_1[\Sigma']$, and both are spectrally bounded between $\frac{1}{4K}$ and $4K$. Since we have the requirements that $(1-\alpha) \cL[\Sigma'] - \tau \cdot T \cdot I \preccurlyeq \tSigma \preccurlyeq (1+\alpha) \cL[\Sigma'] + \tau \cdot T \cdot I$, this implies that $\tSigma_1 \preccurlyeq O(1) \cdot \tSigma_2$ and $\tSigma_2 \preccurlyeq O(1) \cdot \tSigma_1$. 
    %Also, note that both are spectrally bounded above by $4 K \cdot I$ and below by $\frac{1}{4 K} \cdot I$.
\end{proof}

\subsection{SoS bounds for arbitrary samples: Mean estimation} \label{subsec:mean_sos_arbitrary}

In this subsection, we prove \Cref{lem:approx-mean-accuracy}, which is our worst-case robustness result for mean estimation. First, we establish a 1-dimensional Sum-of-Squares result that will be crucial in proving \Cref{lem:approx-mean-accuracy}.

\begin{lemma} \label{lem:1d_mean_arbitrary}
    Let $z_1, \dots, z_n$ be a set of $n$ reals, such that at least $n/4$ of the $\{z_i\}$ are at least $20$. Then, for any degree-6 pseudoexpectation $\pE$ on the variables $\{w_i\}, \{z_i'\}$ such that
\begin{enumerate}
    \item $\forall i$, $\pE$ satisfies $w_i^2-w_i = 0$, \label{const3:w_0_or_1}
    \item $\pE$ satisfies $\sum w_i - 0.99 n = 0$, \label{const3:sum_w_large}
    \item $\forall i$, $\pE$ satisfies $w_i(z_i'-z_i) = 0$, \label{const3:zi_match}
    \item $\pE[\mu']= 0$, where $\mu' = \frac{1}{n} \sum z_i'$, \label{const3:mu_0}
\end{enumerate}
    we must have that $\pE\left[\frac{1}{n} \sum (z_i'-\mu')^2\right] \ge 2$.
\end{lemma}

\begin{proof}
    Using the fact that $\pE$ satisfies $w_i^2 = w_i$, we have that $(1-w_i w_j)^2 = (1-w_i w_j)$, which means $w_iw_j \le 1$ is satisfied. In addition, $\pE$ satisfies $w_i w_j = w_i^2 w_j^2 \ge 0$, and also satisfies $w_i w_j \ge w_i+w_j-1$, since $w_iw_j-(w_i+w_j-1) = (1-w_i)(1-w_j) = (1-w_i)^2(1-w_j)^2$.

    This means
\begin{alignat}{3}
    \pE\left[\sum_{i, j = 1}^n (z_i'-z_j')^2\right]
    &\ge \pE\left[\sum_{i, j} w_i w_j (z_i'-z_j')^2\right] && \nonumber \\
    &\ge \pE\left[\sum_{i, j} w_i w_j (z_i-z_j)^2\right] && \hspace{0.5cm} \text{(Condition 3)} \nonumber \\
    &= \sum_{i, j}(z_i-z_j)^2 \cdot \pE\left[w_i w_j\right] && \hspace{0.5cm} \text{(Linearity)} \nonumber \\
    &\ge \sum_{i, j}(z_i-z_j)^2 \cdot \max (\pE[w_i]+\pE[w_j]-1, 0). \label{eq:mean_arbitrary_1}
\end{alignat}

    Now, $C_1$ and $C_2$ be the $25$th and $75$th percentiles, respectively, of the elements $z_i$ sorted in increasing order. We show that we may assume $C_2-C_1 \le 8$. Otherwise, there exists a set $S$ of $0.25 n$ elements $z_i$ that are at least $C_1+8$, and a set $T$ of $0.25 n$ elements that are at most $C_1$. In this case, we can bound \eqref{eq:mean_arbitrary_1} as at least
\begin{align*}
    &\hspace{0.5cm} 2 \cdot \sum_{i \in S, j \in T} (z_i - z_j)^2 \cdot \max(\pE[w_i]+\pE[w_j]-1, 0) \\
    &\ge 2 \cdot \sum_{i \in S, j \in T} 8^2 \cdot (\pE[w_i]+\pE[w_j]-1) \\
    &= 2 \cdot \left[64 (n/4) \cdot \sum_{i \in S} \pE[w_i] + 64 (n/4) \cdot \sum_{j \in T} \pE[w_j] - 64 (n/4)^2\right]\\
    &\ge 4 \cdot 64(n/4) (n/4 - 0.01 n) - 2 \cdot 64 (n/4)^2 \\
    &\ge 6 n^2,
\end{align*}
    where the penultimate inequality uses the fact that $\pE[w_i] \in [0, 1]$ and $\sum_{i = 1}^{n} \pE[w_i] \ge 0.99 n$. Overall, this means $\pE[\sum_{i, j} (z_i'-z_j')^2] \ge 6 n^2$. But, $\sum_{i, j = 1}^n (z_i'-z_j')^2 = 2n \cdot \sum_{i = 1}^n (z_i'-\mu')^2,$ which means $\pE\left[\sum_{i = 1}^n (z_i'-\mu')^2\right] \ge 3n$, as desired.

    Hence, we may assume that the 25th and 75th percentiles are within $8$ of each other. Re-define $S \subset [n]$ to be the set of indices of size $n/2$ between the $25$th and $75$th percentile. By our assumption in the lemma that at least $n/4$ values are at least $20$, $z_i \in [C-4, C+4]$ for all $i \in S$, for some $C \ge 16.$
    Note that
\begin{equation}
    \pE\left[\sum_{i \in S} w_i z_i'\right] = \pE\left[\sum_{i \in S} w_i z_i\right] = \sum_{i \in S} z_i \pE[w_i] \ge (C-4) \cdot \sum_{i \in S} \pE[w_i] \ge (C-4) \cdot (0.49 n), \label{eq:mean_arbitrary_3}
\end{equation}
    but 
\begin{equation}
    \pE\left[\left(\sum_{i \in S} w_i z_i'\right)^2\right] = \pE\left[\left(\sum_{i \in S} w_i z_i\right)^2\right] = \sum_{i, j \in S} z_i z_j \pE[w_i w_j] \le (C+4)^2 \cdot (0.5 n)^2. \label{eq:mean_arbitrary_4}
\end{equation}

    In addition, if we assume $\pE[\frac{1}{n} \sum_{i=1}^{n} (z_i'-\mu')^2] \le 2$, then since $S$ is fixed and has size $n/2$,
\begin{align}
    \pE\left[\sum_{i \in S} (z_i')^2\right]
    &\le \pE\left[\sum_{i \in S} (z_i')^2\right] + \frac{1}{|S|} \cdot \pE\left[\left(\sum_{i \in S} z_i' - |S| \cdot \mu'\right)^2\right] \nonumber \\
    &= \pE\left[\sum_{i \in S} (z_i'-\mu')^2\right] + \frac{1}{|S|} \cdot \pE\left[\left(\sum_{i \in S} z_i'\right)^2\right] \nonumber \\
    &\le \pE\left[\sum_{i=1}^n (z_i'-\mu')^2\right] + \frac{1}{|S|} \cdot \pE\left[\left(\sum_{i \in S} z_i'\right)^2\right] \nonumber \\
    &\le 2n + \frac{1}{|S|} \cdot \pE\left[\left(\sum_{i \in S} z_i'\right)^2\right]. \label{eq:mean_arbitrary_2}
\end{align}
    Making use of the fact that $\pE$ satisfies $(1-w_i) = (1-w_i)^2$, we have
\begin{alignat*}{3}
    \underbrace{\pE\left[\left(\sum_{i \in S} (1-w_i) z_i'\right)^2\right]}_{A}
    &\le \pE\left[\left(\sum_{i \in S} (1-w_i)\right) \cdot \left(\sum_{i \in S} (1-w_i)(z_i')^2\right)\right] && \hspace{0.5cm} \text{(Cauchy-Schwarz)} \\
    &\le 0.01 n \cdot \pE\left[\sum_{i \in S} (1-w_i)(z_i')^2\right] && \hspace{0.5cm} \text{(Condition 2)} \\
    &\le 0.01 n \cdot \pE\left[\sum_{i \in S} (z_i')^2\right] && \hspace{0.5cm} \text{(Condition 1)} \\
    &\le 0.01 \cdot \left(\pE\left[2 \left(\sum_{i \in S} z_i'\right)^2\right] + 2 n^2\right) && \hspace{0.5cm} \text{(Equation \eqref{eq:mean_arbitrary_2})} \\
    &\le 0.01 \cdot \Biggr(4 \underbrace{\pE\left[\left(\sum_{i \in S} (1-w_i) z_i'\right)^2\right]}_{A} + 4 \pE\left[\left(\sum_{i \in S} w_i z_i'\right)^2\right] + 2 n^2\Biggr). && \hspace{0.5cm} \text{(Cauchy-Schwarz)}
\end{alignat*}
    Hence, we have that $A \le 0.05 \pE\left[\left(\sum_{i \in S} w_i z_i'\right)^2\right] + 0.03 n^2 \le 0.02 C^2 n^2$ for $C \ge 16$, using \eqref{eq:mean_arbitrary_4}.

    So, by Cauchy-Schwarz, we have that $\left|\pE\left[\sum_{i \in S} (1-w_i) z_i'\right]\right| \le 0.15 C n$. But $\pE\left[\sum_{i \in S} w_i z_i'\right] \ge (C-4) \cdot 0.49 n \ge 0.35 C n$ by \eqref{eq:mean_arbitrary_3}, which means $\pE\left[\sum_{i \in S} z_i'\right] \ge 0.2 C n$.

    However, if $\pE\left[\sum_{i=1}^n z_i'\right] = 0$, then $\pE\left[\sum_{i \in S} z_i' - \frac{1}{2} \sum_{i=1}^n z_i'\right] \ge 0.2 C n$. By Cauchy-Schwarz, this means $\pE\left[\left(\sum_{i \in S} z_i' - \frac{1}{2} \sum_{i=1}^n z_i' \right)^2\right] \ge 0.04 C^2 n^2$. Since $S$ is a fixed set of size $n/2$, defining $T := [n] \backslash S$, we have
\begin{alignat*}{3}
    \pE\left[\biggr(\sum_{i \in S} z_i' - \frac{1}{2} \sum_{i=1}^n z_i' \biggr)^2\right] &= \frac{1}{4} \cdot \pE\left[\biggr(\sum_{i \in S} z_i' - \sum_{i \in T} z_i'\biggr)^2\right] \\
    &= \frac{1}{n^2} \pE\left[\biggr(\sum_{i \in S, j \in T} (z_i'-z_j')\biggr)^2\right] \\
    &\le \frac{1}{4} \cdot \pE\left[\sum_{i \in S, j \in T} (z_i'-z_j')^2\right]. && \hspace{0.5cm}\text{(Cauchy-Schwarz)}
\end{alignat*}
    This implies that $\pE\left[\sum_{i \in S, j \in T} (z_i'-z_j')^2\right] \ge 0.16 C^2 n^2$, which means $2n \cdot \pE\left[\sum_{i=1}^n (z_i'-\mu')^2\right] = \pE\left[\sum_{i, j = 1}^n (z_i'-z_j')^2\right] \ge 0.32 C^2 n^2$. So, $\pE\left[\sum_{i=1}^n (z_i'-\mu')^2\right] \ge 0.16 C^2 n \ge 3 n$.
\end{proof}

\begin{proof}[Proof of \Cref{lem:approx-mean-accuracy}]
    Our main goal will be to show that $\hat{\mu}_1 := \cL_1[\mu'], \hat{\mu}_2 := \cL_2[\mu']$ are close in $\ell_2$ distance.
    To do so, we show that for any unit vector $v$, $\langle \cL_1[\mu']-\cL_2[\mu'] ,v \rangle \le O(1)$.

    We first focus on $\cL_1$: suppose $\cL_1$ is an actual pseudoexpectation (i.e., if $\tau = 0$).
    We define $z_i := \langle x_i - \hat{\mu}_1, v \rangle$ and $z_i' := \langle x_i' - \hat{\mu}_1, v \rangle$. If $\cL_1$ is an $(\alpha, \tau, T)$-certificate for $\tau = 0$ and $T \le 0.01 n$, then it is clear that $\cL_1$ satisfies Constraints \ref{const3:w_0_or_1}, \ref{const3:sum_w_large}, and \ref{const3:zi_match} of Lemma \ref{lem:1d_mean_arbitrary}. To check Constraint \ref{const2:variance_bound}, note that $\cL_1\left[\frac{1}{n} \sum z_i'\right] = \frac{1}{n} \sum \cL_1[\langle x_i', v \rangle - \langle \hat{\mu}_1, v \rangle] = 0$.

    Hence, by Lemma \ref{lem:1d_mean_arbitrary}, if the median of $\langle x_i - \hat{\mu}_1, v \rangle$ was greater than $20$, then $\cL_1\left[\frac{1}{n} \sum \langle x_i' - \mu', v \rangle^2\right] \ge 2$, where $\mu' := \frac{1}{n} \sum x_i'$. This, however, contradicts Condition 2e in \Cref{def:certifable-mean}. For general $\tau$, we note that again the coefficients at each step in the sum-of-squares proof are bounded by $\poly(n, d, K)$. So, this implies that if $\cL_1$ is an $(\alpha^*, \tau, \phi, T)$-certifiable mean, then for every unit vector $v$, $\langle \hat{\mu}_1, v \rangle$ is at most $20$ away from the median of $\langle x_i, v \rangle$. (This is true in both directions since we can replace $v$ with $-v$).
    
    Likewise, the same is true for $\cL_2$, which means that $\left|\langle \hat{\mu}_1, v \rangle - \langle \hat{\mu}_2, v \rangle\right| \le 40$ for all vectors $v$. Therefore, $\|\hat{\mu}_1-\hat{\mu}_2\|_2 \le 40$. Finally, we note that $\|\tmu_1 - \hat{\mu}_1\|_\infty, \|\tmu_2 - \hat{\mu}_2\|_\infty \le \phi+\tau \cdot T \le O(\alpha/\sqrt{d})$, so $\|\tmu_1-\tmu_2\|_2 \le 42$.
\end{proof}

\section{Computing Score Functions}
\label{sec:computing-score-functions}
In this section we will describe how we can compute the value of the score functions efficiently.

% In this subsection we will describe a generic form that captures the score functions that we repeatedly make use of.

In our problems, we usually have some family of properties $\set{P_T}$, parameterized by $T$. The higher values of $T$ correspond to more lenient settings and the lower values of $T$ correspond to more stringent settings.  We are interested in how well (or poorly) a parameter $\theta$ satisfies these properties. We can easily define a score function to measure this. These score functions are later used to run the exponential mechanism and design private algorithms. These score functions are defined in the following fashion.

\begin{equation*}
\cS \paren{\theta} := \inf_T \text{ such that $\theta$ satisfies $P_T$ }.
\end{equation*}

As mentioned, because $P_T$'s are increasingly lenient, $\theta$ satisfies $P_T$ for all $T > \cS\paren{\theta}$, and does not satisfy $P_T$ for all $T < \cS\paren{\theta}$.
In our problems we describe $\set{P_T}$ through systems of polynomial inequalities and the existence of linear functionals that approximately satisfy them. We define polynomial constraints ${q_1 \ge 0, \dots, q_k \ge 0}$, which depend on $T$ and $\theta$, and if there exists a linear functional (an approximate pseudo-expectation) that \textit{approximately satisfies} these polynomial constraints, we say that $\theta$ satisfies $P_T$. We first make some assumptions on these generic polynomial constraints and after that we will define approximate satisfiability formally in \cref{def:approximate-satisfiability}.

\begin{assumption}
\label{ass:polynomials}
We make the assumption that in problems that we deal with parameterized families of polynomials $\set{Q_T}_{T = 0}^{T_{\max}}$ that are in the following form and may include the following different types of constraints.
\begin{enumerate}
    \item Regular constraints: $q \ge 0$.
    \item PSD constraints: $\forall h$, where $\normt{h} = 1$: $q h^2 \ge 0$.
    \item $T$-constraint: Each $Q_T$, has exactly one constraint that depends on $T$. We call this constraint the "$T$-constraint".The other constraints do not depend on $T$, and are the same over all $Q_T$'s. Let $q_T$ denote this constraint. This constraint is also a PSD constraint and it appears only in the form of $\forall h: q_T h^2 \ge 0$.
    We also make the assumption that $q_T$ depends linearly on $T$ and 
    \begin{equation*}
        \forall 0 \le T, T' \le T_{\max} : (q_T - q_T') = (T - T') / \paren{2 T_{\max}}.
    \end{equation*}
    Note that this is a polynomial identity. 
    \item Matrix PSD constraints: $q \succcurlyeq 0$.
\end{enumerate}
\end{assumption}

\begin{definition}[approximate satisfiability]
\label{def:approximate-satisfiability}
Suppose $R>1$, and a parameterized family of polynomials $\set{Q_T}$ of up to degree $d$, over $\R^n$ are given as in \cref{ass:polynomials}.
We say a linear functional $\cL$ over the set of polynomials of degree at most $d$ over $R^n$, $\tau$-approximately satisfies $Q_T$ and write
$\cL \entails_{\tau} Q_T$ if and only if
\begin{enumerate}
    \item $\cL 1 = 1$,
    \item $\cL h^2 \ge -\tau \cdot T$, for every polynomial $h$ such that $2\deg{h} \le d$ and $\normt{h} \le 1$.
    \item $\cL q \ge -\tau \cdot T$, for every polynomial $q \in Q_T$ that is a regular constraint.
    \item $\cL q h^2 \ge -\tau \cdot T$, for every polynomial $q \in Q_T$ that is a PSD constraint and every polynomial $h$ such that $2 \deg {h} + \deg {q} \le d$ and $\normt{h} = 1$.
    \item $\cL q \succcurlyeq -\tau \cdot T \cdot I$, for every polynomial $q \in Q_T$, that is a matrix PSD constraint.
    \item $\normt{\cR\paren{\cL}} \le R + \tau \cdot T$.
\end{enumerate}
    %\blu
    In addition, for any $\gamma > 0$ we write $\cL \entails_{\tau, \gamma} Q_T$ if the above conditions hold but replacing $\tau \cdot T$ with $\tau \cdot (T+\gamma)$. (Note that the constraint $Q_T$ has \emph{not} been replaced with $Q_{T+\gamma}$.)
\end{definition}

\begin{remark}
In order to run the ellipsoid algorithm, we should have a full dimensional ball of positive volume. If we attempt to run the ellipsoid algorithm over the set of functionals with $\cL 1 = 1$, this is trivially not going to be the case. Therefore, instead we only consider the space of linear functionals excluding the $S = \emptysetAlt$ index, which corresponds to the monomial $1$.
\end{remark}
\begin{lemma}[efficient functional search]
\label{lem:efficient-functional-search}
Suppose $R > 1$, and 
$Q_T$ is a set of polynomial constraints of up to degree $d$, over $\R^n$ as in \Cref{ass:polynomials}, with fixed parameter $T$. 
Let $\cR\paren{\cL}_{\overline{\emptysetAlt}}$ denote the representation of a functional $\cL$ for every multiset of size up to $d$, excluding the empty set index.
Then, for any $r, \gamma > 0$, there exists an algorithm that runs in time $\poly\paren{n^d, \Size \paren{Q_T}, \log \paren{R'/r}, \log \paren{1/\gamma}}$ that either
\begin{enumerate}
    \item finds the representation of a  linear functional $\cL$ such that $\normt{\cR\paren{\cL}_{\overline{\emptysetAlt}}} \le R'$, and $\cL \entails_{\tau, O(\gamma)} Q_T$; or,
    \item shows that the volume of representations of functionals $\cL$ such that $\normt{\cR\paren{\cL}_{\overline{\emptysetAlt}}} \le R'$ and $\cL \entails_{\tau} Q_T$, when projected to the entries $S \neq \emptysetAlt$, is less than the volume of a ball of radius $r$,
\end{enumerate}
where
$R' = \sqrt{\paren{R + \tau \cdot T}^2 - 1}$.
Note that here
$\cR\paren{\cL} \in \R^{n \choose \le d}$, and $\cR\paren{\cL}_{\overline{\emptysetAlt}} \in \R^{ {n \choose \le d} -1}$, and the volume in the second case is measured with respect to $\R^{ {n \choose \le d} -1}$.
\end{lemma}
In essence, we use reductions to semi-definite programs. For a textbook treatment of this approach see Chapter 3 of \cite{FlemingKP19}.
\begin{proof}
Firstly note that 
under the assumption that $\cR\paren{\cL}_{\emptysetAlt} = 1$, we have that
$\normt{\cR\paren{\cL}} \le R$ is equivalent to $\normt{\cR\paren{\cL}_{\overline{\emptysetAlt}}} \le R'$.
Let
\begin{equation*}
K = \Set{\cR\paren{\cL} \suchthat \cL \entails_{\tau} Q_T},  
K_{\overline{\emptysetAlt}} = \Set{\cR\paren{\cL}_{\overline{\emptysetAlt}} \suchthat \cL \entails_{\tau} Q_T}.
\end{equation*}
It is easy to see that $K \subset \R^{{n \choose \le d}}$ is equal to $K_{\overline{\emptysetAlt}} \subset \R^{{n \choose \le d} - 1}$ with the adjustment that all of its members have the additional $\emptysetAlt$ entry $1$. 
We want to apply the ellipsoid algorithm over the ball of radius $R$ in $\R^{{n \choose \le d} -1}$, if we show that

\begin{enumerate}
    \item $K_{\overline{\emptysetAlt}}$ is convex; and,
    \item $K_{\overline{\emptysetAlt}}$ admits an efficient (approximate) membership and separation oracle,
\end{enumerate}
we are done and we obtain the desired guarantees via the ellipsoid algorithm.

\paragraph{Convexity.}
In order to show that $K_{\overline{\emptysetAlt}}$ is convex, it suffices to show that $K$ is convex.
let $M_1, M_2 \in K$, we need to prove that $\forall \alpha \in \brac{0, 1}$, $M_3 = \alpha M_1 + \paren{1- \alpha} M_2 \in K$.
By triangle inequality it is easy to see that $\normt{M_3} \le \alpha \normt{M_2} + \paren{1- \alpha}\normt{M_2} \le R$. 
Let $\cL_1, \cL_2, \cL_3$ be the corresponding functionals of $M_1, M_2, M_3$. It suffices to show that $\cL_3 \entails Q_T$. Let's verify this.
\begin{enumerate}
    \item $\cL_3 1= \alpha \cL_1 1 + \paren{1- \alpha} \cL_2 1  = 1$.
    \item $\cL_3 q = \alpha \cL_1 q +  \paren{1- \alpha}\cL_2 q \ge -\tau \cdot T$, for every regular constraint $q \in Q_T$.
    \item $\cL_3 h^2 = \alpha \cL_1 h^2 + \paren{1-\alpha} \cL_2 h^2 \ge - \alpha \tau \cdot T - \paren{1- \alpha} \tau \cdot T = -\tau \cdot T$, for every polynomial $h$ such that $2 \deg \paren{h} \le d$.
    \item $\cL_3 q h^2 =  \alpha \cL_1 q h^2 + \paren{1-\alpha} \cL_2 q h^2 \ge - \alpha \tau \cdot T - \paren{1- \alpha} \tau \cdot T = -\tau \cdot T$, for every PSD polynomial constraint $q \in Q_T$ and every polynomial $h$ such that $\deg q + 2 \deg \paren{h} \le d$.
    \item $\cL_3 q = \alpha \cL_1 q + \paren{1-\alpha} \cL_2 q \succcurlyeq - \alpha \tau \cdot T \cdot I - \paren{1- \alpha} \tau \cdot T \cdot I = -\tau \cdot T \cdot I$, for every matrix PSD polynomial constraint $q \in Q_T$ and every polynomial $h$ such that $\deg q + 2 \deg \paren{h} \le d$.
    \item$\normt{\cR\paren{\cL_3}} \le \alpha \normt{\cR\paren{\cL_1}} + \paren{1- \alpha }\normt{\cR\paren{\cL_2}} = R + \tau \cdot T$.
\end{enumerate}
Therefore $K$ is convex as desired.

\paragraph{Membership/Separation oracle.}
Suppose $M \in \R^{{n \choose \le d} - 1}$ is given. 
We need to verify $M \in K_{\overline{\emptysetAlt}}$, or not. Let $M'$ be equal to $M$ with the additional entry $M'_{\emptyset} = 1$.
Then it is easy to see that $M \in K_{\overline{\emptysetAlt}}$, if and only if $M' \in K$.
Suppose $\cL$ is the linear functional with $M'$ as its representation. We need to come up with membership/separation oracles for each of the constraints in \Cref{def:approximate-satisfiability}.
\paragraph{Regular Constraints.} 
\begin{equation*}
    \cL q \ge -\tau \cdot T, \text{ for every regular constraint $q \in Q_T$.}
\end{equation*}
In order to check this constraint we can just compute the value $\iprod{M', \cR\paren{q}}$. If its value is greater than or equal to $-\tau \cdot T$, then that means $\cL q \ge -\tau \cdot T$ is satisfied, and this constraint does not refute $\cL \entails_{\tau, \cO\paren{\gamma}}$, and we would be in the setting where $\cL \entails_{\tau, \cO\paren{\gamma}} Q_T$, if all of the other constraints hold as well.

If this is not the case then let $H \in \R^{{n \choose \le d}}$ be as 
\begin{equation*}
H = \cR\paren{q}.
\end{equation*}
Then $\iprod{M, H_{\overline{\emptysetAlt}}} = \iprod{M', H} - H_{\emptysetAlt} < - \tau \cdot T - H_{\emptysetAlt}$. Moreover, for every $N \in K_{\overline{\emptysetAlt}}$, we have that $\iprod{N, H_{\overline{\emptysetAlt}}} \ge -\tau \cdot T - H_{\emptysetAlt}$. Therefore $H_{\emptysetAlt}$ is a separating hyperplane. Therefore we have an efficient separation oracle as desired.
\begin{equation*}
\end{equation*}

\paragraph{PSD Constraints.}
These constraints are in the following form.
\begin{align*}
\cL q h^2\ge -\tau \cdot T, &\text{ for every polynomial $h$ where $\normt{\cR\paren{h}} \le 1$, and $\deg q + 2 \deg h \le d$,} \\
&\text{ and for every polynomial $q$ that is either $1$ or a PSD constraint in $Q_T$.}
\end{align*}
Suppose $q = \iprod{a, v_d\paren{x}}$, $h = \iprod{b, v_d\paren{x}}$. Then, 
\begin{align*}
\cL q h^2 &= \cL \Paren{\sum_{U} \Paren{a_U x^U}}
\cdot \Paren{\sum_{V} \Paren{b_V x^V}}
\cdot \Paren{\sum_{W} \Paren{b_W x^W}}
\\
&= 
\cL \sum_{U, V, W} a_U b_V b_W \cdot x^{U+V+W} \\
&= 
\sum_{V, W} b_V b_W \Brac{\sum_{U} a_U \cL\paren{x^{U+V+W}}}.
\end{align*}
Define the matrix $X \in \R^{{n \choose \le (d - \deg q)/ 2} \times {n \choose \le (d - \deg q)/ 2}}$
as
\begin{align*}
X_{V, W} &= \sum_U a_U \cL\paren{x^{U+V+W}} \\
&= \sum_U \cR\paren{q}_U \cR\paren{\cL}_{U+V+W}.
\end{align*}
Then 
\begin{equation*}
\cL q h^2 = \transpose{b} X b.
\end{equation*}
Our goal is to verify whether $\cL q h^2$ is larger than $-\tau \cdot T$ for every $h$, where $\normt{\cR\paren{h}} \le 1$ or not. This is equivalent to $\transpose{b}X b$ being larger than $-\tau \cdot T$ for every $b$, where $\normt{b} \le 1$.
We can check this by looking at the spectral value decomposition of $X$. Suppose that the spectral decomposition of $X = P D \transpose{P}$, where $D$ is a diagonal matrix whose entries are the eigenvalues of $X$, and the rows of $P$ are the corresponding eigenvectors.
This decomposition can be computed in polynomial time using standard algorithms for obtaining eigenvalue decompositions.
%\todo{This is going to be off by some $\gamma$ if we take runtime $\poly(n^d, \log 1/\tau, \log 1/\gamma)$} 
More accurately, for any $\gamma > 0$, we can learn the minimum eigenvalue up to error $\tau \cdot \gamma$ in time $\poly(n^d, \log \frac{1}{\tau \cdot \gamma})$.
Then, if (our estimate of) the minimum eigenvalue is at least $-\tau \cdot (T+3\gamma)$, this means that the constraint $\cL q h^2 \ge \tau \cdot (T+4\gamma)$ is satisfied, and this constraint does not refute $\cL \entails_{\tau, \cO\paren{\gamma}}$, and we would be in the setting where $\cL \entails_{\tau, \cO\paren{\gamma}} Q_T$, if all of the other constraints hold as well.

If this is not the case then we know that the minimum eigenvalue is less than $-\tau \cdot (T+2\gamma)$, and we need to return a separating hyperplane that separates $M$ and $K_{\overline{\emptysetAlt}}$. Suppose the minimum eigenvalue of $X$ is less than $-\tau \cdot (T+2\gamma)$. 
%Let the corresponding row in $P$ be $c$. Then $c^\top X c < -\tau \cdot T$.
Then we can find a vector $c$ such that $c^\top X c < -\tau \cdot (T+\gamma)$.
Let the vector $H \in \R^{n \choose \le d}$ be as
\begin{equation*}
H_S = \sum_{U \cup V = S} c_U \cR\paren{q}_V.
\end{equation*}
Note that we can compute this vector efficiently.
Then we have that
\begin{align*}
\iprod{M', H} &= \cL q \iprod{c, v_{\paren{d - \deg q}/2}\paren{x}}^2 \\
 &=
\transpose{c} X c \\
&< -\tau \cdot (T+\gamma).
\end{align*}
Since $M'_{\emptysetAlt} = 1$, we have that $\iprod{M, H_{\overline{\emptysetAlt}}} = \iprod{M', H} - H_{\emptysetAlt} < - \tau \cdot \paren{T + \gamma} - H_\emptysetAlt$.
Now assume $N \in K_{\overline{\emptysetAlt}}$.
Similarly, we can show that $\iprod{N, H_{\overline{\emptysetAlt}}} \ge - \tau \cdot T - H_{\emptysetAlt}$.
Therefore $H_{\overline{\emptysetAlt}}$ is a separating hyperplane. Therefore we have an efficient separation oracle as desired.

%Mahbod: This is uncomfortably similar similar to the previous part
\paragraph{Matrix PSD Constraints.} 
\begin{equation*}
\cL q \succcurlyeq - \tau \cdot T, \text{ for every matrix PSD constraint $q \in Q_T$.}
\end{equation*}
Note that here $q$ is a square matrix with polynomials as its entries. We use $q_{i, j}$ to denote the $(i, j)$- entry of this matrix, which is a polynomial.
In order to check this constraint just define $X$ as
\begin{equation*}
X_{i, j} = \cL q_{i , j} = \iprod{M, \cR\paren{q_{i, j}}}. 
\end{equation*}
In order to check the constraint $\cL q \succcurlyeq -\tau \cdot T$, we can check the spectral value decomposition of $X$. 
Suppose that the spectral decomposition of $X = PD\transpose{P}$, where $D$ is a diagonal matrix whose entries are the eigenvalues of $X$, and the rows of $P$ are the corresponding eigenvectors. This decomposition can be computed in polynomial time using standard algorithms for obtaining eigenvalue decompositions. More accurately, for any $\gamma > 0$, we can compute the minimum eigenvalue of to error $\tau \cdot \gamma$ in time $\poly\paren{n^d, \log \frac{1}{\tau \cdot \gamma}}$. Then, if our estimate of the eigenvalue is at least $-\tau \cdot \paren{T+ 3\gamma}$, this means that the constraint $\cL q \succcurlyeq -\tau \cdot \paren{T + 4\gamma} \cdot I$ is satisfied, and this constraint does not refute $\cL \entails_{\tau, \cO\paren{\gamma}Q_T}$, and we would be in the setting where $\cL \entails_{\tau, \cO\paren{\gamma}} Q_T$, if all of the other constraints hold as well.

If this is not the case then we know that the minimum eigenvalue is less than $-\tau \cdot \paren{T + 2\gamma}$, and we need to return a separating hyperplane that separates $M$ and $K_{\overline{\emptysetAlt}}$. Suppose the minimum eigenvalue of $X$ is less than $-\tau \cdot \paren{T + 2\gamma}$. Then we can find a vector $c$ such that $\transpose{c} X c < -\tau \cdot \paren{T+ \gamma}$.
Now consider $\transpose{c} \cL q c$, and assume $c$ and $q$ are constants and $\cL$ is variable.
We can write this as 
\begin{equation*}
\transpose{c} \cL q c = \sum_{U} H_U \cL\paren{x^U},
\end{equation*}
for some $H_U$'s that depend only on $q$ and $c$. Moreover, give $q$ and $c$ we can compute this $H$ efficiently.
Now since $\transpose{c}X c < -\tau \cdot \paren{T+ \gamma}$, we have that
\begin{equation*}
\iprod{\cR\paren{\cL}, H} = \iprod{M',H} < -\tau \cdot \paren{T+ \gamma},
\end{equation*}
and therefore $\iprod{M, H_{\overline{\emptysetAlt}}} = \iprod{M', H} - H_{\emptysetAlt} < -\tau \cdot \paren{T + \gamma} - H_{\emptysetAlt}$. Similarly, if $N \in K_{\overline{\emptysetAlt}}$, we can show that 
$\iprod{N, H_{\overline{\emptysetAlt}}} \ge -\tau \cdot T - H_{\emptysetAlt}$. Therefore $H_{\overline{\emptysetAlt}}$ is a separating hyperplane. Therefore we have obtained an efficient separation oracle as desired.

\paragraph{Norm Bound Constraints.}
\begin{equation*}
\normt{\cR\paren{\cL}} \le R + \tau \cdot T.
\end{equation*}
In order to check this constraint we just compute $\normt{M}^2$. If its value is less than or equal to $R'^2$, then that means $\normt{\cR\paren{\cL}} \le R + \tau \cdot T$ is satisfied.
If this is not the case then let $H \in \R^{{n \choose \le d}}$ be as 
$H = \cR\paren{L} = M'$.
Note that $\normt{H} > R$, since $\normt{M} > R'$.
Then $\iprod{M, H_{\overline{\emptysetAlt}}} = \iprod{M', H} - H_{\emptysetAlt} = \normt{H}^2 - 1 $.
Moreover, for every $N \in K_{\overline{\emptysetAlt}}$, we have that 
$\iprod{N, H_{\overline{\emptysetAlt}}} \le R \normt{H} - 1$.
Therefore $H_{\overline{\emptysetAlt}}$ is a separating hyperplane.
\end{proof}

\begin{lemma}[robust satisfiability]
\label{lem:robust-satisfiability}
Consider the family of polynomial constraints $\set{Q_T}$ of up to degree $d$ over $\R^n$ as in \Cref{ass:polynomials}.
Moreover, suppose that there exists some linear functional $\cL_0$, such that $\cL_0 \entails Q_{T_0}$. Then there exists a set of linear functionals $\cF$ such that
\begin{equation*}
\Set{\cR\paren{\cL}_{\overline{\emptysetAlt}}\suchthat \cL \in \cF}
\end{equation*}
contains a full-dimensional ball of radius 
$r =\poly\paren{1/\poly\paren{n^{d}}, \tau , \gamma, 1/k, 1 / \normi{R\paren{Q_{T_0}}}}$, and 
for all $\cL \in \cF$, we have that $\cL \entails_{\tau} Q_{T_0 + \gamma}$. Here $\normi{\cR\paren{Q_{T_0}}}$ denotes the infinity norm over all coefficients that appear in $Q_{T_{0}}$..
\end{lemma}
\begin{proof}
Suppose $E \in \R^{{n \choose \le d}}$ be such that $\normt{E_{\overline{\emptysetAlt}}} \le r$ and $E_{\emptysetAlt} = 0$. Let $\cL$ be the linear functional with the representation $\cR\paren{\cL}= \cR\paren{\cL_0} + E$. Our goal is to choose $r$, in a way that for every choice $E_{\overline{\emptysetAlt}}$, where $\normt{E_{\overline{\emptysetAlt}}} \le r$, we can prove that $\cL \entails_{\tau} Q_{T_0 + \gamma}$.
\begin{enumerate}
    \item $\cL 1 = \cL_0 1 = 1$.
    \item For every regular constraint $q$, we have that 
    \begin{equation*}
    \cL q = \cL_0 q + \iprod{E, \cR\paren{q}} \ge -\tau \cdot T_0 - r \normi{q}.
    \end{equation*}
    \item For all $h$ such that $2 \deg h \le d$ and $\normt{h} \le 1$ we have that 
    \begin{equation*}
        \cL h^2 = \cL_0 h^2 + \iprod{E, \cR\paren{h^2}} \ge - \tau \cdot T_0 - r \poly(n^d).
    \end{equation*}
    \item  
    For every PSD constraint $q$, excluding the $T$-constraint, and every polynomial $h$ such that $2\deg h \le d - \deg q$, and $\normt{h} \le 1$ we have that 
    \begin{align*}
    \cL q h^2 &= \cL_0 q h^2 + \iprod{E, \cR\paren{q h^2}} \\
    &\ge -\tau \cdot T_0 - r  \normi{\cR\paren{q}} \cdot \poly(n^{d}).
    \end{align*}
    \item Let $c = 1/2T_{\max}$. For the $T$-constraint $q_{T_0 + \gamma}$, and every polynomial $h$ such that $2 \deg h \le d - \deg q_{T_0 + \gamma}$, and $\normt{h} \le 1$, we have that 
    \begin{align*}
        \cL q_{T_0 + \gamma} h^2 &= \cL \paren{q_{T_0} + c \gamma }h^2  \\
        &= \cL_0 q_{T_0} h^2
    + c \gamma \cL h^2 + \iprod{E, \cR\paren{q_{T_0} h^2}} \\
    &\ge -\tau \cdot T_0 - c \gamma
    \Paren{\tau \cdot T_0 + r \poly\paren{n^d}} - r \cdot \normi{\cR\paren{q_{T_0}}} \cdot \poly\paren{n^d}
    \end{align*}
    \item Let $\cE$ be the corresponding linear functional for $E$. For every $k\times k$ matrix PSD constraint q, we have that 
    \begin{align*}
        \normt{\cE q} & 
        \le \sqrt{k} \normi{\cE q} \\
        &= \sqrt{k} \max_{i, j} \Abs{
\iprod{E, \cR\paren{q_{i, j}}}
        } \\
        &\le r \cdot \sqrt{k} \cdot \poly(n^d) \cdot \max_{i, j}  \normi{\cR\paren{q_{i, j}}}.
    \end{align*}
    Therefore 
    \begin{align*}
        \cL q \succcurlyeq \cL_0 q + \cE q \succcurlyeq -\tau \cdot T - r \cdot \sqrt{k} \cdot \poly\paren{n^d} \cdot \max_{i, j} \normi{\cR\paren{q_{i, j}}}.
    \end{align*}
    \item We have
    \begin{equation*}
    \normt{\cR\paren{cL}} \le \normt{\cR\paren{\cL_0}} + \normt{E} \le R + \tau \cdot T_0 + r.
    \end{equation*}
 \end{enumerate}
Therefore it suffices to take $r$ such that 

\begin{enumerate}
    \item $r \cdot \poly\paren{n^d} \le \tau \gamma$. In order to do this take $r \le \tau \gamma / \poly\paren{n^d}$.
    \item $c \gamma r \poly\paren{n^d} + r \cdot \normi{\cR\paren{q_{T_0}}} \cdot \poly\paren{n^d} \le \tau \gamma / 2$.
    In order to do this take $r$ to be
    \begin{equation*}
    r \le \frac{\tau}{4 \poly\paren{n^d}} \cdot \min\Paren{\frac{\gamma}{\Normi{\cR\paren{q_{T_0}}}}, \frac{1}{T_{\max}}}, 
    \end{equation*}
    \item $r \cdot \sqrt{k} \cdot \poly\paren{n^d} \cdot \max_{i, j} \normi{\cR\paren{q_{i, j}}} \le \tau \cdot \gamma$.
    In order for  this 
    to hold take $r$ to be 
    \begin{equation*}
    r \le \frac{\tau \gamma}{\sqrt{k} \poly\paren{n^d} \max_{i, j} \normi{\cR\paren{q_{i, j}}}}.
    \end{equation*}
\end{enumerate}

Therefore there exists a ball of radius 
$\poly\paren{1/\poly\paren{n^{d}}, \tau , \gamma, 1/k, 1 / \normi{R\paren{Q_{T_0}}}}$
such that for every $\cR\paren{\cL}_{\overline{\emptysetAlt}}$ in that ball we have that
$\cL \entails_{\tau} Q_{T_0+ \gamma}$, as desired.
\end{proof}

\begin{lemma}
\label{lem:approx-gamma-implies-entails}
Consider the family of polynomial constraints $\set{Q_T}$ of up to degree $d$ over $\R^n$ as in \Cref{ass:polynomials}.
Suppose there exists some linear functional $\cL$ such that $\cL \entails_{\tau, \gamma} Q_T$. Then if $\gamma \le T_{\max} / 2$, we have that $\cL \entails_{\tau} Q_{T + 4 \gamma}$.
\end{lemma}
\begin{proof}
All of the inequalities in $\cL \entails_{\tau} Q_{T + 4 \gamma}$ will be trivially satisfied because of $\cL \entails_{t+ \gamma} Q_T$ except for the $T$-constraint. So we should prove the inequality for the $T$-constraint.
Suppose $h$ is a polynomial such that $\normt{h} \le 1$, and $2\deg h \le d - \deg q_T$. Then
\begin{align*}
\cL q_{T + 4\gamma} h^2 
&= \cL q_T h^2 + \frac{4\gamma}{2T_{\max}} \cL h^2
\\
&\ge -\tau \cdot\paren{T + \gamma} - 2 \gamma \tau \cdot \paren{T + \gamma} / T_{\max}
\\
&\ge -\tau \cdot\paren{T + \gamma} - 3 \gamma \tau  \\
&\ge -\tau \cdot \paren{T + 4\gamma},
\end{align*}
as desired.
\end{proof}

\begin{theorem}[computability of score functions] \label{thm:computing-score}
    Consider the family of polynomial constraints $\set{Q_T}$ of up to degree $d$ over $\R^n$ as in \Cref{ass:polynomials}.
Let
\begin{equation*}
T_0 = \inf_T \text{ such that there exists $\cL$ such that 
$\cL \entails_{\tau} Q_T$.}
\end{equation*}
Then we can compute $T_0$ in time $\poly\paren{n^d, \Size\paren{Q_T}, \log\paren{R},
\log\paren{T_{\max}},
\log\paren{1/\gamma}, \log\paren{1/\tau}}$ up to error $\cO\paren{\gamma}$. Note that $R$ is as in \Cref{def:approximate-satisfiability}.
\end{theorem}
\begin{proof}
    We apply binary search in order to estimate $T_0$. Suppose $T$ is given, run the ellipsoid algorithm from \Cref{lem:efficient-functional-search}, either we can find some functional $\cL$ such that $\cL \entails_{\tau, \gamma} Q_{T}$, or a proof that no ball of radius $r\paren{\gamma}$ of functionals $\cL$ that satisfy $\cL \entails_\tau Q_T$ exists. Note that $r\paren{\gamma}$ here is as in \Cref{lem:robust-satisfiability}. 
    If we are in the first case, by \Cref{lem:approx-gamma-implies-entails} we know that $\cL \entails_{\tau} Q_{T+ 4\gamma}$. Therefore, $T + 4\gamma \ge T_0$, and we decrease the value of $T$.
    If we are in the second case, we must have $T < T_0 + \gamma$, since otherwise we know that by \Cref{lem:robust-satisfiability} there should exists a ball of radius $r\paren{\gamma}$.
    This gives us an efficient algorithm for approximating the score function.
\end{proof}

\section{High-Probability Bound for Stability of Covariance} \label{sec:high-prob-stability}

\subsection{Preliminaries}

\begin{lemma} \cite[Corollary 4.8, rephrased]{DiakonikolasKKLMS19} \label{lem:cov-first-moment-concentration}
    There exists $\alpha = O(\eta \log \frac{1}{\eta})$, such that for any $n \ge O\left(\frac{d^2 + \log(1/\beta)}{\alpha^2}\right)$ and $X_1, \dots, X_n \overset{i.i.d.}{\sim} \cN(\textbf{0}, I)$, then with probability at least $1-\tau$, for all symmetric matrices $P \in \BR^{d \times d}$ with Frobenius norm $1$ and all $b \in [0, 1]^n$ with $\E_i b_i \ge 1-\eta$,
\[\frac{1}{n} \sum_{i=1}^{n} b_i \langle x_i x_i^\top - I, P \rangle \le \alpha.\]
\end{lemma}

\begin{lemma}[Hanson-Wright Inequality] \label{lem:hanson-wright}
    Let $x \sim \mathcal{N}(\textbf{0}, I)$ be a $d$-dimensional Gaussian vector. Then, there exists a universal constant $c$ such that for any symmetric matrix $P$ and for all $t > 0$,
\[\BP\left[\left|\langle xx^\top - I,  P\rangle\right| > t\right] \le 2 \exp\left(-c \cdot \min\left(\frac{t^2}{\|P\|_F^2}, \frac{t}{\|P\|_{op}}\right)\right).\]
\end{lemma}

\begin{proposition}[Theorem 4.5, \cite{VershyninBook}] \label{prop:operator_norm_random_matrix}
    Let $A$ be a rectangular $m \times n$-dimensional matrix with each entry i.i.d. $\cN(0, 1)$. Then, there exists a constant $C_0$ such that for any $t \ge 0$, $\BP(\|A\|_{op} > C_0 (\sqrt{m}+\sqrt{n}+t)) \le 2e^{-\Omega(t^2)}$.
\end{proposition}

\begin{proposition} \label{prop:net_size}
    For some fixed $1 \le k \le d$, let $\cP$ be the set of symmetric $d \times d$ matrices with Frobenius norm at most $1$ and all nonzero eigenvalues at least $\sqrt{1/k}$ in absolute value. Then, for any $0 < \gamma < 1/2$, $\cP$ has a $\gamma$-net (in the Frobenius norm distance) of size $(1/\gamma)^{O(k \cdot d)}$.
\end{proposition}

\begin{proof}
    For such a $P \in \cP$, note that $P$ must have rank at most $k$. Therefore, we can write $P = U D U^\top$, where $D$ is a diagonal matrix of Frobenius norm at most $1$ and $U$ is a $d \times k$-dimensional matrix with orthonormal columns. Let $\cT$ be a $\gamma/10$-net of the $d$-dimensional unit sphere, of size $(1/\gamma)^{O(d)}$.
    Define $\cV \in \BR^{d \times k}$ to be the set of $d \times k$-matrices where each column is in $\cT$. Then, every orthogonal $U \in \BR^{d \times k}$ has a corresponding $V \in \cV$ such that each corresponding column in $U, V$ are unit vectors of distance at most $\gamma/10$. Therefore, there exists a set $\cW$ of orthogonal matrices in $\BR^{d \times k}$ such that every $U$ has a corresponding $W$ where $\|U-W\|_F \le k \cdot \gamma/5$. $\cW$ is created by choosing a single representative near each $V \in \cV$, should one exist, which means $|\cW| \le (1/\gamma)^{O(d \cdot k)}$.
    Finally, let $\cT'$ be a $\gamma/5$-net of the unit ball in $k$-dimensions, which corresponds to a $\gamma/5$-net $\cD$ of diagonal matrices of Frobenius norm at most $1$. 
    
    Now, we claim that the set of matrices $WD'W^\top$, for $W \in \cW$ and $D' \in \cD$, form a $\gamma$-net for the set of $P$. Indeed, for any $P = UDU^\top$, we associate $U$ with $W$ such that each column of $U$ and of $W$ differ by at most $\gamma/5$ in $\ell_2$-distance, and $D$ with $D'$ such that $\|D-D'\|_F \le \gamma/5$. We want to show that $\|UDU^\top - WD'W^\top\|_F \le \gamma$.
    
    Note we can bound $\|UDU^\top - WD'W^\top\|_F \le \|UD(U-W)^\top\|_F + \|(U-W) D W^\top\|_F + \|W(D'-D)W^\top\|_F$, so it suffices to bound each of these terms by $\gamma/5$. Since $U$ and $W$ are orthogonal matrices, $\|UM\|_F = \|W M\|_F = \|M\|_F$ and $\|M U^\top\|_F = \|M W^\top\|_F = \|M\|_F$ for any matrix $M$ (fitting the dimensions). Therefore, it suffices to show that $\|D (U-W)^\top\|_F, \|(U-W) D\|_F$, and $\|D'-D\|_F \le \gamma/5$. Indeed, we already know $\|D'-D\|_F \le \gamma/5$, and $\|D (U-W)^\top\|_F = \|(U-W) D\|_F$ since $D$ is diagonal, so we just need to show $\|(U-W) D\|_F \le \gamma/5$.
    To prove this, note that $U-W$ is a $d \times k$-dimensional matrix with very column having $\ell_2$ norm at most $\gamma/5$. When we multiply by $D$, this multiplies the $i$th column of $U-W$ by $D_{ii}$, the $i$th diagonal entry of $D$. Therefore, the $i$th column of $(U-W) D$ has $\ell_2$ norm at most $\gamma/5 \cdot D_{ii}$. Therefore, the Frobenius norm of $(U-W) D$ is at most $\sqrt{\sum_{i=1}^{k} (\gamma/5)^2 \cdot D_{ii}^2} = \gamma/5 \cdot \sqrt{\sum_{i=1}^{k} D_{ii}^2} = \gamma/5$.
    
    Finally, the size of this net is at most $|\cW| \cdot |\cD| \le (1/\gamma)^{O(d \cdot k)} \cdot (1/\gamma)^{O(k)} = (1/\gamma)^{O(d \cdot k)}.$
\end{proof}

\subsection{Main Probability Bound}

\begin{lemma} \label{lem:probability-bash-main-1}
    Let $n \ge \tcO\left(\frac{(d+\log (1/\delta))^2}{\eta^2}\right)$ and let $x_1, \dots, x_n \overset{i.i.d.}{\sim} \cN(\textbf{0}, I)$. Then, with probability at least $1-\delta,$ for any $d \times d$ symmetric matrix $P$ with Frobenius norm $\le 1$, and for any subset $S \subset [n]$ of size at most $\eta \cdot n$,
\[\sum_{i \in S} \langle x_i x_i^\top - I, P \rangle^2 \le O\left(\eta \log^2 \frac{1}{\eta}\right) \cdot n.\]
\end{lemma}

\begin{proof}
    For simplicity, we may assume without loss of generality that $\delta = e^{-d}$. This is because if $\delta > e^{-d}$, we can decrease the failure probability to $e^{-d}$. Likewise, if $\delta < e^{-d}$, then $d < \log(1/\delta)$, so we may increase the dimension to $\log (1/\delta)$ by sampling additional random standard Gaussians for the rest of the coordinates of $x_i$, and then only proving the result for all $P$ with all nonzero values supported on the first $d$ rows and columns.

    Let $\cR$ be a $1/2$-net of the set of symmetric matrices with Frobenius norm at most $1$, and suppose we successfully prove the lemma for all $P \in \cR$. Then, for a general $P$, we can write $P = \sum_{i=0}^\infty 2^{-i} R_i$, for each $R_i \in \cR$. Then, for any $i$,
\[\left\langle x_i x_i^\top - I, \sum_{i=0}^\infty 2^{-i} R_i \right\rangle^2 = \left(\sum_{i=0}^{\infty} 2^{-i} \cdot \langle x_i x_i^\top - I, R_i \rangle\right)^2 \le \left(\sum_{i=0}^{\infty} 2^{-i}\right) \cdot \left(\sum_{i=0}^{\infty} 2^{-i} \langle x_i x_i^\top - I, R_i\rangle^2\right),\]
    using the Cauchy-Schwarz inequality.
    Then, we can write
\[\sum_{i \in S} \langle x_i x_i^\top - I, P\rangle^2 \le 2 \cdot \left(\sum_{i=0}^{\infty} 2^{-i} \cdot \sum_{i \in S}\langle x_i x_i^\top - I, R_i \rangle^2\right) \le 4 \cdot \tcO(\eta) \cdot n.\]
    So it suffices to show the theorem for the net $\cR$.
    
    Next, note that for a sufficiently large constant $C_0$,
\begin{align*}
    \sum_{i \in S} \langle x_i x_i^\top - I, P \rangle^2 &= \sum_{i \in S} \int_{t = 0}^\infty \BI\left[t \le \langle x_i x_i^\top - I, P \rangle^2\right] dt \\
    &= \int_{t = 0}^{\infty} \#\left\{i \in S: \langle x_i x_i^\top - I, P \rangle^2 \ge t\right\} dt \\
    &\le (C_0 \log (1/\eta))^2 \cdot \eta n + \int_{t = (C_0 \log (1/\eta))^2}^\infty \#\left\{i: \langle x_i x_i^\top - I, P \rangle^2 \ge t\right\} \cdot (t \log^2 t) \cdot \frac{1}{t \log^2 t} dt \\
    &\le (C_0 \log (1/\eta))^2 \cdot \eta n + \max_{t \ge (C_0 \log (1/\eta))^2} \left(t \log^2 t \cdot \#\left\{i: \langle x_i x_i^\top - I, P \rangle^2 \ge t\right\}\right) \\
    &\le (C_0 \log (1/\eta))^2 \cdot \eta n + \max_{C \ge C_0 \log (1/\eta)} \left(C^2 \log^2 C \cdot \#\left\{i: \left|\langle x_i x_i^\top - I, P \rangle\right| \ge C\right\}\right).
\end{align*}
    The second-to-last line uses the fact that $\int_{3}^{\infty} \frac{1}{t \log^2 t} dt < 1$, and the last line is just a substitution $C = \sqrt{t}$.
    
    Therefore, it will suffice to show that for all $C \ge C_0 \log (1/\eta)$, with probability at least $1-e^{-d}/C$ the following holds. For all $P \in \cR$, the number of $i \in [n]$ such that $|\langle x_i x_i^\top - I, P \rangle| \ge C$ is at most $n \cdot \eta/(C^2 \log^2 C)$. The probability bound is sufficient since it suffices to prove this for all $C$ that is a power of $2$, and the sum of $e^{-d}/C$ over $C$ a power of $2$ is $e^{-d}$.
\end{proof}

So, to prove Lemma \ref{lem:probability-bash-main-1}, it suffices to prove the following lemma.

\begin{lemma} \label{lem:probability-bash-main-2}
    Suppose $\frac{n}{\log^{20} n} \ge O\left(\frac{d^2}{\eta^2}\right)$ and let $x_1, \dots, x_n \overset{i.i.d.}{\sim} \cN(\textbf{0}, I)$. Then, there exists a sufficiently large constant $C_0$ and a $1/2$-net $\cR$ for $d \times d$ symmetric matrices with Frobenius norm at most $1$, such that for any $C \ge C_0 \log (1/\eta)$, with probability at least $1-e^{-d}/C,$ for any $P \in \cR$, the number of indices $i \in [n]$ such that $\langle x_i x_i^\top - I, P \rangle \ge C$ is at most $n \cdot \eta/(C^2 \log^2 C)$.
\end{lemma}

\begin{proof}
    First, assume that $C \le \sqrt{\eta n/(\log^9 n \cdot d)}$.
    For $j \ge 1$, let $\cP_j$ be a $\gamma_j := (1/(10 j^2))$-net of the matrices in $\cP$ with all nonzero eigenvalues in the range $[-2/\sqrt{2^j}, -1/\sqrt{2^j}] \cup [1/\sqrt{2^j}, 2/\sqrt{2^j}]$. Also, let $\cQ_j$ be a $1/10$-net of the set of matrices in $\cP$ with all eigenvalues below $1/\sqrt{2^j}$ in absolute value.
    
    Now, for some fixed $P$, suppose we can write $P = P_1 + P_2 + \cdots + P_{\lceil \log_2 C^2 \rceil} + Q$, where each $P_j \in \cP_j$ and $Q \in \cQ_{\lceil \log_2 C^2 \rceil}$. Then, if the event that $\langle P, xx^\top - I \rangle \ge C$ holds, then we must have that either $\langle P_j, xx^\top - I \rangle \ge C/(4 j^2)$ for some $j$ or $\langle Q, xx^\top - I \rangle \ge C/2$. For any fixed choice of $\{P_j\}_{1 \le j \le \lceil \log_2 C^2 \rceil}$ and $Q$, the probability that this event occurs for each $P_j$ is at most $\exp\left(-c_1 \min\left(\frac{C^2}{j^4}, \frac{C \cdot 2^{j/2}}{j^2}\right)\right) \le \exp\left(-c_1 \cdot \frac{C \cdot 2^{j/2}}{j^4}\right)$, by the Hanson-Wright inequality and since $2^{j/2} \le 2C$ for $j \le \lceil \log_2 C^2 \rceil$. The probability that this event holds for $Q$, by Hanson-Wright, is at most $\exp\left(-c_1 \cdot \min\left(C^2, C \cdot 2^{\lceil \log_2 C^2 \rceil/2}\right)\right) \le \exp\left(-c_1 \cdot C^2\right).$
    
    For a fixed $P = P_1 + P_2 + \cdots + P_{\lceil \log_2 C^2 \rceil} + Q$, and for $x_1, \dots, x_n \overset{i.i.d.}{\sim} \cN(\textbf{0}, I)$, we bound the probability of the event that $\langle P, x_ix_i^\top - I \rangle \ge C$ for at least $\eta \cdot n/(C^2 \cdot \log^2 C)$ different choices of $i \in [n]$.
    For simplicity we define $\bar{C} = C^2 \cdot (\log C)^2/\eta$. 
    %By our assumption that $C \le \sqrt{\eta n/(\log^9 n \cdot d)}$, we have that $\bar{C} = o(n/(d \log^7 n))$.
    Now, if the event holds, then either some $\langle P_j, x_i x_i^\top - I \rangle \ge C/(4 j^2)$ for $n/(4 \bar{C} j^2)$ indices, or $\langle Q, x_i x_i^\top - I \rangle \ge C/2$ for $n/(2 \bar{C})$ different choices of $i \in [n]$. For fixed $j \le \lceil \log_2 C^2 \rceil$, the probability of this occurring for $P_j$ is at most
\begin{align*}
    {n \choose n/(4 \bar{C} j^2)} \cdot \exp\left(-c_1 \cdot \frac{C \cdot 2^{j/2}}{j^4} \cdot \frac{n}{4\bar{C} j^2}\right) &\le O(\bar{C} j^2)^{n/(4\bar{C} j^2)} \cdot \exp\left(-c_1 \cdot \frac{C \cdot 2^{j/2}}{j^4} \cdot \frac{n}{4\bar{C} j^2}\right) \\
    &\le \exp\left(-c_2 \cdot \frac{n \cdot 2^{j/2} \cdot C}{\bar{C} \cdot j^6}\right) \\
    &= \exp\left(-c_2 \cdot \frac{n \cdot 2^{j/2} \cdot \eta}{C (\log C)^2 \cdot j^6}\right).
\end{align*}
    where we used the fact that $\log (\bar{C} j^2) \le O(\log (C/\eta)),$ which is much smaller than $C \le O(C \cdot 2^{j/2}/j^4)$ since $C \ge C_0 \log (1/\eta)$.

    Likewise, the probability of this occurring for $Q$ is at most
\begin{align*}
    {n \choose n/(2\bar{C})} \cdot \exp\left(-c_1 \cdot C^2 \cdot \frac{n}{2\bar{C}}\right) &\le O(\bar{C})^{n/(2\bar{C})} \cdot \exp\left(-c_1 \cdot C^2 \cdot \frac{n}{2\bar{C}}\right) \\
    &\le \exp\left(-c_2 \cdot C^2 \cdot \frac{n}{\bar{C}}\right) \\
    &= \exp\left(-c_2 \cdot \frac{\eta \cdot n}{(\log C)^2}\right).
\end{align*}
    
    Finally, recall that $|\cP_j| \le O(j^2)^{2^j \cdot d} = e^{O(\log j \cdot 2^j \cdot d)}$ and $|\cQ_{\lceil \log_2 C^2 \rceil}| = e^{O(d^2)}$.
    So overall, the probability of there even existing such a $P$ that can be written as $P_1+\cdots+P_{\lceil \log_2 C^2 \rceil}+Q$ where each $P_j \in \cP_j$ and $Q \in \cQ_{\lceil \log_2 C^2 \rceil}$ is at most 
\begin{equation} \label{eq:chaining_bash_1}
    \sum_{j=1}^{\lceil \log_2 C^2 \rceil} \exp\left(-c_2 \cdot \frac{n \cdot 2^{j/2} \cdot \eta}{C (\log C)^2 \cdot j^6}\right) \cdot \exp\left(C_1 \cdot \log j \cdot 2^j \cdot d\right) + \exp\left(-c_2 \cdot \frac{n \cdot \eta}{(\log C)^2}\right) \cdot \exp\left(C_1 \cdot d^2\right).
\end{equation}
    Since $C \le \sqrt{\eta n/(\log^{9} n \cdot d)}$ and $2^j \le 2C^2$, this means $\frac{n \cdot 2^{j/2} \cdot \eta}{C (\log C)^2 \cdot j^6} \gg \log j \cdot 2^j \cdot d$. To see why, this is equivalent to $\eta \cdot \frac{n}{d} \gg C (\log C)^2 \cdot j^6 \log j \cdot 2^{j/2}$, and since $2^{j/2} \le 2 C$ and $C \ll n$, this is implied by $\eta \cdot \frac{n}{d} \gg C^2 (\log n)^{9}$.
    In addition, assuming that $n \gg (\log n)^2 \cdot d^2/\eta$, we also have that $\frac{n \cdot \eta}{(\log C)^2} \gg d^2$. Therefore, we can further bound \eqref{eq:chaining_bash_1} by
\[n \cdot \max_{j \le \lceil \log_2 C^2 \rceil} \exp\left(-c_3 \cdot \frac{n \cdot 2^{j/2} \cdot \eta}{C (\log C)^2 \cdot j^6}\right) + \exp\left(-c_3 \cdot \frac{n \cdot \eta}{(\log C)^2}\right) \le (n+1) \cdot \exp\left(-c_4 \frac{\eta \cdot n}{C (\log C)^2}\right) \le e^{-d}/C.\]
    
    So, our probability bound is sufficient, but we need to make sure that the set $\cR$ of matrices that can be written as $P_1 + \cdots + P_{\lceil \log_2 C^2 \rceil}+Q$ for $P_j \in \cP_j$ and $Q \in \cQ_{\lceil \log_2 C^2 \rceil}$ is a $1/2$-net. However, by looking at the singular value decomposition of any symmetric matrix $\tilde{P}$ with $\|P\|_F = 1$, we can write it as $\tilde{P}_1 + \cdots + \tilde{P}_{\lceil \log_2 C^2 \rceil} + \tilde{Q}$, where $\tilde{P}_j$ has all nonzero eigenvalues in $[-2/\sqrt{2^j}, -1/\sqrt{2^j}] \cup [1/\sqrt{2^j}, 2/\sqrt{2^j}]$ and $\tilde{Q}$ has all eigenvalues at most $1/\sqrt{2^{\lceil \log_2 C^2 \rceil}}$ in absolute value. In addition, each $\tilde{P}_j$ is within distance $1/(10 j^2)$ of some $P_j \in \cP_j$ and $\tilde{Q}$ is within distance $1/10$ of some $Q \in \cQ_{\lceil \log_2 C^2 \rceil}$. So, by the triangle inequality, $\cR$ is a $1/2$-net.

    Next, suppose $C \ge \sqrt{\eta n/(\log^{9} n \cdot d)}$, but $C \le \sqrt{\eta n/(\log^2 n)}$ so $\bar{C} \le n$. Then, for any fixed choice of $n/\bar{C} = \eta \cdot n/(C^2 (\log C)^2)$ indices $S$, the probability that there exists $P \in \BR^{d \times d}$ such that $\|P\|_F = 1$ and $\langle x_i x_i^\top - I, P \rangle \ge C$ for all $i \in S$ is at most
\begin{align*}
    \BP\left(\exists P: \sum_{i \in S} \langle x_i x_i^\top - I, P \rangle \ge C \cdot \frac{n}{\bar{C}}\right)
    &= \BP\left(\left\|\sum_{i \in S} (x_i x_i^\top - I)\right\|_F \ge C \cdot \frac{n}{\bar{C}}\right) \\
    &\le \BP\left(\left\|\sum_{i \in S} x_i x_i^\top \right\|_F \ge \frac{n\cdot C}{\bar{C}} - \frac{n \sqrt{d}}{\bar{C}}\right) \\
    &\le \BP\left(\left\|\sum_{i \in S} x_i x_i^\top \right\|_F \ge \frac{n\cdot C}{2 \bar{C}}\right).
\end{align*}
    The second line is true because the Frobenius norm of $I$ is $\sqrt{d}$ so the Frobenius norm of $|S| \cdot I$ is $\frac{n \sqrt{d}}{\bar{C}}$. The third line is true because $C \ge 2\sqrt{d}$ if $n \ge 4 d^2 \log^9 n/\eta$.

    So, for the final event to occur, it is equivalent for $\|AA^\top\|_F \ge \frac{n \cdot C}{2 \bar{C}}$, where $A$ is the $(n/\bar{C}) \times d$-dimensional matrix with each row of $A$ being $x_i$ for $i \in S$. Since $A$ and therefore $AA^\top$ have rank at most $n/\bar{C}$, this requires $\|A\|_{op}^2 = \|AA^\top\|_{op} \ge \frac{n \cdot C}{2 \bar{C}} \cdot \sqrt{\frac{\bar{C}}{n}} \ge \frac{\sqrt{n} \cdot C}{2 \sqrt{C}} = \frac{\sqrt{\eta \cdot n}}{2 \log C}$.
    
    Assuming $n \gg d^2 \log^2 n/\eta$, then $\frac{\sqrt{\eta \cdot n}}{2 \log C} \gg d$. Also, assuming $n \gg d^2 \log^{18} n/\eta$, then $\sqrt{\eta n} \gg d \log^9 n,$ which means $\sqrt{\eta n} \ll \frac{\eta n}{\log^9 n \cdot d} \le C^2 \cdot \log C$. This means that $\frac{n}{\bar{C}} = \frac{\eta \cdot n}{C^2 (\log C)^2} \ll \frac{\sqrt{\eta \cdot n}}{2 \log C}.$ So, this means $\frac{\sqrt{\eta \cdot n}}{2 \log C} \gg d + \frac{n}{\bar{C}}$, which means that by \Cref{prop:operator_norm_random_matrix}, the probability of $\|A\|_{op}^2 \ge \frac{\sqrt{\eta \cdot n}}{2 \log C}$ is at most $2 \exp\left(-\Omega\left(\frac{\sqrt{\eta \cdot n}}{\log C}\right)\right)$.
    
    Therefore, for any fixed choice of $n/C^2$ indices, the probability that there exists $P \in \BR^{d \times d}$ such that $\|P\|_F = 1$ and $\langle x_i x_i^\top - I, P \rangle \ge C$ for all $i \in S$ is at most $2 \exp\left(-\Omega\left(\frac{\sqrt{\eta \cdot n}}{\log C}\right)\right)$. There are at most ${n \choose n/\bar{C}} \le e^{\log n \cdot n/\bar{C}} \le e^{d \log^{10} n/\log^2 C}$. Note that $\frac{d \log^{10} n}{\log^2 C} \ll \frac{\sqrt{\eta \cdot n}}{\log C}$ for any $n \gg d \log^{20} n/\eta$. So, this means that the overall failure probability is at most $2 \exp\left(-\Omega\left(\frac{\sqrt{\eta \cdot n}}{\log C}\right)\right) \le e^{-d}/C$.
    %n^{n/\bar{C}} \le e^{n/(d \cdot \log^6 n)}$ possible choices of indices, so if $n \gtrsim d^2 \log^{26} d$, then by a union bound the probability that there even exist $n/C^2$ such indices is at most $2e^{-\Omega(n^{1/2})}$ also.
    
    The final case is if $C \ge \sqrt{\eta \cdot n/(\log^2 n)}$. In this case, the probability that even a single index has $\|x_i x_i^\top - I\|_F \ge C$ means $\|x_i x_i^\top\|_F \ge \frac{C}{2}$ which means $\|x_i\|_2^2 \ge \frac{C}{2}$. We can again apply Hanson-Wright to conclude that, since $C \gg d$, the probability that $\|x_i\|_2^2 \ge \frac{C}{2}$ is at most $2e^{-\Omega(-C)},$ which means the probability that this is true for even a single $x_i$ is at most $2ne^{-\Omega(C)} \le e^{-d}/C$.
\end{proof}

\subsection{Proof of \Cref{lem:resilience-of-moments-covariance}}

First, we note the following corollary of \Cref{lem:probability-bash-main-1}.

\begin{corollary} \label{cor:probability-3}
    With probability at least $1-\beta$, every $d \times d$ symmetric matrix $P$ with Frobenius norm exactly $1$, $\frac{1}{n} \sum_{i=1}^{n} \langle x_i x_i^\top - I, P \rangle^2 = 2 \pm O\left(\eta \cdot \log^2 \frac{1}{\eta}\right)$.
\end{corollary}

\begin{proof}
    Suppose $x_1, \dots, x_n$ has the property of Lemma \ref{lem:probability-bash-main-1}. Now, for a fixed $P$ with $\|P\|_F = 1$, note that $\sum_{i=1}^{n} \langle x_i x_i^\top - I, P \rangle^2 \le \sum_{i=1}^{n} \min\left(\langle x_i x_i^\top - I, P \rangle^2, C_0 \log^2 \frac{1}{\eta}\right) + O(n \cdot \eta \cdot \log^2 \frac{1}{\eta})$. This is because the number of indices $i$ such that $\langle x_i x_i^\top - I, P \rangle^2 \ge C_0 \log^2 \frac{1}{\eta}$ is at most $O(\eta \cdot n)$ by Lemma \ref{lem:probability-bash-main-1}, and for those indices, we know that $\sum \langle x_i x_i^\top - I, P \rangle^2$ is at most $O(n \cdot \eta \cdot \log^2 \frac{1}{\eta})$.
    
    Now, note that for any fixed $P$ with $\|P\|_F = 1$, $\E_{x \sim \cN(\textbf{0}, I)} \langle x x^\top - I, P \rangle^2 = 2$. Indeed, this is simple to see if $P$ is diagonal (using the fact that the fourth moment of a Gaussian is $3$), and for general symmetric $P$ we can diagonalize $P$ and use the same diagonalization on each $x_i$, to show this is true. Therefore, since $\BP(\langle x x^\top - I, P \rangle^2 \ge t^2) \le 2e^{-\Omega(t)}$ by Hanson-Wright, this means if $C_0$ is sufficiently large, $\E_{x \sim \cN(\textbf{0}, I)} \min\left(\langle x x^\top - I, P \rangle^2, C_0 \log^2 \frac{1}{\eta}\right) \in [2-\eta, 2]$. In addition, this variable is bounded between $0$ and $C_0 \log^2 \frac{1}{\eta}$, so by Hoeffding's inequality, the probability that $\frac{1}{n} \cdot \sum_{i=1}^{n} \min\left(\langle x x^\top - I, P \rangle^2, C_0 \log^2 \frac{1}{\eta}\right)$ is not in the range $[2-2 \eta, 2+2\eta]$ is at most $e^{-2 n \eta^2/\log^4 \frac{1}{\eta}}.$
    
    We can union bound over a $1/n^2$-net of symmetric matrices with Frobenius norm $1$, which has size $e^{O(d^2 \log d)}$, to say that if $n \gg \frac{d^2}{\eta^2} \cdot \log n \log^4 \frac{1}{\eta}$, then with probability at least $e^{- n \cdot \eta^2/\log^4 \frac{1}{\eta}}$, every $P$ in the net satisfies $\frac{1}{n} \cdot \sum_{i=1}^{n} \min\left(\langle x x^\top - I, P \rangle^2, C_0 \log^2 \frac{1}{\eta}\right)\in [2-2\eta, 2+2\eta]$.
    
    For a general $P,$ write $P = P_0+P',$ where $P_0$ is in the net and $\|P'\|_F \le 1/n^2$. Assuming the event of Lemma \ref{lem:probability-bash-main-1}, for every choice of $P'$ and every choice of $x_i$, $\langle x_i x_i^\top - I, P' \rangle \le \frac{1}{n}$. Therefore, the difference between $\min\left(\langle x x^\top - I, P \rangle^2, C_0 \log^2 \frac{1}{\eta}\right)$ and $\min\left(\langle x x^\top - I, P_0 \rangle^2, C_0 \log^2 \frac{1}{\eta}\right)$ is always at most $O\left(\frac{1}{n} \cdot \log^2 \frac{1}{\eta}\right) \le \eta$. So, for every symmetric matrix $P$ with Frobenius norm $1$, we have that $\frac{1}{n} \cdot \sum_{i=1}^{n} \min\left(\langle x x^\top - I, P \rangle^2, C_0 \log^2 \frac{1}{\eta}\right) \in [2-3 \eta, 2+3\eta]$ with probability at least $1-\beta$, as long as $n \ge \tcO\left(\frac{(d+\log (1/\beta))^2}{\eta^2}\right)$.
    
    Therefore, $\sum_{i=1}^{n} \langle x_i x_i^\top - I, P \rangle^2 = \left(2 \pm \tcO(\eta \cdot \log^2 \frac{1}{\eta})\right) \cdot n$, as desired.
\end{proof}

\begin{proof}[Proof of \Cref{lem:resilience-of-moments-covariance}]
    Part 1 and the first half of Part 3 are immediate from \Cref{lem:cov-first-moment-concentration}. The second half of Part 3 follows from \Cref{lem:probability-bash-main-1} and Part 2 follows from Corollary \ref{cor:probability-3}. Finally, Part 4 follows from \Cref{lem:cov-first-moment-concentration} in the same way that Part 4 of \Cref{cor:resilience} follows from \Cref{lem:resilience-of-moments}. For instance, we can set $\eta = 0.01$ to obtain that for any subset $S$ of size at most $0.01 n$, $\frac{1}{n} \cdot \sum_{i \in S} c_i \langle x_i x_i^\top - I, P \rangle \le O(1)$ for any choice of $c_i \in \{-1, 1\}$, which means $\frac{1}{n} \cdot \sum_{i \in S} |\langle x_i x_i^\top - I, P \rangle| \le O(1)$. We can then partition $[n]$ into $100$ such sets $S$.
\end{proof}

\section{Mean Estimation in $\ell_\infty$}
\label{sec:infty-mean}

In this section, we start by providing an algorithm for robust Gaussian mean estimation in $\ell_\infty$ distance (Proposition~\ref{prop:ell_infty_robust}).
We then show that this robust algorithm allows us to derive a pure DP algorithm with better sample complexity than a black-box application of Lemma~\ref{lem:intro-black-box} (Proposition~\ref{prp:better-linf}).

\begin{proposition}
\label{prop:ell_infty_robust}
    There is a robust estimator $\hat{\mu_0} \, : \, (\R^d)^n \rightarrow \R$ such that for every $\mu \in \R^d$ and small-enough $\eta > 0$, with high probability over $x_1,\ldots,x_n \sim \cN(\mu,I)$, letting $\overline{x} = \tfrac 1 n \sum_{i=1}^n x_i$, given any $\eta$-corruption $y_1,\ldots,y_n$ of $x_1,\ldots,x_n$, $\|\hat{\mu}(y_1,\ldots,y_n) - \overline{x}\|_2 \leq O(\sqrt{\eta d (\log n) / n} + \eta \sqrt{\log n})$ and $\|\hat{\mu} - \overline{x}\|_\infty \leq O(\eta \sqrt{ \log n})$, as long as $n \gg d$.
\end{proposition}

To prove the proposition, we establish a few facts about $x_1,\ldots,x_n \sim \cN(\mu,I)$.

\begin{fact}
\label{fact:ell_infty_concentration}
  The following all hold with high probability for $x_1,\ldots,x_n \sim \cN(\mu,I)$, with $\mu \in \R^d$, and letting $\overline{x} = \tfrac 1 n \sum_{i=1}^n x_i$, if $n \gg d$.
  \begin{enumerate}
       \item For a big-enough constant $C$ and all $t \in \{ C \sqrt{\log n}, 2C \sqrt{\log n}, 4C \sqrt{\log n},\ldots, d \}$
       \[
       \sup_{\|v\| = 1} \sum_{i \leq n} 1[\iprod{x_i-\overline{x},v} > t] \leq O(d (\log d + \log \log n) /t^2)\mcom
       \]
      \item $\|x_i - \overline{x}\| \leq O(\sqrt d + \sqrt{\log n})$
      \item every coordinate $i \in [d]$ and $j,k \leq n$ have $|(x_j)_i - (x_k)_i| \leq O(\sqrt{\log nd})$.
  %    \item $\|\overline{x} - \mu\|_\infty \leq O(\sqrt{(\log d) / n})$.
  %    \item $\|\tfrac 1 n \sum_{i \leq n} (x_i - \overline{x})(x_i - \overline{x})^\top \| \leq O(\max(n,d)/n)$.
  \end{enumerate}
\end{fact}
\begin{proof}
  First, the following simultaneously occur with high probability by standard Gaussian concentration arguments:
  \begin{itemize}
      \item each $x_i$ has $\|x_i\| \leq \sqrt{d} + O(\sqrt{\log n})$, and
      \item $\|\overline{x}\| \leq O(\sqrt{d/n})$.
%      \item $|\iprod{x_i, \overline{x}}| \leq \tfrac {\|x_i\|^2}{n} + \|x_i\| \cdot O(\sqrt{\log n / n}) \leq O \Paren{ \frac{d + \log n}{n} + \sqrt{\frac{d \log n}{n}} + \frac{\log n}{\sqrt{n}}}$
  \end{itemize}
  
  Now, let $S$ be a $\delta$-net of the $\ell_2$ unit sphere; we can take $S$ to have $2^{O(d \log (1/\delta))}$ elements.
  For any $x_1,\ldots,x_n$ and $t > 0$, let $n_t = \sup_{v \in S} \sum_{i \leq n} 1[\iprod{x_i,v} > t]$.
  We claim that
  \[
  \sum_{i \leq n} 1[\iprod{x_i - \overline{x},v} > t] \leq n_{t - \delta \cdot \max_i \|x_i\| - \|\overline{x}\|} \mper
  \]
  To see this, we can write $v = w + \Delta$, where $w \in S$ and $\|\Delta\| \leq \delta$.
  Then $\iprod{x_i - \overline{x},v} = \iprod{x_i,w} - \iprod{\overline{x},v} + \iprod{x_i, \Delta} > t$ only if $\iprod{x_i,w} > t + \iprod{\overline{x},v} - \iprod{x_i, \Delta} \geq t - \|\overline{x}\| - \delta \|x_i\|$.
  If $\max_i \|x_i\| \leq 1/(2\delta)$ and $n \gg d$, then we get $\sum_{i \leq n} 1[\iprod{x_i,v} > t] \leq n_{t-1}$.
  
  We just need to establish a high-probability upper bound on $n_{t-1}$ for $\delta \ll 1/(\sqrt{d} + O(\sqrt{\log n}))$ and $t \in \{C \sqrt{\log n}, 2 \sqrt{\log n},\ldots,d\}$.
  If $x_1,\ldots,x_n \sim \cN(\textbf{0},I)$, then for any fixed $v \in S$ and fixed $t$, we have
  \[
  \Pr \left ( \sum_{i \leq n} 1[\iprod{x_i,v} > t]  > s \right ) \leq n^s \exp(-\Omega(s t^2)) \mper
  \]
  via a union bound over $n^s$ choices of $s$ indices $i \in [n]$.
  If $t \geq C \sqrt{\log n}$ and $s = O(d \max(\log d,\log \log n) / t^2)$, we can take a union bound over the net $S$ and get that for any fixed $t$, $\sum_{i \leq n} 1[\iprod{x_i,v} > t] \leq  O(d \max (\log d,\log \log n) / t^2)$ with probability at least $1 - e^{-\Omega(d)}$; the proof for (1) is finished by a union bound over $O(\log d)$ choices of $t$.
  
  The proof for (2) is standard Gaussian concentration, and the proof for (3) is a union bound over $n^2d$ pairs $(x_j)_i, (x_k)_i$.
\end{proof}

\begin{proof}[Proof of Proposition~\ref{prop:ell_infty_robust}]
    Define the estimator $\hat{\mu}$ as: given $y_1,\ldots,y_n$, find any $x'_1,\ldots,x'_n$ which (a) agree with the $y_i$s on $(1-\eta)n$ vectors and (b) have both properties in Fact~\ref{fact:ell_infty_concentration}, and output $\hat{\mu} = \frac{1}{n} \sum_{i=1}^n x_i'$.
    If no such set $\{x_i'\}$ exists, output $\emptyset$.
    
    With high probability over $x_1,\ldots,x_n \sim \cN(\mu,I)$, by Fact~\ref{fact:ell_infty_concentration}, such a set of $x'$s exists, since the $x$s are such a set.
    
    Let's bound $\|\overline{x} - \overline{x'}\|_2$.
    Let $B \subseteq [n]$, $|B| \leq 2 \eta n$, be the indices where $x_i \neq x_i'$.
    For any unit $v \in \R^d$,
    \begin{align*}
      \iprod{\overline{x} - \overline{x'},v} = \frac 1 n \sum_{i \leq n} \iprod{x_i - x_i',v} & = \frac 1 n \sum_{i \in B} \iprod{x_i - \overline{x},v} - \frac 1 n \sum_{i \in B} \iprod{x_i' - \overline{x'},v} + \frac{|B|}{n} \iprod{\overline{x} - \overline{x'},v}\mper
    \end{align*}
    For each of the sums, we group terms in the average by their magnitudes.
    Terms smaller than $O(\sqrt{\log n})$ can only contribute $O(\eta \sqrt{\log n})$.
    At most $\eta n$ terms are smaller than $\sqrt{d (\log d + \log \log n) / (\eta n)}$; they contribute at most $\sqrt{\eta d (\log d + \log \log n) / n}$.
    So we have
    \[
    \frac 1 n \sum_{i \in B} \iprod{x_i - \overline{x},v} \leq O(\eta \sqrt{\log n}) + O\Paren{\sqrt{\frac{\eta d (\log d + \log \log n)}{n}}} + \sum_{i \, : \, |\iprod{x_i - \overline{x},v}| > \sqrt{d (\log d + \log \log n) / \eta n}} \iprod{x_i - \overline{x},v}\mper
    \]
    The remaining terms on the RHS we can group by their magnitudes; for each $t = C 2^j \sqrt{\log n}$ there are at most $O(d (\log d + \log \log n) / t^2)$ terms of magnitude $t$, so the total contribution to the average is also $O(\sqrt{\eta d (\log d + \log \log d)/n})$.
    The same argument applies symmetrically to $\tfrac 1 n \sum_{i \in B} \iprod{x_i' - \overline{x'},v}$; this proves our bound on $\|\overline{x} - \overline{x'}\|$.
    
    We turn to the bound on $\|\overline{x} - \overline{x'}\|_\infty$.
    Fix a coordinate $j \in [d]$.
    Then we have
    \[
    \overline{x}(j) - \overline{x'}(j) = \frac 1 n \sum_{i \in B} x_i(j) - x'_i(j) = \frac 1 n \sum_{i \in B} x_i(j) - \overline{x}(j) - (x'_i(j) - \overline{x'}) + \frac{|B|}{n} (\overline{x}(j) - \overline{x'}(j))
    \]
    Each term in the average on the RHS is at most $O(\sqrt{\log nd})$, so we obtain
    \[
    |\overline{x}(j) - \overline{x'}(j)| \leq O(\eta \sqrt{\log nd})\mper
    \]
\end{proof}

\begin{proposition}
\label{prp:better-linf}
    There is an $\e$-DP estimator which takes $n$ \iid\ samples $y_1,\ldots,y_n \sim \cN(\mu, I)$, assuming $\|\mu\| \leq R$, and with high probability produces $\hat{\mu}$ such that $\|\hat{\mu} - \mu\|_\infty \leq \alpha$, as long as $n \geq \tilde{O}( \tfrac{d \log R}{\e} + \tfrac{d^{2/3}}{\alpha \e^{2/3}} + \tfrac{\sqrt{d}}{\alpha \e} + \tfrac{\log d}{\alpha^2})$.
\end{proposition}
\begin{proof}
    Before we describe the $\e$-DP estimator, we establish a few geometry statements.
    Define $B$ to be the intersection between the $\ell_\infty$ ball of radius $\alpha$ and the $\ell_2$ ball of radius $c \frac{\alpha \sqrt{d}}{\sqrt{\log d}}$, for some small constant $c$.
    Let $W_d$ be the volume of the $d$-dimensional unit $\ell_2$ ball.
    We claim that
    \[
      \frac 12 \cdot W_d \cdot \left(c \frac{\alpha \sqrt{d}}{\sqrt{\log d}}\right)^d  \leq \vol(B) \leq W_d \cdot \left(c \frac{\alpha \sqrt{d}}{\sqrt{\log d}}\right)^d\mper
    \]
    The upper bound is simply because $B$ is contained in the $\ell_2$ ball of radius $c \alpha \sqrt{d} / \sqrt{\log d}$.
    For the lower bound, note that, having taken $c$ small enough, for a random $z$ in the $\ell_2$ ball of radius $c \alpha \sqrt{d} / \sqrt{\log d}$, we have $\Pr(\|z\|_\infty \leq \alpha) \geq 1/2$, and hence $\Pr(z \in B) \geq 1/2$, so $B$ contains at least half the volume of the $\ell_2$ ball of this radius.
    
    Now we describe the estimator $\hat{\mu}$.
    Let $\hat{\mu}_0$ be the robust estimator whose guarantees are described in Proposition~\ref{prop:ell_infty_robust}.
    Given a dataset $\cY$, we define 
    \[
    S(\tmu; \cY) =  \min_{\cY'} d(\cY,\cY') \text{ such that } \hat{\mu}_0(\cY') - \tmu \in B\mper
    \]
    In words, the score of $\tmu$ is the minimum distance from $\cY$ to a dataset $\cY'$ which causes the robust estimator $\hat{\mu}_0$ to output a point which is both $\ell_\infty$ and $\ell_2$-close to $\tmu$.
    The estimator $\hat{\mu}$ is given by outputting a random draw from the exponential mechanism with score function $S(\cdot ; \cY)$, over the $R$-radius $\ell_2$ ball.
    
    Privacy holds by construction, so we just have to analyze accuracy.
    We claim that any $\tilde{\mu}$ with $S(\tilde{\mu};\cY) \ll \alpha n / \sqrt{\log n}$ has $\|\tilde{\mu} - \overline{\mu}\|_\infty \leq \alpha/2$, where $\overline{\mu} = \tfrac 1 n \sum_{i \leq n} y_i$; indeed, this follows from the $\ell_\infty$ accuracy guarantee of $\hat{\mu}_0$.
    And, since $n \gg (\log d) / \alpha^2$, with high probability we have $\|\overline{\mu} - \mu\|_\infty \leq \alpha/2$.
    So, we just need to show that the estimator outputs $\tilde{\mu}$ with score $\ll \alpha / \sqrt{\log n}$ with high probability.
    
    First of all, there's a set $\tilde{\mu}$s of volume at least $\vol(B)$ with score $0$ -- the set $B$, centered at $\hat{\mu}_0(\cY)$.
    
    Now consider the set of $\tmu$ with score $\eta n$ for $\eta \geq \alpha / \sqrt{\log n}$.
    By the robustness guarantee of $\hat{\mu}_0$ and the definition of $B$, any $\tmu$ with score $\eta n$ has $\|\tmu - \overline{\mu}\|_2 \leq O(\max(\sqrt{\eta d \log n/n}, \eta \sqrt{\log n}) + c \alpha \sqrt{d}/\sqrt{\log d})$, so is contained in a ball around $\overline{\mu}$ of volume at most
    \[
      O \Paren{ \frac{\sqrt{\frac{\eta d \log n}{n}} + \eta \sqrt{\log n} + \frac{c \alpha \sqrt{d}}{\sqrt{\log d}}}{ \frac{ c \alpha \sqrt{d}}{\sqrt{\log d}}}}^d \cdot \vol(B) \leq \exp \Paren{ O \Paren{ \frac{d \sqrt{\eta} \log n}{ \alpha \sqrt{n}} + \frac{\sqrt{d} \eta \log n}{\alpha}}} \cdot \vol(B)\mper
    \]
    Following the same argument as in Lemma~\ref{lem:intro-black-box}, the mechanism outputs $\tilde{\mu}$ with $S(\tmu; \cY) \ll \alpha n / \sqrt{\log n}$ with high probability so long as for every $1/2 > \eta \geq \Omega( \alpha / \sqrt{\log n})$,
    \[
      \frac{O \Paren{ \frac{d \sqrt{\eta} \log n}{ \alpha \sqrt{n}} + \frac{\sqrt{d} \eta \log n}{\alpha}} + \log (\eta n)}{\eta \e} \ll n\mper
    \]
    This occurs so long as $n \gg \tilde{O}(\tfrac{d^{2/3}}{ \alpha \e^{2/3}} + \tfrac{\sqrt{d}}{\alpha \e})$.
\end{proof}

\end{document}